\documentclass[10pt,twocolumn,superscriptaddress,aps,amsmath,amssymb,citeautoscript,floatfix,prx]{revtex4-2}
\pdfoutput=1
\usepackage[english]{babel}
 
\usepackage{letltxmacro}
\usepackage{latexsym}
\usepackage{booktabs}
\usepackage{siunitx}
\usepackage{multirow}
\usepackage{algpseudocode}
\usepackage{array}
\usepackage{tabularx}

\LetLtxMacro{\ORIGselectlanguage}{\selectlanguage}
\makeatletter
\DeclareRobustCommand{\selectlanguage}[1]{%
  \@ifundefined{alias@\string#1}
    {\ORIGselectlanguage{#1}}
    {\begingroup\edef\x{\endgroup
       \noexpand\ORIGselectlanguage{\@nameuse{alias@#1}}}\x}%
}
\newcommand{\definelanguagealias}[2]{%
  \@namedef{alias@#1}{#2}%
}
\makeatother
 
\makeatletter
\def\maketitle{
\@author@finish
\title@column\titleblock@produce
\suppressfloats[t]}
\makeatother
 
\definelanguagealias{en}{english}
\definelanguagealias{English}{english}
\usepackage{graphicx}
\usepackage{amsmath}
\usepackage{amsfonts}
\usepackage{amssymb}
\usepackage{amsthm}
\usepackage{cancel}
\usepackage{mathtools}
\usepackage{bm}

\usepackage[percent]{overpic}
\usepackage[dvipsnames]{xcolor}
\usepackage{soul}
\usepackage{wasysym}
\usepackage{dsfont}
\usepackage{physics}
\AtBeginDocument{\RenewCommandCopy\qty\SI} 
\usepackage{hyperref}
\hypersetup{colorlinks,allcolors=black}
\usepackage{comment}
\usepackage{enumitem}
\usepackage{textgreek}
\usepackage{tcolorbox}
\usepackage{xcite}
 
\newcounter{alg}
 
\hypersetup{
  colorlinks   = true,
  urlcolor     = blue,
  linkcolor    = blue,
  citecolor    = red
}
\usepackage[mathscr]{euscript}
 
\usepackage{varwidth}
\usepackage[lined,boxed,ruled,norelsize,linesnumbered]{algorithm2e}
\usepackage{verbatim}
\usepackage[normalem]{ulem}
\usepackage{cleveref}

\newtheorem{theorem}{Theorem}

\newtheorem*{lemma*}{Lemma}
\newtheorem{corollary}{Corollary}
\newtheorem*{theorem*}{Theorem}

\theoremstyle{plain}

\theoremstyle{plain}

\theoremstyle{plain}
\newtheorem*{lem*}{\protect\lemmaname}
\theoremstyle{plain}
\newtheorem*{thm*}{\protect\theoremname}
\theoremstyle{plain}

\theoremstyle{plain}

\renewcommand{\thealg}{\arabic{alg}}
\newtcolorbox[use counter=alg,
              crefname={algorithm}{algorithms},
              Crefname={Algorithm}{Algorithms}]
{alg}[2][]{%
  floatplacement=#1,
  float,
  colback=cyan!5!white,
  colframe=cyan!50!black,
  colbacktitle=cyan!85!black,
  fonttitle=\bfseries,
  title=Algorithm~\thealg: #2
}

\definecolor{mycrimson}{HTML}{A60808}

\usepackage{bbm}
\usepackage{tikz-cd}

\newcommand{\C}{\mathbb{C}}
\newcommand{\R}{\mathbb{R}}

\newcommand{\End}{\mathrm{End}}
\newcommand{\Hom}{\mathrm{Hom}}
\newcommand{\Lie}{\mathrm{Lie}}

\newcommand{\id}{\mathrm{id}}
\newcommand{\gl}{\mathfrak{gl}}
\newcommand{\su}{\mathfrak{su}}
\newcommand{\uu}{\mathfrak{u}}
\newcommand{\slc}{\mathfrak{sl}}
\newcommand{\HH}{\mathcal{H}}
\newcommand{\I}{\mathbbm{1}}

\theoremstyle{plain}
\newtheorem{smtheorem}{Theorem}[section]
\newtheorem{smproposition}[smtheorem]{Proposition}
\newtheorem{smlemma}[smtheorem]{Lemma}

\theoremstyle{definition}
\newtheorem{smdefinition}[smtheorem]{Definition}
\newtheorem{smremark}[smtheorem]{Remark}
\newtheorem{smexample}[smtheorem]{Example}
\theoremstyle{plain}

\begin{document}
 
\title{
Covariant Approximate Quantum Codes for Protected Analog Computation
}
 
\author{Mariia Elovenkova}
\altaffiliation{\href{mailto:melovenkova@g.harvard.edu}{melovenkova@g.harvard.edu}}
\affiliation{Department of Physics, Harvard University, Cambridge, MA 02138, USA}

\author{Hong-Ye Hu}
\altaffiliation{\href{mailto:hongyehu.physics@gmail.com}{hongyehu.physics@gmail.com}}
\affiliation{Department of Physics, Harvard University, Cambridge, MA 02138, USA}

\author{Susanne F. Yelin}
\altaffiliation{\href{mailto:syelin@g.harvard.edu}{syelin@g.harvard.edu}}
\affiliation{Department of Physics, Harvard University, Cambridge, MA 02138, USA}
 
\begin{abstract}

Quantum error correction compatible with continuous symmetries is a fundamental problem in quantum information and a possible route to robust analog quantum simulation. 
Because the Eastin-Knill theorem forbids exact codes with continuous transversal symmetries, we construct explicit \(SU(d)\)-covariant approximate codes that exploit permutation symmetry to 
spread logical information uniformly across all physical subsystems. For one-, two-, and three-qudit erasures at known locations, we prove worst-case purified-distance scaling \(O(1/N)\), 
which matches known approximate Eastin-Knill lower bounds up to constants, and we extend the reduced-state analysis to general flagged local noise. Exact permutation symmetry of this kind is available only up to three erased sites; 
for an arbitrary fixed number $k$ of erased sites, we instead show that Haar-random and efficiently sampled encodings achieve worst-case distance $O(\sqrt{k}/N)$ with probability exponentially close to one. 
For single-qudit erasure, we construct an explicit near-optimal decoder from the Petz recovery map. We then use these codes as building blocks for encoded analog dynamics. 
Symmetry-preserving Hamiltonians generate block-structured dynamical Lie algebras implementable transversally, while controlled symmetry-breaking terms serve as non-transversal resources for universal dynamics. 
These results provide explicit non-Abelian covariant codes and a framework for robust analog quantum simulation.

\end{abstract}
\maketitle 
\addtocontents{toc}{\protect\setcounter{tocdepth}{-10}}

\begin{figure*}[htbp]
    \centering
    \includegraphics[width=1\linewidth]{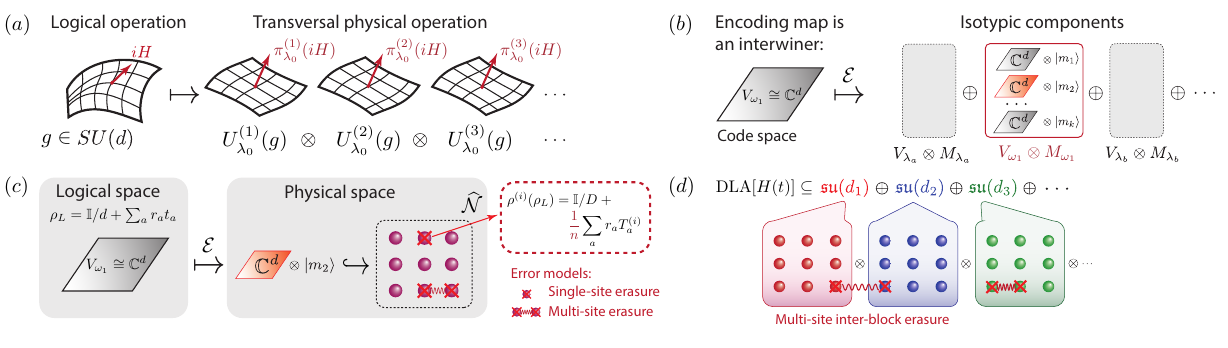}
\caption{
Overview of the covariant encoding construction and its applications.
(a) Covariance relates the logical $SU(d)$ action to a physical transversal action. The logical generator $iH$ is represented physically as a sum of local generators acting on the $n$ physical subsystems, each carrying the same single-site representation $\lambda_0$.
(b) The encoding isometry embeds the logical space $V_{\omega_1}\cong\mathbb C^d$ into the physical Hilbert space. After decomposing the physical space into isotypic components $V_\lambda\otimes M_\lambda$, covariance forces the code to lie inside the component $V_{\omega_1}\otimes M_{\omega_1}$. Thus, choosing the covariant code reduces to choosing a multiplicity vector $\ket m\in M_{\omega_1}$, which selects one embedded copy of the logical representation.
(c) Error protection is analyzed through the complementary channel, which describes the information leaked to the environment. For a single-site erasure, the environment sees the reduced state $\rho^{(i)}(\rho_L)$ of the erased subsystem. In our codes, the logical-state-dependent part of this reduced state is suppressed as the number of physical subsystems grows, so the environment output becomes close to a constant channel.
(d) The same codes can be used block-wise for symmetry-constrained analog simulation. The invariant part of the dynamical Lie algebra acts transversally within each encoded block, while multi-site erasures, including erasures involving different blocks, are treated using the same reduced-state and complementary-channel method.
}
\label{fig:theme}

\end{figure*}

\let\oldaddcontentsline\addcontentsline
\renewcommand{\addcontentsline}[3]{}

\section{Introduction}\label{sec:main_intro}

Digital quantum computing has made substantial progress. Programmable processors have reached regimes that are hard to simulate classically, 
with growing evidence for useful pre-fault-tolerant computations in specific settings~\cite{Arute2019,Kim2023}. 
These experiments do not yet use full quantum error correction, but digital quantum computing has a clear fault-tolerance paradigm: 
encode logical qubits into many physical qubits, extract error information and correct errors repeatedly, and implement protected logical gates. 
Although the overheads remain large~\cite{Fowler2019, Gidney2021}, the conceptual route to long fault-tolerant digital algorithms is in place. 
Analog platforms offer a complementary route to quantum simulation and computation~\cite{Bloch2012,Browaeys2020,Daley2022,Schaefer2020}, with recent experiments preparing strongly correlated states, 
including Fermi-Hubbard systems, in regimes challenging for classical methods~\cite{Xu2025,Mazurenko2017,Shao2024,Chalopin2024}.
 More generally, analog devices implement continuously generated Hamiltonian dynamics, and universal analog computation can be formulated by asking whether the available Hamiltonians generate a 
 sufficiently large dynamical Lie algebra (DLA)~\cite{Tavis1968,DAlessandro2021,Chen2020,Albertini2018,Albertini2020,Wang2016,Wang2012,Zanardi2004,AlbertiniDAlessandro2025,Hu2025}. 
 The promise of analog computation is therefore clear, but the corresponding framework for protecting native continuous-time dynamics is much less developed than in the digital circuit model.

Noise control is central in both settings. Digital fault tolerance is supported by a mature theory and by rapid experimental progress in logical memories and small logical processors, including break-even demonstrations where increasing the number of physical qubits reduces 
logical error rates~\cite{Krinner2022,AcharyaGoogle2025,RyanAnderson2021,daSilvaPaetznick2024,Reichardt2024Tesseract,Bluvstein2024,Reichardt2024NeutralAtom}. 
For analog simulation, one can in principle digitize the target evolution and run it in a protected circuit, but this replaces continuous Hamiltonian evolution by long gate sequences and adds Trotterization
 or phase-estimation costs to the fault-tolerance overhead~\cite{Reiher2017ReactionMechanisms,Kivlichan2020Trotter}. 
 This motivates noise-control strategies that are native to analog dynamics. Recent work has identified noise-stable regimes with quantum advantage~\cite{Trivedi2024AnalogStability}, 
 provided accuracy and hardness guarantees for noisy open-system simulations~\cite{Kashyap2025AnalogOpenSimulation}, and proposed direct error-reduction methods based on Hamiltonian 
 reshaping or rescaling~\cite{Guo2025HamiltonianReshaping}, penalty-Hamiltonian encodings~\cite{Cao2024RobustAnalog}, and continuous syndrome monitoring with feedback~\cite{Atalaya2021ContinuousQEC}. 
 These ideas show that analog noise control need not simply reproduce digital fault tolerance. 
 What is still missing is a general encoding principle that suppresses local noise while preserving the continuous encoded dynamics that make analog platforms natural.

A fundamental obstruction is the Eastin--Knill theorem~\cite{Eastin2009}. In an encoded analog device, the natural fault-tolerant way to implement a continuous logical operation is to realize its generator transversally, 
as a sum of local physical generators across the code block. 
Equivalently, the desired logical dynamics can be viewed as a continuous group of logical transformations, and a transversal implementation realizes the same group on the physical subsystems by a product action. 
In this sense, the transversal implementation defines a continuous symmetry of the encoded system: applying the group action before encoding should be equivalent to encoding 
first and then applying the corresponding physical action. This is the \textit{covariance} condition. 
However, exact finite-dimensional quantum error-correcting codes that correct local errors cannot support nontrivial continuous transversal symmetries, 
as stated in the Eastin--Knill theorem~\cite{Eastin2009}. Thus continuous analog dynamics and exact transversal quantum error correction are in tension. 
One can relax this obstruction by using infinite-dimensional covariant codes~\cite{Hayden2021}, by keeping exact finite-dimensional codes while restricting the correctable 
errors~\cite{Gupta2024}, or by allowing a controlled recovery error, leading to approximate quantum error correction (AQEC)~\cite{schumacher2002approximate, Beny2010}. 
We follow the third route: finite-dimensional covariant codes with approximate correctability.

Covariant codes can support continuous groups of transversal logical operations at the cost of controlled approximation error~\cite{Faist2020}. 
Quantitative trade-offs between covariance and correctability are now known, including information-theoretic lower bounds for covariant codes under erasure noise~\cite{Faist2020,Zhou2021,Faist2023,Kong2022,Yang2022}. 
Near-optimal scaling has been obtained from symmetric random constructions, and explicit examples have been developed for $U(1)$ covariance, W-state-type codes, and related questions in encoding complexity~\cite{Lin2025,Faist2020,Alexander2025,Yi2024}. 
These works establish the basic principles and limits of covariant AQEC, but many constructions are randomized, use reference frames or ancilla systems~\cite{Kong2022,Yang2022}, or address specific code families. 
Covariant approximate quantum error-correcting codes (AQECCs) are also relevant beyond computation. They model the protection of quantum information carrying a physical symmetry charge; 
they arise naturally in holographic quantum error correction and in discussions of approximate global symmetries in quantum gravity~\cite{Almheiri2015,Pastawski2015,HarlowOoguri2019}; 
and they are closely connected to quantum reference frames and metrology, where the same constraints can be expressed through reference-frame accuracy or quantum Fisher information~\cite{Hayden2021,Woods2020,Kubica2021}. 
Explicit non-Abelian covariant AQECCs therefore provide concrete models for studying how symmetry, locality, and approximate correctability can coexist.

Covariant AQEC is therefore a natural candidate for protecting analog dynamics. Rather than digitizing the evolution into a long protected circuit, one chooses an encoded subspace 
that carries the required continuous evolution and suppresses the information that local noise leaks to the environment. 
Approximate correctability is what makes this possible in finite dimensions: it evades the Eastin--Knill obstruction while keeping transversal 
implementations of the continuous symmetry generators. This matters beyond analog computation, because the same trade-off between symmetry and correctability constrains any fault-tolerant scheme that has to respect a continuous symmetry. 
Explicit and flexible constructions, however, remain scarce. Many existing approaches assume Abelian or otherwise restricted symmetries, rely on randomization, or treat only single-site noise, whereas physical settings usually involve non-Abelian 
symmetries and errors acting on several subsystems or with a more general local structure. This leads to the question we address here: can one construct finite-dimensional covariant AQECCs for non-Abelian symmetries such as $SU(d)$, with controllable performance under general multi-qudit noise models?

We answer this question positively. As summarized in Section~\ref{sec:main_summary}, we construct explicit finite-dimensional $SU(d)$-covariant AQECCs in which the logical system transforms as the fundamental representation and the physical system consists of many local $SU(d)$ degrees of freedom. 
The main idea is to use the freedom available inside the physical representation space to impose permutation symmetries on the code. These symmetries spread the logical information nearly uniformly across the physical subsystems, so that the reduced states seen by local noise depend only weakly on the logical input. 
We first illustrate this mechanism for $SU(2)$ in Section~\ref{sec:main_su2}, and then develop the general $SU(d)$ construction in Section~\ref{sec:main_sud}. 
This gives $O(1/n)$ worst-case purified-distance scaling for one-, two-, and three-qudit flagged erasures. For single-qudit erasure this scaling is near-optimal, matching known covariant-code lower bounds up to constants~\cite{Faist2020,Zhou2021}. 
Making all $k$-site reduced states exactly equal requires a permutation symmetry that can carry any $k$ sites onto any other $k$ sites, and no such symmetry is available once $k$ exceeds three, so this route stops at three erased sites. In Section~\ref{subsec:main_multisite_erasure_noise} we reach an arbitrary fixed number $k$ of erased sites by replacing exact symmetry with typicality. 
Covariance alone leaves a large family of admissible encodings, and we show that a Haar-random member of this family is, with probability exponentially close to one in $n$, simultaneously good for all $\binom nk$ erased sets, with worst-case distance $O(\sqrt{k(d^2-1)}/n)$.
Read as a statement about covariant codes as a whole, this says that the encodings leaking an atypically large amount of logical information to some erased subsystem form an exponentially small fraction of the family. 
Typicality by itself does not produce a code one can write down, because specifying a Haar-random encoding takes exponentially many parameters, and sampling one costs as many random bits. We therefore narrow the search to a finite subfamily of encodings, each generated from a short seed.
Sampling from this subfamily uses $O(n\log d)$ random bits, every parameter of the sampled encoding is computable in $\operatorname{poly}(n,\log d)$ time, and the resulting code meets the same $O(\sqrt{k(d^2-1)}/n)$ bound with probability exponentially close to one. The same reduced-state 
method also gives general $O_{d,k}(n^{-1/2})$ bounds for arbitrary flagged local noise. For single-qudit erasure, we additionally give an explicit near-optimal decoder based on the Petz recovery map~\cite{Petz1988,Barnum2002,Ng2010}. In Section~\ref{sec:main_covar_analog} we use these codes as building blocks for analog quantum simulation. 
A simulator does not need every unitary on its Hilbert space, since the evolutions it can reach form the dynamical Lie algebra generated by its drift and control Hamiltonians, which is typically much smaller than the algebra of all Hamiltonians on that system. We therefore build an encoding covariant only with respect to this smaller algebra, so that every Hamiltonian the simulator can reach is implemented transversally 
on the physical qudits. When the available Hamiltonians share a site-permutation symmetry, the algebra splits into blocks, and we encode each block separately, reusing the codes above together with known $U(1)$-covariant constructions~\cite{Lin2025}. The same $1/n$ protection persists for the two-qudit erasure models analyzed here, including erasures that involve qudits from different blocks.
Finally, in Section~\ref{sec:main_univ_analog}, we explain how controlled symmetry-breaking Hamiltonians can be used as non-transversal resource operations to enlarge the protected symmetry-preserving dynamics to universal analog control.

\section{Summary of Main Results}\label{sec:main_summary}
 
The following sections derive the main results using a representation-theoretical framework. 
This section summarizes them with minimal use of that formalism, for readers interested in the formulas and bounds rather than the derivations. The subsections are presented in the same order as the corresponding sections of the paper.
 
\subsection{Approximate error correction framework}\label{subsec:summary_framework}
 
A code is an isometry $\mathcal{V}$ embedding of a small logical Hilbert space $\mathcal H_L$ into a large physical Hilbert space $\mathcal H_P=A_1\otimes\cdots\otimes A_n$ of $n$ qudits, $\mathcal E(\rho_L):=\mathcal{V}\rho_L\mathcal{V}^\dagger$. 
We say the code $\varepsilon$-approximately protects against a noise channel $\mathcal N$ if some recovery channel $\mathcal R$ can undo the noise approximately,
$$
d(\mathcal R\circ\mathcal N\circ\mathcal E,\mathrm{id}_L)\le \varepsilon,
$$
where $d$ is a fidelity-based distance between channels, Eq.~\eqref{eq:distance_def}. For a family of codes indexed by $n$ (usually being the number of physical qudits), the question is how fast $\varepsilon$ shrinks as $n$ grows. 
The main noise model is erasure at known locations, i. e. with probability $p_S$, the qudits in a set $S$ are discarded and replaced by a blank state, with a flag recording which qudits were lost, Eq.~\eqref{eq:erasure_noise}; we also treat the more general case of an arbitrary, possibly unknown error acting on a known set of qudits, Eq.~\eqref{eq:arbitrary_multiqudit_noise}.
 
A convenient way to think about correctability, due to \citet{Beny2010}, is to ask what an environment having access only to the discarded qudits could learn about the logical state. If that environment's state is essentially the same no matter what was encoded, the logical information was never really leaked, and a good recovery map is guaranteed to exist, Eq.~\eqref{eq:beny}. 
The relevant object is therefore the reduced state of the encoded qudits seen by this observer,
$$
\rho^{(S)}(\rho_L):=\operatorname{Tr}_{\hat S}\!\left(\mathcal{V}\rho_L\mathcal{V}^\dagger\right),
$$
and most of the technical work in the paper is a calculation of $\rho^{(S)}(\rho_L)$ for small sets $S$ of one, two, or three qudits, followed by an estimate of how strongly it depends on $\rho_L$.
 
The approximate codes we build are also required to satisfy a precise covariance condition with respect to a continuous symmetry $SU(d)$ (or a product of such groups). 
Let $U_L(g)$ denote the action of $g\in SU(d)$ on the logical space, and let $U_P(g)$ denote the transversal physical action obtained by applying one fixed local rotation to every qudit at once, 
Eq.~\eqref{eq:group_phys_space_rep}. Covariance means that the encoding isometry $\mathcal{V}$ satisfies
$$
\mathcal{V}U_L(g)=U_P(g)\mathcal{V},\qquad \forall g\in SU(d),
$$
Eq.~\eqref{eq:group_covar_cond}. This means that rotating the logical state and then encoding it gives exactly the same physical state as encoding it first and then rotating every qudit transversally. 
Thus, an approximate quantum error-correcting code satisfying the covariance condition provides a code space on which the logical $SU(d)$ symmetry is implemented by continuous transversal physical operations.

Equation~\eqref{eq:group_covar_cond} is a strong constraint, because it forces $\mathcal{V}$ to map the logical space onto one specific copy of itself sitting inside the physical Hilbert space, among the several equivalent copies that occur there once $n$ is large enough. 
What this constraint does not fix is \emph{which} copy to choose when more than one is available, and this is exactly the freedom we exploit. We pick the copy so that it is also left invariant by permuting the physical qudits, which forces the logical information to be shared equally among them; an asymmetric choice, 
by contrast, is exactly what would let a handful of erased qudits give away a disproportionate share of the logical information.
 
\subsection{A first example: one logical qubit shared among many spins}\label{subsec:summary_su2}
 
The simplest version of the construction encodes a logical spin-$\tfrac12$ qubit into $n$ physical spin-$j$ systems ($n$ odd, $j$ half-integer), with logical $SU(2)$ rotations implemented as the total spin $J_a=\sum_k J_a^{(k)}, a=x, y, z$. 
Requiring the code to be invariant under cyclic permutations of the $n$ spins forces the logical information to be shared exactly equally among them. The reduced state of any single spin is then
$$
\rho^{(i)}(\rho_L)=\frac{I_{2j+1}}{2j+1}+\frac{3}{2nj(j+1)(2j+1)}\sum_a r_aJ_a^{(i)},
$$
for $\rho_L=\tfrac12(I+\mathbf r\cdot\boldsymbol\sigma)$. The first term, maximally mixed, carries no information about $\rho_L$; the second term, which does, is suppressed by $1/n$. 
Translating this into a recovery bound for flagged single-spin erasure gives
$$
d(\mathcal R\mathcal N\mathcal E,\mathrm{id})\le \frac{3}{4\sqrt{2j(j+1)}}\,\frac1n+O_{j}(n^{-2}).
$$
Known lower bounds for any $SU(2)$-covariant code under erasure scale the same way in $n$~\cite{Faist2020}, so this code is near-optimal in its $n$-scaling.
 
\subsection{The general \texorpdfstring{$SU(d)$}{SU(d)} code, a decoder, and general noise}\label{subsec:summary_sud}
 
The same mechanism works for a $d$-dimensional logical system encoded into $n$ copies of a fixed local representation of $SU(d)$ (each physical qudit carries an antisymmetric tensor power $V_{\omega_r}=\bigwedge^r\mathbb C^d$ of the fundamental representation, 
with generators acting as $T_a^{(i)}$; Eq.~\eqref{eq:su_d_gens}). The cyclically symmetric code again gives a one-qudit reduced state of "maximally mixed plus a $1/n$-suppressed, logical-dependent piece,"
$$
\rho^{(i)}(\rho_L)=\frac{I_{V_{\omega_r}}}{\dim V_{\omega_r}}+\frac{1}{2\kappa_{\omega_r}n}\sum_a r_aT_a^{(i)},\quad \kappa_{\omega_r}=\frac12\binom{d-2}{r-1},
$$
and a single-qudit erasure bound
$$
d(\mathcal R\mathcal N\mathcal E,\mathrm{id})\le \frac{d-1}{2\sqrt2}\sqrt{\frac{d+1}{r(d-r)}}\,\frac1n+O_{d}(n^{-2}),
$$
which again matches the known $SU(d)$ lower bound~\cite{Faist2020} in its $1/n$ scaling, up to a constant depending on $d$ and $r$.
 
We then provide an explicit recovery map, the Petz (transpose) map built directly from the encoding isometry $\mathcal{V}$, and show that it achieves the same $1/n$ scaling, though we do not pin down its optimal constant.
 
The same reduced state also controls protection against non-erasure noise. If a single qudit undergoes an arbitrary unknown local noise channel, then the environment does not see $\rho^{(i)}(\rho_L)$ directly, 
but only its image under the corresponding local complementary channel. Thus, all information about the logical state still passes through the one-site reduced state $\rho^{(i)}(\rho_L)$ above. Since the logical-state-dependent part of this reduced state is suppressed with $n$, the code remains approximately correctable, with the weaker scaling $O_{d}(1/\sqrt n)$. The loss from $1/n$ to $1/\sqrt n$ comes from using a fully general, but less tight, proof technique. The same argument covers noise on several qudits at once, so we treat both cases together in Section~\ref{subsec:main_arb_3_noise}.
 
\subsection{Protecting against multi-qudit noise}\label{subsec:summary_multierasure}
 
If two or three qudits can be erased simultaneously, at unknown but flagged locations, cyclic symmetry is no longer enough: every pair, respectively every triple, of qudits must look alike to the environment. 
This requires the code to be invariant under a subgroup of the permutation group that can move any pair (or triple) of sites to any other pair (or triple) -- a $2$- or $3$-transitive subgroup of all permutations of the $n$ sites. 
Such subgroups are rare; we use the affine group $\mathrm{AGL}(1,n)$, Eq.~\eqref{eq:AGL_def}, for two-qudit erasures and the projective linear group $\mathrm{PGL}(2,n-1)$, Eq.~\eqref{eq:PGL_def}, 
for three-qudit erasures, both of which admit the needed invariant codes once $n$ (respectively $n-1$) is a suitable prime power. 
(The full permutation group $S_n$ would also work in principle, but it is too restrictive, as we demonstrate in the corresponding section.)
 
With this extra symmetry, the three-qudit reduced state splits into a logical-independent piece and a logical-dependent piece of size $\Theta(1/n)$, $\rho^{(ijk)}(\rho_L)=\tau^{(ijk)}+\Delta^{(ijk)}(\rho_L)$, Eq.~\eqref{eq:3q_reduced_state}, 
and the resulting bound for flagged three-qudit erasure is
$$
d(\mathcal R\mathcal N\mathcal E,\mathrm{id})\le \frac{\sqrt{3(d^2-1)}}{2\sqrt2}\,\frac1n+O_{d}(n^{-2}).
$$
As in the single-qudit case, the same reduced state controls protection against arbitrary (non-erasure) noise on up to three known qudits, at the weaker rate $O_{d}(1/\sqrt n)$.

\subsection{Erasure on an arbitrary fixed number of qudits}\label{subsec:summary_general_k_erasure}

The construction above stops at three erased sites because the $k$-transitive subgroups of $S_n$ needed to force all $k$-site reduced states to coincide exactly are not available for $k>3$. 
Apart from the alternating group, which is so large that no covariant code is invariant under it, the classification of finite multiply transitive permutation groups leaves only the Mathieu groups, 
and these occur at four values of $n$ and so support no construction with $n\to\infty$. To go beyond three sites, we drop exact equality of the $k$-site reduced states and ask only that they be close to one another, for a well-chosen code.

Closeness has to be measured against something, and the natural reference is an average over a whole family of codes. 
We start from a simple seed encoding that places the logical qudit at one site and fills the remaining sites with independent $SU(d)$-singlets, 
Eq.~\eqref{eq:seed-isometry}. The covariance condition leaves a whole family of admissible encodings, Eq.~\eqref{eq:Vm}, and we average over all of them.

This average, taken with respect to the Haar measure on that family, defines the mean $k$-site channel. 
This channel belongs to no single encoding. It is the uniform mixture of all covariant encodings, and Proposition~\ref{prop:orbit-twirl} identifies that mixture with the seed encoding twirled over the full permutation group $S_n$. 
This is what averaging buys. No single member of the family is invariant under all of $S_n$, for the same reason that ruled out the alternating group above, but the mixture is. 
Its $k$-site reduced state is therefore the same for every set of $k$ sites, and it stays close to maximally mixed up to a logical-dependent correction of order $1/n$, 
just as in the one-, two-, and three-site cases above. The mean channel is thus the common point around which well-chosen codes should cluster.

A concentration argument (Theorem~\ref{thrm:main_haar_typical}) then shows that a single Haar-random draw from this family is
simultaneously close to the mean channel for every one of the $\binom nk$ possible erased sets, with probability approaching one exponentially fast in $n$. 
Drawing exactly from the family would require exponentially many random bits, since specifying one of its members takes exponentially many parameters; 
we instead give an explicit randomness-efficient sampling procedure (Theorem~\ref{thrm:main_efficient_sampling}) that reproduces the same guarantee from a seed of only $O(n\log d)$ random bits.

A Haar-random code already achieves the resulting erasure bound (Theorem~\ref{thrm:main_multisite_erasure_haar}), and combining it with the efficient sampling gives 
a family of $SU(d)$-covariant codes, correctable with high probability over the random seed against flagged erasure on any fixed number $k$ of sites, uniformly over the distribution of erased sets 
(Corollary~\ref{thrm:main_multisite_erasure}):
$$
d(\mathcal R_p\mathcal N\mathcal E,\mathrm{id})\le \frac{\sqrt{k(d^2-1)}}{2\sqrt2}\,\frac1n+O_{d,k}(n^{-2}).
$$
The $1/n$ protection against flagged erasure is therefore not limited to at most three sites, and it extends, with the same scaling in $n$ and an explicit $\sqrt k$ dependence on the number of erased sites, to erasure on any fixed number of physical sites. We make no optimality claim for $k>1$, 
since the covariant lower bounds we use~\cite{Faist2020,Zhou2021} are stated for single-qudit erasure, so neither the $1/n$ constant nor the $\sqrt k$ dependence is matched by a converse for $k$-site erasure.

This also changes how the construction should be read. Covariance by itself leaves the whole family of admissible encodings, and what the concentration argument says is that the members for which some erased subsystem reveals an atypically large amount of 
logical information form an exponentially small fraction of it. Poorly performing encodings do exist, the seed encoding itself among them, since erasing the one site that carries the logical qudit reveals the logical state completely. But they are non-generic, so the protection does not come from 
fine-tuning one special code, it becomes clear that approximate decoupling from any fixed-size erased subsystem is a typical property of covariant codes.
The efficient sampler then shows that such an encoding can also be built. Full Haar randomness is not what the error-correcting property needs: the concentration argument uses only second and fourth moments, so four-wise independent phases generated from a short seed do the same work.

\subsection{Using these codes to protect analog quantum dynamics}\label{subsec:summary_analog}
 
We also use this machinery to protect continuous-time (analog) quantum simulation under a symmetry. If the Hamiltonians driving an analog simulator are invariant under some finite group $G$ of qubit permutations, 
every reachable evolution is confined to a block-diagonal invariant algebra
$$
\mathfrak{su}(\mathcal H)^G\cong \mathfrak{su}(d_1)\oplus\cdots\oplus\mathfrak{su}(d_k)\oplus\mathfrak u(1)^{\oplus(r-1)},
$$
Eq.~\eqref{eq:inv_alg_decomp}, where the block sizes $d_i$ are fixed by $G$ and are typically far smaller than the full Hilbert-space dimension. 
We then build a \emph{block encoding}, Eq.~\eqref{eq:block_enc_def}: each non-Abelian block $\mathfrak{su}(d_i)$ is encoded separately, using one of the $SU(d_i)$-covariant codes above, 
so that every symmetry-respecting Hamiltonian acts transversally, block by block, on the encoded system.
 
This block code still tolerates noise that mixes different blocks. For a two-block code (block sizes $n_1,n_2$) under flagged two-qudit erasure -- whether both erased 
qudits lie in the same block or one in each -- the recovery error again scales as
$$
d(\mathcal R\mathcal N\mathcal E,\mathrm{id})\le\frac{\sqrt{C_1(d_1^2-1)+C_2(d_2^2-1)}}{n}+O_{d_1,d_2}(n^{-2})
$$
for $n_1\approx n_2\approx n$, Eq.~\eqref{eq:fidelity_2}, with constants $C_1,C_2$ independent of the block dimensions. 
Stacking covariant codes block by block therefore preserves the same $1/n$ protection, even for erasures affecting two blocks.
 
\subsection{Towards universal analog computation}\label{subsec:summary_universal}
 
Symmetric Hamiltonians alone can only ever generate the smaller, block-diagonal invariant algebra above, which in general doesn't 
coincide with the full algebra needed for universal control. We give a sufficient condition, Theorem~\ref{thm:breakers_suff_cond}, 
under which adding a modest set of \emph{symmetry-breaking} Hamiltonians restores universality: if the resulting "coupling graph" 
linking the different symmetry blocks is connected, and each link is "rich enough" in a precise sense 
(the Full First-Factor Span condition of Appendix~\ref{sec:supple_Universal_comp}), then the symmetric Hamiltonians together with the 
symmetry-breaking ones generate the full algebra $\mathfrak{su}(\mathcal H)$ on the whole Hilbert space. 
For instance, three qubits with $S_3$ permutation symmetry become fully controllable once a single extra term, $Z_1+X_2$, is added to the symmetric Hamiltonians.
 
This suggests a possible route toward robust universal analog computation, rather than a complete architecture by itself. 
The block encoding protects the encoded code space, within which symmetry-preserving Hamiltonians can be 
implemented transversally on the corresponding covariant code blocks. Universality then requires additional symmetry-breaking Hamiltonians, 
which are not transversal under the block encoding and therefore act as resource operations. 
The usefulness of this approach depends on choosing a symmetry that gives a favorable tradeoff: 
the invariant block structure should reduce the required physical resources, such as the number of physical qudits or their local dimensions, 
while the required symmetry-breaking resource Hamiltonians should remain sufficiently simple to implement.

\section{Preliminaries}\label{sec:main_prelim}

\subsection{Approximate Quantum Error Correction}\label{subsec:main_AQEC}

We first recall the approximate error-correction framework used throughout the paper. 
An exact quantum error-correcting code for a given noise model is an encoding for which there exists a recovery 
channel that perfectly restores the logical information after the noise has acted. This requirement is too restrictive 
for the setting considered here. In particular, the Eastin--Knill theorem and its variants imply that finite-dimensional codes 
correcting local errors exactly cannot support nontrivial continuous logical symmetries transversally~\cite{Eastin2009}.

Approximate quantum error correction relaxes this condition. For a fixed noise channel, 
an encoding is $\varepsilon$-approximately correctable if there exists a recovery channel that reverses the noisy encoded evolution 
up to error at most $\varepsilon$, calculated by a chosen distance between quantum channels~\cite{Beny2010}. 
For an asymptotic family of codes, the central requirement is that this error decreases as the number $n$ of physical subsystems grows.

The noise model has to be specified with some care. If the noise is artificially weak, approximate correctability can arise for trivial reasons.
 For example, if exactly one physical qubit is erased uniformly at random among $n$ qubits, then the trivial encoding
$$
\ket\psi
\longmapsto
\ket\psi\ket 0^{\otimes(n-1)}
$$
has vanishing error as $n\to\infty$, simply because the logical qubit is erased only with probability $1/n$~\cite{Yi2024}. 
This does not constitute a robust error-correction mechanism. Thus, throughout the paper, we consider asymptotic code families 
for which the error parameter $\varepsilon$ decreases with the system size, and we work with the most general noise models for which 
the performance can be analyzed explicitly.

For a fixed code and noise model, one could in principle optimize over all recovery channels. 
This optimization is generally difficult~\cite{PhysRevA.75.012338}. We instead use the complementary-channel 
characterization of approximate error correction. The complementary channel describes the information leaked to the environment. 
Approximate correctability is equivalent to the statement that this environment output is close to independent of the logical input, 
or equivalently that the complementary channel is close to a constant channel, i.e., a quantum channel that maps all states to some fixed 
state~\cite{Beny2010, Faist2020}.

\paragraph{Correctability and distance.}

We now define the channel distance used in the paper. For two quantum channels $\mathcal N$ and $\mathcal M$, their worst-case entanglement fidelity is
$$
\begin{aligned}
&F(\mathcal N,\mathcal M)
\\
&:=
\min_\rho
f\Big(
(\mathcal N\otimes\mathrm{id})
(\ket{\psi_\rho}\bra{\psi_\rho}),
(\mathcal M\otimes\mathrm{id})
(\ket{\psi_\rho}\bra{\psi_\rho})
\Big),
\end{aligned}
$$
where $\ket{\psi_\rho}$ is a purification of $\rho$, and
$$
f(\rho,\sigma)
:=
\operatorname{Tr}
\sqrt{\sqrt\rho\sigma\sqrt\rho}
$$
is the root fidelity between states. We use the corresponding fidelity-induced channel distance
\begin{equation}\label{eq:distance_def}
d(\mathcal N,\mathcal M)
:=
\sqrt{1-F(\mathcal N,\mathcal M)} .
\end{equation}

Let
$$
\mathcal E:
\mathcal D(\mathcal H_L)
\to
\mathcal D(\mathcal H_P)
$$
be an encoding channel, where $\mathcal H_L$ and $\mathcal H_P$ are the logical and physical Hilbert spaces. 
For a fixed noise channel $\mathcal N$, we say that $\mathcal E$ is $\varepsilon$-correctable under $\mathcal N$~\cite{Beny2010} 
if there exists a recovery channel $\mathcal R$ such that
$$
d(\mathcal R\circ\mathcal N\circ\mathcal E,\mathrm{id}_L)
\le
\varepsilon .
$$

\paragraph{Flagged local noise.}

Throughout the paper, the physical Hilbert space of $n$ qudits is decomposed as
$$
\mathcal H_P
=
A_1\otimes\cdots\otimes A_n,
$$
where $A_i$ is the Hilbert space of the $i$th qudit. For a subset $S\subseteq{1,\ldots,n}$, we write
$$
A_S
:=
\bigotimes_{i\in S}A_i,
$$
and $\hat S$ denotes the complementary subset. The hat denotes complementation throughout: $\widehat{ijk}$ for the complement of $\{i,j,k\}$, $\widehat j$ for the retained tensor factors other than $j$, 
and $\widehat{\mathcal N}$ for the complementary channel of $\mathcal N$. 
To simplify notation, when $S={i,j,k}$, we write $A_{ijk}$, $\operatorname{Tr}_{ijk}$, $p_{ijk}$, and $\ket{ijk}$ instead of $A_S$, $\operatorname{Tr}_{A_S}$, $p_S$, and $\ket S$. 
Similarly, for a single site $S={i}$, we often write $i$ instead of $A_i$.

The main noise model considered in the paper is subsystem erasure at known locations, or equivalently flagged subsystem erasure. In its general form, it is
\begin{equation}\label{eq:erasure_noise}
\mathcal N(\sigma)
=
\sum_S
p_S
\ket S\bra S_F
\otimes
\ket e\bra e_{A_S}
\otimes
\operatorname{Tr}_S(\sigma),
\end{equation}
where $F$ is a classical flag register recording which subsystem $A_S$ was erased~\cite{Faist2020}. 
We also consider an arbitrary flagged multi-qudit noise of the form
\begin{equation}\label{eq:arbitrary_multiqudit_noise}
\mathcal N(\sigma)
=
\sum_{S\in\mathcal S}
p_S
\ket S\bra S_F
\otimes
\left(
\mathcal N_S
\otimes
\mathrm{id}_{\hat S}
\right)(\sigma),
\end{equation}
where $\mathcal N_S$ is an arbitrary noise channel acting on the subsystem $A_S=\bigotimes_{i\in S}A_i$, extending beyond erasure errors.

\paragraph{Complementary-channel criterion.}

Let
$$
\Phi:\operatorname{End}(\mathcal H_A)\to\operatorname{End}(\mathcal H_B)
$$
be a quantum channel with Stinespring dilation
$$
\Phi(\sigma)
=
\operatorname{Tr}_E(W\sigma W^\dagger),
\qquad
W:\mathcal H_A\to\mathcal H_B\otimes\mathcal H_E .
$$
The complementary channel is obtained by tracing out the channel output:
$$
\widehat\Phi(\sigma)
:=
\operatorname{Tr}_B(W\sigma W^\dagger).
$$
Thus $\Phi$ describes the system output, while $\widehat\Phi$ describes the corresponding environment output.

For a state $\rho_0$ on the environment output space, the corresponding constant channel is
\begin{equation}\label{eq:constant_channel_def}
\Lambda_0(\rho)
:=
\operatorname{Tr}(\rho)\rho_0 .
\end{equation}
The following theorem relates approximate correctability to decoupling from the environment.

\begin{theorem}[\citet{Beny2010}]\label{thm:beny}
Let $\mathcal E$ be an encoding channel, let $\mathcal N$ be a noise channel, and let $\widehat{\mathcal N\circ\mathcal E}$ be a complementary channel to $\mathcal N\circ\mathcal E$. Then
\begin{equation}\label{eq:beny}
    \min_{\mathcal R}
    d\left(
        \mathcal R\circ\mathcal N\circ\mathcal E,
        \mathrm{id}_L
    \right)
    =
    \min_{\Lambda_0\,\mathrm{constant}}
    d\left(
        \widehat{\mathcal N\circ\mathcal E},
        \Lambda_0
    \right).
\end{equation}
\end{theorem}

The theorem states that an encoding is approximately correctable precisely when the environment output can be made close to a fixed state that is independent of the logical input. 
Therefore, throughout the paper rather than constructing a recovery map first, we control what the environment can learn.

\paragraph{Reduced states seen by the environment.}

For the erasure channel, defined in Eq.~\eqref{eq:erasure_noise}, suppose that the encoding is isometric:
\begin{equation}\label{eq:enc_isom}
\mathcal E(\rho_L)
=
\mathcal{V}\rho_L\mathcal{V}^\dagger .
\end{equation}
Hence, a complementary channel to the erasure channel, derived in~\cite{Faist2020}, is
\begin{equation}\label{eq:dual_to_erasure}
\widehat{\mathcal N\circ\mathcal E}(\rho_L)
=
\sum_S
p_S
\ket S\bra S_{F_E}
\otimes
\operatorname{Tr}_{\hat S}
\left(
\mathcal{V}\rho_L\mathcal{V}^\dagger
\right).
\end{equation}
Thus, the environment receives the erased subsystem, together with a flag recording which subsystem was erased.

For the arbitrary flagged multi-qudit noise model~\eqref{eq:arbitrary_multiqudit_noise}, let $\widehat{\mathcal N}_S$ be a complementary channel to $\mathcal N_S$. 
Then, as shown in Lemma~\ref{lem:dual_to_arb},
\begin{equation}\label{eq:dual_to_arbitrary_multiqudit}
\widehat{\mathcal N\circ\mathcal E}(\rho_L)
=
\sum_{S\in\mathcal S}
p_S
\ket S\bra S_{F_E}
\otimes
\widehat{\mathcal N}_S
\left(
\operatorname{Tr}_{\hat S}
\left(
\mathcal{V}\rho_L\mathcal{V}^\dagger
\right)
\right).
\end{equation}

Eqs.~\eqref{eq:dual_to_erasure} and~\eqref{eq:dual_to_arbitrary_multiqudit} show that local reduced states of the encoded state are the central objects in the complementary-channel analysis. 
We therefore use the notation
\begin{equation}\label{eq:reduced_state}
\rho^{(S)}(\rho_L)
:=
\operatorname{Tr}_{\hat S}
\left(
\mathcal{V}\rho_L\mathcal{V}^\dagger
\right),
\end{equation}
where $S\subseteq{1,\ldots,n}$. In particular,
$$
\rho^{(i)}(\rho_L),
\qquad
\rho^{(ij)}(\rho_L),
\qquad
\rho^{(ijk)}(\rho_L)
$$
denote the one-, two-, and three-qudit reduced states, respectively. The main calculations below are calculations of these 
reduced states and of how much they depend on the logical input.

\subsection{Covariant encoding}\label{subsec:main_covar_enc}

We now introduce the covariant encoding framework used throughout the paper. Approximate correctability allows us to keep 
continuous logical symmetries implemented locally, while relaxing exact correction to correction with a controlled error. 
The compatibility condition between an encoding and a symmetry is \textit{covariance}~\cite{Hayden2021}. 
In representation-theoretic language, covariance says that the encoding map is an \textit{intertwiner} between the logical and physical representations.

For the codes constructed below, the logical system carries the fundamental representation of $SU(d)$. 
Constructing a covariant code therefore amounts to embedding a copy of this representation into the physical Hilbert space. 
Representation theory identifies where such copies can occur, and the remaining freedom is a choice of vector in the corresponding multiplicity 
space~\cite{Denys2024}. We use this freedom to impose additional permutation symmetries on the code space. 
Specifically, we choose the multiplicity vector to be invariant under a suitable subgroup of the permutation group on $n$ elements $G_P\subseteq S_n$ 
acting by permutations of the physical subsystems. Since this permutation action commutes with the collective $SU(d)$ action, 
it acts only on multiplicity spaces. Thus, a $G_P$-invariant multiplicity vector gives a $G_P$-invariant code space. 
The choice of $G_P$ depends on the noise model and will be specified in the later sections.

\paragraph{Covariance.}

Let $G$ be a compact connected Lie group, and let $\mathfrak g$ be its Lie algebra. Let
$$
\mathcal{V}:\mathcal H_L\to\mathcal H_P
$$
be the encoding isometry associated with the encoding channel in Eq.~\eqref{eq:enc_isom}. 
Let $U_L$ and $U_P$ be unitary representations of $G$ on $\mathcal H_L$ and $\mathcal H_P$, respectively. 
We say that $\mathcal{V}$ is $G$-\textit{covariant with respect to $U_L$ and $U_P$} if
\begin{equation}\label{eq:group_covar_cond}
\mathcal{V}U_L(g)
=
U_P(g)\mathcal{V},
\qquad
\forall g\in G .
\end{equation}
Equivalently, let $\pi_L$ and $\pi_P$ be the corresponding Lie-algebra representations,
$$
\pi_L(X)
:=
\left.\frac{d}{dt}\right|_{t=0}
U_L(e^{tX}),
\qquad
\pi_P(X)
:=
\left.\frac{d}{dt}\right|_{t=0}
U_P(e^{tX}).
$$
Then $\mathcal{V}$ is \textit{$\mathfrak g$-covariant with respect to $\pi_L$ and $\pi_P$} if
\begin{equation}\label{eq:algebra_covar_cond}
\mathcal{V}\pi_L(X)
=
\pi_P(X)\mathcal{V},
\qquad
\forall X\in\mathfrak g .
\end{equation}
Eq.~\eqref{eq:algebra_covar_cond} follows from Eq.~\eqref{eq:group_covar_cond} by differentiation. 
Conversely, for a connected Lie group, Eq.~\eqref{eq:algebra_covar_cond} determines the corresponding covariance condition at the group level.

We will say that the encoding channel $\mathcal E$ is \textit{$G$-covariant}, or \textit{$\mathfrak g$-covariant}, 
when its encoding isometry satisfies Eq.~\eqref{eq:group_covar_cond}, or equivalently Eq.~\eqref{eq:algebra_covar_cond}. 
We will often specify covariance at the Lie-algebra level, because the physical representations used below are most naturally described by their infinitesimal generators. 
This is also the natural language in the analog-simulation setting, where the relevant algebras have the form
$$
\mathfrak{su}(d_1)
\oplus
\mathfrak{su}(d_2)
\oplus
\cdots .
$$
When the representations are clear from context, we simply say that the code is \textit{$G$-covariant}.

\paragraph{Intertwiners.}

The covariance condition strongly restricts the possible encoding maps. It is precisely the condition that $\mathcal{V}$ is an \textit{intertwiner}. 
Recall that if $(V,\rho_{V})$ and $(W,\rho_{W})$ are representations of a group $G$, then a linear map $A:V\to W$ is a \textit{$G$-intertwiner} if
\begin{equation}\label{eq:intertwiner_def}
A\circ\rho_{V}(g)
=
\rho_{W}(g)\circ A,
\qquad
\forall g\in G .
\end{equation}
Here $G$ may be finite or continuous. At the Lie-algebra level, for infinitesimal representations $d\rho_{V}$ and $d\rho_{W}$, this condition becomes
$$
A\circ d\rho_{V}(X)
=
d\rho_{W}(X)\circ A,
\qquad
\forall X\in\mathfrak g .
$$

The structure of intertwiners is controlled by Schur's lemma. If $V_\lambda$ and $V_\mu$ are irreducible complex representations of $G$, then
$$
\operatorname{Hom}_G(V_\lambda,V_\mu)
=
\begin{cases}
0, & \lambda\not\simeq\mu,\\
\mathbb C\cdot I_{V_\lambda}, & \lambda\simeq\mu .
\end{cases}
$$
Thus an intertwiner cannot map an irreducible representation to an inequivalent irreducible representation -- it can only map it to an equivalent copy.

More generally, suppose that two finite-dimensional $G$-representations decompose as
$$
W
\cong
\bigoplus_{\lambda\in\Lambda}
V_\lambda\otimes M_\lambda,
\qquad
V
\cong
\bigoplus_{\lambda\in\Lambda}
V_\lambda\otimes N_\lambda,
$$
where $\Lambda$ is the set of irreducible representations of $G$, and $M_\lambda$ and $N_\lambda$ are multiplicity spaces. Then every intertwiner $A:W\to V$ has the form
\begin{equation}\label{eq:intertwiner_decomp}
A
\cong
\bigoplus_{\lambda\in\Lambda}
I_{V_\lambda}\otimes A_\lambda,
\qquad
A_\lambda\in\operatorname{Hom}(M_\lambda,N_\lambda).
\end{equation}
Thus the intertwiner acts trivially on the irrep factor $V_\lambda$ and only acts nontrivially on the corresponding multiplicity spaces.

This observation will be used repeatedly. Besides the encoding isometry itself, the linear extensions of reduced-state maps and the compression map
$$
X
\longmapsto
\mathcal{V}^\dagger X \mathcal{V}
$$
are also intertwiners. This is why representation theory fixes much of the structure of the reduced states and compressed local operators appearing below. 
Throughout the paper, we use the terms covariant map and intertwiner interchangeably: the term covariant is standard in quantum information theory, 
especially for channels and encodings, while intertwiner is the standard representation-theoretic term.

\paragraph{The $SU(d)$ representation setting.}

The main Lie group considered in this work is $SU(d)$. Since $SU(d)$ is compact, connected, and simply connected, 
its irreducible representations are in one-to-one correspondence with irreducible representations of $\mathfrak{su}(d)$. 
We will therefore freely pass between $SU(d)$ and $\mathfrak{su}(d)$ representations. 
Background on these representation-theoretic facts can be found in Ref.~\cite{Hall2015}.

The physical system consists of $n$ identical subsystems, each carrying the same irreducible representation of $SU(d)$. 
Thus, the physical Hilbert space is
\begin{equation}\label{eq:phys_space_rep}
\mathcal H_P
:=
(V_{\lambda_0})^{\otimes n},
\end{equation}
where $V_{\lambda_0}$ is a fixed irreducible representation of $SU(d)$. The transversal group action is
\begin{equation}\label{eq:group_phys_space_rep}
U_P(g)
=
\bigotimes_{i=1}^n
U_{\lambda_0}^{(i)}(g),
\qquad
g\in SU(d),
\end{equation}
where $U_{\lambda_0}^{(i)}(g)$ acts as $U_{\lambda_0}(g)$ on the $i$th tensor factor. At the Lie-algebra level, the derived representation is
\begin{equation}\label{eq:alg_phys_space_rep}
\pi_P(X)
=
\sum_{i=1}^n
\pi_{\lambda_0}^{(i)}(X),
\qquad
X\in\mathfrak{su}(d),
\end{equation}
where $\pi_{\lambda_0}^{(i)}(X)$ acts as $\pi_{\lambda_0}(X)$ on the $i$th tensor factor and as the identity on all other tensor factors. 
For illustration of Lie group and corresponding Lie algebra actions, see Fig.~\ref{fig:theme}(a).

The logical Hilbert space is the fundamental representation of $SU(d)$:
\begin{equation}\label{eq:log_space_rep}
\mathcal H_L
:=
V_{\omega_1}
\cong
\mathbb C^d .
\end{equation}
The notation $V_{\omega_1}$ is explained in Appendix~\ref{sec:supple_sud}. We denote the corresponding group action by
$$
U_L(g)
=
u(g),
\qquad
g\in SU(d),
$$
and the corresponding Lie-algebra representation by
$$
\pi_L(X)
=
\bar X,
\qquad
X\in\mathfrak{su}(d).
$$
The bar is used only when we want to emphasize that the generator acts on the logical fundamental representation, 
rather than on one of the physical tensor factors.

\paragraph{Multiplicity-vector form of the encoding.}

Decompose the physical Hilbert space into irreducible $SU(d)$-representations:
\begin{equation}\label{eq:isotipic_decomp}
(V_{\lambda_0})^{\otimes n}
\cong
\bigoplus_{\lambda\in\Lambda}
V_\lambda\otimes M_\lambda .
\end{equation}
Here $V_\lambda\otimes M_\lambda$ is the isotypic component associated with the irrep $V_\lambda$, and $M_\lambda$ is its multiplicity space~\cite{Kong2022}. Since the logical space is $V_{\omega_1}$, Eq.~\eqref{eq:intertwiner_decomp} implies that the image of any covariant encoding must lie inside the corresponding isotypic component
$$
V_{\omega_1}\otimes M_{\omega_1}
\subseteq
\mathcal H_P .
$$
Thus a covariant encoding selects one copy of $V_{\omega_1}$ inside the physical Hilbert space. Equivalently, it is specified by a \textit{multiplicity vector} $\ket{m_{\omega_1}}\in M_{\omega_1}$:
\begin{equation}\label{eq:mult_vec_enc}
\mathcal{V}:
\mathcal H_L
\cong
V_{\omega_1}
\longrightarrow
V_{\omega_1}\otimes\ket{m_{\omega_1}}
\subseteq
\mathcal H_P;
\end{equation}
see Fig.~\ref{fig:theme}(b) for schematic illustration. The choice of $\ket{m_{\omega_1}}$ determines the particular covariant code. 
The same multiplicity-space viewpoint was used in Ref.~\cite{Denys2024} to reconstruct known quantum error-correcting codes covariant under 
finite symmetry groups.

This construction is possible only if the fundamental isotypic component is present in the physical representation, that is, only if
$$
\dim M_{\omega_1}>0 .
$$
For each explicit construction below, we will specify conditions ensuring this nonzero multiplicity.

It is useful to describe the encoding in terms of the Schur transform. For a representation $\mathcal H$ of $SU(d)$, 
the Schur transform is a unitary
$$
U_{\mathrm{Schur}}:
\mathcal H
\longrightarrow
\bigoplus_{\lambda\in\Lambda}
V_\lambda\otimes M_\lambda
$$
such that
\begin{equation}\label{eq:schur_transform}
U_{\mathrm{Schur}}
U(g)
U_{\mathrm{Schur}}^\dagger
=
\bigoplus_{\lambda\in\Lambda}
U_\lambda(g)\otimes I_{M_\lambda}.
\end{equation}
The Schur transform is not canonical, because it depends on choices of bases in the irreducible spaces $V_\lambda$ and in 
the multiplicity spaces $M_\lambda$. Its construction in quantum information settings is discussed in Ref.~\cite{Bacon2006}.

With this notation, the encoding associated with the multiplicity vector $\ket{m_{\omega_1}}\in M_{\omega_1}$ is
\begin{equation}\label{eq:schur_transform_enc}
\mathcal{V}
=
U_{\mathrm{Schur}}^\dagger
\left(
I_{V_{\omega_1}}
\otimes
\ket{m_{\omega_1}}
\right).
\end{equation}

\paragraph{Permutation symmetry of the code space.}

In addition to the $SU(d)$ action, we will use finite permutation symmetries of the physical subsystems. Let $G_P\subseteq S_n$ 
act on $\mathcal H_P$ by permuting tensor factors:
\begin{equation}\label{eq:group_act_on_phys_space}
P_g
\left(
\ket{i_1}\otimes\cdots\otimes\ket{i_n}
\right)
:=
\ket{i_{g^{-1}(1)}}\otimes\cdots\otimes\ket{i_{g^{-1}(n)}} .
\end{equation}
This permutation action commutes with the transversal $SU(d)$ action in Eq.~\eqref{eq:group_phys_space_rep}, 
because $SU(d)$ acts identically on each tensor factor. Hence $G_P$ preserves the isotypic decomposition~\eqref{eq:isotipic_decomp}. 
By the double-centralizer structure of the commutant, the action of $G_P$ has the form~\cite{GoodmanWallach2009}
\begin{equation}\label{eq:group_act_on_code}
P_g
=
\bigoplus_{\lambda\in\Lambda}
I_{V_\lambda}\otimes R_\lambda(g),
\end{equation}
where $R_\lambda$ is a representation of $G_P$ on the multiplicity space $M_\lambda$. Thus $G_P$ acts trivially on the 
$SU(d)$-irrep factor and nontrivially only on the multiplicity space.

The purpose of introducing $G_P$ is to impose an additional symmetry on the code space. After covariance fixes the relevant isotypic component,
the remaining freedom is the choice of the multiplicity vector. We choose this vector from the $G_P$-invariant subspace
\begin{equation}\label{eq:g_inv_subspace}
M_{\omega_1}^{G_P}
:=
\left\{
\ket m\in M_{\omega_1}
:
R_{\omega_1}(g)\ket m=\ket m
\ \forall g\in G_P
\right\}.
\end{equation}

If $\ket m\in M_{\omega_1}^{G_P}$, then the code space
$$
V_{\omega_1}\otimes\ket m
$$
is invariant under $G_P$, because $G_P$ acts as the identity on $V_{\omega_1}$ and preserves $\ket m$ in the multiplicity space.

Finally, the same permutation action induces an action on operators by conjugation. For $X\in\operatorname{End}(\mathcal H_P)$ and $g\in G_P$, define
\begin{equation}\label{eq:group_act_oper}
g\cdot X
:=
P_gXP_g^\dagger,
\end{equation}
where $P_g$ is the permutation operator from Eq.~\eqref{eq:group_act_on_phys_space}.

\section{\texorpdfstring{$SU(2)$}{SU(2)}-covariant Approximate Quantum Error Correction Codes}\label{sec:main_su2}

We begin the explicit construction with the familiar case of $SU(2)$. The logical system is a spin-$1/2$ degree of freedom, 
while the physical system consists of many spin-$j$ subsystems. Covariance fixes the symmetry sector in which the code can live, but it still leaves a multiplicity-space freedom. We use this remaining freedom to distribute the logical spin uniformly over the physical sites. For single-site erasure it means that the environment receives a one-site reduced state, and the code is good when this reduced state is nearly independent of the logical input.

\paragraph{Setup and reduced-state strategy.}

Let $V_j$ denote the irreducible spin-$j$ representation of $SU(2)$. We take the logical Hilbert space~\eqref{eq:log_space_rep} to be
\[
    \mathcal H_L = V_{1/2}\cong \mathbb C^2,
\]
and the physical Hilbert space~\eqref{eq:phys_space_rep} to be
\[
    \mathcal H_P = (V_j)^{\otimes n},
\]
where $n$ is odd and $j\in \frac12+\mathbb Z$ is half-integer. These conditions are exactly those ensuring that the spin-$1/2$ representation 
appears in $(V_j)^{\otimes n}$, or equivalently that the multiplicity space $M_{V_{1/2}}$ is nonzero:
\[
    \dim M_{V_{1/2}}>0
    \quad\Longleftrightarrow\quad
    n\equiv 1 \pmod 2
    \quad\text{and}\quad
    j\in \frac12+\mathbb Z .
\]
We use the standard spin-operator notation for the physical representation~\eqref{eq:alg_phys_space_rep} of the generators of $\mathfrak{su}(2)$,
\[
    \pi_j^{(k)}\left(\frac{i\sigma_a}{2}\right)
    \equiv J_a^{(k)},
    \qquad a=x,y,z,
\]
so that the transversal physical action is
\begin{equation}\label{eq:su2_cov_rep}
    \frac{i\sigma_a}{2}
    \longmapsto
    iJ_a
    :=
    \sum_{k=1}^n iJ_a^{(k)},
    \qquad a=x,y,z .
\end{equation}

This $SU(2)$ setting is useful both as a physically transparent example and as a complete illustration of the general method. 
The main object is the one-site reduced state $\rho^{(i)}(\rho_L)$ defined in Eq.~\eqref{eq:reduced_state}. For flagged single-qudit erasure, this is 
precisely the state received by the environment, as follows from the complementary channel in Eq.~\eqref{eq:dual_to_erasure}. 
Thus, the task is to choose the code so that $\rho^{(i)}(\rho_L)$ approaches a fixed state, independent of $\rho_L$, as $n$ grows. 
A detailed derivation of the final expression~\eqref{eq:reduced_state_su2} is given in Proposition~\ref{prop:SU2cov}, here we spell out the representation-theoretic mechanism.

\paragraph{Reduced-state structure from covariance.}
The form of $\rho^{(i)}(\rho_L)$ is strongly constrained by covariance. It satisfies
\begin{equation}\label{eq:reduced-state_cov_cond}
    \rho^{(i)}\!\left(u(g)\rho_Lu(g)^\dagger\right)
    =
    U_j^{(i)}(g)\rho^{(i)}(\rho_L)\bigl(U_j^{(i)}(g)\bigr)^\dagger .
\end{equation}
Equivalently, after extending $\rho^{(i)}$ linearly from density matrices to all operators on the logical space, the resulting map
\[
    \operatorname{End}(V_{1/2})
    \xrightarrow{\;\Phi^{(i)}\;}
    \operatorname{End}(V_j)
\]
is an $SU(2)$-intertwiner~\eqref{eq:intertwiner_def}, where both operator spaces carry the adjoint action. We can therefore determine 
its possible form from representation theory. On the logical space,
\[
    \operatorname{End}(V_{1/2})
    \cong
    V_{1/2}\otimes V_{1/2}^*
    \cong
    V_0\oplus V_1,
\]
where $V_0$ is the trivial representation, spanned by the identity, and $V_1$ is the adjoint representation, spanned by the Pauli generators. On a physical spin-$j$ site,
\[
    \operatorname{End}(V_j)
    \cong
    V_j\otimes V_j^*
    \cong
    \bigoplus_{\ell=0}^{2j} V_\ell .
\]
Thus Eq.~\eqref{eq:intertwiner_decomp} implies that $\Phi^{(i)}$ maps the trivial component of $\operatorname{End}(V_{1/2})$ to the unique trivial component of $\operatorname{End}(V_j)$, 
and maps the adjoint component of $\operatorname{End}(V_{1/2})$ to the unique adjoint component of $\operatorname{End}(V_j)$.

Writing the logical state in Bloch form,
\[
    \rho_L=\frac12(I+\mathbf r\cdot\boldsymbol\sigma),
\]
the identity term belongs to the trivial representation and therefore maps to a multiple of the identity on $V_j$. 
Normalization fixes this multiple to be $1/(2j+1)$. The traceless term belongs to the adjoint representation, 
and hence its image must be proportional to the unique adjoint component on $\operatorname{End}(V_j)$, spanned by the physical spin operators $J_a^{(i)}$. 
Therefore the reduced state has the form
\begin{equation}\label{eq:reduced_state_beta}
    \rho^{(i)}(\rho_L)
    =
    \frac{I_{2j+1}}{2j+1}
    +
    \beta^{(i)}
    \sum_{a=x,y,z} r_a J_a^{(i)} .
\end{equation}
At this stage the code space has not yet been specified; all dependence on the multiplicity vector~\eqref{eq:mult_vec_enc} is contained in the scalar coefficients $\beta^{(i)}$.

To determine these coefficients, we first compress local spin operators to the code space. Since the encoding is covariant, the total physical spin acts as the logical spin,
\begin{equation}\label{eq:total_spin_act}
    \mathcal{V}^\dagger \sum_{i=1}^n J_a^{(i)} \mathcal{V}
    =
    \frac{\sigma_a}{2} .
\end{equation}
Moreover, the compression map $X\mapsto \mathcal{V}^\dagger X \mathcal{V}$ is again an $SU(2)$-intertwiner $\operatorname{End}(V_j)\rightarrow \operatorname{End}(V_{1/2})$. 
Since the operators $J_a^{(i)}$ span an adjoint component on the $i$th site, Eq.~\eqref{eq:intertwiner_decomp} implies that their compression must be proportional 
to the unique adjoint component on the logical spin-$1/2$ space. Hence
\begin{equation}\label{eq:alpha_def}
    \mathcal{V}^\dagger J_a^{(i)} \mathcal{V}
    =
    \alpha^{(i)}\frac{\sigma_a}{2},
\end{equation}
for real scalars $\alpha^{(i)}$ satisfying $\sum_i\alpha^{(i)}=1$, where the latter identity follows from Eq.~\eqref{eq:total_spin_act}. 
Comparing Hilbert--Schmidt inner products in Eqs.~\eqref{eq:reduced_state_beta} and~\eqref{eq:alpha_def} gives
\begin{equation}\label{eq:beta_through_alpha}
    \beta^{(i)}
    =
    \frac{3\alpha^{(i)}}{2j(j+1)(2j+1)} .
\end{equation}

\paragraph{Cyclic invariance.}

It remains to choose the multiplicity vector so that the coefficients $\alpha^{(i)}$ are uniform. 
We do this by choosing $\ket m$ to be an eigenvector of the cyclic shift $s\in C_n\subseteq S_n$. 
Equivalently, the code space is invariant under cyclic permutations of the physical qudits. If $P_s$ denotes the corresponding permutation operator, 
then $P_s J_a^{(i)}P_s^\dagger=J_a^{(i+1)}$, while on the code space $P_s\mathcal{V}=e^{i\phi}\mathcal{V}$. Therefore
\[
    \mathcal{V}^\dagger J_a^{(i+1)}\mathcal{V}
    =
    \mathcal{V}^\dagger P_sJ_a^{(i)}P_s^\dagger \mathcal{V}
    =
    \mathcal{V}^\dagger J_a^{(i)}\mathcal{V} .
\]
All coefficients $\alpha^{(i)}$ are equal. Since they sum to one, $\alpha^{(i)}=1/n$ for every site.

For the cyclically invariant $SU(2)$-covariant code, the one-site reduced state is therefore
\begin{equation}\label{eq:reduced_state_su2}
    \rho^{(i)}(\rho_L)
    =
    \frac{I_{2j+1}}{2j+1}
    +
    \frac{3}{2nj(j+1)(2j+1)}
    \sum_{a=x,y,z} r_a J_a^{(i)},
\end{equation}
where $i=1,\ldots,n$ and $\rho_L=\frac12(I+\mathbf r\cdot\boldsymbol\sigma)$. This expression makes the error-correction mechanism transparent, 
because the logical-state-dependent part of the erased site's density matrix is suppressed by $1/n$, so the environment's state approaches the maximally mixed state as $n$ grows.

\paragraph{Single-site erasure bound.}
We now translate this reduced-state statement into an AQEC bound. Since all one-site reduced states are identical up to the natural identification 
of the physical sites, the complementary channel for flagged single-site erasure becomes
\begin{equation}\label{eq:comp_chan_su2}
    \widehat{\mathcal N\circ\mathcal E}(\rho_L)
    =
    \sum_i p_i\ket i\!\bra i_{F_E}\otimes \rho^{(i)}(\rho_L)
    \cong
    \omega_{\mathrm{flag}}\otimes \rho^{(1)}(\rho_L),
\end{equation}
where
\[
    \omega_{\mathrm{flag}}
    :=
    \sum_i p_i\ket i\!\bra i_{F_E} .
\]
We compare this channel with the constant channel
\[
    \Lambda_0(\rho)
    :=
    \operatorname{Tr}(\rho)\,
    \omega_{\mathrm{flag}}\otimes \frac{I_{2j+1}}{2j+1} .
\]
The explicit fidelity calculation, given in Proposition~\ref{prop:SU2covFidelity}, yields
\begin{equation}\label{eq:fidelity_su2}
    F\!\left(\widehat{\mathcal N\circ\mathcal E},\Lambda_0\right)
    =
    1-
    \frac{9}{32n^2j(j+1)}
    +O_{j}(n^{-3}) .
\end{equation}
Together with the complementary-channel characterization of AQEC, this yields the following result.

\begin{theorem}\label{thrm:su2scale}
Let $\mathcal H_L=V_{1/2}$ and $\mathcal H_P=(V_j)^{\otimes n}$, with $n$ odd and $j\in\frac12+\mathbb Z$. Equip $\mathcal H_L$ with the fundamental spin-$1/2$ 
representation of $\mathfrak{su}(2)$ and $\mathcal H_P$ with the transversal spin-$j$ representation
\begin{equation}
    \frac{i\sigma_a}{2}
    \longmapsto
    iJ_a
    \equiv
    \sum_{k=1}^n iJ_a^{(k)},
    \qquad a=x,y,z .
\end{equation}
Let $\mathcal E$ be an $\mathfrak{su}(2)$-covariant encoding whose code space is invariant under cyclic permutations of the physical qudits. 
Then, for the flagged single-qudit erasure channel
\[
    \mathcal N(\sigma)
    =
    \sum_{i=1}^n p_i\ket i\!\bra i_F
    \otimes \ket e\!\bra e_{A_i}
    \otimes \operatorname{Tr}_i(\sigma),
\]
there exists a recovery channel $\mathcal R$ such that
\[
    d(\mathcal R\mathcal N\mathcal E,\mathrm{id})
    \le
    \frac{3}{4\sqrt{2j(j+1)}}\frac{1}{n}
    +O_{j}(n^{-2}) .
\]
In particular, the code is $O\!\left(1/(\sqrt{j(j+1)}\,n)\right)$-correctable against flagged single-qudit erasure.
\end{theorem}

\begin{proof}
By Eq.~\eqref{eq:fidelity_su2} and the definition of the channel distance in Eq.~\eqref{eq:distance_def},
\[
    d\!\left(\widehat{\mathcal N\circ\mathcal E},\Lambda_0\right)
    =
    \frac{3}{4\sqrt{2j(j+1)}}\frac{1}{n}
    +O_{j}(n^{-2}) .
\]
The complementary-channel theorem, Eq.~\eqref{eq:beny}, gives
\[
    \min_{\mathcal R}
    d(\mathcal R\mathcal N\mathcal E,\mathrm{id})
    =
    \min_{\Lambda\,\mathrm{constant}}
    d\!\left(\widehat{\mathcal N\circ\mathcal E},\Lambda\right)
    \le
    d\!\left(\widehat{\mathcal N\circ\mathcal E},\Lambda_0\right),
\]
and the stated bound follows.
\end{proof}

Known lower bounds for covariant codes under erasure noise~\cite{Faist2020} give
\[
    \min_{\mathcal R}
    d(\mathcal R\mathcal N\mathcal E,\mathrm{id})
    \ge
    \frac{1}{4\sqrt{2}\,j}\frac{1}{n} .
\]
Thus the construction is optimal in its scaling with $n$ and $j$, up to constant factors.

\section{\texorpdfstring{$SU(d)$}{SU(d)}-covariant Approximate Quantum Error Correction Codes}\label{sec:main_sud}

We now extend the construction to $SU(d)$. The logical system carries the fundamental representation, while each physical subsystem carries an antisymmetric 
tensor-power representation. As in the $SU(2)$ case, covariance fixes the symmetry sector in which the code can live, but leaves a freedom in the corresponding multiplicity space. 
We use this freedom to choose a code space that is invariant under suitable permutations of the physical subsystems. For single-qudit flagged erasure, 
cyclic permutation symmetry is sufficient, since it spreads the logical information uniformly over the physical sites, 
so that each erased-site reduced state becomes nearly independent of the logical input.

\paragraph{Single-site physical representation.}

We first define the single-site physical representation. For $1\le r\le d-1$, let
\begin{equation}\label{eq:vomegar_def}
V_{\omega_r}:=\bigwedge^r \mathbb C^d .
\end{equation}
This space has basis vectors
$$
e_I:=e_{i_1}\wedge\cdots\wedge e_{i_r},
\qquad
1\le i_1<\cdots<i_r\le d,
$$
where $\wedge$ denotes the wedge product. Its dimension is
$$
\dim V_{\omega_r}=\binom{d}{r}.
$$
The action of $X\in\mathfrak{su}(d)$ on $V_{\omega_r}$ is defined by
\begin{equation}\label{eq:sud_vomegar_def}
\pi_{\omega_r}(X)
\left(v_1\wedge\cdots\wedge v_r\right)
:=
\sum_{m=1}^r
v_1\wedge\cdots\wedge Xv_m\wedge\cdots\wedge v_r .
\end{equation}
This is an irreducible representation; see Ref.~\cite{FultonHarris1991}. The notation $V_{\omega_r}$ is explained in Appendix~\ref{sec:supple_sud}.

We use the standard generators of the fundamental representation of $\mathfrak{su}(d)$. Let ${t_a}_{a=1}^{d^2-1}$ be traceless Hermitian $d\times d$ matrices satisfying
\begin{equation}
t_a t_b
=
\frac{1}{2d}\delta_{ab}I_d
+
\frac12
\sum_{c=1}^{d^2-1}
\left(i f_{abc}+d_{abc}\right)t_c ,
\end{equation}
where $f_{abc}$ are completely antisymmetric structure constants and $d_{abc}$ are completely symmetric coefficients. 
Further properties of these generators are collected in Appendix~\ref{sec:supple_sud}.

We take the logical Hilbert space~\eqref{eq:log_space_rep} to be the fundamental representation
$$
\mathcal H_L:=V_{\omega_1}\cong \mathbb C^d ,
$$
and the physical Hilbert space~\eqref{eq:phys_space_rep} to be
$$
\mathcal H_P:=(V_{\omega_r})^{\otimes n},
$$
where $r\neq 0,d$ and $nr\equiv 1\pmod d$. We denote the action of the generator $t_a$ on the $j$th physical subsystem by
\begin{equation}\label{eq:su_d_gens}
    \pi_{\omega_r}^{(j)}(t_a)\equiv T_a^{(j)},
    \qquad
    a=1,\ldots,d^2-1 .
\end{equation}

Thus the transversal physical representation~\eqref{eq:alg_phys_space_rep} is specified on generators by
$$
t_a
\longmapsto
T_a
:=
\sum_{i=1}^n T_a^{(i)},
\qquad
a=1,\ldots,d^2-1 .
$$
In this section the representation $V_{\omega_r}$ is fixed, so we omit the subscript $\omega_r$ from the local generators. 
In later sections and in the main theorems, we will often specify the physical representation by giving the images of the Lie-algebra generators.

The condition $nr\equiv 1\pmod d$ ensures that $\mathcal H_P$ contains a copy of the fundamental representation. 
Equivalently, the $V_{\omega_1}$-isotypic component is nonzero, so it can contain the image of an $SU(d)$-covariant encoding of the logical space. 
Lemma~\ref{lem:fundamental_in_extirior_power} proves that this condition is sufficient for $(V_{\omega_r})^{\otimes n}$ to contain $V_{\omega_1}$. 
It is not necessary in general. For example, the $SU(2)$ construction in Section~\ref{sec:main_su2} used arbitrary half-integer spin-$j$ irreducible representations 
as physical subsystems.

\paragraph{One-site reduced state.}

The main object is again the one-site reduced state $\rho^{(i)}(\rho_L)$ defined in Eq.~\eqref{eq:reduced_state}. 
A detailed derivation of its final form is given in Proposition~\ref{prop:SUdcov}. Here we describe the representation-theoretic structure of the calculation.

Write the logical state in the generalized Bloch form
\begin{equation}\label{eq:general_bloch}
\rho_L
=
\frac{I}{d}
+
\left(\rho_L-\frac{I}{d}\right)
=
\frac{I}{d}
+
\sum_a r_a t_a .
\end{equation}
The reduced-state map extends linearly to a map
$$
    \operatorname{End}(V_{\omega_1})
    \xrightarrow{\;\Phi^{(i)}\;}
    \operatorname{End}(V_{\omega_r}) .
$$
By covariance, this map is an $SU(d)$-intertwiner~\eqref{eq:intertwiner_def}, where both operator spaces carry the adjoint action. 
As shown in Lemma~\ref{lem:adjoint_in_end}, $\operatorname{End}(V_{\omega_r})$ contains unique copies of both the trivial representation and the adjoint representation. 
Therefore, using the intertwiner decomposition~\eqref{eq:intertwiner_decomp}, the identity component of $\operatorname{End}(V_{\omega_1})$ must map 
to the identity component of $\operatorname{End}(V_{\omega_r})$, while the adjoint component must map to the adjoint component. Hence,
\begin{equation}\label{eq:linear_ext_of_red_state}
\begin{gathered}
\Phi^{(i)}\left(\frac{I_d}{d}\right)
=
\frac{I_{V_{\omega_r}}}{\dim V_{\omega_r}},
\\
\Phi^{(i)}\left(\sum_a r_a t_a\right)
=
\beta^{(i)}
\sum_a r_a T_a^{(i)} .
\end{gathered}
\end{equation}
Thus the reduced state has the form
\begin{equation}\label{eq:reduced_state_sud}
\rho^{(i)}(\rho_L)
=
\frac{I_{V_{\omega_r}}}{\dim V_{\omega_r}}
+
\beta^{(i)}
\sum_a r_a T_a^{(i)}.
\end{equation}

The coefficients $\beta^{(i)}$ are determined by the compression of local generators to the code space. Since the encoding is covariant, 
the total physical action restricts to the logical fundamental action. For each local generator, the compression map $X\mapsto \mathcal{V}^\dagger X \mathcal{V}$ is again 
an $SU(d)$-intertwiner $\operatorname{End}(V_{\omega_r})\rightarrow \operatorname{End}(V_{\omega_1})$, and therefore
$$
\mathcal{V}^\dagger T_a^{(i)}\mathcal{V}
=
\alpha^{(i)}\bar t_a ,
$$
where
$$
\sum_i\alpha^{(i)}=1,
$$
and $\bar t_a$ denotes the fundamental generator acting on the logical space. The relation between $\alpha^{(i)}$ and $\beta^{(i)}$ 
is obtained by comparing Hilbert--Schmidt inner products, exactly as in the $SU(2)$ case, giving
$$
\beta^{(i)}
=
\frac{\alpha^{(i)}}{2\kappa_{\omega_r}},
\qquad
\kappa_{\omega_r}
=
\frac12\binom{d-2}{r-1}.
$$

It remains to choose the multiplicity vector. We choose $\ket m$ to be an eigenvector of the cyclic shift operator. 
Equivalently, the code space $V_{\omega_1}\otimes\ket m$ is invariant under cyclic permutations of the physical qudits. 
The same argument as in Section~\ref{sec:main_su2} then implies that all coefficients $\alpha^{(i)}$ are equal. Since they sum to one,
$$
\alpha^{(i)}=\frac1n
\qquad
\text{for every } i .
$$
For this cyclically invariant $SU(d)$-covariant code, the one-site reduced state becomes
$$
\rho^{(i)}(\rho_L)
=
\frac{I_{V_{\omega_r}}}{\dim V_{\omega_r}}
+
\frac{1}{2\kappa_{\omega_r}n}
\sum_a r_a T_a^{(i)} .
$$
The logical-state-dependent part is suppressed by $1/n$, so the erased-site state approaches the maximally mixed state as $n$ grows; 
see Fig.~\ref{fig:theme}(c) for a schematic illustration.

\paragraph{Single-site erasure bound.}

We now translate this reduced-state statement into an AQEC bound. Since the one-site reduced states are identical up to the natural identification of physical sites, 
the complementary channel~\eqref{eq:dual_to_erasure} for flagged single-qudit erasure is
$$
\widehat{\mathcal N\circ\mathcal E}(\rho_L)
\cong
\omega_{\mathrm{flag}}
\otimes
\left(
\frac{I_{V_{\omega_r}}}{\dim V_{\omega_r}}
+
\frac{1}{2\kappa_{\omega_r}n}
\sum_a r_a T_a^{(1)}
\right),
$$
where
$$
\omega_{\mathrm{flag}}
:=
\sum_i p_i\ket i\bra i_{F_E}.
$$
We compare this channel with the constant channel~\eqref{eq:constant_channel_def}
$$
\Lambda_0(\rho)
:=
\operatorname{Tr}(\rho)\;
\omega_{\mathrm{flag}}
\otimes
\frac{I_{V_{\omega_r}}}{\dim V_{\omega_r}} .
$$
The explicit fidelity calculation, given in Proposition~\ref{prop:SUdcovfidelity}, yields
\begin{equation}\label{eq:fidelity_sud}
F\left(\widehat{\mathcal N\circ\mathcal E},\Lambda_0\right)
=
1
-
\frac{(d-1)^2(d+1)}{8r(d-r)}
\frac{1}{n^2}
+
O_{d}(n^{-3}) .
\end{equation}
Together with the complementary-channel characterization of AQEC, this gives the following result.

\begin{theorem}\label{thrm:sudscale}
Let $\mathcal H_L=V_{\omega_1}$ and $\mathcal H_P=(V_{\omega_r})^{\otimes n}$, with $r\neq 0,d$ and $nr\equiv 1\pmod d$. 
Equip $\mathcal H_L$ with the fundamental representation of $\mathfrak{su}(d)$ and $\mathcal H_P$ with the transversal representation
\begin{equation}
t_a
\longmapsto
T_a
\equiv
\sum_{i=1}^n T_a^{(i)},
\qquad
a=1,\ldots,d^2-1,
\end{equation}
where $T_a^{(i)}$ is the representation of $t_a$ on the $i$th physical space $V_{\omega_r}$, acting as in Eq.~\eqref{eq:sud_vomegar_def}. 
Let $\mathcal E$ be an $\mathfrak{su}(d)$-covariant encoding whose code space is invariant under cyclic permutations of the physical qudits. 
Then, for the flagged single-qudit erasure channel
$$
\mathcal N(\sigma)
=
\sum_{i=1}^n
p_i\ket i\bra i_F
\otimes
\ket e\bra e_{A_i}
\otimes
\operatorname{Tr}_i(\sigma),
$$
there exists a recovery channel $\mathcal R$ such that
$$
d(\mathcal R\mathcal N\mathcal E,\mathrm{id})
\le
\frac{d-1}{2\sqrt 2}
\sqrt{\frac{d+1}{r(d-r)}}
\frac1n
+
O_{d}(n^{-2}) .
$$
In particular, the code is
$$
O\left(
(d-1)
\sqrt{\frac{d+1}{r(d-r)}}
\frac1n
\right)
$$
-correctable against flagged single-qudit erasure.
\end{theorem}

\begin{proof}
By Eq.~\eqref{eq:fidelity_sud} and the definition of the channel distance in Eq.~\eqref{eq:distance_def},
$$
d\left(\widehat{\mathcal N\circ\mathcal E},\Lambda_0\right)
=
\frac{d-1}{2\sqrt 2}
\sqrt{\frac{d+1}{r(d-r)}}
\frac1n
+
O_{d}(n^{-2}) .
$$
The complementary-channel theorem, Eq.~\eqref{eq:beny}, gives
$$
\min_{\mathcal R}
d(\mathcal R\mathcal N\mathcal E,\mathrm{id})
=
\min_{\Lambda,\mathrm{constant}}
d\left(\widehat{\mathcal N\circ\mathcal E},\Lambda\right)
\le
d\left(\widehat{\mathcal N\circ\mathcal E},\Lambda_0\right),
$$
and the stated bound follows.
\end{proof}

For $r=1$, the upper bound coincides with the scaling obtained in Ref.~\cite{Kong2022} for a randomized code construction. The lower bound of Ref.~\cite{Faist2020} gives
$$
\min_{\mathcal R}
d(\mathcal R\mathcal N\mathcal E,\mathrm{id})
\ge
\frac{1}{2n}.
$$
Thus the construction is optimal in its scaling with $n$, up to a constant factor depending on $d$ and $r$.

\subsection{Covariant Decoder}\label{subsec:main_decoder}

We now discuss an explicit decoder for the covariant codes constructed above. In approximate quantum error correction, 
the recovery map is not fixed uniquely by the code and the noise model. Instead, one has to choose a recovery channel that gives a small error with 
respect to the chosen distance measure. Determining the optimal recovery, including the sharp constants in its dependence on the physical parameters, can be difficult. 
We therefore use a natural recovery map that is known to be near-optimal.

\paragraph{Petz recovery map.}

We use the transpose channel, also called the Petz recovery map. For a channel $\mathcal N$ and a reference state $\omega$, the Petz map is
$$
\mathcal P_{\omega,\mathcal N}(\sigma)
=
\omega^{1/2}
\mathcal N^*
\left(
\mathcal N(\omega)^{-1/2}
\sigma
\mathcal N(\omega)^{-1/2}
\right)
\omega^{1/2},
$$
where $\mathcal N^*$ is the adjoint of $\mathcal N$ with respect to the Hilbert--Schmidt inner product, and the inverse is taken on the support of $\mathcal N(\omega)$.

\paragraph{Near-optimality.}

The analysis uses two main facts. First, for the reference state $\omega=I_d/d$, the Petz-recovered logical channel
$$
\Phi
:=
\mathcal R_{\mathrm{Petz}}
\circ
\mathcal N
\circ
\mathcal E
$$
is $SU(d)$-covariant. For related uses of channel covariance, see Refs.~\cite{Zhou2021, Alexander2025}. Since the logical representation is irreducible, 
Schur's lemma implies that $\Phi$ is a depolarizing channel:
$$
\Phi(\rho)
=
\lambda_{\mathrm{Petz}}\rho
+
\left(1-\lambda_{\mathrm{Petz}}\right)
\frac{I_d}{d}\operatorname{Tr}\rho .
$$
The depolarizing parameter is
\begin{equation}\label{eq:lambda_P}
\lambda_{\mathrm{Petz}}
=
\frac{2}{d^2-1}
\sum_{i=1}^n p_i
\sum_{a=1}^{d^2-1}
\operatorname{Tr}
\left[
S_i^{-1/2}
\mathcal C_i(t_a)
S_i^{-1/2}
\mathcal C_i(t_a)
\right],
\end{equation}
where
$$
\mathcal C_i(\rho_L)
:=
\operatorname{Tr}_i\left(\mathcal{V}\rho_L\mathcal{V}^\dagger\right),
\qquad
S_i
:=
\operatorname{Tr}_i\left(\mathcal{V}\mathcal{V}^\dagger\right).
$$
The corresponding entanglement fidelity is
$$
F(\Phi,\mathrm{id})
=
\sqrt{
\lambda_{\mathrm{Petz}}
+
\frac{1-\lambda_{\mathrm{Petz}}}{d^2}
} .
$$
Thus, determining the Petz recovery error with sharp constants is equivalent to determining the asymptotics of $\lambda_{\mathrm{Petz}}$.

The exact expression in Eq.~\eqref{eq:lambda_P} is difficult to analyze directly. It requires control of the large-$n$ behavior of 
the operators $S_i^{-1/2}$ and $\mathcal C_i(t_a)$ on the $(n-1)$-site space. These operators have support in representation sectors of $(V_{\omega_r})^{\otimes(n-1)}$ 
where the adjoint representation can appear with large multiplicity. This multiplicity structure makes a direct asymptotic calculation substantially more 
involved than the one-site reduced-state calculation. Similar difficulties also appear in the analysis of multi-qudit erasures below.

For this reason, we use the near-optimality of the transpose channel established in Ref.~\cite{Ng2010}. Combined with the bound in Theorem~\ref{thrm:sudscale}, 
this shows that the Petz recovery achieves the same $1/n$ scaling of the recovery error, up to a constant factor. The explicit form of the Petz decoder and its 
asymptotic performance are stated in Theorem~\ref{thm:petz_decoder_explicit}:

\begin{theorem}
For the encoding $\mathcal E$ and the flagged single-qudit erasure channel $\mathcal N$ defined in Theorem~\ref{thrm:sudscale}, 
the Petz recovery map with reference state $\omega=I_d/d$ has the form
$$
\mathcal R_{\mathrm{Petz}}(X)
=
\sum_{i:,p_i>0}
\mathcal{V}^\dagger
\left(
S_i^{-1/2}
\langle i,e|X|i,e\rangle_{F,A_i}
S_i^{-1/2}
\otimes I_{A_i}
\right)
\mathcal{V},
$$
where $\mathcal{V}$ is the encoding isometry and
$$
S_i
:=
\operatorname{Tr}_i\left(\mathcal{V}\mathcal{V}^\dagger\right).
$$
Moreover, this recovery is near-optimal in the sense that
$$
    d\left(
    \mathcal R_{\mathrm{Petz}}
    \circ
    \mathcal N
    \circ
    \mathcal E,
    \mathrm{id}
    \right)
    =
    O_{d}(n^{-1}) .
$$
\end{theorem}

This result gives the asymptotic scaling in $n$, but not the sharp constants depending on the representation-theoretic parameters. 
We therefore do not claim constant-level optimality of the Petz recovery map in this setting. A more detailed analysis of Eq.~\eqref{eq:lambda_P}, 
including its dependence on the relevant multiplicity spaces, is left for future work.

The implementation of the decoder is closely related to the Schur transform in Eq.~\eqref{eq:schur_transform}. 
The Schur transform appears explicitly in the encoding isometry $\mathcal{V}$ in Eq.~\eqref{eq:schur_transform_enc}, and it also enters the Petz map through the operators $S_i$.

\subsection{Flagged two- and three-qudit erasure}\label{subsec:main_muti_erasure_noise}

We now consider flagged erasure errors affecting more than one physical site. In this subsection each physical subsystem carries 
the fundamental representation of $SU(d)$, so that
$$
\mathcal H_P=(V_{\omega_1})^{\otimes n}.
$$
The transversal physical representation is specified on generators by
$$
t_a
\longmapsto
\sum_{i=1}^n t_a^{(i)},
\qquad
a=1,\ldots,d^2-1,
$$
where $t_a^{(i)}$ denotes the fundamental action of $t_a$ on the $i$th physical space $V_{\omega_1}$.

\paragraph{Why higher transitivity is needed.}

For single-site erasure, the main object was the one-site reduced state. For multi-site erasure, the same role is played by the reduced state on the erased set of sites. 
Thus, for $k$-site erasure, the relevant states are
$$
\rho^{(i_1,\ldots,i_k)}(\rho_L).
$$
As before, these reduced-state maps are covariant. Hence their possible form is constrained by representation theory, while the choice of code space, equivalently 
the choice of multiplicity vector, enters through the coefficients in the resulting expansion.

The main new issue is that cyclic symmetry is no longer sufficient. In the one-site case, cyclic invariance forced
$$
\rho^{(i)}(\rho_L)
\cong
\rho^{(j)}(\rho_L)
\qquad
\text{for all } i,j,
$$
and this made the logical-state-dependent part of each one-site reduced state scale as $1/n$. For $k$-site erasure, 
the analogous condition is that all $k$-site reduced states be equal up to the canonical identification of tensor factors:
$$
\rho^{(i_1,\ldots,i_k)}(\rho_L)
\cong
\rho^{(j_1,\ldots,j_k)}(\rho_L)
$$
for all ordered $k$-tuples of distinct physical indices. This requires a larger permutation symmetry than the cyclic group.

We therefore choose the multiplicity vector $\ket m\in M_{\omega_1}$ to be invariant under a subgroup $G_P\subseteq S_n$ 
acting on the physical systems as in Eq.~\eqref{eq:group_act_on_phys_space}. The subgroup has to satisfy two requirements. 
First, it must be sufficiently transitive to identify all relevant reduced states. 
Second, it must be small enough that the invariant subspace $M_{\omega_1}^{G_P}$ in Eq.~\eqref{eq:g_inv_subspace} is nonzero.

The first requirement is met if $G_P$ is $k$-transitive. This means that for any two ordered $k$-tuples of distinct indices $(i_1,\ldots,i_k)$ and $(j_1,\ldots,j_k)$, 
there exists $g\in G_P$ such that
$$
g(i_1,\ldots,i_k)=(j_1,\ldots,j_k).
$$
If the code space is invariant under this action, then $P_g\mathcal{V}\rho_L\mathcal{V}^\dagger P_g^\dagger=\mathcal{V}\rho_L\mathcal{V}^\dagger$, 
and the induced action on local operators~\eqref{eq:group_act_oper} gives
$$
\rho^{(j_1,\ldots,j_k)}(\rho_L)
=
P_g
\rho^{(i_1,\ldots,i_k)}(\rho_L)
P_g^\dagger .
$$

\paragraph{Permutation groups and invariant multiplicities.}

The full symmetric group $S_n$ would enforce this condition for every $k$, but it is too restrictive for the present construction. 
The multiplicity space $M_{\omega_1}$ does not contain a nonzero $S_n$-invariant subspace, since $M_{\omega_1}$ is an irreducible $S_n$-representation; 
see Lemma~\ref{lem:M_module_irreducibility}. Therefore, one has to use a proper subgroup of $S_n$.

For the construction used here, the relevant proper highly transitive subgroups are the affine group
\begin{equation}\label{eq:AGL_def}
\mathrm{AGL}(1,n)
:=
\left\{
x\mapsto ax+b:
a\in\mathbb F_n^\times,;
b\in\mathbb F_n
\right\}
\end{equation}
and the projective linear group
\begin{equation}\label{eq:PGL_def}
\operatorname{PGL}(2,n-1)
:=
\operatorname{GL}\left(2,\mathbb F_{n-1}\right)
/\mathbb F_{n-1}^\times .
\end{equation}
The group $\mathrm{AGL}(1,n)$ is $2$-transitive on $\mathbb F_n$, while $\operatorname{PGL}(2,n-1)$ is $3$-transitive on the projective line
$$
\mathbb P^1(\mathbb F_{n-1})
=
\mathbb F_{n-1}\cup{\infty}.
$$

These groups are the useful proper highly transitive subgroups for our construction. 
Their use is guided by the classification of finite multiply transitive permutation groups, 
which relies on the Classification of Finite Simple Groups~\cite{GorensteinLyonsSolomon1994}. 
In particular, apart from the alternating group $A_n$, there is no infinite family of proper $k$-transitive subgroups of $S_n$ for $k>3$. Beyond $S_n$ and $A_n$, 
the only $4$-transitive groups are the Mathieu groups $M_{11}, M_{12}, M_{23}, M_{24}$, and the only $5$-transitive ones are $M_{12}$ and $M_{24}$; 
no other group is $6$-transitive. The Mathieu groups occur only at $n\in\{11,12,23,24\}$, so they cannot support a construction with $n\to\infty$. 
The alternating group $A_n$ is a proper highly transitive subgroup of $S_n$, but it is almost as large as $S_n$ and is too restrictive for the second requirement, 
since it does not provide the nonzero invariant multiplicity subspaces needed for our code construction. Thus, for our purposes, the useful choices are $\mathrm{AGL}(1,n)$ 
for $k=2$ and $\operatorname{PGL}(2,n-1)$ for $k=3$. For further background on these permutation groups, see Ref.~\cite{DixonMortimer1996}.

The existence of nonzero invariant subspaces
$$
M_{\omega_1}^{\mathrm{AGL}(1,n)}
\neq
0,
\qquad
M_{\omega_1}^{\operatorname{PGL}(2,n-1)}
\neq
0
$$
for sufficiently large $n$ is proved in Lemmas~\ref{lem:2_transitive_subgroup} and~\ref{lem:3_transitive_subgroup}. Therefore, one can choose multiplicity vectors
$$
\ket{v_1}\in M_{\omega_1}^{\mathrm{AGL}(1,n)},
\qquad
\ket{v_2}\in M_{\omega_1}^{\operatorname{PGL}(2,n-1)}.
$$
The corresponding code spaces
$$
V_{\omega_1}\otimes\ket{v_1},
\qquad
V_{\omega_1}\otimes\ket{v_2}
$$
are invariant under $\mathrm{AGL}(1,n)$ and $\operatorname{PGL}(2,n-1)$, respectively.

We summarize the resulting existence statements. The formal proofs are given in Theorems~\ref{thrm:2_transitive_covariant_code} and~\ref{thrm:3_transitive_covariant_code}.

\begin{theorem}\label{thrm:main_2_transitive_covariant_code}
Let $\mathcal H_L=V_{\omega_1}$ and $\mathcal H_P=(V_{\omega_1})^{\otimes n}$, where $n\equiv 1\pmod d$ and $n=p^k$ for a prime $p$ and a positive integer $k$. 
Equip $\mathcal H_L$ with the fundamental representation of $\mathfrak{su}(d)$ and $\mathcal H_P$ with the transversal representation
$$
t_a
\longmapsto
\sum_{i=1}^n t_a^{(i)},
\qquad
a=1,\ldots,d^2-1,
$$
where $t_a^{(i)}$ denotes the fundamental action of $t_a$ on the $i$th physical space $V_{\omega_1}$. 
For sufficiently large $n$, there exists an $\mathfrak{su}(d)$-covariant encoding $\mathcal E$ with respect to these representations whose code space is 
invariant under the action of $\mathrm{AGL}(1,n)$, defined in Eq.~\eqref{eq:AGL_def}.
\end{theorem}

\begin{theorem}\label{thrm:main_3_transitive_covariant_code}
Let $\mathcal H_L=V_{\omega_1}$ and $\mathcal H_P=(V_{\omega_1})^{\otimes n}$. Assume that $d=p^r$, where $p$ is prime and $r$ is a positive integer, 
and that $n-1=p^k$, where $k\ge r$ is a positive integer. Equip $\mathcal H_L$ with the fundamental representation of $\mathfrak{su}(d)$ and $\mathcal H_P$ 
with the transversal representation
$$
t_a
\longmapsto
\sum_{i=1}^n t_a^{(i)},
\qquad
a=1,\ldots,d^2-1,
$$
where $t_a^{(i)}$ denotes the fundamental action of $t_a$ on the $i$th physical space $V_{\omega_1}$. For sufficiently large $n$, there exists an $\mathfrak{su}(d)$-covariant 
encoding $\mathcal E$ with respect to these representations whose code space is invariant under the action of $\operatorname{PGL}(2,n-1)$, defined in Eq.~\eqref{eq:PGL_def}.
\end{theorem}

These results explain why exact permutation symmetry of this kind lets us treat erasure on at most three sites. 
This restriction is a limitation of the symmetry-based construction used here, not a no-go statement for more general codes. 
Requiring all $k$-site reduced states to coincide is a sufficient condition for approximate error correction, but it is not necessary. 
Section~\ref{subsec:main_multisite_erasure_noise} removes this restriction by giving up exact symmetry. 
A typical, rather than exactly symmetric, code is shown to correct erasure on any fixed number of sites.

\paragraph{Three-site reduced states.}

We now focus on flagged erasure of three physical sites. Since $\operatorname{PGL}(2,n-1)$ is also $2$-transitive, the same construction also applies to two-site erasures. 
If the goal is only to correct two-site erasures, the weaker $\mathrm{AGL}(1,n)$ symmetry already suffices. 
The two-site case is discussed separately in Section~\ref{subsec:main_covar_analog_noise}.

The three-site reduced-state calculation is more involved than the one-site calculation. 
The linear extension of the reduced-state map $\rho^{(ijk)}$ is still an intertwiner, but the relevant operator space now contains several copies of the adjoint representation. 
Covariance therefore no longer fixes the reduced state up to a single scalar coefficient, unlike in the one-site case~\eqref{eq:linear_ext_of_red_state}.

The detailed derivation is given in Section~\ref{subsec:supple_multi_erasure_noise}. Here we summarize the main steps. 
First, Lemmas~C.30 and~C.31 decompose the relevant trivial and adjoint components in a basis of few-body operators. 
Next, Proposition~C.28 determines the compression of these few-body operators to the code space. For distinct physical sites $i,j,k$, one obtains
$$
\begin{gathered}
\mathcal{V}^{\dagger} t_a^{(i)} \mathcal{V}=\frac{1}{n} \bar{t}_a, \quad \mathcal{V}^{\dagger} t_{a_1}^{(i)} t_{a_2}^{(j)} \mathcal{V}=-\frac{1}{2 d n} \delta_{a_1 a_2} I_L, \\
\mathcal{V}^{\dagger} t_{a_1}^{(i)} t_{a_2}^{(j)} t_{a_3}^{(k)} \mathcal{V}=\frac{1}{2 d n(n-2)} d_{a_1 a_2 a_3} I_L-\\-\frac{1}{2 d n(n-2)}\left(\delta_{a_1 a_2} \bar{t}_{a_3}+\delta_{a_1 a_3} \bar{t}_{a_2}+\delta_{a_2 a_3} \bar{t}_{a_1}\right),
\end{gathered}
$$
where $\bar t_a$ denotes the fundamental generator acting on the logical space. 
Finally, Proposition~C.34 combines these compression identities with the relevant Hilbert--Schmidt inner-product identities and gives the three-site reduced state
\begin{equation}\label{eq:3q_reduced_state}
    \rho^{(ijk)}(\rho_L)
    =
    \tau^{(ijk)}
    +
    \Delta^{(ijk)}(\rho_L).
\end{equation}
Here $\tau^{(ijk)}$ is independent of the logical input, while $\Delta^{(ijk)}(\rho_L)$ contains all dependence on $\rho_L$. Explicitly,
\begin{widetext}
\begin{align}
\tau^{(ijk)}
&=
\frac{1}{d^3}\mathbf 1
-\frac{2}{d^2 n}
\sum_{m=1}^{d^2-1}
\left(
t_m^{(i)}t_m^{(j)}
+t_m^{(i)}t_m^{(k)}
+t_m^{(j)}t_m^{(k)}
\right)
\nonumber\\
&\hspace{1cm}
+\frac{4}{dn(n-2)}
\sum_{m,l,p=1}^{d^2-1}
d_{mlp}
t_m^{(i)}t_l^{(j)}t_p^{(k)},
\\
\Delta^{(ijk)}(\rho_L)
&=
\sum_b r_b
\left[
\frac{1}{nd^2}
\left(
t_b^{(i)}
+t_b^{(j)}
+t_b^{(k)}
\right)
\right.
\nonumber\\
&\hspace{1cm}
\left.
-\frac{2}{dn(n-2)}
\sum_{m=1}^{d^2-1}
\left(
t_b^{(i)}t_m^{(j)}t_m^{(k)}
+t_m^{(i)}t_b^{(j)}t_m^{(k)}
+t_m^{(i)}t_m^{(j)}t_b^{(k)}
\right)
\right].
\end{align}
\end{widetext}

\paragraph{Three-site erasure bound.}

After identifying all triples of physical sites using the $3$-transitive symmetry, the complementary channel~\eqref{eq:dual_to_arbitrary_multiqudit} takes the form
$$
\widehat{\mathcal N\circ\mathcal E}(\rho_L)
\cong
\omega_{\mathrm{flag}}
\otimes
\rho^{(i_0j_0k_0)}(\rho_L),
$$
where $i_0,j_0,k_0$ are fixed distinct indices and
$$
\omega_{\mathrm{flag}}
=
\sum_{i<j<k}
p_{ijk}
\ket{ijk}\bra{ijk}_{F_E}.
$$
The appropriate constant channel is obtained by keeping the input-independent part of the three-site reduced state:
$$
\Lambda_0(\rho_L)
=
\operatorname{Tr}(\rho_L)\;
\omega_{\mathrm{flag}}
\otimes
\tau^{(i_0j_0k_0)} .
$$
Unlike in the one-site case, this state is not simply maximally mixed; it contains nontrivial few-body terms.

The entanglement fidelity is computed explicitly in Proposition~\ref{prop:3_transitive_fidelity}:
\begin{equation}\label{eq:fidelity_3}
F\left(\widehat{\mathcal N\circ\mathcal E},\Lambda_0\right)
=
1
-
\frac{3(d^2-1)}{8n^2}
+
O_{d}(n^{-3}) .
\end{equation}
Together with the complementary-channel characterization of AQEC, this gives the following result.

\begin{theorem}
Let $\mathcal E$ be the $\mathfrak{su}(d)$-covariant encoding from Theorem~\ref{thrm:main_3_transitive_covariant_code}. Then, for the flagged three-site erasure channel
$$
\mathcal N(\sigma)
=
\sum_{1\le i<j<k\le n}
p_{ijk}
\ket{ijk}\bra{ijk}_F
\otimes
\ket e\bra e_{ijk}
\otimes
\operatorname{Tr}_{ijk}(\sigma),
$$
there exists a recovery channel $\mathcal R$ such that
$$
d(\mathcal R\mathcal N\mathcal E,\mathrm{id})
\le
\frac{\sqrt{3(d^2-1)}}{2\sqrt 2}
\frac1n
+
O_{d}(n^{-2}) .
$$
In particular, the code is $O(\sqrt{d^2-1}/n)$-correctable against flagged erasure on three sites.
\end{theorem}

\begin{proof}
By Eq.~\eqref{eq:fidelity_3} and the definition of the channel distance in Eq.~\eqref{eq:distance_def},
$$
d\left(\widehat{\mathcal N\circ\mathcal E},\Lambda_0\right)
=
\frac{\sqrt{3(d^2-1)}}{2\sqrt 2}
\frac1n
+
O_{d}(n^{-2}) .
$$
The complementary-channel theorem, Eq.~\eqref{eq:beny}, gives
$$
\min_{\mathcal R}
d(\mathcal R\mathcal N\mathcal E,\mathrm{id})
=
\min_{\Lambda,\mathrm{constant}}
d\left(\widehat{\mathcal N\circ\mathcal E},\Lambda\right)
\le
d\left(\widehat{\mathcal N\circ\mathcal E},\Lambda_0\right),
$$
and the stated bound follows.
\end{proof}

\subsection{Flagged erasure on an arbitrary fixed number of qudits}\label{subsec:main_multisite_erasure_noise}

We now extend the erasure analysis to an arbitrary fixed number $k$ of erased sites, completing the discussion of Section~\ref{subsec:main_muti_erasure_noise}. 
As before, $\mathcal H_L=V_{\omega_1}$ and $\mathcal H_P=(V_{\omega_1})^{\otimes n}$, with the transversal representation
$$
t_a
\longmapsto
\sum_{i=1}^n t_a^{(i)},
\qquad
a=1,\ldots,d^2-1.
$$

\paragraph{Why exact symmetry fails for general $k$.}

For $k\le3$, Section~\ref{subsec:main_muti_erasure_noise} chose the multiplicity vector invariant under a subgroup of $S_n$ 
that is exactly $k$-transitive, so that all $k$-site reduced states coincide exactly. As explained there, this route stops at $k=3$: 
for $k>3$ the classification of finite multiply transitive permutation groups leaves only the alternating group $A_n$, 
which is too large to leave a nonzero invariant multiplicity subspace, and the Mathieu groups, which occur only at $n\in\{11,12,23,24\}$ and so give no family with $n\to\infty$~\cite{GorensteinLyonsSolomon1994}. 
For fixed $k>3$ we therefore abandon exact symmetry and ask that a \emph{typical} member of a suitable family of codes have $k$-site reduced states that are all close to one another, 
rather than exactly equal, and we show that such a typical code can be found and sampled efficiently.

\paragraph{A background-singlet family of codes.}

Let $\{e_1,\ldots,e_d\}$ be the standard basis of $V_{\omega_1}=\mathbb C^d$, and let
$$
\ket{\zeta_d}
:=
\frac1{\sqrt{d!}}
\sum_{\pi\in S_d}
\operatorname{sgn}(\pi)\,
e_{\pi(1)}\otimes\cdots\otimes e_{\pi(d)}
$$
be the normalized totally antisymmetric state of $d$ qudits. This is the unique $SU(d)$-invariant vector in $(V_{\omega_1})^{\otimes d}$, 
and the reduced state of any $s\le d$ of its tensor factors is proportional to a projector onto an antisymmetric subspace (Lemma~\ref{lem:singlet-red-states}). 
In particular, every one-site reduced state is the maximally mixed state $\tau:=I_d/d$.

Throughout this section $n=qd+1$; see Section~\ref{subsec:main_muti_erasure_noise} for the motivation of this choice.
We partition the $n=qd+1$ physical positions into a distinguished position $0$ and $q$ blocks $B_1,\ldots,B_q$ of size $d$, and define the seed isometry
$$
\mathcal{V}_0\ket\psi
:=
\ket\psi_0\otimes\ket{\zeta_d}_{B_1}\otimes\cdots\otimes\ket{\zeta_d}_{B_q}.
$$
This isometry places the logical qudit at site $0$ and fills the rest of the system with $q$ independent $SU(d)$-singlets, which we call \textit{background sites}. It is exactly $SU(d)$-covariant (Proposition~\ref{prop:seed-covariance}). 
Identifying $\mathcal{V}_0$ through the Schur--Weyl inclusion $W_{\omega_1}:V_{\omega_1}\otimes M_{\omega_1}\to(V_{\omega_1})^{\otimes n}$ gives a unit vector $\ket{m_0}\in M_{\omega_1}$ with $\mathcal{V}_0=W_{\omega_1}(I_{V_{\omega_1}}\otimes\ket{m_0})$. 
For a general unit vector $\ket m\in M_{\omega_1}$, define
$$
\mathcal{V}_m
:=
W_{\omega_1}(I_{V_{\omega_1}}\otimes\ket m),
\qquad
\mathcal E_m(X):=\mathcal{V}_mX\mathcal{V}_m^\dagger,
$$
Every $\mathcal{V}_m$ is an exactly $SU(d)$-covariant isometric encoding, and $\mathcal{V}_0=\mathcal{V}_{m_0}$ is one particular member of this family, 
indexed by the multiplicity vector $\ket m$.

Averaging the encoding over a Haar-random $\ket m\in M_{\omega_1}$ reproduces the permutation twirl of the seed encoding:
$$
\mathbb E_{m\sim\mathrm{Haar}(M_{\omega_1})}
\bigl[\mathcal{V}_mX\mathcal{V}_m^\dagger\bigr]
=
\frac1{n!}
\sum_{\pi\in S_n}
P_\pi\,\mathcal{V}_0X\mathcal{V}_0^\dagger P_\pi^\dagger,
$$
proved in Proposition~\ref{prop:orbit-twirl}. This identity is the source of the \textit{full} permutation symmetry 
that an exactly $k$-transitive subgroup enforced directly in the construction for $k\le3$.

\paragraph{The mean $k$-site channel.}

For a subset $S\subseteq\{1,\ldots,n\}$ of size $k$, define the erased-subsystem channel
$$
\Phi_{S,m}(X)
:=
\operatorname{Tr}_{\hat S}(\mathcal{V}_mX\mathcal{V}_m^\dagger).
$$
This is the reduced-state map $\rho^{(S)}$ of Eq.~\eqref{eq:reduced_state}, written here as a linear map and for a general multiplicity 
vector $m$. Its Haar average over $m$ defines the mean $k$-site channel $\overline\Phi_{k,n}:=\mathbb E_m\Phi_{S,m}$, which by the orbit-twirl 
identity above is, after the canonical identification of retained sites, independent of $S$. Proposition~\ref{prop:exact-mean-channel} evaluates it exactly:
$$
\overline\Phi_{k,n}(X)
=
\frac{n-k}n\operatorname{Tr}(X)\,\beta_{k,N}
+
\frac1n\sum_{j=1}^k
X^{(j)}\otimes\beta_{k-1,N}^{(\widehat j)},
$$
where $N:=n-1=qd$ is the number of background sites and $\beta_{\ell,N}$ is the $\ell$-site reduced state of the permutation twirl of the background, 
and Proposition~\ref{prop:beta-partition} evaluates it as a sum over set partitions of the $\ell$ sample labels. 
Two sampled sites fall into the same singlet block with probability $O_{d,k}(1/n)$, so $\beta_{\ell,N}$ is close to $\tau^{\otimes\ell}$ in trace norm for fixed 
$\ell$ (Lemma~\ref{lem:collision}), and $\overline\Phi_{k,n}(X)$ is therefore close to $\operatorname{Tr}(X)\,\tau^{\otimes k}$ up to a logical-dependent correction of order $1/n$, 
exactly as in the one-, two-, and three-site cases above.

Let 
\begin{equation}\label{eq:lambda_k_n}
\lambda_{k,n}:=\overline\Phi_{k,n}(\tau), \quad \Lambda_{k,n}(X):=\operatorname{Tr}(X)\lambda_{k,n}.
\end{equation}
Both $\overline\Phi_{k,n}$ and $\Lambda_{k,n}$ are $SU(d)$-covariant, since every $\beta_{\ell,N}$ is $SU(d)$-invariant, so by Lemma~\ref{fidelity_opt} 
the worst-case entanglement fidelity between them is attained on the maximally mixed logical input, and a direct expansion of the root fidelity around this point 
(Proposition~\ref{prop:k-loc_fidelity}) gives
$$
F(\overline\Phi_{k,n},\Lambda_{k,n})
=
1-\frac{k(d^2-1)}{8n^2}
+O_{d,k}(n^{-3}).
$$
At $k=1$ and $k=3$ this reproduces, term by term, the coefficients found above by the exact symmetry-based calculation, Eq.~\eqref{eq:fidelity_sud} at $r=1$ and Eq.~\eqref{eq:fidelity_3}.

\paragraph{Haar-typical multiplicity vectors.}

The mean channel $\overline\Phi_{k,n}$ is an average over the whole family $\{\mathcal E_m\}$, not the channel produced by any single code. 
Turning this into a statement about one encoding requires a single unit vector $\ket m\in M_{\omega_1}$ for which \emph{every} one of the $\binom nk$ erased-subsystem channels $\Phi_{S,m}$ is 
simultaneously close to $\overline\Phi_{k,n}$ in diamond norm $\|\cdot\|_\diamond$, the completely bounded trace norm on channels. 
Because the multiplicity space $M_{\omega_1}\cong[q+1,q^{d-1}]$ has dimension $D_q$ given by Lemma~\ref{lem:dimension_formula}, 
which grows exponentially in $n$ for fixed $d$ by the lower bound of Lemma~\ref{lem:dimension_bound}, a Haar-random unit vector makes every 
matrix element of $\Phi_{S,m}$ concentrate sharply around its mean, and a union bound over the finitely many matrix elements and erased sets shows that a single random draw 
is simultaneously good for all of them, with probability approaching one exponentially fast in $n$.

\begin{theorem}\label{thrm:main_haar_typical}
Let $r_k:=d^{k+1}$ and $M_{n,k}:=\binom nk d^{2k+2}$. Let $\ket m$ be Haar-random in $M_{\omega_1}$. Then
$$
\mathbb P\left[
\max_{|S|=k}
\|\Phi_{S,m}-\overline\Phi_{k,n}\|_\diamond
>r_k^{3/2}D_q^{-1/4}
\right]
\le
\frac{M_{n,k}D_q^{1/2}}{D_q+1}.
$$
For fixed $d,k$, the right-hand side is $e^{-\Omega_d(n)}$.
\end{theorem}

This is proved as Theorem~\ref{thm:haar-uniform}. The argument writes each matrix element of $\Phi_{S,m}-\overline\Phi_{k,n}$ as a quadratic form
in $\ket m$ (Lemma~\ref{lem:quadratic-form}), bounds its variance using the Haar second-moment formula (Lemma~\ref{lem:haar-second-moment}),
and converts the resulting entrywise bound into a diamond-norm bound through the unnormalized Choi matrix (Lemma~\ref{lem:entry-to-diamond}), 
before taking a union bound over the $M_{n,k}$ matrix entries appearing across all $\binom nk$ erased sets.

A useful way to read this result is that the error-correcting property does not rely on fine-tuning a special multiplicity vector. 
Covariance first restricts the encoding to a fixed isotypic component, and the remaining freedom is the choice of a direction in its multiplicity space. 
Theorem~\ref{thrm:main_haar_typical} shows that, with respect to the natural Haar measure on this space, the set of multiplicity vectors
for which some erased subsystem reveals an atypically large amount of logical information has exponentially small measure. 
Specially aligned and poorly performing vectors certainly exist: the seed vector $\ket{m_0}$ itself is one, since $\mathcal{V}_0$ 
leaves the logical qudit unencoded at site $0$, so erasing that single site reveals the logical state completely. 
Such vectors are non-generic. Approximate decoupling from any fixed-size erased subsystem is therefore a typical property of the covariant sector rather than 
a feature of an isolated construction.

\paragraph{Multi-site erasure bound for a Haar-random code.}
Let $\mathcal N$ be the flagged $k$-site erasure channel~\eqref{eq:erasure_noise} with a distribution $p=(p_S)_{|S|=k}$ 
on the $k$-element subsets of $\{1,\ldots,n\}$. By Eq.~\eqref{eq:dual_to_erasure}, the complementary channel to $\mathcal N\circ\mathcal E_m$ is
$$
\widehat{\mathcal N\circ\mathcal E_m}(\rho_L)
=
\sum_{|S|=k}
p_S
\ket S\bra S_{F_E}
\otimes
\Phi_{S,m}(\rho_L).
$$
As in the three-site case, the environment can access the logical state only through the reduced states on the erased set, here $\Phi_{S,m}(\rho_L)$ for every possible $S$. 
The difference is that these reduced states are no longer exactly equal to one another. 
Theorem~\ref{thrm:main_haar_typical} controls all of them simultaneously by the single mean channel $\overline\Phi_{k,n}$, 
and Proposition~\ref{prop:k-loc_fidelity} controls how far that mean channel is from a constant one. 
Combining the two gives the erasure bound.

\begin{theorem}\label{thrm:main_multisite_erasure_haar}
Let $d\ge2$ and $k\ge1$ be fixed. For every sufficiently large $n=qd+1$, let $\ket m$ be Haar-random in $M_{\omega_1}$ and let $\mathcal E_m$ be as above. With probability at least $1-e^{-\Omega_d(n)}$ over $\ket m$, the following holds: for every distribution $p$ on the $k$-element subsets of $\{1,\ldots,n\}$ there exists a recovery channel $\mathcal R_p$ such that
$$
d\left(
\mathcal R_p\circ\mathcal N\circ\mathcal E_m,
\mathrm{id}
\right)
\le
\frac{\sqrt{k(d^2-1)}}{2\sqrt2}
\frac1n
+
O_{d,k}(n^{-2}).
$$
In particular, with probability at least $1-e^{-\Omega_d(n)}$ over the draw of $\ket m$, the code $\mathcal E_m$ is $O(\sqrt{k(d^2-1)}/n)$-correctable against flagged $k$-site erasure, uniformly over the erasure distribution $p$.
\end{theorem}

\begin{proof}
Condition on the event of Theorem~\ref{thrm:main_haar_typical}, which has probability at least $1-M_{n,k}D_q^{1/2}/(D_q+1)=1-e^{-\Omega_d(n)}$ and on which
$$
\max_{|S|=k}\|\Phi_{S,m}-\overline\Phi_{k,n}\|_\diamond
\le
r_k^{3/2}D_q^{-1/4}.
$$
We now introduce the two auxiliary channels
$$
\overline\Lambda(\rho_L)
:=
\sum_{|S|=k}
p_S
\ket S\bra S_{F_E}
\otimes
\overline\Phi_{k,n}(\rho_L),
$$
$$
\Lambda_0(\rho_L)
:=
\operatorname{Tr}(\rho_L)
\sum_{|S|=k}
p_S
\ket S\bra S_{F_E}
\otimes
\lambda_{k,n},
$$
where $\lambda_{k,n}$ is taken from Eq.~\eqref{eq:lambda_k_n}, so that $\Lambda_0$ is a constant channel 
in the sense of Eq.~\eqref{eq:constant_channel_def}. The three channels $\widehat{\mathcal N\circ\mathcal E_m}$, 
$\overline\Lambda$, and $\Lambda_0$ are block diagonal with respect to the same flag decomposition and carry the same classical weights $p_S$. 
Root fidelity of two such block-diagonal states is the $p$-weighted sum of the block root fidelities, 
so for every reference-assisted input the Fuchs--van de Graaf inequality $1-f(\rho,\sigma)\le\frac12\|\rho-\sigma\|_1$ gives, block by block,
$$
    1-F\left(\widehat{\mathcal N\circ\mathcal E_m},\overline\Lambda\right)
    \le
    \frac12
    \max_{|S|=k}
    \|\Phi_{S,m}-\overline\Phi_{k,n}\|_\diamond
    \le
    \frac12r_k^{3/2}D_q^{-1/4},
$$
while $d(\overline\Lambda,\Lambda_0)=d(\overline\Phi_{k,n},\Lambda_{k,n})$. By Eq.~\eqref{eq:distance_def}, $d(\cdot,\cdot)$ is the supremum over purified logical inputs 
of the Bures-type state distance $\sqrt{1-f}$, and a supremum of metrics is a metric, so the triangle inequality gives
$$
    d\left(\widehat{\mathcal N\circ\mathcal E_m},\Lambda_0\right)
    \le
    \sqrt{\tfrac12r_k^{3/2}D_q^{-1/4}}
    +
    d\left(\overline\Phi_{k,n},\Lambda_{k,n}\right).
$$
The first term is $e^{-\Omega_d(n)}$, because $D_q$ grows exponentially in $n$ by Lemma~\ref{lem:dimension_bound} while $r_k$ is fixed. By Eq.~\eqref{eq:mean-fidelity} and Eq.~\eqref{eq:distance_def}, the second term equals $\sqrt{k(d^2-1)/8}\,n^{-1}+O_{d,k}(n^{-2})$. 
The complementary-channel theorem, Eq.~\eqref{eq:beny}, gives
\begin{align*}
    \min_{\mathcal R_p}
    d\left(\mathcal R_p\circ\mathcal N\circ\mathcal E_m,\mathrm{id}\right)
    &=
    \min_{\Lambda\,\mathrm{constant}}
    d\left(\widehat{\mathcal N\circ\mathcal E_m},\Lambda\right)\\
    &\le
    d\left(\widehat{\mathcal N\circ\mathcal E_m},\Lambda_0\right),
\end{align*}
and the stated bound follows.
\end{proof}

\paragraph{Efficient sampling of multiplicity vectors.}

Theorem~\ref{thrm:main_haar_typical} shows that a typical multiplicity vector works, but sampling a Haar-random unit vector in $M_{\omega_1}$ naively requires about $D_q=e^{\Theta(n)}$ random bits, 
one for every amplitude in the Young--Yamanouchi basis of $M_{\omega_1}$ introduced in Section~\ref{sec:efficient_sampling} (any other convenient basis will do). 
The proof of Theorem~\ref{thrm:main_haar_typical} only ever uses second and fourth moments of the quadratic forms $\bra mA\ket m$ of Lemma~\ref{lem:quadratic-form}, 
so full independence of the amplitudes is unnecessary and four-wise independence of their phases suffices to repeat the same concentration argument. 
This lets an exponentially long random string be replaced by a seed of length $O(n\log d)$, from which every phase can be computed efficiently.

\begin{theorem}\label{thrm:main_efficient_sampling}
Fix $d$ and $k$, and let $n=qd+1$. There is an explicit randomized algorithm using $O(n\log d)$ random bits that outputs a seed $s$ specifying a flat multiplicity vector
$$
\ket{m_s}
=
\frac1{\sqrt{D_q}}
\sum_{T\in\operatorname{SYT}(q+1,q^{d-1})}
z_s(T)\ket T,
$$
where $z_s(T)\in\{1,-1,i,-i\}$ for all $T$, such that each phase $z_s(T)$ is computable in time $\operatorname{poly}(n,\log d)$, and
$$
\mathbb P_s\left[
\max_{|S|=k}
\|\Phi_{S,m_s}-\overline\Phi_{k,n}\|_\diamond
>r_k^{3/2}D_q^{-1/4}
\right]
\le
\frac{M_{n,k}}{\sqrt{D_q}}
=
e^{-\Omega_d(n)}.
$$
\end{theorem}

This is proved as Theorem~\ref{thrm:randomized_multiplicity_vector}. The construction labels every tableau $T$ by its row word and assigns it a phase using a random cubic polynomial over $\mathbb F_{2^h}$, $h=n\log d$, evaluated at the tableau's label.
The four coefficients of the polynomial form the seed $s$. A Vandermonde-determinant argument shows that the resulting phases are exactly four-wise independent, which is all the concentration bound of Theorem~\ref{thrm:main_haar_typical} requires. 
Only $4h=O(n\log d)$ random bits are used, and each phase is computable in $\operatorname{poly}(n,\log d)$ time.

This strengthens Theorem~\ref{thrm:main_haar_typical} in an operationally relevant way. 
Haar randomness is a useful benchmark because it shows that good multiplicity vectors are abundant, but writing down a generic Haar-random vector means specifying $D_q=e^{\Theta(n)}$ amplitudes, 
so by itself it does not give an efficient code-construction procedure. 
The bounded-independence construction replaces the full continuous randomness by a short classical seed from which every phase is generated efficiently. 
The resulting ensemble retains exactly the moment information the concentration argument uses, namely second and fourth moments of the quadratic forms in Lemma~\ref{lem:quadratic-form}, 
while consuming only $O(n\log d)$ random bits and allowing each phase to be evaluated in $\operatorname{poly}(n,\log d)$ time. 
In this sense the construction identifies how much randomness the error-correcting property actually needs.

\paragraph{Multi-site erasure bound for an efficiently sampled code.}

Let $\ket{m_s}$ be the multiplicity vector produced by Theorem~\ref{thrm:main_efficient_sampling} and let $\mathcal E_s(X):=\mathcal{V}_{m_s}X\mathcal{V}_{m_s}^\dagger$. 
The proof of Theorem~\ref{thrm:main_multisite_erasure_haar} uses the Haar distribution only through the diamond-norm bound of Theorem~\ref{thrm:main_haar_typical}, 
which Theorem~\ref{thrm:main_efficient_sampling} reproduces for $\ket{m_s}$ with the same threshold $r_k^{3/2}D_q^{-1/4}$ and a failure probability that
is again $e^{-\Omega_d(n)}$. The same argument therefore applies to this case with $\ket m$ replaced by $\ket{m_s}$.

\begin{corollary}\label{thrm:main_multisite_erasure}
Let $d\ge2$ and $k\ge1$ be fixed. For every sufficiently large $n=qd+1$, let $\ket{m_s}$ and $\mathcal E_s$ be as above. 
With probability at least $1-e^{-\Omega_d(n)}$ over the seed $s$, the following holds: for every distribution $p$ on the $k$-element subsets of 
$\{1,\ldots,n\}$ there exists a recovery channel $\mathcal R_p$ such that
$$
d\left(
\mathcal R_p\circ\mathcal N\circ\mathcal E_s,
\mathrm{id}
\right)
\le
\frac{\sqrt{k(d^2-1)}}{2\sqrt2}
\frac1n
+
O_{d,k}(n^{-2}).
$$
In particular, with probability at least $1-e^{-\Omega_d(n)}$ over the seed $s$, the code $\mathcal E_s$ is $O(\sqrt{k(d^2-1)}/n)$-correctable against flagged $k$-site erasure, 
uniformly over the erasure distribution $p$.
\end{corollary}

\begin{proof}
Condition on the event of Theorem~\ref{thrm:main_efficient_sampling} instead of that of Theorem~\ref{thrm:main_haar_typical}. 
It has probability at least $1-M_{n,k}D_q^{-1/2}=1-e^{-\Omega_d(n)}$ and gives the same bound
$$
\max_{|S|=k}\|\Phi_{S,m_s}-\overline\Phi_{k,n}\|_\diamond
\le
r_k^{3/2}D_q^{-1/4}.
$$
The rest of the proof of Theorem~\ref{thrm:main_multisite_erasure_haar} uses no other property of the multiplicity vector and applies unchanged.
\end{proof}

\subsection{Arbitrary flagged multi-qudit noise}\label{subsec:main_arb_3_noise}

We now extend the erasure results to arbitrary flagged noise acting on a fixed number $k$ of physical sites, the single-site case included. 
The argument does not depend on the value of $k$. The complementary channel~\eqref{eq:dual_to_arbitrary_multiqudit} 
for arbitrary flagged noise on a set $S$ depends on the logical input only through the reduced state $\rho^{(S)}(\rho_L)$ on that set. 
All that the argument uses is that this reduced state splits as
$$
\rho^{(S)}(\rho_L)
=
\tau_S
+
\Delta_S(\rho_L),
$$
with $\tau_S$ independent of the logical input and $\Delta_S=O_{d,k}(1/n)$ uniformly over the sets $S$ on which the noise can act. 
This is Lemma~\ref{lem:arb_noise_from_leakage}, which turns such a uniform bound into an $O_{d,k}(1/\sqrt n)$ correctability statement for arbitrary 
flagged noise supported on those sets.

What makes the reduced state alone sufficient is linearity. By Lemma~\ref{lem:dual_to_arb}, the complementary channel on the flag $S$ is the composition $\widehat{\mathcal N}_S\circ\rho^{(S)}$, 
so the environment sees the logical state only after it has passed through the reduced-state map, whatever the local noise does. 
Both maps are linear, so the splitting survives the composition,
$$
\widehat{\mathcal N}_S\left(\rho^{(S)}(\rho_L)\right)
=
\widehat{\mathcal N}_S(\tau_S)
+
\widehat{\mathcal N}_S\left(\Delta_S(\rho_L)\right),
$$
and the logical-dependent term keeps the $1/n$ prefactor carried by $\Delta_S$. 
Applying $\widehat{\mathcal N}_S$ cannot enlarge that term either, because the trace norm is contractive under channels. 
An unknown local noise channel therefore neither creates nor amplifies the dependence on the logical input, 
therefore it acts on a signal that is already suppressed by $1/n$, and the complementary channel approaches a constant one at the rate at which 
the reduced states do.

Two constructions of this paper supply the required uniform bound. For $k\le3$, the exactly transitive codes of Section~\ref{subsec:main_muti_erasure_noise} 
make the reduced states on all $k$-tuples identical up to the canonical identification of sites, and their logical-state-dependent part is controlled 
by Eq.~\eqref{eq:3q_reduced_state}. For arbitrary fixed $k$, the codes of Section~\ref{subsec:main_multisite_erasure_noise} give the same bound with
 the exact equality replaced by the diamond-norm estimate of Theorem~\ref{thrm:main_haar_typical}.

This allows one to compare the complementary channel of arbitrary channel with a constant channel obtained by replacing 
each reduced state by its input-independent part. The proof uses a trace-norm estimate together with the Fuchs--van de Graaf inequality, 
rather than an explicit fidelity calculation. The resulting bounds are stated in Theorems~\ref{thrm:3_transitive_scale} and~\ref{thrm:k_arb_scale}:

\begin{theorem}
Let $\mathcal H_L=V_{\omega_1}$ and $\mathcal H_P=(V_{\omega_1})^{\otimes n}$. Assume that $d=p^r$, where $p$ is prime and $r$ is a positive integer, and that $n-1=p^s$, 
where $s\ge r$ is a positive integer. Equip $\mathcal H_L$ with the fundamental representation of $\mathfrak{su}(d)$ and $\mathcal H_P$ with the transversal representation
$$
t_a
\longmapsto
\sum_{i=1}^n t_a^{(i)},
\qquad
a=1,\ldots,d^2-1,
$$
where $t_a^{(i)}$ denotes the fundamental action of $t_a$ on the $i$th physical space $V_{\omega_1}$. 
For sufficiently large $n$, there exists an $\mathfrak{su}(d)$-covariant encoding $\mathcal E$ with respect 
to these representations whose code space is invariant under the action of $\operatorname{PGL}(2,n-1)$. 
For any flagged noise channel supported on sets of at most three sites,
$$
\mathcal N(\sigma)
=
\sum_{1\le|S|\le3}
p_S
\ket S\bra S_F
\otimes
\left(
\mathcal N_S\otimes\mathrm{id}_{\hat S}
\right)(\sigma),
$$
with arbitrary channels $\mathcal N_S$ acting on the sites of $S$, there exists a recovery channel $\mathcal R$ such that
$$
d(\mathcal R\mathcal N\mathcal E,\mathrm{id})
=
O_{d}\left(\frac{1}{\sqrt n}\right).
$$
In particular, the code is $O_{d}(1/\sqrt n)$-correctable against arbitrary 
flagged noise on any set of at most three qudits, single-qudit noise included, uniformly over the distribution of noisy sets.
\end{theorem}

\begin{theorem}\label{thrm:main_k_arb_scale}
Let $d\ge2$ and $k\ge1$ be fixed, let $n=qd+1$, and let $\mathcal E_s$ be the efficiently sampled $SU(d)$-covariant encoding of Theorem~\ref{thrm:main_efficient_sampling}. 
With probability at least $1-e^{-\Omega_d(n)}$ over the seed $s$, the following holds: for every flagged $k$-qudit noise channel
$$
\mathcal N(\sigma)
=
\sum_{|S|=k}
p_S
\ket S\bra S_F
\otimes
\left(
\mathcal N_S\otimes\mathrm{id}_{\hat S}
\right)(\sigma),
$$
with arbitrary channels $\mathcal N_S$ acting on the erased sites, there exists a recovery channel $\mathcal R$ such that
$$
d(\mathcal R\mathcal N\mathcal E_s,\mathrm{id})
\le
\sqrt{\frac kn}
+
e^{-\Omega_d(n)} .
$$
In particular, with probability at least $1-e^{-\Omega_d(n)}$ over the seed $s$, 
the code $\mathcal E_s$ is $O(\sqrt{k/n})$-correctable against arbitrary flagged $k$-site noise, uniformly over the distribution of noisy sets.
\end{theorem}

The loss from the $1/n$ erasure scaling to the $1/\sqrt n$ scaling comes from the use of the trace-norm bound and the Fuchs--van de Graaf inequality 
in place of an explicit fidelity calculation. The argument applies to arbitrary local noise channels, but it is not expected to give sharp constants or optimal scaling in general. 
A direct fidelity calculation for more general noise models is left for future work.

\section{Covariant encoding for Analog Quantum Simulation}\label{sec:main_covar_analog}

We now discuss how the covariant AQECCs constructed in Sections~\ref{sec:main_su2} and~\ref{sec:main_sud} can be used for analog quantum simulation. 
The aim is to run the simulated evolution directly on the encoded state: 
each generator of the target dynamics has to be implemented on the code space by an operation the physical system already performs, 
here a sum of single-site generators, so that the encoding intertwines the simulated dynamics with the physical one. 
This is the covariance condition of Section~\ref{subsec:main_covar_enc}, and the Eastin--Knill theorem is what forces it to hold only approximately. 
Applied to every Hamiltonian on the simulated Hilbert space $\mathcal H$, it would require covariance with respect to all of $\mathfrak{su}(\mathcal H)$.

In analog simulation, however, one usually does not need to implement all unitaries on the full Hilbert space. Instead, the available evolutions are generated by a restricted set of Hamiltonians, and the Lie algebra generated by these Hamiltonians is the dynamical Lie algebra (DLA). In many symmetric physical systems, the DLA is contained in an algebra of the form
$$
\mathfrak{su}(d_1)
\oplus
\cdots
\oplus
\mathfrak{su}(d_k)
\oplus
\mathfrak{u}(1)^{\oplus(r-1)} .
$$
An AQECC covariant with respect to this smaller algebra therefore makes all dynamics generated within it act transversally on the encoded space.

In this section, we construct such codes by building on the $SU(d)$-covariant AQECCs from the previous sections. We call the resulting construction a block encoding. We then analyze its performance under two-qudit erasure noise, including both in-block and inter-block erasures. For a schematic illustration see Fig.~\ref{fig:theme}(d).

\subsection{Dynamical Lie algebra and Symmetries}\label{subsec:main_DLA}

\paragraph{Symmetric analog dynamics.}

Consider a quantum system consisting of $N$ qubits, with Hilbert space
$$
\mathcal H=(\mathbb C^2)^{\otimes N},
\qquad
\dim\mathcal H=d=2^N .
$$
The dynamics are governed by a time-dependent Hamiltonian
\begin{equation}\label{eq:system_ham}
H(t)
=
H_0
+
\sum_{j=1}^m u_j(t)H_j,
\end{equation}
where $H_0$ is the intrinsic drift Hamiltonian, ${H_j}_{j=1}^m$ are control Hamiltonians, and $u_j(t)\in\mathbb R$ are tunable control fields.

In this section, we focus on \textit{global controls}. By this we mean controls that act collectively on the physical sites and preserve the relevant permutation symmetry. 
Collective spin operators give a standard example:
$$
\sum_{i=1}^N \sigma_\alpha^{(i)},
\qquad
\alpha=x,y,z .
$$
Such controls are invariant under permutations of the physical sites. In Section~\ref{sec:main_univ_analog}, we extend the discussion to \textit{local controls}, 
which act on individual sites and generally break the permutation symmetry.

The set of reachable unitary transformations is determined by the \textit{dynamical Lie algebra}~\cite{DAlessandro2021, Albertini2003}
$$
\mathrm{Lie}_{\mathbb R}
\{iH_0,iH_1,\ldots,iH_m\}
\subseteq
\mathfrak{su}(\mathcal H).
$$
Here $\mathfrak{su}(\mathcal H)$ denotes the Lie algebra of traceless skew-Hermitian linear operators on $\mathcal H$. After choosing a basis of $\mathcal H$, 
this is the usual matrix Lie algebra $\mathfrak{su}(d)$.

The DLA depends on the specific Hamiltonians in Eq.~\eqref{eq:system_ham}. In general, determining it can be difficult. 
In many symmetric systems, however, the DLA is contained in an invariant Lie algebra associated with a finite site-permutation symmetry. 
This invariant algebra has a canonical block structure, which is the starting point for the block-encoding construction below.

\paragraph{Invariant Lie algebra.}

Let $G\subseteq S_N$ be a finite group acting on $\mathcal H$ by permuting the qubit sites:
$$
P_g
\ket{i_1}\otimes\cdots\otimes\ket{i_N}
=
\ket{i_{g^{-1}(1)}}\otimes\cdots\otimes\ket{i_{g^{-1}(N)}} .
$$
Here $G$ acts on the logical simulation Hilbert space. This should be distinguished from the finite permutation groups used earlier in Eq.~\eqref{eq:group_act_on_phys_space}, 
which acted on the physical subsystems of a code in order to impose symmetry constraints on the code space.

The induced action on operators is by conjugation, as in Eq.~\eqref{eq:group_act_oper}. The $G$-invariant linear operators are
$$
\operatorname{End}(\mathcal H)^G
:=
{
X\in\operatorname{End}(\mathcal H):
P_gXP_g^\dagger=X
\text{ for all } g\in G
}.
$$
The corresponding \textit{invariant Lie algebra}~\cite{DAlessandroHartwig2021, Albertini2020} is
$$
\mathfrak{su}(\mathcal H)^G
:=
\operatorname{End}(\mathcal H)^G
\cap
\mathfrak{su}(\mathcal H)
\subseteq
\mathfrak{su}(\mathcal H).
$$

In physical examples, $G$ is usually a site-permutation symmetry of the Hamiltonians. For instance, if the drift and control Hamiltonians have the form
$$
\begin{gathered}
H_x
=
\sum_{i=1}^N\sigma_x^{(i)},
\qquad
H_z
=
\sum_{i=1}^N\sigma_z^{(i)},
\\
H_{zz}
=
\sum_{(i,j)\in E(\Gamma)}
\sigma_z^{(i)}\sigma_z^{(j)},
\end{gathered}
$$
where $\Gamma$ is an interaction graph, then the site-permutation symmetry is the automorphism group of $\Gamma$. 
For a complete graph this group is $S_N$, while for a cycle graph it contains the dihedral symmetry $D_N$. 
Table~\ref{tab:physical_symmetry_groups} lists several standard examples of spin systems with finite site-permutation symmetries.

\begin{table*}[t]
\centering
\label{tab:physical_symmetry_groups}
\renewcommand{\arraystretch}{1.25}
\setlength{\tabcolsep}{3pt}

\begin{tabular}{llll}
\hline
\parbox[t]{0.19\textwidth}{Physical system}
&
\parbox[t]{0.38\textwidth}{Controlled Hamiltonian $ H(t) = H_{0} + \sum_{j=1}^k u_j(t) H_j$}
&
\parbox[t]{0.22\textwidth}{Site-permutation symmetry}
&
\parbox[t]{0.13\textwidth}{References}
\\
\hline

\parbox[t]{0.19\textwidth}{Uniform open spin chain}
&
\parbox[t]{0.38\textwidth}{
\[
    \sum_{i=1}^{N-1} \sigma_z^{(i)} \sigma_z^{(i+1)}
    +u_x(t)S_x +u_z(t) S_z
\]
}
&
\parbox[t]{0.22\textwidth}{
Reflection symmetry \(C_2\cong\mathbb Z_2\)
}
&
\parbox[t]{0.13\textwidth}{
\cite{Wang2012,Wang2016,Wiersema2024}
}
\\[1.2em]

\parbox[t]{0.19\textwidth}{Uniform spin ring}
&
\parbox[t]{0.38\textwidth}{
\[
\sum_{i=1}^{N} \sigma_z^{(i)} \sigma_z^{(i+1)}
    +u_x(t)S_x +u_z(t)S_z, \quad \sigma_z^{(N+1)}=\sigma_z^{(1)}
\]
}
&
\parbox[t]{0.22\textwidth}{
Dihedral group \(D_N\); translations give \(C_N\)
}
&
\parbox[t]{0.13\textwidth}{
\cite{Wang2012,Kokcu2024,Wiersema2024}
}
\\[1.2em]

\parbox[t]{0.19\textwidth}{Spin network on a graph \(\Gamma\)}
&
\parbox[t]{0.38\textwidth}{
\[
\sum_{(i,j)\in E(\Gamma)} J_{ij} \;\sigma_z^{(i)} \sigma_z^{(j)}
    +u_x(t)S_x +u_z(t)S_z
\]
}
&
\parbox[t]{0.22\textwidth}{
Weighted graph automorphism group \(\operatorname{Aut}(\Gamma,J)\)
}
&
\parbox[t]{0.13\textwidth}{
\cite{Wang2012,Chen2020,Kokcu2024,Wiersema2024}
}
\\[1.2em]

\parbox[t]{0.19\textwidth}{Fully connected collective-spin model}
&
\parbox[t]{0.38\textwidth}{
\[
S_z^2+u_x(t)S_x+u_z(t)S_z
\]
}
&
\parbox[t]{0.22\textwidth}{
Full permutation group \(S_N\)
}
&
\parbox[t]{0.13\textwidth}{
\cite{Chen2020, Wiersema2024}
}
\\[1.2em]

\parbox[t]{0.19\textwidth}{Tavis--Cummings model}
&
\parbox[t]{0.38\textwidth}{
\[
\omega(t) a^\dagger a+\frac{\hbar\Omega(t)}{2} S_z
+\frac{g(t)\hbar{\sqrt N}}2(a+a^\dagger)S_x
\]
}
&
\parbox[t]{0.22\textwidth}{
Permutation symmetry \(S_N\) of identical atoms
}
&
\parbox[t]{0.13\textwidth}{
\cite{Tavis1968}
}
\\[1.2em]

\parbox[t]{0.19\textwidth}{Rydberg atom array with uniform drive}
&
\parbox[t]{0.38\textwidth}{
\[
\frac{\Omega(t)}{2}S_x
-\Delta(t)\sum_i n_i
+\sum_{i<j}V_{ij}n_in_j
\]
}
&
\parbox[t]{0.22\textwidth}{
Geometric automorphisms of the array preserving \(V_{ij}\)
}
&
\parbox[t]{0.13\textwidth}{
\cite{Browaeys2020,Daley2022}
}
\\

\hline
\end{tabular}
\caption{
Examples of spin Hamiltonian systems with finite site-permutation symmetries. 
In each row, both the drift Hamiltonian and the listed global controls are invariant under the corresponding group \(G\subseteq S_N\). 
Hence, the DLA generated by these Hamiltonians is contained in
\(\mathfrak{su}(\mathcal H)^G\). Here $S_{\alpha}=\sum_i \sigma_{\alpha}^{(i)}, \; \alpha=x, y, z$, and $n_i=\left(\sigma_z^{(i)}+1\right) / 2$.
}
\end{table*}

If all Hamiltonians in Eq.~\eqref{eq:system_ham} are invariant under the action of $G$, then
$$
iH_0,iH_1,\ldots,iH_m
\in
\mathfrak{su}(\mathcal H)^G .
$$
Since $\mathfrak{su}(\mathcal H)^G$ is closed under commutators, the DLA is contained in the invariant Lie algebra:
$$
\mathrm{Lie}_{\mathbb R}
\{iH_0,iH_1,\ldots,iH_m\}
\subseteq
\mathfrak{su}(\mathcal H)^G .
$$
Thus $\mathfrak{su}(\mathcal H)^G$ provides a symmetry-determined ambient algebra for the reachable dynamics. 
We do not attempt to characterize when the DLA is equal to $\mathfrak{su}(\mathcal H)^G$. 
This controllability question depends on the detailed structure of the Hamiltonians and is known only in special cases; see Refs.~\cite{Albertini2020, Kokcu2024, Wiersema2024}.

\paragraph{Block structure from representation theory.}

As shown in Lemma~\ref{lem:reductive_decomposition}, the invariant Lie algebra has the reductive decomposition
\begin{equation}\label{eq:inv_alg_decomp}
\mathfrak{su}(\mathcal H)^G
\cong
\mathfrak{su}(d_1)
\oplus
\cdots
\oplus
\mathfrak{su}(d_k)
\oplus
\mathfrak{u}(1)^{\oplus(r-1)} .
\end{equation}
This decomposition follows from the decomposition of $\mathcal H$ into irreducible representations of $G$:
\begin{equation}\label{eq:H_G_irrep_decomp}
\mathcal H
\cong
\bigoplus_{\alpha\in\mathcal I_{\mathcal H}}
V_\alpha\otimes M_\alpha .
\end{equation}
Here $V_\alpha$ is an irreducible representation of $G$, $M_\alpha$ is the corresponding multiplicity space, and $\mathcal I_{\mathcal H}$ is the set of irreducible representations of $G$ that appear in $\mathcal H$.

By definition, $\mathfrak{su}(\mathcal H)^G$ consists of traceless skew-Hermitian operators on $\mathcal H$ that commute with the action of $G$. 
Schur's lemma implies that every such operator acts as the identity on each irrep factor $V_\alpha$ and acts nontrivially only on the corresponding multiplicity space $M_\alpha$. 
Equivalently,
$$
\operatorname{End}(\mathcal H)^G
\cong
\bigoplus_{\alpha\in\mathcal I_{\mathcal H}}
I_{V_\alpha}
\otimes
\operatorname{End}(M_\alpha).
$$
After imposing skew-Hermiticity and the global tracelessness condition, this gives
$$
\mathfrak{su}(\mathcal H)^G
\cong
\bigoplus_{\alpha\in\mathcal I_{\mathcal H}}
\mathfrak{su}(M_\alpha)
\oplus
\mathfrak{u}(1)^{\oplus(|\mathcal I_{\mathcal H}|-1)} .
$$
Thus, if $m_\alpha:=\dim M_\alpha$, then
$$
\mathfrak{su}(\mathcal H)^G
\cong
\mathfrak{su}(m_1)
\oplus
\cdots
\oplus
\mathfrak{su}(m_{|\mathcal I_{\mathcal H}|})
\oplus
\mathfrak{u}(1)^{\oplus(|\mathcal I_{\mathcal H}|-1)} .
$$
This identifies the parameters in Eq.~\eqref{eq:inv_alg_decomp}: the numbers $d_i$ are the dimensions of the multiplicity spaces associated with 
the irreducible representations of $G$ appearing in $\mathcal H$, and $r=|\mathcal I_{\mathcal H}|$. In what follows, we will also write $d_i=m_i$.

\subsection{\texorpdfstring{$\left(\oplus_i \mathfrak{s}\mathfrak{u}\left(d_i\right)\right)
\oplus\left(\oplus_j \mathfrak{u}(1)\right)$}{(sum i su(d i)) plus (sum j u(1))}-covariant block encoding}
\label{subsec:main_covar_block_enc}

A fully $\mathfrak{su}(\mathcal H)$-covariant encoding would make every Hamiltonian on $\mathcal H$ act transversally on the code space. 
For analog simulation with symmetries, this is usually more covariance than is needed. If $\dim\mathcal H=d$, then a nontrivial transversal implementation of $SU(d)$ 
requires each physical subsystem to carry a nontrivial representation of $SU(d)$. Since the smallest nontrivial irreducible representation of $SU(d)$ is the fundamental 
representation, the local physical dimension is at least $d$. Thus, a fully $\mathfrak{su}(\mathcal H)$-covariant encoding inherits the full Hilbert-space dimension of 
the simulated system.

In the symmetry-adapted setting of Section~\ref{subsec:main_DLA}, the reachable dynamics are contained in the smaller algebra
$$
\mathfrak{su}(d_1)
\oplus
\cdots
\oplus
\mathfrak{su}(d_k)
\oplus
\mathfrak{u}(1)^{\oplus(r-1)} .
$$
It is therefore sufficient to construct a code covariant with respect to this algebra, rather than with respect to all of $\mathfrak{su}(\mathcal H)$. 
We do this by encoding each non-abelian block separately, using the $SU(d)$-covariant codes constructed above. 
This gives a block encoding adapted to the symmetry decomposition of the invariant algebra containing the DLA.

\paragraph{Definition of the block encoding.}

For the $a$th non-abelian block, let $V_{d_a,\lambda_a}$ be an irreducible representation of $\mathfrak{su}(d_a)$, chosen as in Section~\ref{sec:main_sud}, and define
$$
A_a
:=
(V_{d_a,\lambda_a})^{\otimes n_a},
\qquad
a=1,\ldots,k .
$$
For the abelian factors, choose physical spaces
$$
B_s
:=
(W_s)^{\otimes n_{k+s}},
\qquad
s=1,\ldots,r-1 .
$$
The total physical Hilbert space of the block encoding is
$$
\mathcal H_P^{\mathrm{block}}
:=
A_1
\otimes
\cdots
\otimes
A_k
\otimes
B_1
\otimes
\cdots
\otimes
B_{r-1}.
$$

The physical representation of
$$
\mathfrak{su}(d_1)
\oplus
\cdots
\oplus
\mathfrak{su}(d_k)
\oplus
\mathfrak{u}(1)^{\oplus(r-1)}
$$
is defined blockwise. For
$$
(t_1,\ldots,t_k,u_1,\ldots,u_{r-1})
\in
\mathfrak{su}(d_1)
\oplus
\cdots
\oplus
\mathfrak{su}(d_k)
\oplus
\mathfrak{u}(1)^{\oplus(r-1)},
$$
we set
\begin{equation}\label{eq:block_enc_def}
\begin{aligned}
&(t_1,\ldots,t_k,u_1,\ldots,u_{r-1})
\longmapsto
\\
&\qquad
\pi_{d_1,\lambda_1}(t_1)
+
\cdots
+
\pi_{d_k,\lambda_k}(t_k)\\
&\qquad+
R_1(u_1)
+
\cdots
+
R_{r-1}(u_{r-1}) .
\end{aligned}
\end{equation}
Here each term acts nontrivially only on its corresponding block and as the identity on all other blocks. The non-abelian block actions are transversal:
\begin{equation}
\pi_{d_j,\lambda_j}(t_j)
=
\sum_{i=1}^{n_j}
\pi_{d_j,\lambda_j}^{(i+\sum_{s=1}^{j-1}n_s)}(t_j).
\end{equation}
The operator $R_s(u_s)$ denotes the chosen representation of the generator $u_s\in\mathfrak{u}(1)$ on the abelian block $B_s=(W_s)^{\otimes n_{k+s}}$.

We do not analyze $U(1)$-covariant codes in this work. The abelian blocks should therefore be viewed as placeholders 
for any suitable $U(1)$-covariant construction. Such constructions can be obtained, for example, from Ref.~\cite{Lin2025}, 
and can be used to specify the choices of $W_s$ and $R_s$.

The resulting code space is a product code:
\begin{equation}\label{eq:product_code}
C_1
\otimes
\cdots
\otimes
C_{k+r-1}.
\end{equation}
For $i\le k$, the factor $C_i$ is the code space of an $\mathfrak{su}(d_i)$-covariant encoding. For $k<i\le k+r-1$, the factor $C_i$ is the code space of a 
$\mathfrak{u}(1)$-covariant encoding. This product structure is natural from representation theory, because irreducible representations of a direct sum of Lie algebras 
are tensor products of irreducible representations of the summands.

As a simple example, consider an $\mathfrak{su}(2)\oplus\mathfrak{su}(3)$-covariant encoding with
$$
\mathcal H_P^{\mathrm{block}}
=
(V_j)^{\otimes n}
\otimes
(V_{3,\omega_1})^{\otimes n},
$$
where $V_j$ is the spin-$j$ representation of $\mathfrak{su}(2)$ and $V_{3,\omega_1}$ is the fundamental representation of $\mathfrak{su}(3)$. If $\lambda_1$ denotes the first 
standard Gell-Mann matrix, then the generator $(\sigma_x,\lambda_1)$ is represented as
$$
\sum_{i=1}^n J_x^{(i)}
+
\sum_{i=1}^n \lambda_1^{(i+n)} .
$$

\paragraph{Encoding symmetry-adapted Hamiltonians.}

We now relate this block encoding to analog simulation. The Hamiltonians ${H_0,H_1,\ldots,H_m}$ need not be block diagonal in the computational basis. 
However, if $G$ is a site-permutation symmetry of the Hamiltonians, then they become block diagonal in a symmetry-adapted basis associated with 
the decomposition~\eqref{eq:H_G_irrep_decomp}. A systematic construction of such a basis using Young symmetrizers was given in Ref.~\cite{DAlessandroHartwig2021}.

Let $U_{\mathrm{sym}}$ denote the corresponding symmetry-adapted transform. Then every $G$-invariant Hamiltonian has the block form
$$
iH
\longmapsto
U_{\mathrm{sym}}^\dagger iH U_{\mathrm{sym}}
=
\underbrace{i\widetilde H_{\alpha}}_{\in\mathfrak{su}(d_1)}
\oplus
\underbrace{i\widetilde H_{\beta}}_{\in\mathfrak{su}(d_2)}
\oplus
\cdots .
$$
Once this decomposition is known, each block Hamiltonian is encoded using the corresponding block representation in Eq.~\eqref{eq:block_enc_def}. 
In Example~\ref{ex:Sn_block_diagonalization}, we work this out for $G=S_3$. In that case the relevant invariant Lie algebra is
$$
\mathfrak{su}(4)
\oplus
\mathfrak{su}(2)
\oplus
\mathfrak{u}(1)\subseteq \mathfrak{su}(2^3),
$$
and we give the explicit block diagonal form of
\begin{equation}\label{eq:s3_ex_hams}
\begin{gathered}
H_x
=
\sigma_x^{(1)}
+
\sigma_x^{(2)}
+
\sigma_x^{(3)},
\qquad
H_y
=
\sigma_y^{(1)}
+
\sigma_y^{(2)}
+
\sigma_y^{(3)},
\\
H_{zz}
=
\sigma_z^{(1)}\sigma_z^{(2)}
+
\sigma_z^{(1)}\sigma_z^{(3)}
+
\sigma_z^{(2)}\sigma_z^{(3)} .
\end{gathered}
\end{equation}

\paragraph{Resource scaling.}

The possible advantage of the block encoding is clearest for large symmetry groups. For $G=S_N$, Lemma~\ref{lem:Sn_full_dimension} gives
$$
\dim\left(\mathfrak{su}(\mathcal H)^{S_N}\right)
=
\binom{N+3}{3}
-
1 .
$$
A related result was obtained in Ref.~\cite{Albertini2018}; we include the derivation in the Appendix for completeness. 
Therefore, the invariant Lie algebra has dimension polynomial in $N$, whereas
$$
\dim\mathfrak{su}(\mathcal H)
=
4^N-1
$$
for $\mathcal H=(\mathbb C^2)^{\otimes N}$.

For a general subgroup $G\subseteq S_N$, Lemma~\ref{lem:subgroup_su_inclusion} shows that
$$
\mathfrak{su}(\mathcal H)^{S_N}
\subseteq
\mathfrak{su}(\mathcal H)^G .
$$
Consequently,
\begin{equation}\label{eq:inv_alg_dim}
\dim\mathfrak{su}(\mathcal H)^G
\ge
\dim\mathfrak{su}(\mathcal H)^{S_N}
=
\binom{N+3}{3}
-
1 .
\end{equation}
Larger symmetry groups impose more constraints and can therefore lead to smaller invariant Lie algebras. 
Compared with a fully $\mathfrak{su}(\mathcal H)$-covariant encoding, the symmetry-adapted block encoding reduces the dimensions of the non-abelian blocks 
that need to be encoded covariantly. The amount of reduction depends on the symmetry group and on the multiplicities in the decomposition~\eqref{eq:H_G_irrep_decomp}.

There are two main reasons to consider block encodings covariant with respect to
$$
\mathfrak{su}(d_1)
\oplus
\cdots
\oplus
\mathfrak{su}(d_k)
\oplus
\mathfrak{u}(1)^{\oplus(r-1)}
$$
rather than fully $\mathfrak{su}(\mathcal H)$-covariant encodings. First, in the presence of symmetries, passing from the full algebra $\mathfrak{su}(\mathcal H)$ to the invariant algebra 
$\mathfrak{su}(\mathcal H)^G$ can reduce the size of the non-abelian sectors that must be protected covariantly. 
Second, from a practical perspective, stacking several $\mathfrak{su}(d_i)$-covariant code blocks may scale more favorably than increasing 
the local physical dimension to match the full dimension of $\mathcal H$, as would be required for a fully $\mathfrak{su}(\mathcal H)$-covariant encoding.

In the block-encoding setting, the numbers $n_i$ of physical qudits assigned to the different $\mathfrak{su}(d_i)$ blocks can be chosen independently. 
This gives additional flexibility in distributing physical resources among the different symmetry sectors. 
By contrast, in the fully $\mathfrak{su}(\mathcal H)$-covariant setting, one is tied to a single global code block with a fixed global block size $n$. 
Therefore, the block encoding can reduce the required local physical dimension while retaining transversal implementation of the symmetry-preserving dynamics.

These advantages are useful only if the resulting block encoding still has good approximate error-correction properties. We analyze this question next.

\subsection{Correlated flagged two-qudit erasure}\label{subsec:main_covar_analog_noise}

\paragraph{From product noise to correlated erasures.}

For product noise across the blocks, the block code inherits the AQEC properties of its individual components. Indeed, consider a noise channel of the form
$$
\mathcal N
=
\bigotimes_{i=1}^{2|\mathcal I_{\mathcal H}|-1}
\mathcal N_i
$$
for the
$$
\mathfrak{su}(m_1)
\oplus
\cdots
\oplus
\mathfrak{su}(m_{|\mathcal I_{\mathcal H}|})
\oplus
\mathfrak{u}(1)^{\oplus(|\mathcal I_{\mathcal H}|-1)}
$$
-covariant block code constructed in Section~\ref{subsec:main_covar_block_enc}. Since the block code is a product code~\eqref{eq:product_code}, 
and each $\mathcal N_i$ acts on a single block, the correctability analysis reduces to the corresponding analysis for the individual covariant codes. 
In particular, when the $\mathcal N_i$ are erasure-type channels of the form considered in Eqs.~\eqref{eq:erasure_noise} and~\eqref{eq:arbitrary_multiqudit_noise}, 
the block code is corrected blockwise.

This product noise model does not capture correlated erasures involving qudits from different blocks. 
We therefore consider a two-qudit erasure model that includes both in-block and inter-block erasures.

\paragraph{Two-block code and reduced states.}

We focus on an
$\mathfrak{su}(d_1)\oplus\mathfrak{su}(d_2)$-covariant block encoding. The two blocks have sizes $n_1$ and $n_2$, and the physical Hilbert space is
$$
\mathcal H_P
\cong
(\mathbb C^{d_1})^{\otimes n_1}
\otimes
(\mathbb C^{d_2})^{\otimes n_2}.
$$
Each physical qudit carries the fundamental representation of the corresponding Lie algebra. Therefore, the transversal physical action is
$$
t_a
\longmapsto
\sum_{i=1}^{n_1} t_a^{(i)},
\qquad
s_b
\longmapsto
\sum_{j=1}^{n_2} s_b^{(j)},
$$
where $t_a^{(i)}$ acts on the $i$th qudit of the first block and $s_b^{(j)}$ acts on the $j$th qudit of the second block. For each block, 
we choose the code space as in Theorem~\ref{thrm:main_2_transitive_covariant_code}, so that two-qudit reduced states inside each block are identical 
up to the canonical identification of sites. The precise number-theoretic conditions on $n_1$ and $n_2$ are those of Theorem~\ref{thrm:main_2_transitive_covariant_code}; 
below we only use that both $n_1$ and $n_2$ can be taken arbitrarily large.

It is convenient to write a logical state in a form adapted to the two-block structure:
$$
\begin{aligned}
\rho_L
&=
\frac{I}{d_1d_2}
+
\frac{1}{d_2}
\sum_a r_{1,a},\bar t_a\otimes I
+
\frac{1}{d_1}
\sum_b r_{2,b},I\otimes\bar s_b
\\
&\qquad
+
\sum_{a,b}
C_{ab},\bar t_a\otimes\bar s_b .
\end{aligned}
$$
Here $r_{1,a}$ and $r_{2,b}$ are the Bloch-vector components of the one-block reduced logical states, and $C_{ab}$ is the inter-block correlation matrix. 
The barred generators act on the logical spaces.

The relevant two-qudit reduced states fall into three classes: two erased qudits in the first block, two erased qudits in the second block, and one erased qudit in each block. 
The detailed derivation is given in Proposition~\ref{prop:two_block_fidelity}. For two sites $i,j$ in the first block, the reduced state is
$$
\begin{aligned}
\rho_{1,\mathrm{in}}^{(ij)}(\rho_L)
&=
\frac{1}{d_1^2}\mathbf 1
-
\frac{2}{d_1n_1}
\sum_{m=1}^{d_1^2-1}
t_m^{(i)}t_m^{(j)}
\\
&\qquad
+
\frac{1}{n_1d_1}
\sum_{b=1}^{d_1^2-1}
r_{1,b}
\left(
t_b^{(i)}
+
t_b^{(j)}
\right).
\end{aligned}
$$
For two sites $i,j$ in the second block,
$$
\begin{aligned}
\rho_{2,\mathrm{in}}^{(ij)}(\rho_L)
&=
\frac{1}{d_2^2}\mathbf 1
-
\frac{2}{d_2n_2}
\sum_{m=1}^{d_2^2-1}
s_m^{(i)}s_m^{(j)}
\\
&\qquad
+
\frac{1}{n_2d_2}
\sum_{b=1}^{d_2^2-1}
r_{2,b}
\left(
s_b^{(i)}
+
s_b^{(j)}
\right).
\end{aligned}
$$
These in-block expressions can be obtained by tracing out one site from the three-qudit reduced state in Eq.~\eqref{eq:3q_reduced_state}.
 For an inter-block pair, with $i$ in the first block and $j$ in the second block, the reduced state is
$$
\begin{aligned}
\rho_{\mathrm{inter}}^{(ij)}(\rho_L)
&=
\frac{I}{d_1d_2}
+
\frac{1}{n_1d_2}
\sum_a r_{1,a},t_a^{(i)}\otimes I
\\
&\qquad
+
\frac{1}{n_2d_1}
\sum_b r_{2,b},I\otimes s_b^{(j)} \\
&\qquad
+
\frac{1}{n_1n_2}
\sum_{a,b}
C_{ab},t_a^{(i)}\otimes s_b^{(j)} .
\end{aligned}
$$
This follows from the product form of the block encoding and the one-site reduced-state formula~\eqref{eq:reduced_state_sud}.

\paragraph{Flagged two-qudit erasure.}

 erasure channel. Let $n=n_1+n_2$, where the first block consists of sites $1,\ldots,n_1$ and the second block consists of sites $n_1+1,\ldots,n$. 
 The flagged two-qudit erasure channel is
$$
\mathcal N(\sigma)
=
\sum_{1\le i<j\le n}
q_{ij}
\ket{ij}\bra{ij}_F
\otimes
\ket e\bra e_{ij}
\otimes
\operatorname{Tr}_{ij}(\sigma),
$$
where $q_{ij}\ge 0$, $\sum_{i<j}q_{ij}=1$, and the flag register records the erased pair.

There are three types of erased pairs: two sites in the first block, two sites in the second block, and one site in each block.
 We denote the corresponding total probabilities by
$$
\begin{gathered}
p_1
:=
\sum_{1\le i<j\le n_1}
q_{ij},
\qquad
p_2
:=
\sum_{n_1<i<j\le n}
q_{ij},
\\
p_{12}
:=
\sum_{\substack{1\le i\le n_1\ n_1<j\le n}}
q_{ij}.
\end{gathered}
$$
Thus $p_{12}=1-p_1-p_2$. The remaining information about which pair was erased is stored in orthogonal flag states, denoted by $\omega_{1,\mathrm{in}}$, $\omega_{2,\mathrm{in}}$, and $\omega_{\mathrm{inter}}$, 
supported on the three corresponding flag subspaces.

Using the permutation symmetries imposed separately on the two blocks, the reduced states within each of the three sectors can be identified up to relabeling of sites. 
Therefore, the complementary channel has the form
$$
\begin{aligned}
\widehat{\mathcal N\circ\mathcal E}(\rho_L)
&=
p_1,
\omega_{1,\mathrm{in}}
\otimes
\rho_{1,\mathrm{in}}(\rho_L)
\\
&\qquad
+
p_2,
\omega_{2,\mathrm{in}}
\otimes
\rho_{2,\mathrm{in}}(\rho_L)
\\
&\qquad
+
p_{12},
\omega_{\mathrm{inter}}
\otimes
\rho_{\mathrm{inter}}(\rho_L).
\end{aligned}
$$
Here the site labels have been suppressed because the reduced states are identified within each sector.

We compare this complementary channel with the constant channel
$$
\Lambda_0(\rho)
=
\operatorname{Tr}(\rho)
\left(
p_1\omega_{1,\mathrm{in}}\otimes\tau_1
+
p_2\omega_{2,\mathrm{in}}\otimes\tau_2
+
p_{12}\omega_{\mathrm{inter}}\otimes\tau_{12}
\right),
$$
where
$$
\begin{gathered}
\tau_1
:=
\frac{I}{d_1^2}
-
\frac{2}{d_1n_1}
\sum_m t_m^{(i)}t_m^{(j)},
\\
\tau_2
:=
\frac{I}{d_2^2}
-
\frac{2}{d_2n_2}
\sum_m s_m^{(i)}s_m^{(j)},
\\
\tau_{12}
:=
\frac{I}{d_1d_2}.
\end{gathered}
$$
These are the logical-input-independent parts of the first in-block, second in-block, and inter-block reduced states, respectively. 

The explicit fidelity calculation in Proposition~\ref{prop:two_block_fidelity} gives
\begin{equation}\label{eq:fidelity_2}
\begin{aligned}
F\left(\widehat{\mathcal N\circ\mathcal E},\Lambda_0\right)
&=
1
-
\frac{p_1(d_1^2-1)}{4n_1^2}
-
\frac{p_2(d_2^2-1)}{4n_2^2}
\\
&\qquad
-
p_{12}
\left(
\frac{d_1^2-1}{8n_1^2}
+
\frac{d_2^2-1}{8n_2^2}
\right)
\\
&\qquad
+
O_{d_1,d_2}\left(n_1^{-2}n_2^{-2}+n_1^{-3}+n_2^{-3}\right).
\end{aligned}
\end{equation}
Together with the complementary-channel characterization of AQEC, this gives the following result.

\begin{theorem}
Let
$$
\mathcal H_L
=
V_{d_1,\omega_1}
\otimes
V_{d_2,\omega_1},
\qquad
\mathcal H_P
=
(V_{d_1,\omega_1})^{\otimes n_1}
\otimes
(V_{d_2,\omega_1})^{\otimes n_2}.
$$
Assume that each $n_\ell$ is a prime power satisfying $n_\ell\equiv 1\pmod{d_\ell}$, for $\ell=1,2$. 
Equip $\mathcal H_L$ with the product fundamental representation of $\mathfrak{su}(d_1)\oplus\mathfrak{su}(d_2)$ and $\mathcal H_P$ with the transversal representation
$$
t_a
\longmapsto
\sum_{i=1}^{n_1} t_a^{(i)},
\qquad
s_b
\longmapsto
\sum_{j=1}^{n_2} s_b^{(j)} .
$$
Let $\mathcal E$ be an $\mathfrak{su}(d_1)\oplus\mathfrak{su}(d_2)$-covariant encoding with respect to these representations, 
such that each block is invariant under the action of $\mathrm{AGL}(1,n_\ell)$~\eqref{eq:AGL_def}. Then, for the flagged two-qudit erasure channel
$$
\mathcal N(\sigma)
=
\sum_{1\le i<j\le n_1+n_2}
q_{ij}
\ket{ij}\bra{ij}_F
\otimes
\ket e\bra e_{ij}
\otimes
\operatorname{Tr}_{ij}(\sigma),
$$
there exists a recovery channel $\mathcal R$ such that, when $n_1\approx n_2\approx n$,
$$
d(\mathcal R\mathcal N\mathcal E,\mathrm{id})
\le
\frac{
\sqrt{
C_1(d_1^2-1)
+
C_2(d_2^2-1)
}
}{n}
+
O_{d_1,d_2}(n^{-2}),
$$
where $C_1$ and $C_2$ are constants independent of the encoding parameters. In particular, the code is
$$
O\left(
\frac{
\sqrt{
C_1(d_1^2-1)
+
C_2(d_2^2-1)
}
}{n}
\right)
$$
-correctable against flagged erasure on two sites.
\end{theorem}

\begin{proof}
By Eq.~\eqref{eq:fidelity_2} and the definition of the channel distance in Eq.~\eqref{eq:distance_def}, 
the distance between the complementary channel and $\Lambda_0$ is bounded by
$$
d\left(\widehat{\mathcal N\circ\mathcal E},\Lambda_0\right)
\le
\frac{
\sqrt{
C_1(d_1^2-1)
+
C_2(d_2^2-1)
}
}{n}
+
O_{d_1,d_2}(n^{-2})
$$
for $n_1\approx n_2\approx n$, with constants $C_1$ and $C_2$ independent of the encoding parameters. The complementary-channel theorem, Eq.~\eqref{eq:beny}, gives
$$
\min_{\mathcal R}
d(\mathcal R\mathcal N\mathcal E,\mathrm{id})
=
\min_{\Lambda,\mathrm{constant}}
d\left(\widehat{\mathcal N\circ\mathcal E},\Lambda\right)
\le
d\left(\widehat{\mathcal N\circ\mathcal E},\Lambda_0\right),
$$
and the stated bound follows.
\end{proof}

As a special case, the same result gives the two-qudit erasure analysis for the usual in-block $SU(d)$-covariant encoding, 
as discussed in Section~\ref{subsec:main_muti_erasure_noise}.

\section{Covariant Encoding for Universal Analog Computations}\label{sec:main_univ_analog}

We now discuss universal analog computation in the presence of symmetries. 
In analog quantum computation, universality means that the available Hamiltonians generate all traceless Hamiltonians on the Hilbert space, 
or equivalently that the corresponding dynamical Lie algebra is $\mathfrak{su}(\mathcal H)$. When the available Hamiltonians preserve a symmetry, 
their DLA is contained in the invariant Lie algebra $\mathfrak{su}(\mathcal H)^G$, and therefore cannot 
be universal by itself unless this invariant algebra already coincides with $\mathfrak{su}(\mathcal H)$.

A natural way to restore universality is to add symmetry-breaking Hamiltonians. 
In this section, we give a sufficient condition under which such Hamiltonians, together with the invariant Lie algebra, 
generate the full algebra $\mathfrak{su}(\mathcal H)$. This structure also motivates a framework for robust universal analog computation. 
The block encoding from Section~\ref{subsec:main_covar_block_enc} makes the invariant algebra act transversally on the encoded space, 
while the symmetry-breaking Hamiltonians are treated as resource Hamiltonians because they are not transversal for the block encoding.

\subsection{Universal Analog Computations}\label{subsec:main_univ_analog}

\paragraph{Lie-algebraic universality.}

We first recall the standard notion of universal controllability in the analog setting. The propagator $U(t)$ satisfies
$$
\dot U(t)
=
-iH(t)U(t),
\qquad
U(0)=I,
$$
where $H(t)$ is defined in Eq.~\eqref{eq:system_ham}. The system is unitarily controllable if, for every target unitary $U_{\mathrm{tar}}\in SU(d)$, 
there exist a finite time $T$ and control functions $u_k(t)$ such that
$$
U(T)=U_{\mathrm{tar}}
$$
up to a global phase.

The standard Lie-algebraic controllability criterion states that the system is universally controllable if and only if 
the dynamical Lie algebra generated by the available Hamiltonians is $\mathfrak{su}(d)$~\cite{Albertini2003}:
$$
\operatorname{Lie}_{\mathbb R}
\{iH_0,iH_1,\ldots,iH_m\}
\cong
\mathfrak{su}(d).
$$
When this condition holds, we call the set ${iH_0,iH_1,\ldots,iH_m}$ universal. Controllability can also be formulated in terms of $U(d)$ rather than $SU(d)$, 
in which case the relevant Lie algebra is $\mathfrak{u}(d)$. Throughout this work, however, we consider controllability up to global phase. 
This is sufficient for closed-system quantum dynamics, since scalar Hamiltonian terms generate only global phases. 
We therefore work with traceless Hamiltonians and the Lie algebra $\mathfrak{su}(d)$.

\paragraph{Restoring universality by symmetry breaking.}

As discussed in Section~\ref{subsec:main_DLA}, if the drift Hamiltonian $H_0$ and the available global control Hamiltonians ${H_j}_{j=1}^m$ are invariant under a finite symmetry group $G$, 
then their DLA is contained in the invariant Lie algebra
$$
\mathfrak{su}(\mathcal H)^G
\cong
\mathfrak{su}(d_1)
\oplus
\cdots
\oplus
\mathfrak{su}(d_k)
\oplus
\mathfrak{u}(1)^{\oplus(r-1)}.
$$
This algebra is generally a proper subalgebra of $\mathfrak{su}(\mathcal H)$. 
Its precise block structure depends on the symmetry group $G$ and on the representation of $G$ on $\mathcal H$.

The question is whether universality can be restored by adding a small set of local symmetry-breaking Hamiltonians to the symmetric Hamiltonians. 
In principle, this is always possible by adding a full generating set of $\mathfrak{su}(\mathcal H)$. Such a solution, 
however, does not use the structure of the symmetric system and is not useful as a design principle. 
The relevant question is instead how to choose a small and physically reasonable set of symmetry-breaking Hamiltonians that, together with the symmetric dynamics, 
generates the full algebra.

Related universality criteria for symmetry-breaking controls have been studied in Refs.~\cite{Hu2025, AlbertiniDAlessandro2025}. 
Here we use a sufficient condition tailored to the representation-theoretic decomposition induced by $G$. Let $\mathcal G^{\mathrm{br}}$ denote the set of added symmetry-breaking Hamiltonians, 
or equivalently the corresponding skew-Hermitian generators when inserted into the real Lie closure. 
To state the condition, we introduce two objects in Appendix~\ref{sec:supple_Universal_comp}: 
the coupling graph $\Gamma_{\mathrm{co}}$ and the Full First-Factor Span condition. The coupling graph records which symmetry sectors are connected by the breaking Hamiltonians, 
while the Full First-Factor Span condition ensures that these couplings are large enough to generate the required off-diagonal directions between sectors.

The resulting sufficient condition is the following.

\begin{theorem}\label{thm:breakers_suff_cond}
If the coupling graph $\Gamma_{\mathrm{co}}$ is connected and the Full First-Factor Span condition holds for every edge in $\Gamma_{\mathrm{co}}$, then
$$
\operatorname{Lie}_{\mathbb R}
\{\mathfrak{su}(\mathcal H)^G\cup\mathcal G^{\mathrm{br}}\}
=
\mathfrak{su}(\mathcal H).
$$
\end{theorem}

The proof and the precise definitions of $\Gamma_{\mathrm{co}}$ and the Full First-Factor Span condition are given in Appendix~\ref{sec:supple_Universal_comp}. 
In Appendix~E.3, we also illustrate Theorem~\ref{thm:breakers_suff_cond} on a three-qubit system with $S_3$ permutation symmetry. 
In that example, adding the single symmetry-breaking Hamiltonian
$$
Z_1+X_2
$$
is sufficient to generate $\mathfrak{su}(\mathcal H)$ together with the invariant algebra.

The theorem is stated using $\mathfrak{su}(\mathcal H)^G$ rather than the actual DLA of the symmetric Hamiltonians. 
This distinction is important: the DLA generated by a particular set of symmetric Hamiltonians may be a proper subalgebra of $\mathfrak{su}(\mathcal H)^G$, 
as discussed in Section~\ref{subsec:main_DLA}. The invariant algebra is nevertheless the natural object for the criterion, because it is determined directly 
by the representation theory of the $G$-action and is therefore easier to characterize. In the rest of this section, 
we focus on cases where the symmetric Hamiltonians generate the full invariant algebra $\mathfrak{su}(\mathcal H)^G$. 
This assumption isolates the additional role of the symmetry-breaking Hamiltonians and is sufficient for the robust universal simulation framework considered below.

\subsection{Breaking Hamiltonians as a resource}\label{subsec:main_H_br_resource}

\paragraph{Transversal and non-transversal parts.}

Suppose that, for a given family of symmetric Hamiltonians
${iH_0,iH_1,\ldots,iH_m}$, we have chosen a set of symmetry-breaking Hamiltonians $\mathcal G^{\mathrm{br}}$ satisfying the universality condition of Theorem~\ref{thm:breakers_suff_cond}. 
We now describe how the covariant codes constructed above can be used as building blocks for this universal analog-control setting. 
The idea is to use a code space compatible with the continuous symmetry, i.e.,
the code space is preserved by the symmetry action, so the corresponding
encoded operations can be implemented transversally whenever they arise
from this symmetry. This provides a natural way to encode information
while keeping as much of the relevant dynamics as possible in transversal
form.

One option is to use a fully $\mathfrak{su}(d)$-covariant code, where $d=\dim\mathcal H$, described in previous sections. 
Such a code has the AQEC performance established in the previous sections and makes all generators of $\mathfrak{su}(d)$ act transversally. 
However, as discussed above, this approach can be costly, because the local physical dimension must scale with the dimension of the logical Hilbert space. 
This makes the fully covariant construction difficult to scale for large analog systems.

The block-encoding approach gives a more symmetry-adapted alternative. 
Instead of imposing covariance with respect to all of $\mathfrak{su}(\mathcal H)$, we use a code covariant with respect to
$$
\mathfrak{su}(d_1)
\oplus
\cdots
\oplus
\mathfrak{su}(d_k)
\oplus
\mathfrak{u}(1)^{\oplus(r-1)} .
$$
For this code, the Hamiltonians in the invariant algebra $\mathfrak{su}(\mathcal H)^G$ are implemented transversally. 
The symmetry-breaking Hamiltonians in $\mathcal G^{\mathrm{br}}$ are different: 
they do not preserve the symmetry-sector decomposition and therefore are not transversal for the block encoding. We therefore treat them as resource Hamiltonians. 
In general, these resource Hamiltonians may be nonlocal after encoding and may be harder to implement than the transversal invariant dynamics.

\paragraph{Three-qubit example.}

Consider again the example of three qubits with $S_3$ permutation symmetry. The corresponding invariant Lie algebra is
$$
\mathfrak{su}(4)
\oplus
\mathfrak{su}(2)
\oplus
\mathfrak{u}(1).
$$
As shown in Example~\ref{ex:Sn_block_diagonalization}, the global control Hamiltonians $H_x$ and $H_y$, together with the intrinsic Hamiltonian $H_{zz}$ defined 
in Eq.~\eqref{eq:s3_ex_hams}, are block diagonal in the symmetry-adapted basis. By contrast, the symmetry-breaking Hamiltonian
$$
Z_1+X_2
$$
does not preserve this block decomposition and is not block diagonal in the same basis; see Example~\ref{ex:s3_breaking_not_block_diagonal}. 
As shown in Ref.~\cite{Albertini2018}, the DLA generated by ${iH_x,iH_y,iH_{zz}}$ coincides with $\mathfrak{su}(2^3)^{S_3}$. Therefore the enlarged set
$$
\{iH_x,iH_y,iH_{zz},i(Z_1+X_2)\}
$$
is universal.

This example also clarifies the role of the assumption that the symmetric DLA coincides with $\mathfrak{su}(\mathcal H)^G$. For the purpose of universal logical computation, 
the available symmetric Hamiltonians do not have to arise from a fixed microscopic physical model. What is needed is a set of Hamiltonians 
that generates the desired invariant algebra. One may therefore choose a convenient finite symmetry group $G$ and a corresponding generating set for $\mathfrak{su}(\mathcal H)^G$, 
and then add symmetry-breaking Hamiltonians that restore universality. For example, in the case $G=S_N$, Ref.~\cite{Albertini2018} gives explicit generators of 
$\mathfrak{su}(\mathcal H)^{S_N}$:
$$
\begin{gathered}
iH_x
=
i\sum_{j=1}^N \sigma_x^{(j)},
\qquad
iH_y
=
i\sum_{j=1}^N \sigma_y^{(j)},
\\
iH_{zz}
=
i\sum_{1\le i<j\le N}
\sigma_z^{(i)}\sigma_z^{(j)} .
\end{gathered}
$$
In this example, the generators are also naturally motivated by physical collective-spin interactions.

\paragraph{Design trade-off.}

The design problem is therefore to choose the symmetry group $G$, equivalently the decomposition
$$
\mathfrak{su}(\mathcal H)^G
\cong
\mathfrak{su}(d_1)
\oplus
\cdots
\oplus
\mathfrak{su}(d_k)
\oplus
\mathfrak{u}(1)^{\oplus(r-1)},
$$
together with a symmetry-breaking resource set $\mathcal G^{\mathrm{br}}$. These choices involve a trade-off. 
A larger symmetry group gives a smaller invariant algebra and therefore smaller blocks to encode covariantly. 
For instance, choosing $G=S_N$ makes $\mathfrak{su}(\mathcal H)^G$ polynomial-dimensional in $N$, as in Eq.~\eqref{eq:inv_alg_dim}. 
However, imposing such a large symmetry can make the required symmetry-breaking resource set more complicated. 
Conversely, choosing a smaller symmetry group, such as $\mathbb Z_2$, gives less reduction in the invariant algebra, 
but universality may be restored with only one symmetry-breaking Hamiltonian; see Ref.~\cite{Hu2025}.

Theorem~\ref{thm:breakers_suff_cond} gives only a sufficient condition for universality. It does not determine the minimal number of symmetry-breaking Hamiltonians, 
nor does it characterize the simplest possible structure of such Hamiltonians. Understanding this optimization problem, 
including the trade-off between block-encoding efficiency and the complexity of the resource Hamiltonians, is left for future work.

\section{Discussion and Outlook}\label{sec:main_disc}

To summarize, in the first part of the paper, we constructed several families of \(SU(d)\)-covariant AQECC. The recovery error of these codes scales
as \(1/n\) for one-qudit, two-qudit, and three-qudit flagged erasure models, and for single-qudit erasure this scaling is near-optimal. 
We then constructed a decoder for the single-qudit flagged erasure model using the Petz recovery map. 
This decoder is near-optimal and achieves \(1/n\) scaling, up to constants depending on the physical and logical dimensions. We also analyze the same codes
under arbitrary flagged noise on at most three qudits, obtaining the more general, but weaker, scaling \(1/\sqrt n\). 
Because exactly $k$-transitive permutation groups do not exist for $k>3$, this symmetry-based construction cannot be extended directly beyond three erased sites. 
We removed this restriction by replacing exact permutation symmetry with a Haar-random multiplicity vector: a typical draw gives a code whose $k$-site reduced states are 
all close to one another, rather than exactly equal, for any fixed $k$, and we gave an efficient classical procedure for sampling such a typical code from a seed of $O(n\log d)$ 
random bits. This yields a family of \(SU(d)\)-covariant codes correctable against flagged erasure on any fixed number \(k\) of sites, with recovery error scaling as 
\(O(\sqrt{k(d^2-1)}/n)\). The trace-norm argument used for non-erasure noise applies to these codes as well, giving \(O(\sqrt{k/n})\) correctability against arbitrary 
flagged noise on any fixed number \(k\) of sites. We read this less as the construction of one more code than as a statement about the covariant sector itself. 
Once covariance has fixed the isotypic component, the only remaining freedom is a direction in the multiplicity space, and the directions along which some erased subsystem 
learns an atypically large amount about the logical state have exponentially small Haar measure. Badly aligned directions exist, the unencoded seed vector among them, 
but they are non-generic, so approximate decoupling from any fixed-size erased subsystem is typical rather than fine-tuned. 
The efficient sampler adds a second reading of the same kind. Haar randomness establishes that good directions are abundant, 
but a generic one is not something that can be written down, and the sampling procedure shows that the abundance survives when the continuous ensemble is replaced 
by four-wise independent phases from a seed of $O(n\log d)$ bits. What the error-correcting property depends on is therefore not Haar randomness but the low-order moments 
the concentration argument consumes, and those are already available inside a structured describable family.

In the second part of the paper, we introduced block encodings that are covariant with respect to algebras of the form
$
    \mathfrak{su}(d_1)\oplus\cdots\oplus\mathfrak{su}(d_k)
    \oplus
    \mathfrak{u}(1)^{\oplus(r-1)}.
$
These encodings are motivated by quantum simulation, where the relevant dynamics are determined by the DLA of the available
Hamiltonians. We analyzed the resulting block codes under two-qudit erasure noise, including both in-block and inter-block erasures, and obtained \(1/n\) scaling 
of the recovery error.

In the third part of the paper, we proposed a framework for robust universal analog computation. We first gave a sufficient condition under which adding a set of symmetry-breaking 
Hamiltonians to the invariant Lie subalgebra generates a universal DLA. When combined with the block encoding developed
earlier, the invariant subalgebra acts transversally on the encoded space, while the symmetry-breaking Hamiltonians become non-transversal resource Hamiltonians needed to complete universality. 
This suggests a possible route toward fault-tolerant analog computation in which a large symmetry-preserving sector is protected by covariance, while a controlled
set of non-transversal resource Hamiltonians supplies universality.

There is an interesting resemblance between the \(SU(d)\)-covariant AQECC codes constructed here and stabilizer codes. 
First, for stabilizer codes, correcting flagged \(k\)-qubit erasures is equivalent to correcting arbitrary \(k\)-qubit errors at known locations. A similar mechanism
appears in our setting. Indeed, the complementary channel for arbitrary flagged local noise factors through the reduced encoded state: it is obtained by first taking the local reduced state, 
which is precisely the information seen by the environment in the erasure model, and then applying the complementary channel of
the local noise. Since this second step is linear, the coefficient suppressing the logical-state-dependent part of the reduced state is preserved. 
Thus, the same reduced-state structure that makes the code good against erasures also implies approximate correctability against arbitrary flagged local noise, although the general trace-norm argument gives a weaker \(1/\sqrt n\) bound unless the
fidelity can be evaluated more sharply.

Second, there is an analogy with the behavior of tensor products of stabilizer codes. For stabilizer codes, taking a tensor product of two codes with 
the same distance preserves the distance scale. In our setting, the block encoding behaves similarly in terms of asymptotic behavior. Namely, the
\(\mathfrak{su}(d_1)\oplus\mathfrak{su}(d_2)\)-covariant block code is obtained from the corresponding \(SU(d_1)\)- and \(SU(d_2)\)-covariant approximate codes, 
and its performance under two-qudit erasure noise preserves the same \(1/n\)
scaling. This remains true for both in-block and inter-block erasures.

Although continuous covariance imposes an $\Omega(1/n)$ lower bound on
the worst-case recovery distance, this scaling should be viewed as a
resource--accuracy tradeoff rather than as an accuracy floor.  A target
error $\varepsilon$ can still be reached with
$n=O(\varepsilon^{-1})$ physical subsystems.  Such inverse-polynomial
overhead can be compatible with quantum-simulation applications,
particularly when the objective is to reproduce local or few-body
observables over finite evolution times rather than the entire
many-body state with extremely high global fidelity. Indeed, broad
classes of finite-time local-observable simulation tasks are known to
be stable against local implementation errors
~\cite{Trivedi2024AnalogStability}.  The inverse-linear scaling also has
a natural connection to quantum metrology, as metrological bounds on
quantum Fisher information provide one route to approximate
Eastin--Knill inequalities~\cite{Kubica2021,Zhou2021}, and the
$d=\Theta(1/n)$ scaling in our fidelity-based distance corresponds to an
infidelity $1-F=\Theta(1/n^2)$, reminiscent of Heisenberg scaling.
This connection does not identify recovery error with metrological
estimation error, but it emphasizes that the symmetry constraint
imposes a fundamental scaling law whose cost can be systematically
reduced by increasing the available physical resources.

\section{Future work}\label{sec:main_future_work}

Several important directions remain open.
As discussed above, it is not necessary for all \(k\)-qudit reduced states to be identical, even up to the identification of tuples of sites, 
in order to distribute logical information sufficiently uniformly across the physical system: more complicated distributions of reduced states 
can still lead to good AQEC performance. In Section~\ref{subsec:main_multisite_erasure_noise} we used this observation to construct a covariant 
AQECC with good scaling under general \(k\)-qudit erasure noise, using the reduced-state method developed in this work: replacing the exact permutation symmetry, 
which does not exist for $k>3$, by a Haar-random, efficiently sampled multiplicity vector gives a code whose \(k\)-site reduced states are all close to a common mean, 
with recovery error scaling as
$
    \frac{\sqrt{k(d^2-1)}}{n}
$
for fixed \(k\), confirming the \(\sqrt k/n\) scaling expected from the reduced-state method. Two related questions remain open: 
whether the same scaling can be reached by a code with no probability of failure, and whether the constant obtained here is optimal.

Another natural direction is to analyze recovery maps more explicitly. The Petz, or transpose, recovery map can be generalized directly to the \(k\)-qudit erasure setting. 
However, understanding its exact performance remains an important open problem. In this work, we establish near-optimality at the level of scaling in \(n\), 
but we do not determine the sharp constants depending on the physical representation, the logical dimension. 
A more detailed analysis of the Petz recovery map may therefore clarify whether it is optimal beyond the level of asymptotic scaling.

Our analysis of arbitrary local noise also generalizes formally to \(k\)-qudit noise, provided the corresponding \(k\)-qudit reduced states have the appropriate form. 
In particular, if one constructs an encoding for which the \(k\)-qudit reduced states approach a fixed input-independent state as \(n\to\infty\), then an \(O_{d,k}(n^{-1/2})\) 
correctability bound for arbitrary flagged \(k\)-qudit noise follows essentially automatically from the trace-norm argument. However, this bound is likely not sharp. 
In contrast to erasure channels, arbitrary local noise channels are not generally covariant, which makes a direct fidelity calculation substantially harder, 
see~\cite{Alexander2025, Zhou2021}. Improving the \(O_{d,k}(n^{-1/2})\) estimate, and understanding whether \(O_{d,k}(n^{-1})\) scaling can be recovered for broader classes of local noise, 
remains an important question.

It would also be interesting to consider more general choices of the individual physical representation. 
In this work, we mostly focus on irreducible physical subsystems, such as \(V_{\omega_r}\). 
However, there are natural covariant codes whose individual physical systems are reducible. 
For example, the \(W\)-state-type code~\cite{Faist2020, Alexander2025} uses a physical subsystem carrying a direct sum of the fundamental representation 
and a trivial representation,
$
    V_{\omega_1}\oplus V_0.
$
It would be useful to understand the broader family of codes to which this example belongs and to analyze such codes under general noise models.

There are also several open questions motivated by quantum simulation. For a fixed physical system, one would like to find an encoding adapted to its actual DLA. 
In general, the DLA need not have the block form
$
    \mathfrak{su}(d_1)\oplus\cdots\oplus\mathfrak{su}(d_k)
    \oplus
    \mathfrak{u}(1)^{\oplus(r-1)}.
$
Thus, if the goal is to protect a given analog simulation, it may be necessary to construct covariant AQECC for other Lie algebras, 
such as \(\mathfrak{so}(d)\), \(\mathfrak{sp}(d)\), or other Lie algebras that can arise as DLA of physical systems, see~\cite{Kokcu2024, Wiersema2024}.

Finally, many questions remain open on the side of universal analog computation. Even without encoding, it is important to understand the necessary conditions under which a set of 
symmetry-breaking Hamiltonians, together with the invariant Lie subalgebra, generates the full algebra required for universality. 
Ideally, one would like a complete criterion giving necessary and sufficient conditions. From the encoding perspective, the goal is to find both an efficient encoding 
and a suitable set of breaking Hamiltonians: the encoding should minimize the overhead, while the breaking Hamiltonians should be as physically reasonable as possible, 
for example not too nonlocal and not too numerous. Understanding this tradeoff is a necessary step toward the longer-term goal of establishing a threshold theorem for such codes 
and developing a framework for fault-tolerant analog quantum computation.
 
\begin{acknowledgments}
We thank Nikita Romanov and Fangjun Hu for carefully reading the manuscript and for helpful suggestions that improved its presentation. We are also grateful to Anna Knörr and Liyuan Chen for useful discussions about the Petz recovery map and for pointing us to references that helped us construct the explicit decoder. We would like to acknowledge the NSF for funding through the CUA PFC (PHY-2317134), the Q-SEnSE QLCI (OMA-2016244), and the DOE through the QUACQ (DE-SC0025572).
\end{acknowledgments}
 
\nocite{KretschmannSchlingemannWerner2008}
\nocite{Macdonald1995}
\nocite{Robbins1955}
\nocite{LarsenShalev2008}

\bibliography{main.bib}
\let\addcontentsline\oldaddcontentsline

\clearpage
\onecolumngrid

\addtocontents{toc}{\protect\setcounter{tocdepth}{3}}

\begin{center}
{\LARGE\bfseries APPENDIX}
\end{center}
\vspace{1em}

\tableofcontents
\clearpage

\begin{appendix}

\section*{Nomenclature}\label{sec:notation}
\subsection*{Hilbert spaces, channels, and encoding maps}

\begin{tabular}{@{}ll@{}}

\parbox[t]{0.22\textwidth}{$\mathcal{H}$} & \parbox[t]{0.72\textwidth}{A finite-dimensional Hilbert space.} \\[0.6em]

\parbox[t]{0.22\textwidth}{$\mathcal{H}_L$} & \parbox[t]{0.72\textwidth}{A finite-dimensional logical Hilbert space.} \\[0.6em]

\parbox[t]{0.22\textwidth}{$\mathcal{H}_P$} & \parbox[t]{0.72\textwidth}{A finite-dimensional physical Hilbert space.} \\[0.6em]

\parbox[t]{0.22\textwidth}{$\mathcal{N}, \mathcal{M}$} & \parbox[t]{0.72\textwidth}{Quantum channels, i.e. completely positive trace-preserving linear maps.} \\[0.6em]

\parbox[t]{0.22\textwidth}{$\mathcal{E}$} & \parbox[t]{0.72\textwidth}{An encoding map from a logical Hilbert space to a physical Hilbert space.} \\[0.6em]

\parbox[t]{0.22\textwidth}{$\mathcal{V}$} & \parbox[t]{0.72\textwidth}{An encoding isometry, corresponding to encoding map $\mathcal{E}$, i. e. $\mathcal{E}(\rho_L)=\mathcal{V}\rho_L\mathcal{V}^{\dagger}$.} \\[0.6em]

\parbox[t]{0.22\textwidth}{$\rho^{(i)}(\rho_L)$} & \parbox[t]{0.72\textwidth}{The reduced state on the $i$-th physical qudit after encoding, i. e. $\rho^{(i)}(\rho_L)=\mathrm{Tr}_{\hat{i}}(\mathcal{V}\rho_L\mathcal{V}^{\dagger})$.} \\[0.6em]

\parbox[t]{0.22\textwidth}{$A_i$, $A_S$} & \parbox[t]{0.72\textwidth}{The Hilbert space of the $i$-th physical qudit, and $A_S=\bigotimes_{i\in S}A_i$ for a subset $S\subseteq\{1,\ldots,n\}$.} \\[0.6em]

\parbox[t]{0.22\textwidth}{$\hat S$} & \parbox[t]{0.72\textwidth}{The complement of $S$ in $\{1,\ldots,n\}$. The hat marks complementation generally: $\widehat{ijk}$ for the complement of $\{i,j,k\}$, $\widehat j$ for the retained tensor factors other than $j$, and $\widehat{\mathcal N}$ for the complementary channel.} \\[0.6em]

\parbox[t]{0.22\textwidth}{$\rho^{(S)}(\rho_L)$} & \parbox[t]{0.72\textwidth}{The reduced state on $A_S$ after encoding, i.\,e. $\rho^{(S)}(\rho_L)=\operatorname{Tr}_{\hat S}(\mathcal{V}\rho_L\mathcal{V}^\dagger)$, see Eq.~\eqref{eq:reduced_state}.} \\[0.6em]

\parbox[t]{0.22\textwidth}{$\widehat{\mathcal{N}}$} & \parbox[t]{0.72\textwidth}{The complementary channel to $\mathcal{N}$, see definition \ref{def:complementary_channel}.} \\[0.6em]

\parbox[t]{0.22\textwidth}{$f(\rho, \sigma)$} & \parbox[t]{0.72\textwidth}{The fidelity between states $\rho$ and $\sigma$, see definition \ref{def:root_fidelity}.} \\[0.6em]

\parbox[t]{0.22\textwidth}{$F_{\rho}(\mathcal{N}, \mathcal{M})$} & \parbox[t]{0.72\textwidth}{The entanglement fidelity between channels $\mathcal{N}$ and $\mathcal{M}$, see definition \ref{def:rho_entanglement_fidelity}.} \\[0.6em]

\parbox[t]{0.22\textwidth}{$F(\mathcal{N}, \mathcal{M})$} & \parbox[t]{0.72\textwidth}{The worst-case entanglement fidelity between channels $\mathcal{N}$ and $\mathcal{M}$, see definition \ref{def:worst_case_entanglement_fidelity}.} \\[0.6em]

\parbox[t]{0.22\textwidth}{$\ket{\psi_\rho}$} & \parbox[t]{0.72\textwidth}{A purification of the state $\rho$ on the logical space, used as the reference-assisted input in $F_\rho$.} \\[0.6em]

\parbox[t]{0.22\textwidth}{$d(\mathcal N,\mathcal M)$} & \parbox[t]{0.72\textwidth}{The fidelity-induced channel distance $\sqrt{1-F(\mathcal N,\mathcal M)}$, see Eq.~\eqref{eq:distance_def}.} \\[0.6em]

\parbox[t]{0.22\textwidth}{$\|\cdot\|_\diamond$} & \parbox[t]{0.72\textwidth}{The diamond norm, i.\,e. the completely bounded trace norm on linear maps between operator spaces.} \\[0.6em]

\parbox[t]{0.22\textwidth}{$\Lambda_0$} & \parbox[t]{0.72\textwidth}{A constant channel, i.\,e. $\Lambda_0(\rho)=\operatorname{Tr}(\rho)\rho_0$ for a fixed state $\rho_0$, see Eq.~\eqref{eq:constant_channel_def}.} \\[0.6em]

\parbox[t]{0.22\textwidth}{$\mathcal R_p$} & \parbox[t]{0.72\textwidth}{A recovery channel for flagged $k$-site erasure with distribution $p=(p_S)_{|S|=k}$ over the erased sets, Eq.~\eqref{eq:erasure_noise}.} \\[0.6em]

\parbox[t]{0.22\textwidth}{$\widehat{\mathcal N}_S$} & \parbox[t]{0.72\textwidth}{A complementary channel to the local noise channel $\mathcal N_S$ acting on $A_S$.} \\[0.6em]

\parbox[t]{0.22\textwidth}{$\tau_S$, $\Delta_S$} & \parbox[t]{0.72\textwidth}{The logical-independent and logical-dependent parts of the reduced state on $A_S$, i.\,e. $\rho^{(S)}(\rho_L)=\tau_S+\Delta_S(\rho_L)$, see Eq.~\eqref{eq:leakage-splitting}.} \\[0.6em]

\parbox[t]{0.22\textwidth}{$\delta$} & \parbox[t]{0.72\textwidth}{The uniform leakage of a code over a family of subsets, see Eq.~\eqref{eq:uniform-leakage}.} \\[0.6em]

\parbox[t]{0.22\textwidth}{$\mathcal{V}_0$} & \parbox[t]{0.72\textwidth}{The seed isometry placing the logical qudit at site $0$ and filling the rest with $SU(d)$-singlets, see Eq.~\eqref{eq:seed-isometry}.} \\[0.6em]

\parbox[t]{0.22\textwidth}{$\mathcal{V}_m$, $\mathcal E_m$} & \parbox[t]{0.72\textwidth}{The encoding isometry built from the multiplicity vector $\ket m$ and the corresponding encoding $\mathcal E_m(X)=\mathcal{V}_mX\mathcal{V}_m^\dagger$, see Eq.~\eqref{eq:Vm}.} \\[0.6em]

\parbox[t]{0.22\textwidth}{$\mathcal E_s$} & \parbox[t]{0.72\textwidth}{The encoding $\mathcal E_s(X)=\mathcal{V}_{m_s}X\mathcal{V}_{m_s}^\dagger$ built from the efficiently sampled multiplicity vector $\ket{m_s}$.} \\[0.6em]

\parbox[t]{0.22\textwidth}{$\Phi_{S,m}$} & \parbox[t]{0.72\textwidth}{The erased-subsystem channel $\Phi_{S,m}(X)=\operatorname{Tr}_{\hat S}(\mathcal{V}_mX\mathcal{V}_m^\dagger)$, i.\,e. the reduced-state map $\rho^{(S)}$ for the code $\mathcal E_m$, see Eq.~\eqref{eq:local-channel}.} \\[0.6em]

\parbox[t]{0.22\textwidth}{$\overline\Phi_{k,n}$} & \parbox[t]{0.72\textwidth}{The mean $k$-site channel, i.\,e. the Haar average of $\Phi_{S,m}$ over the multiplicity vector $\ket m$.} \\[0.6em]

\parbox[t]{0.22\textwidth}{$\lambda_{k,n}$, $\Lambda_{k,n}$} & \parbox[t]{0.72\textwidth}{The state $\overline\Phi_{k,n}(\tau)$ and the constant channel $\Lambda_{k,n}(X)=\operatorname{Tr}(X)\lambda_{k,n}$ it defines, see Eq.~\eqref{eq:lambda-kn}.} \\[0.6em]

\end{tabular}

\subsection*{Spin systems and \(SU(2)\) notation}

\begin{tabular}{@{}ll@{}}

\parbox[t]{0.22\textwidth}{$V_j$} & \parbox[t]{0.72\textwidth}{The spin-$j$ irreducible representation of $SU(2)$.} \\[0.6em]

\parbox[t]{0.22\textwidth}{$J_x, J_y, J_z$} & \parbox[t]{0.72\textwidth}{The spin-$j$ generators of $SU(2)$ acting on $V_j$, see Eq.~\eqref{eq:spinj}, Eq.~\eqref{eq:Jz}, and Eq.~\eqref{eq:Jpm}.} \\[0.6em]

\end{tabular}

\subsection*{\(SU(d)\) generators and covariance notation}

\begin{tabular}{@{}ll@{}}

\parbox[t]{0.22\textwidth}{$t_a$} & \parbox[t]{0.72\textwidth}{The generators of $\mathfrak{su}(d)$, traceless Hermitian complex $d\times d$ matrices, see Eq.~\eqref{eq:sudgenerators}.} \\[0.6em]

\parbox[t]{0.22\textwidth}{$\bar{t}_a$} & \parbox[t]{0.72\textwidth}{The generators of $\mathfrak{su}(d)$ in the fundamental representation, used to emphasize action on the logical Hilbert space.} \\[0.6em]

\parbox[t]{0.22\textwidth}{$t_a^{(i)}$} & \parbox[t]{0.72\textwidth}{The generators of $\mathfrak{su}(d)$ in the fundamental representation acting on the $i$-th physical qudit.} \\[0.6em]

\parbox[t]{0.22\textwidth}{$T_a^{(i)}$} & \parbox[t]{0.72\textwidth}{The generators of $\mathfrak{su}(d)$ in the arbitrary representation acting on the $i$-th physical qudit.} \\[0.6em]

\parbox[t]{0.22\textwidth}{$T_a$} & \parbox[t]{0.72\textwidth}{The global generators of $\mathfrak{su}(d)$, acting on the whole physical space, see Eq.~\eqref{eq:global_generators}.} \\[0.6em]

\parbox[t]{0.22\textwidth}{$U_{\lambda_0}^{(i)}(g)$} & \parbox[t]{0.72\textwidth}{Representation of an element $g\in SU(d)$ on physical site $i$ equipped with irreducible representation $V_{\lambda_0}$ of $SU(d)$.} \\[0.6em]

\parbox[t]{0.22\textwidth}{$\pi_{\lambda_0}^{(i)}(X)$} & \parbox[t]{0.72\textwidth}{Representation of an element $X\in \mathfrak{su}(d)$ on physical site $i$ equipped with irreducible representation $V_{\lambda_0}$ of $\mathfrak{su}(d)$.} \\[0.6em]

\parbox[t]{0.22\textwidth}{$U_L(g)$, $U_P(g)$} & \parbox[t]{0.72\textwidth}{The unitary representations of $g\in G$ on the logical and physical Hilbert spaces, so that covariance reads $\mathcal{V}U_L(g)=U_P(g)\mathcal{V}$, see Eq.~\eqref{eq:group_covar_cond}.} \\[0.6em]

\parbox[t]{0.22\textwidth}{$u(g)$} & \parbox[t]{0.72\textwidth}{Fundamental representation of $g\in SU(d)$ on the logical space.} \\[0.6em]

\parbox[t]{0.22\textwidth}{$\mathrm{Ad}$} & \parbox[t]{0.72\textwidth}{The adjoint representation of $SU(d)$.} \\[0.6em]

\parbox[t]{0.22\textwidth}{$\delta_{a b}$} & \parbox[t]{0.72\textwidth}{The Kronecker delta.} \\[0.6em]

\parbox[t]{0.22\textwidth}{$d_{abc}$} & \parbox[t]{0.72\textwidth}{The symmetric structure constants of $\mathfrak{su}(d)$, see Eq.~\eqref{eq:sudgenerators}, Eq.~\eqref{eq:sudnormalisation}, Eq.~\eqref{eq:sudidentities}.} \\[0.6em]

\parbox[t]{0.22\textwidth}{$f_{abc}$} & \parbox[t]{0.72\textwidth}{The antisymmetric structure constants of $\mathfrak{su}(d)$, see Eq.~\eqref{eq:sudgenerators}, Eq.~\eqref{eq:sudidentities}.} \\[0.6em]

\parbox[t]{0.22\textwidth}{$E_{rs}$} & \parbox[t]{0.72\textwidth}{The matrix units, see Eq.~\eqref{eq:matr_units}.} \\[0.6em]

\parbox[t]{0.22\textwidth}{$F_{rs}$} & \parbox[t]{0.72\textwidth}{The traceless matrix units, see Eq.~\eqref{eq:matr_units}.} \\[0.6em]

\end{tabular}

\subsection*{Exterior powers and fundamental \(SU(d)\) representations}

\begin{tabular}{@{}ll@{}}

\parbox[t]{0.22\textwidth}{$\bigwedge^r \mathbb{C}^d$} & \parbox[t]{0.72\textwidth}{The $r$-th exterior power of $\mathbb{C}^d$.} \\[0.6em]

\parbox[t]{0.22\textwidth}{$u \wedge v$} & \parbox[t]{0.72\textwidth}{The wedge product of vectors $u$ and $v$ in $\mathbb{C}^d$.} \\[0.6em]

\parbox[t]{0.22\textwidth}{$\omega_r$} & \parbox[t]{0.72\textwidth}{Fundamental weights of $SU(d)$, see Remark~\ref{rem:fundamental_weights}.} \\[0.6em]

\parbox[t]{0.22\textwidth}{$V_{\omega_r}$} & \parbox[t]{0.72\textwidth}{The irreducible representation of $SU(d)$ with highest weight $\omega_r$, which is isomorphic to $\bigwedge^r \mathbb{C}^d$, see definition \ref{def:rfundSUd} and Remark~\ref{rem:fundamental_weights}.} \\[0.6em]

\parbox[t]{0.22\textwidth}{$\kappa_{\omega_r}$} & \parbox[t]{0.72\textwidth}{Dynkin index of the representation $V_{\omega_r}$, or, equivalently, the normalization constant, see Eq.~\eqref{eq:sudnormalisation_2}.} \\[0.6em]

\parbox[t]{0.22\textwidth}{$\ket{\zeta_d}$} & \parbox[t]{0.72\textwidth}{The normalized totally antisymmetric state of $d$ qudits, the unique $SU(d)$-invariant vector in $(V_{\omega_1})^{\otimes d}$, see Eq.~\eqref{eq:singlet}.} \\[0.6em]

\parbox[t]{0.22\textwidth}{$\tau$} & \parbox[t]{0.72\textwidth}{The maximally mixed state $I_d/d$ on a single physical qudit.} \\[0.6em]

\parbox[t]{0.22\textwidth}{$\tau_s$} & \parbox[t]{0.72\textwidth}{The projector onto $\bigwedge^sV_{\omega_1}$, equal to the reduced state of any $s$ tensor factors of $\ket{\zeta_d}$, see Eq.~\eqref{eq:tau-s}; $\tau_1=\tau$.} \\[0.6em]

\parbox[t]{0.22\textwidth}{$\beta_{\ell,N}$} & \parbox[t]{0.72\textwidth}{The $\ell$-site reduced state of the background $\ket{\zeta_d}\bra{\zeta_d}^{\otimes q}$ on a uniformly random ordered tuple of $\ell$ distinct sites, averaged over the tuple; evaluated in Proposition~\ref{prop:beta-partition}, Eq.~\eqref{eq:beta-partition}.} \\[0.6em]

\end{tabular}

\subsection*{Groups and Lie algebras}

\begin{tabular}{@{}ll@{}}

\parbox[t]{0.22\textwidth}{$GL_d(\mathbb{C})$} & \parbox[t]{0.72\textwidth}{The group of invertible $d \times d$ complex matrices.} \\[0.6em]

\parbox[t]{0.22\textwidth}{$SU(d)$} & \parbox[t]{0.72\textwidth}{The group of $d \times d$ unitary matrices with determinant $1$.} \\[0.6em]

\parbox[t]{0.22\textwidth}{$\mathfrak{g}$} & \parbox[t]{0.72\textwidth}{A Lie algebra.} \\[0.6em]

\parbox[t]{0.22\textwidth}{$\mathfrak{g}_\C$} & \parbox[t]{0.72\textwidth}{The complexification of a Lie algebra $\mathfrak{g}$, see Definition~\ref{def:complexification}.} \\[0.6em]

\parbox[t]{0.22\textwidth}{$\mathfrak{gl}(d)$} & \parbox[t]{0.72\textwidth}{The Lie algebra of $d \times d$ complex matrices.} \\[0.6em]

\parbox[t]{0.22\textwidth}{$\mathfrak{sl}(d)$} & \parbox[t]{0.72\textwidth}{The Lie algebra of $d \times d$ complex matrices with zero trace.} \\[0.6em]

\parbox[t]{0.22\textwidth}{$\mathfrak{u}(d)$} & \parbox[t]{0.72\textwidth}{The Lie algebra of $d \times d$ skew-Hermitian matrices.} \\[0.6em]

\parbox[t]{0.22\textwidth}{$\mathfrak{su}(d)$} & \parbox[t]{0.72\textwidth}{The Lie algebra of $d \times d$ skew-Hermitian matrices with zero trace.} \\[0.6em]

\parbox[t]{0.22\textwidth}{$\mathfrak{su}(\mathcal{H})$} & \parbox[t]{0.72\textwidth}{The Lie algebra of skew-Hermitian traceless operators on \(\mathcal H\); if \(\dim\mathcal H=d\), then \(\mathfrak{su}(\mathcal H)\cong\mathfrak{su}(d)\).} \\[0.6em]

\parbox[t]{0.22\textwidth}{$\mathfrak{u}(1)$} & \parbox[t]{0.72\textwidth}{One-dimensional Lie algebra of purely imaginary numbers.} \\[0.6em]

\parbox[t]{0.22\textwidth}{$\Lie_{\R}\{i H_0, i H_1, \ldots, i H_k\}$} & \parbox[t]{0.72\textwidth}{The real Lie algebra generated by the set of Hamiltonians by nested commutators.} \\[0.6em]

\end{tabular}

\subsection*{Linear maps, intertwiners, and representation spaces}

\begin{tabular}{@{}ll@{}}

\parbox[t]{0.22\textwidth}{$V$, $W$} & \parbox[t]{0.72\textwidth}{Finite-dimensional complex vector spaces, or $G$-representations, when no particular one is meant; in the $SU(d)$ sections $V\cong\mathbb C^d$ carries the defining action, so that $S_\lambda(V)$ and $\Lambda^rV$ are its Schur functors and exterior powers.} \\[0.6em]

\parbox[t]{0.22\textwidth}{$\Hom_{\C}(V, W)$} & \parbox[t]{0.72\textwidth}{The space of complex linear maps from $V$ to $W$.} \\[0.6em]

\parbox[t]{0.22\textwidth}{$\End_{\C}(V)$} & \parbox[t]{0.72\textwidth}{The space of complex linear maps from $V$ to itself.} \\[0.6em]

\parbox[t]{0.22\textwidth}{$\End_{\C}(V)^G$} & \parbox[t]{0.72\textwidth}{The space of $G$-invariant complex linear maps from $V$ to itself.} \\[0.6em]

\parbox[t]{0.22\textwidth}{$\Hom_G(V, W)$} & \parbox[t]{0.72\textwidth}{The space of $G$-covariant linear maps from $V$ to $W$.} \\[0.6em]

\parbox[t]{0.22\textwidth}{$\mathcal{I}$} & \parbox[t]{0.72\textwidth}{Set of all irreducible representations of finite group $G$.} \\[0.6em]

\parbox[t]{0.22\textwidth}{$\mathcal{I}_{\mathcal{H}}$} & \parbox[t]{0.72\textwidth}{Set of all irreducible representations of finite group $G$ appearing in $\mathcal{H}$.} \\[0.6em]

\parbox[t]{0.22\textwidth}{$\Lambda$} & \parbox[t]{0.72\textwidth}{Set of all irreducible representations of $SU(d)$.} \\[0.6em]

\parbox[t]{0.22\textwidth}{$M_{\lambda}$} & \parbox[t]{0.72\textwidth}{Multiplicity space of the irreducible representation $V_{\lambda}$.} \\[0.6em]

\parbox[t]{0.22\textwidth}{$V_{\lambda}$} & \parbox[t]{0.72\textwidth}{Irreducible representation of $SU(d)$ with highest weight $\lambda$.} \\[0.6em]

\parbox[t]{0.22\textwidth}{$V_{d_{a},\lambda_a}$} & \parbox[t]{0.72\textwidth}{Irreducible representation of $SU(d_a)$ with highest weight $\lambda_a$ -- used when one has to distinguish between representations of $SU(d_a)$ for different $d_a$.} \\[0.6em]

\parbox[t]{0.22\textwidth}{$\chi^{\lambda}(g)$} & \parbox[t]{0.72\textwidth}{Character of the irreducible representation $V_{\lambda}$ of $G$, evaluated at $g \in G$.} \\[0.6em]

\parbox[t]{0.22\textwidth}{$q$, $N$} & \parbox[t]{0.72\textwidth}{The number of background singlet blocks and the number of background sites, so that $n=qd+1$ and $N=n-1=qd$.} \\[0.6em]

\parbox[t]{0.22\textwidth}{$W_{\omega_1}$} & \parbox[t]{0.72\textwidth}{The Schur--Weyl inclusion $V_{\omega_1}\otimes M_{\omega_1}\to(V_{\omega_1})^{\otimes n}$, an $SU(d)$- and $S_n$-intertwiner, see Eq.~\eqref{eq:W-inclusion}.} \\[0.6em]

\parbox[t]{0.22\textwidth}{$\ket m$} & \parbox[t]{0.72\textwidth}{A unit multiplicity vector in $M_{\omega_1}$; the remaining freedom in a covariant encoding once the isotypic component is fixed.} \\[0.6em]

\parbox[t]{0.22\textwidth}{$\ket{m_0}$} & \parbox[t]{0.72\textwidth}{The multiplicity vector of the seed isometry $\mathcal{V}_0$, see Eq.~\eqref{eq:seed-multiplicity-vector}.} \\[0.6em]

\parbox[t]{0.22\textwidth}{$\ket{m_s}$} & \parbox[t]{0.72\textwidth}{The flat multiplicity vector generated from the seed $s$, see Eq.~\eqref{eq:family_of_multiplicity_vecs}.} \\[0.6em]

\parbox[t]{0.22\textwidth}{$D_q$} & \parbox[t]{0.72\textwidth}{The dimension of $M_{\omega_1}\cong[q+1,q^{d-1}]$, see Lemma~\ref{lem:dimension_formula} for its value and Lemma~\ref{lem:dimension_bound} for the lower bound $D_q=e^{\Omega_d(n)}$ used throughout}. \\[0.6em]

\parbox[t]{0.22\textwidth}{$P_g$, $P_\pi$} & \parbox[t]{0.72\textwidth}{The permutation operator on the physical space, see Eq.~\eqref{eq:group_act_on_phys_space}; written $P_\pi$ when the permutation is denoted $\pi\in S_n$.} \\[0.6em]

\end{tabular}

\subsection*{Partitions, Schur functors, and symmetric polynomials}

\begin{tabular}{@{}ll@{}}

\parbox[t]{0.22\textwidth}{$\lambda=(\lambda_1,\lambda_2,\ldots,\lambda_d)$} & \parbox[t]{0.72\textwidth}{Partition of an integer, i.e. $\lambda_1\geq \lambda_2 \geq \cdots \geq \lambda_d\geq 0$ and $\sum_{i=1}^d \lambda_i=n$. Enumerates irreducible representations of $GL_d(\mathbb{C})$/$S_n$.} \\[0.6em]

\parbox[t]{0.22\textwidth}{$l(\lambda)$} & \parbox[t]{0.72\textwidth}{Length of the partition $\lambda$, i.e. number of non-zero parts in $\lambda$.} \\[0.6em]

\parbox[t]{0.22\textwidth}{$\lambda \unrhd \mu$} & \parbox[t]{0.72\textwidth}{Dominance order on partitions, i.e. $\lambda \unrhd \mu$ if and only if $\sum_{i=1}^k \lambda_i \geq \sum_{i=1}^k \mu_i$ for all $k$.} \\[0.6em]

\parbox[t]{0.22\textwidth}{$\lambda'$} & \parbox[t]{0.72\textwidth}{Conjugate partition to $\lambda$.} \\[0.6em]

\parbox[t]{0.22\textwidth}{$[\lambda]=[\lambda_1,\lambda_2,\ldots,\lambda_d]$} & \parbox[t]{0.72\textwidth}{Irreducible representation of $S_n$ corresponding to partition $\lambda$, namely the Specht module.} \\[0.6em]

\parbox[t]{0.22\textwidth}{$S_{\lambda}(V)$} & \parbox[t]{0.72\textwidth}{The Schur functor, i.e. irreducible representation of $GL_d(\mathbb{C})$, associated with the partition $\lambda$.} \\[0.6em]

\parbox[t]{0.22\textwidth}{$S_\lambda(V) \downarrow_{S U(d)}$} & \parbox[t]{0.72\textwidth}{Restriction of the $GL_d(\mathbb{C})$ representation $S_\lambda(V)$ to $SU(d)$.} \\[0.6em]

\parbox[t]{0.22\textwidth}{$K_{\lambda, \mu}$} & \parbox[t]{0.72\textwidth}{Kostka number.} \\[0.6em]

\parbox[t]{0.22\textwidth}{$s_{\lambda}$} & \parbox[t]{0.72\textwidth}{Schur polynomial corresponding to the partition $\lambda$.} \\[0.6em]

\parbox[t]{0.22\textwidth}{$h_{k}$} & \parbox[t]{0.72\textwidth}{Complete homogeneous symmetric polynomial of degree $k$.} \\[0.6em]

\parbox[t]{0.22\textwidth}{$h_{\lambda}$} & \parbox[t]{0.72\textwidth}{A shorthand for $h_{\lambda_1} h_{\lambda_2} \cdots h_{\lambda_d}$, where $\lambda=(\lambda_1,\lambda_2,\ldots,\lambda_d)$.} \\[0.6em]

\parbox[t]{0.22\textwidth}{$e_k$} & \parbox[t]{0.72\textwidth}{Elementary symmetric polynomial of degree $k$.} \\[0.6em]

\parbox[t]{0.22\textwidth}{$\lambda_q$} & \parbox[t]{0.72\textwidth}{The partition $(q+1,q^{d-1})$ of $n=qd+1$, whose Specht module is the multiplicity space, $[\lambda_q]\cong M_{\omega_1}$.} \\[0.6em]

\parbox[t]{0.22\textwidth}{$\operatorname{SYT}(\lambda)$} & \parbox[t]{0.72\textwidth}{The set of standard Young tableaux of shape $\lambda$.} \\[0.6em]

\parbox[t]{0.22\textwidth}{$\ket T$} & \parbox[t]{0.72\textwidth}{The Young--Yamanouchi basis vector of $[\lambda]$ indexed by the standard Young tableau $T$.} \\[0.6em]

\parbox[t]{0.22\textwidth}{$\operatorname{Part}([\ell])$} & \parbox[t]{0.72\textwidth}{The set of set partitions of $\{1,\ldots,\ell\}$.} \\[0.6em]

\parbox[t]{0.22\textwidth}{$(x)_r$} & \parbox[t]{0.72\textwidth}{The falling factorial $x(x-1)\cdots(x-r+1)$.} \\[0.6em]

\end{tabular}

\subsection*{Finite fields and finite groups}

\begin{tabular}{@{}ll@{}}

\parbox[t]{0.22\textwidth}{$\mathbb{F}_n$} & \parbox[t]{0.72\textwidth}{Finite field with $n$ elements.} \\[0.6em]

\parbox[t]{0.22\textwidth}{$\mathrm{AGL}(1, n)$} & \parbox[t]{0.72\textwidth}{The affine general linear group of degree 1 over the finite field $\mathbb{F}_n$, see Eq.~\eqref{eq:AGL}.} \\[0.6em]

\parbox[t]{0.22\textwidth}{$\operatorname{PGL}(2, n-1)$} & \parbox[t]{0.72\textwidth}{The projective general linear group of degree 2 over the finite field $\mathbb{F}_{n-1}$, see Eq.~\eqref{eq:3_transitive_subgroup}.} \\[0.6em]

\parbox[t]{0.22\textwidth}{$\mathbb F_{2^h}$, $h$} & \parbox[t]{0.72\textwidth}{The finite field used by the sampler and its degree $h=n\log d$, chosen so that every tableau label fits in one field element.} \\[0.6em]

\parbox[t]{0.22\textwidth}{$w(T)$, $x_T$} & \parbox[t]{0.72\textwidth}{The tableau word of $T$, recording the row index of each entry, and the element of $\mathbb F_{2^h}$ it labels.} \\[0.6em]

\parbox[t]{0.22\textwidth}{$s=(c_0,c_1,c_2,c_3)$} & \parbox[t]{0.72\textwidth}{The random seed of the sampler, four independent uniform elements of $\mathbb F_{2^h}$.} \\[0.6em]

\parbox[t]{0.22\textwidth}{$p_s(x)$} & \parbox[t]{0.72\textwidth}{The random cubic polynomial $c_0+c_1x+c_2x^2+c_3x^3$ over $\mathbb F_{2^h}$ generating the phases.} \\[0.6em]

\parbox[t]{0.22\textwidth}{$L$} & \parbox[t]{0.72\textwidth}{A surjective $\mathbb F_2$-linear map $\mathbb F_{2^h}\to\mathbb F_2^2$, used to extract two bits from $p_s(x_T)$.} \\[0.6em]

\parbox[t]{0.22\textwidth}{$z_s(T)$} & \parbox[t]{0.72\textwidth}{The fourth root of unity assigned to the tableau $T$ by the seed $s$; the family $\{z_s(T)\}_T$ is four-wise independent.} \\[0.6em]

\end{tabular}

\subsection*{Concentration and asymptotic notation}

\begin{tabular}{@{}ll@{}}

\parbox[t]{0.22\textwidth}{$\mathbb E_m$} & \parbox[t]{0.72\textwidth}{Expectation over a Haar-random unit vector $\ket m$ in the multiplicity space $M_{\omega_1}$.} \\[0.6em]

\parbox[t]{0.22\textwidth}{$r_k$} & \parbox[t]{0.72\textwidth}{The output dimension $d^{k+1}$ of the Choi matrix of a $k$-site channel, used in the entrywise-to-diamond bound, see Eq.~\eqref{eq:entry-diamond}.} \\[0.6em]

\parbox[t]{0.22\textwidth}{$M_{n,k}$} & \parbox[t]{0.72\textwidth}{The number of matrix entries controlled by the union bound, $M_{n,k}=\binom nkd^{2k+2}=\binom nkr_k^2$.} \\[0.6em]

\parbox[t]{0.22\textwidth}{$A_{S;ab;uv}$} & \parbox[t]{0.72\textwidth}{The operator on $M_{\omega_1}$ representing one matrix element of $\Phi_{S,m}$ as a quadratic form in $\ket m$, see Eq.~\eqref{eq:test-operator}.} \\[0.6em]

\parbox[t]{0.22\textwidth}{$O_{d,k}(\cdot)$, $\Omega_d(\cdot)$} & \parbox[t]{0.72\textwidth}{Asymptotic bounds in $n$ with implied constants depending only on the subscripted parameters, which are held fixed.} \\[0.6em]

\end{tabular}

\subsection*{Block decompositions and controllability notation}

\begin{tabular}{@{}ll@{}}

\parbox[t]{0.22\textwidth}{$\mathcal{W}_{\beta \alpha}$} & \parbox[t]{0.72\textwidth}{The linear space of homomorphisms between isotypic components $\alpha$ and $\beta$, see Eq.~\eqref{eq:endomorphism_decomposition}.} \\[0.6em]

\parbox[t]{0.22\textwidth}{$\mathbf{T}_{\beta \alpha}$} & \parbox[t]{0.72\textwidth}{The linear space of homomorphisms between irreducible representations $\alpha$ and $\beta$, see Eq.~\eqref{eq:tensor_product_factors}.} \\[0.6em]

\parbox[t]{0.22\textwidth}{$\mathbf{M}_{\beta \alpha}$} & \parbox[t]{0.72\textwidth}{The multiplicity space of homomorphisms between multiplicity spaces of irreducible representations $\alpha$ and $\beta$, see Eq.~\eqref{eq:tensor_product_factors}.} \\[0.6em]

\parbox[t]{0.22\textwidth}{$\mathcal{B}_{\beta \alpha}$} & \parbox[t]{0.72\textwidth}{The breaker generated subspace, see Definition~\ref{def:breaker_generated_subspace}.} \\[0.6em]

\parbox[t]{0.22\textwidth}{$S_{\beta \alpha}$} & \parbox[t]{0.72\textwidth}{The first-factor support space, see Definition~\ref{def:breaker_generated_subspace}.} \\[0.6em]

\parbox[t]{0.22\textwidth}{$\Gamma_{\mathrm{co}}$} & \parbox[t]{0.72\textwidth}{The coupling graph, see Definition~\ref{def:coupling_graph}.} \\[0.6em]

\parbox[t]{0.22\textwidth}{$\mathcal{L}$} & \parbox[t]{0.72\textwidth}{The Lie algebra generated by $G$-invariant Lie subalgebra and breaking Hamiltonians, see Eq.~\eqref{eq:DLA_with_breakers}.} \\[0.6em]

\parbox[t]{0.22\textwidth}{$\mathcal{L}_{\C}$} & \parbox[t]{0.72\textwidth}{The complexification of $\mathcal{L}$, see Eq.~\eqref{eq:complexified_DLA_with_breakers}.} \\[0.6em]

\parbox[t]{0.22\textwidth}{$\mathcal{G}^{\text {br }}$} & \parbox[t]{0.72\textwidth}{The set of all breaking Hamiltonians.} \\[0.6em]

\parbox[t]{0.22\textwidth}{$\Pi_{\beta \alpha}(X)$} & \parbox[t]{0.72\textwidth}{The projection of operator $X$ onto the homomorphism space $\mathcal{W}_{\beta \alpha}$, see Definition~\ref{def:projection_onto_block}.} \\[0.6em]

\end{tabular}

\section{Preliminaries}\label{sec:suppl_prelim}

\begin{smdefinition}\label{def:complementary_channel}
Let \(\mathcal N:\End(\mathcal H_A)\to \End(\mathcal H_B)\) be a quantum channel, and let
$
    W:\mathcal H_A\to \mathcal H_B\otimes \mathcal H_E
$
be a Stinespring isometry for \(\mathcal N\), so that
$
    \mathcal N(\rho)=\Tr_E\left(W\rho W^\dagger\right).
$
\textit{The complementary channel} associated with this Stinespring isometry is the channel
$
    \widehat{\mathcal N}:\End(\mathcal H_A)\to \End(\mathcal H_E)
$
defined by
$$
    \widehat{\mathcal N}(\rho)
    :=
    \Tr_B\left(W\rho W^\dagger\right).
$$
\end{smdefinition}

\begin{smdefinition}\label{def:root_fidelity}
For two quantum states \(\rho,\sigma\in \End(\mathcal H)\), we define their fidelity by
\[
    f(\rho,\sigma)
    :=
    \Tr\sqrt{\sqrt{\rho}\,\sigma\,\sqrt{\rho}}.
\]
\end{smdefinition}

\begin{smdefinition}\label{def:rho_entanglement_fidelity}
Let
$
    \mathcal N,\mathcal M:\End(\mathcal H_A)\to \End(\mathcal H_B)
$
be quantum channels, and let \(\rho\in\End(\mathcal H_A)\) be a quantum state. Let
$
    |\psi_\rho\rangle\in \mathcal H_A\otimes \mathcal H_R
$
be any purification of \(\rho\). The input-state-dependent entanglement fidelity between \(\mathcal N\) and \(\mathcal M\) is defined as
\[
    F_\rho(\mathcal N,\mathcal M)
    :=
    f\left(
    (\mathcal N\otimes \mathrm{id}_R)(|\psi_\rho\rangle\langle\psi_\rho|),
    (\mathcal M\otimes \mathrm{id}_R)(|\psi_\rho\rangle\langle\psi_\rho|)
    \right).
\]
This quantity is independent of the chosen purification of \(\rho\).
\end{smdefinition}

\begin{smdefinition}\label{def:worst_case_entanglement_fidelity}
Let
$
    \mathcal N,\mathcal M:\End(\mathcal H_A)\to \End(\mathcal H_B)
$
be quantum channels. The worst-case entanglement fidelity between \(\mathcal N\) and \(\mathcal M\) is defined by
\[
    F(\mathcal N,\mathcal M)
    :=
    \min_{\rho\in\mathcal D(\mathcal H_A)}
    F_\rho(\mathcal N,\mathcal M),
\]
where \(\mathcal D(\mathcal H_A)\) denotes the set of density operators on \(\mathcal H_A\).
\end{smdefinition}

\section{\texorpdfstring{$SU(2)$}{SU(2)}-covariant codes}\label{sec:supple_su2}

Let \(V_j\) denote the spin-\(j\) irreducible representation of \(SU(2)\), so that
$
    \dim V_j=2j+1.
$
We write \(J_x,J_y,J_z\) for the Hermitian spin-\(j\) generators acting on \(V_j\), normalized by
\begin{equation}\label{eq:spinj}
    [J_x,J_y]=iJ_z,
    \qquad
    J_x^2+J_y^2+J_z^2=j(j+1)I_{2j+1}.
\end{equation}
Equivalently, in the standard angular-momentum basis
\(\{|m\rangle:m=-j,-j+1,\ldots,j\}\), one has
\begin{equation}\label{eq:Jz}
    J_z|m\rangle=m|m\rangle,
\end{equation}
and \(J_\pm:=J_x\pm iJ_y\) act by
\begin{equation}\label{eq:Jpm}
    J_\pm |m\rangle
    =
    \sqrt{j(j+1)-m(m\pm1)}\,|m\pm1\rangle .
\end{equation}

\begin{smproposition}\label{prop:SU2cov}
    Let
    $
        \mathcal H_L=V_{1/2},\;
        \mathcal H_P=(V_j)^{\otimes n}.
    $
    Equip \(\mathcal H_L\) with the fundamental spin-\(1/2\) representation of
    \(\mathfrak{su}(2)\), and equip \(\mathcal H_P\) with the transversal spin-\(j\)
    representation given by
    \begin{equation}
    \begin{gathered}
    \frac{i \sigma_a}{2} \rightarrow i J_a \equiv \sum_{k=1}^n i J_a^{(k)}, \quad a=x, y, z,
    \end{gathered}
    \end{equation} 
    Then there exists a code space $\mathfrak{su}(2)$-covariant encoding $\mathcal{E}$ with respect to defined physical and logical representations,
    such that the state seen by the environment after erasure of any single physical qudit is of the form
    \begin{equation}\label{eq:reducedStateSU2}
        \rho^{(i)}\left(\rho_L\right)=\frac{I_{2 j+1}}{2 j+1}+\frac{3}{2 n j(j+1)(2 j+1)} \sum_a r_a J_a^{(i)}, \quad i=1, \ldots, n.
    \end{equation}
    Namely, this code space is one of the copies of the spin-$\frac{1}{2}$ fundamental representation in $\left(V_j\right)^{\otimes n}$ which is invariant under cyclic permutations of physical qudits.
    
\end{smproposition}     

\begin{proof}
    We first prove that
    \[
        \mathcal{V}^{\dagger} J_a^{(i)} \mathcal{V}=\alpha^{(i)} \frac{\sigma_a}{2}, \quad a=x, y, z, \quad \sum_{i=1}^n \alpha^{(i)} = 1,
    \]
    where $\mathcal{V}$ is an encoding isometry. We emphasize that this is a general fact that holds for any $SU(2)$-covariant encoding, and does not depend on the choice of code space.

    Let's denote $K_a^{(i)}:=\mathcal{V}^{\dagger} J_a^{(i)} \mathcal{V}$, so we have to prove
    \[
        K_a^{(i)}=\alpha^{(i)} \frac{\sigma_a}{2}.
    \]
    Recalling the covariance condition, we have
    \[
    \left(\bigotimes_{i=1}^n U^{(i)}(g)\right) \mathcal{V}=\mathcal{V} u(g), \quad  \mathcal{V}^{\dagger}\left(\bigotimes_{i=1}^n U^{(i)}(g)\right)= u(g) \mathcal{V}^{\dagger},
    \]
    where $u(g)$ is the fundamental representation of $SU(2)$ on the logical qubit and $U^{(i)}(g)$ is the spin-$j$ representation of $SU(2)$ on the physical qudits.  
    Components $\{J_a^{(i)}\}_{a=x, y, z}$ form a vector operator, meaning that they transform under the adjoint action of $SU(2)$ as follows
    \[
    \left(\bigotimes_{k=1}^n U^{(k)}(g)\right)  J_a^{(i)} \left(\bigotimes_{k=1}^n U^{(k)}(g)\right)^{\dagger} =\sum_{b=x, y, z} R_{a b}(g) J_b^{(i)},
    \]
    where $R_{a b}(g)$ is the spin-$1$ representation corresponding to the group element $g\in SU(2)$. 
    \[
    u(g) K_a^{(i)} u(g)^{\dagger}=u(g) \mathcal{V}^{\dagger} J_a^{(i)} \mathcal{V} u(g)^{\dagger}=\mathcal{V}^{\dagger}\left(\bigotimes_{k=1}^n U^{(k)}(g)\right)  J_a^{(i)}\left(\bigotimes_{k=1}^n U^{(k)}(g)\right)  \mathcal{V}=\sum_b R_{a b}(g) \mathcal{V}^{\dagger} J_b^{(i)} \mathcal{V}
    \]
    So
    \[
    u(g) K_a^{(i)} u(g)^{\dagger}=\sum_b R_{a b}(g) K_b^{(i)},
    \]
    which means that $\{K_a^{(i)}\}_{a=x, y, z}$ also form a vector operator (e. g. forms spin-1 representation) on logical space. 
    There is linear transformation that maps $\{K_a^{(i)}\}_{a=x, y, z}$ to $\{\frac{\sigma_a}{2}\}_{a=x, y, z}$,
    \[
    K_a^{(i)}=\sum_b M_{a b}^{(i)} \frac{\sigma_b}{2},
    \]
    which also transform as vector operators
    $
    u(g) \frac{\sigma_a}{2} u(g)^{\dagger}=\sum_b R_{a b}(g) \frac{\sigma_b}{2}.
    $
    From simple algebraic manipulations, we get 
    \[
    u(g) K_a^{(i)} u(g)^{\dagger}=\sum_b M_{a b}^{(i)} u(g) \frac{\sigma_b}{2} u(g)^{\dagger}=\sum_{b c} M_{a b}^{(i)} R_{b c}(g) \frac{\sigma_c}{2},
    \]
    and on the other hand 
    \[
    u(g) K_a^{(i)} u(g)^{\dagger}=\sum_b R_{a b}(g) K_b^{(i)}=\sum_{b c} R_{a b}(g) M_{b c}^{(i)} \frac{\sigma_c}{2},
    \]
    which leads to the following condition on the matrix $M^{(i)}$:
    \[
    M^{(i)} R(g)=R(g) M^{(i)}, \quad \forall g \in SU(2).
    \]
    From Schur's lemma, since $R(g)$ is irreducible, we have that $M^{(i)}$ is proportional to the identity matrix, meaning that
    \[
    K_a^{(i)}=\alpha^{(i)} \frac{\sigma_a}{2},
    \]
    which is what we wanted to prove. From the fact that 
    \[
    \sum_{i=1}^n \mathcal{V}^{\dagger} J_a^{(i)} \mathcal{V} = \mathcal{V}^{\dagger} \left(\sum_{i=1}^n J_a^{(i)}\right) \mathcal{V} = \mathcal{V}^{\dagger} J_a \mathcal{V} = \frac{\sigma_a}{2},
    \]
    we get
    $
    \sum_{i=1}^n \alpha^{(i)} \frac{\sigma_a}{2} = \frac{\sigma_a}{2},
    $
    which implies $\sum_{i=1}^n \alpha^{(i)} = 1$.
    Now we are to prove that
    \begin{equation}\label{supple_reduced_state}
        \rho^{(i)}\left(\rho_L\right)=\frac{I_{2 j+1}}{2 j+1}+\beta^{(i)} \sum_{a=x, y, z} r_a J_a
    \end{equation}
    where
    $
        \beta^{(i)}=\frac{3 \alpha^{(i)}}{2 j(j+1)(2 j+1)}.
    $
    Consider the Bloch vector representation of the logical state $\rho_L$ 
    \[
    \rho_L=\frac{1}{2}(I+\mathbf{r} \cdot \boldsymbol{\sigma}).
    \]
    We first prove that the map $\rho_L \rightarrow \rho^{(i)}\left(\rho_L\right)$ is a covariant map with respect to $SU(2)$. Indeed,
    defining the encoded state
    $
    \rho_{\mathrm{enc}}\left(\rho_L\right)=\mathcal{V} \rho_L \mathcal{V}^{\dagger}
    $
    we have
    \[
    \rho_{\text{enc}}\left(u(g) \rho_L u(g)^{\dagger}\right)=\left(\bigotimes_{k=1}^n U^{(k)}(g)\right)\rho_{\text{enc}}\left(\rho_L\right) \left(\bigotimes_{k=1}^n U^{(k)}(g)\right)^{\dagger},
    \]
    Now trace out all sites except the site $i$. Since partial trace commutes with conjugation on the kept subsystem,
    $
    \rho^{(i)}\left(u(g) \rho_L u(g)^{\dagger}\right)=U^{(i)}(g) \rho^{(i)}\left(\rho_L\right) (U^{(i)}(g))^{\dagger},
    $
    from which covariance of the map $\rho_L \rightarrow \rho^{(i)}\left(\rho_L\right)$ follows.
    Since the reduced state map is linear
    $
        \rho^{(i)}\left(\rho_L\right)=A_0+\sum_a r_a A_a, \quad \forall \rho_L  \vec{r}
    $
    where $A_0, A_x, A_y, A_z$ are some operators on the $i$-th physical qudit independent on the logical state. Let's find these operators. 
    First, we find $A_0$ by plugging in the maximally mixed state $\rho_L = \frac{I}{2}$, which gives
    $
    u(g) \frac{I}{2} u(g)^{\dagger}=\frac{I}{2},
    $
    and from the covariance condition, one gets
    $
    U^{(i)}(g) A_0 (U^{(i)}(g))^{\dagger}=A_0 \quad \forall g \in SU(2),
    $
    which basically means that $A_0$ is an intertwiner operator of this spin-$j$ representation of $SU(2)$, and from Schur's lemma we get
    $
    A_0=c I_{2 j+1},
    $
    and constant $c$ can easily be found from the normalization condition $\Tr \rho^{(i)}\left(\frac{I}{2}\right)=1$, which gives $c=\frac{1}{2 j+1}$. 
    
    Now we find $A_a$ for $a=x, y, z$. We have
    \[
    u(g) \rho_L u(g)^{\dagger}=\frac{1}{2} I+\frac{1}{2} \sum_a(R(g) \mathbf{r})_a \sigma_a,
    \]
    where $R(g)$ is spin-1 representation corresponding to the group element $g\in SU(2)$. From the covariance condition, we have
    \[
    A_0+\sum_a(R(g) \mathbf{r})_a A_a=U^{(i)}(g)\left(A_0+\sum_a r_a A_a\right) U^{(i)}(g)^{\dagger},
    \]
    and since $A_0$ is proportional to the identity, we get
    \[
    \sum_a(R(g) \mathbf{r})_a A_a=\sum_a r_a U^{(i)}(g) A_a U^{(i)}(g)^{\dagger},
    \]
    and since $\vec{r}$ is arbitrary
    \[
    U^{(i)}(g) A_a (U^{(i)}(g))^{\dagger}=\sum_b R_{a b}(g) A_b.
    \]
    This means that $\{A_a\}_{a=x, y, z}$ also form a vector operator on the physical space.
    By the same argument as before, we have that $A_a$ is proportional to $J_a$
    \[
    \vec{A}=\beta^{(i)} \vec{J},
    \] 
    which proves the desired form of $\rho^{(i)}\left(\rho_L\right)$. We are now to find coefficients $\beta^{(i)}$. 
    Using formula~\eqref{supple_reduced_state} we get
    \[
    \operatorname{Tr}\left(\rho^{(i)}\left(\rho_L\right) J_a\right)=\operatorname{Tr}\left(\frac{I}{2 j+1} J_a\right)+\beta^{(i)} \sum_b r_b \operatorname{Tr}\left(J_b J_a\right).
    \]
    where inner products can be computed from the standard derivations of spin-$j$ representation of $\mathfrak{su}(2)$
    \[
    \operatorname{Tr}\left(\frac{I}{2 j+1} J_a\right)=0, \quad \operatorname{Tr}\left(J_b J_a\right)=\delta_{a b} \frac{1}{3} j(j+1)(2 j+1),
    \]
    so we get
    \[
    \operatorname{Tr}\left(\rho^{(i)}\left(\rho_L\right) J_a\right)=\beta^{(i)} \frac{1}{3} j(j+1)(2 j+1) r_a .
    \]
    On the other hand, using properties of $\Tr$ and earlier derivations, we obtain
    \[
    \operatorname{Tr}\left(\rho^{(i)}\left(\rho_L\right) J_a\right)=\operatorname{Tr}\left(\mathcal{V} \rho_L \mathcal{V}^{\dagger} J_a^{(i)}\right)=\operatorname{Tr}\left(\rho_L \mathcal{V}^{\dagger} J_a^{(i)} \mathcal{V}\right)=\alpha^{(i)} \operatorname{Tr}\left(\rho_L \frac{\sigma_a}{2}\right)=\alpha^{(i)} \frac{r_a}{2},
    \]
    which leads to the following expression for $\beta^{(i)}$:
    \[
    \beta^{(i)}=\frac{3 \alpha^{(i)}}{2 j(j+1)(2 j+1)} .
    \]
    We now move to the second part of the proof, which is to show that we can always find a code space such that $\alpha^{(i)} = \frac{1}{n}$ for all $i$.
    Consider the cyclic permutation operator $T$ on the physical space defined as follows
    \[
        T\left(v_1 \otimes v_2 \otimes \cdots \otimes v_n\right)=v_n \otimes v_1 \otimes \cdots \otimes v_{n-1} \Longrightarrow T J_a^{(i)} T^{\dagger}=J_a^{(i+1)},
    \]
    where $v_j \in V_j$. From the Shur-Weyl duality, we have that $T$ acts as identity on the representation space $V_{1 / 2}$,
     and as some unitary operator $\tau$ on the multiplicity space $M_{1 / 2}$, meaning that
    $
        T_{|_{V_{1 / 2} \otimes M_{1 / 2}}}=I_{V_{1 / 2}} \otimes \tau.
    $
        Since $T^n=I$, we have $\tau^n=I$, which means that $\tau$ is diagonal in some basis of $M_{1 / 2}$, and its eigenvalues are $n$-th roots of the unity. Let's pick up
        some of these eigenvalues and the corresponding eigenvectors, which we denote as $e^{i \phi}$ and $|m\rangle$ respectively, so we have
    $
        \tau|m\rangle=e^{i \phi}|m\rangle .
    $
     We finally define our code space as follows
    $
        \mathcal{C}:=V_{1 / 2} \otimes|m\rangle.
    $
    The reason for this definition is that $T \mathcal{V}=e^{i \phi} \mathcal{V}$, which leads us to
    the following chain of equations
    \[
        \mathcal{V}^{\dagger} J_a^{(i)} \mathcal{V}=\mathcal{V}^{\dagger} (T^{j-i})^{\dagger} J_a^{(j)} T^{j-i} \mathcal{V}=\mathcal{V}^{\dagger} J_a^{(j)} \mathcal{V}
    \]
    for every $i,j$, from which it instantly follows that
    $
        \alpha^{(1)}=\alpha^{(2)}=\cdots=\alpha^{(n)}=: \alpha,
    $
    and using condition $\sum_i\alpha^{(i)}=1$ we get $\alpha=\frac{1}{n}$, which completes the proof.
\end{proof}

\begin{smlemma}\label{fidelity_opt}
Let $G$ be a compact group, let
$
U: G \rightarrow \mathcal{U}\left(\mathcal{H}_{L}\right)
$
be an irreducible unitary representation on the $\mathcal{H}_{L}$ and let
$
\mathcal{N}, \mathcal{M}: \operatorname{End}\left(\mathcal{H}_L\right) \rightarrow \operatorname{End}\left(\mathcal{H}_{P}\right)
$
be quantum channels such that there exists a unitary representation
$
U_P: G \rightarrow \mathcal{U}\left(\mathcal{H}_{P}\right)
$
with
\begin{equation}\label{eq:channel_covar}
    \mathcal{N}\left(U(g) X U(g)^{\dagger}\right)=U_P(g) \mathcal{N}(X) U_P(g)^{\dagger}, \quad \mathcal{M}\left(U(g) X U(g)^{\dagger}\right)=U_P(g) \mathcal{M}(X) U_P(g)^{\dagger}
\end{equation}
for all $g \in G$ and $X \in \operatorname{End}\left(\mathcal{H}_{L}\right)$.
Then the worst-case entanglement fidelity
$
\rho \longmapsto F_\rho(\mathcal{N}, \mathcal{M})
$
is convex and $G$-invariant. Consequently,

$$
F(\mathcal{N}, \mathcal{M}):=\min _\rho F_\rho(\mathcal{N}, \mathcal{M})=F_{I_{L} / \operatorname{dim} \mathcal{H}_{L}}(\mathcal{N}, \mathcal{M}) .
$$
\end{smlemma}

\begin{proof}
    We first prove convexity in $\rho$. Let $W_{\mathcal N}:\mathcal H_L\to\mathcal H_P\otimes\mathcal H_E$ and $W_{\mathcal M}:\mathcal H_L\to\mathcal H_P\otimes\mathcal H_E$ be Stinespring isometries of $\mathcal N$ and $\mathcal M$. Then $(W_{\mathcal N}\otimes I_R)\ket{\psi_\rho}$ and $(W_{\mathcal M}\otimes I_R)\ket{\psi_\rho}$ purify $(\mathcal N\otimes\operatorname{id})\psi_\rho$ and $(\mathcal M\otimes\operatorname{id})\psi_\rho$, and every other purification is obtained by acting on $\mathcal H_E$ alone. Uhlmann's theorem~\cite{KretschmannSchlingemannWerner2008} therefore gives
    $
        F_\rho(\mathcal N,\mathcal M)=\max_W\left|\operatorname{Tr}\left(\rho X_W\right)\right|,
        \quad
        X_W:=W_{\mathcal M}^\dagger\left(I_P\otimes W\right)W_{\mathcal N},
    $
    the maximum running over unitaries $W$ on $\mathcal H_E$, where we used $\bra{\psi_\rho}(X\otimes I_R)\ket{\psi_\rho}=\operatorname{Tr}(\rho X)$ for the purification below. 
    For each fixed $W$ the map $\rho\mapsto\left|\operatorname{Tr}(\rho X_W)\right|$ is the absolute value of an affine function of $\rho$, hence convex, and a pointwise maximum of convex functions is convex. 
    
    Now, we prove $G$-invariance. Recall that for an arbitrary state $\rho$ we can write its purification as
    $
        \left|\psi_\rho\right\rangle=\left(I \otimes \sqrt{\rho^T}\right)|\Phi\rangle,
    $
    where $|\Phi\rangle$ is a maximally entangled state. Using this identity, we will show that
    $
        \left|\psi_{U(g) \rho U(g)^{\dagger}}\right\rangle=\left(U(g) \otimes \overline{U(g)}\right)\left|\psi_\rho\right\rangle.
    $
    Indeed, from the uniqueness of the square root operator, one can easily check that
    $
        \sqrt{U(g) \rho (g)^{\dagger}}=U(g) \sqrt{\rho} U(g)^{\dagger},
    $
    so we can write the following chain of equalities
    $
        \sqrt{\left(U(g) \rho U(g)^{\dagger}\right)^T}=\left(U(g) \sqrt{\rho} U(g)^{\dagger}\right)^T=\overline{U(g)} \sqrt{\rho^T} U(g)^T.
    $
    From the defined purification, we have
    $
        \left|\psi_{U(g) \rho U(g)^{\dagger}}\right\rangle=\left(I \otimes \overline{U(g)} \sqrt{\rho^T} U(g)^T\right)|\Phi\rangle.
    $
    For a maximally entangled state, we have the following identity
    $
        \left(I \otimes U(g)^T\right)|\Phi\rangle=\left(U(g) \otimes I\right)|\Phi\rangle,
    $
    using which we can write the following chain of equalities
    \[
        \begin{aligned}
        \left|\psi_{U(g) \rho U(g)^{\dagger}}^{\dagger}\right\rangle & =\left(I \otimes \overline{U(g)} \sqrt{\rho^T}\right)\left(I \otimes U(g)^T\right)|\Phi\rangle \\
        & =\left(I \otimes \overline{U(g)} \sqrt{\rho^T}\right)\left(U(g) \otimes I\right)|\Phi\rangle \\
        & =\left(U(g) \otimes \overline{U(g)} \sqrt{\rho^T}\right)|\Phi\rangle \\
        & =\left(U(g) \otimes \overline{U(g)}\right)\left(I \otimes \sqrt{\rho^T}\right)|\Phi\rangle \\
        & =\left(U(g) \otimes \overline{U(g)}\right)\left|\psi_\rho\right\rangle .
        \end{aligned}
    \]
    which proves the desired identity. Now, using this identity and covariance of channels~\eqref{eq:channel_covar}, we get the following
    \[
        (\mathcal{N} \otimes \mathrm{id})\left(\left|\psi_{U(g) \rho U(g)^{\dagger}}\right\rangle\left\langle\psi_{U(g) \rho U(g)^{\dagger}}\right|\right)=\left(U_P(g) \otimes \overline{U_P(g)}\right)(\mathcal{N} \otimes \mathrm{id})\left(\left|\psi_\rho\right\rangle\left\langle\psi_\rho\right|\right)\left(U_P(g)^{\dagger} \otimes \overline{U_P(g)}^{\dagger}\right)
    \]
    and analogously for $\mathcal{M}$. Recalling that fidelity is invariant under unitary conjugation, we get
    \[
        F_{U(g) \rho U(g)^{\dagger}}(\mathcal{N}, \mathcal{M})=F_\rho(\mathcal{N}, \mathcal{M}),
    \]
    which proves $G$-invariance. 

    We finally prove the optimality of the chaotic state. Let $dg$ be normalized Haar measure on $G$, and define the twirl
$
\bar{\rho}:=\int_G U(g) \rho U(g)^{\dagger} dg
$
By convexity,
$
F_{\bar{\rho}}(\mathcal{N}, \mathcal{M}) \leq \int_G F_{U(g) \rho U(g)^{\dagger}}(\mathcal{N}, \mathcal{M}) d g
$
By $G$-invariance, the integrand is constant and equal to $F_\rho(\mathcal{N}, \mathcal{M})$. Hence
$
F_{\bar{\rho}}(\mathcal{N}, \mathcal{M}) \leq F_\rho(\mathcal{N}, \mathcal{M})
$
for every $\rho$.
Now we use the irreducibility of $U$. By Schur's lemma, the twirled state must be a scalar multiple of the identity:
$
    \bar{\rho}=\frac{I_{L}}{\operatorname{dim} \mathcal{H}_{L}}
$
Therefore
$
F_{I_{L} / \operatorname{dim} \mathcal{H}_{L}}(\mathcal{N}, \mathcal{M}) \leq F_\rho(\mathcal{N}, \mathcal{M}) \quad \forall \rho,
$
which proves the desired optimality of the chaotic state.
\end{proof}

\begin{smproposition}\label{prop:SU2covFidelity}
    For the encoding $\mathcal{E}$ defined in Proposition~\ref{prop:SU2cov} and noise channel $\mathcal{N}(\sigma)=\sum_{i=1}^n p_i\ket{i}\bra{i}_F \otimes \ket{e}\bra{e}_{A_i} \otimes \operatorname{Tr}_i(\sigma)$ 
    the worst-case entanglement fidelity for the dual channel is
    \[
        F(\widehat{\mathcal{N} \circ \mathcal{E}}, \Lambda_0)=1-\frac{9}{32 n^2 j(j+1)}+O_{j}\left(n^{-3}\right),
    \]
    where we take $\Lambda_0(\rho)=\operatorname{Tr}(\rho) \omega_{\text{flag}}\otimes\frac{I_{V_{j}}}{2 j+1}$, where $\omega_{\text{flag}}=\sum p_i|i\rangle\langle i|_{F_E}$.
\end{smproposition}

\begin{proof}
The single-qudit erasure noise channel is defined as follows
$
    \widehat{\mathcal{N} \circ \mathcal{E}}(\rho)=\sum_i p_i|i\rangle\langle i|_{F_E} \otimes \rho^{(i)}(\rho) .
$
For our choice of code, all $\rho^{(i)}(\sigma)$ (\ref{eq:reducedStateSU2}) can safely identified to be the same, so
$$
\widehat{\mathcal{N} \circ \mathcal{E}}(\rho)=\omega_{\text{flag}} \otimes \left(\frac{I_{V_j}}{2 j+1}+\frac{3}{2 n j(j+1)(2 j+1)} \sum_a r_a J_a^{(1)}\right),
$$
where
$
\omega_{\text{flag}}=\sum p_i|i\rangle\langle i|_{F_E}.$ From now on, we omit index $1$ for brevity, 
but one should not confuse these local operators with global operators defined in Eq.~\eqref{eq:reducedStateSU2} -- we will not use them in this proposition. 
We define a constant channel $\Lambda$ to be the channel that maps any state to the chaotic state $\frac{I_{V_j}}{2 j+1}$, so we have
\[
    \Lambda_0(\rho)=\operatorname{Tr}(\rho)\omega_{\text{flag}}\otimes \frac{I_{V_j}}{2 j+1} \equiv \operatorname{Tr}(\rho)\cdot\omega_{\text{flag}}\otimes\tau.
\]
It is obvious that the flag-state $\omega_{\text{flag}}$ doesn't influence the computation of the fidelity, so we drop it from now on. 
From Lemma~\ref{fidelity_opt} we know that the optimal state for computing worst-case entanglement fidelity is the chaotic state $\frac{I_{{V_{1/2}}}}{2}$, 
which purification is the maximally entangled state $\ket{\Psi}=\frac{1}{\sqrt{2}}(|00\rangle+|11\rangle)$. 
Now we are to compute the actions of channels $\widehat{\mathcal{N} \circ \mathcal{E}}$ and $\Lambda_0$ on the $\ket{\Psi}\bra{\Psi}$. 
We first express basis operators $\ket{i}\bra{j}$ in terms of Pauli matrices
\[
    \ket{0}\bra{0}=\frac{I}{2} + \frac{1}{2} \sigma_z, \quad \ket{1}\bra{1}=\frac{I}{2} - \frac{1}{2} \sigma_z, \quad \ket{0}\bra{1}=\frac{1}{2} \sigma_x + \frac{i}{2} \sigma_y, \quad \ket{1}\bra{0}=\frac{1}{2} \sigma_x - \frac{i}{2} \sigma_y
\]
And, using notation $\tau:=\frac{I_{V_j}}{2j+1},$ and computation of reduced state $\rho^{(i)}\left(\rho_L\right)$ from Proposition~\ref{prop:SU2cov}

\[
    \widehat{\mathcal{N} \circ \mathcal{E}}(\ket{0}\bra{0})=\tau+\beta J_z, \quad \widehat{\mathcal{N} \circ \mathcal{E}}(\ket{1}\bra{1})=\tau-\beta J_z, \quad \widehat{\mathcal{N} \circ \mathcal{E}}(\ket{0}\bra{1})=\beta J_+, \quad \widehat{\mathcal{N} \circ \mathcal{E}}(\ket{1}\bra{0})=\beta J_-
\]
where we denoted $\beta:=\frac{3}{2 n j(j+1)(2 j+1)}$, gives us the action on the purification of the logical state
\[
    \begin{gathered}
    \sigma_{1/2}:=(\widehat{\mathcal{N} \circ \mathcal{E}} \otimes \mathrm{id})(\ket{\Psi}\bra{\Psi})=\frac{1}{2}\left(\tau+\beta J_z\right) \otimes|0\rangle\langle 0|+\frac{1}{2}\left(\tau-\beta J_z\right) \otimes|1\rangle\langle 1| \\
    +\beta \frac{1}{2}\left(J_{+} \otimes|0\rangle\langle 1|+J_{-} \otimes|1\rangle\langle 0|\right) .
    \end{gathered}
\]
The constant channel $\Lambda_0$ maps any state to the chaotic state $\tau$, so we have
\[
    \eta_{1/2}=(\Lambda_0 \otimes \mathrm{id})(\ket{\Psi}\bra{\Psi})=\tau \otimes \frac{I_2}{2}
\]

Now let's express 
$
    f(\eta_{1/2},\sigma_{1/2}) = \operatorname{Tr} \sqrt{\sqrt{\eta_{1/2}} \sigma_{1/2} \sqrt{\eta_{1/2}}} .
$
Direct computations result in the following
\[ 
\begin{gathered}
    \sqrt{\eta_{1/2}} \sigma_{1/2} \sqrt{\eta_{1/2}}=\frac{1}{4(2 j+1)}\begin{pmatrix}
\left(\tau+\beta J_z\right) & \beta J_+ \\
\beta J_- & \left(\tau-\beta J_z\right)
    \end{pmatrix}.
 \end{gathered}
\]  
which leads us to the following expression for fidelity
\[
\begin{gathered}
    f(\eta_{1/2},\sigma_{1/2})=\operatorname{Tr} \sqrt{ \frac{1}{4(2 j+1)}\begin{pmatrix}
\left(\tau+\beta J_z\right) & \beta  J_+ \\
\beta  J_- & \left(\tau-\beta J_z\right)
    \end{pmatrix}} =\\
    = \frac{1}{2\sqrt{(2 j+1)}}\operatorname{Tr} \sqrt{ \begin{pmatrix}
\tau & 0 \\
0 & \tau
    \end{pmatrix}+\beta\begin{pmatrix}
J_z & J_+ \\
J_- & -J_z
    \end{pmatrix}}=\\
    =\frac{1}{2 (2j+1)} \operatorname{Tr}\left[I_{2(2j+1)}+\frac{\beta (2j+1)}{2} \begin{pmatrix}
J_z & J_+ \\
J_- & -J_z
    \end{pmatrix}-\frac{\beta^2 (2j+1)^2}{8} \begin{pmatrix}
J_z & J_+ \\
J_- & -J_z
    \end{pmatrix}^2+O_{j}\left(\beta^3\right)\right] =\\ =
    1+\frac{\beta}{4} \operatorname{Tr} \begin{pmatrix}
J_z & J_+ \\
J_- & -J_z
    \end{pmatrix}-\frac{\beta^2 (2j+1)}{16} \operatorname{Tr} \begin{pmatrix}
J_z & J_+ \\
J_- & -J_z
    \end{pmatrix}^2+O_{j}\left(\beta^3\right)=\\=1-\frac{(2 j+1)^2 j(j+1)}{8} \beta^2+O_{j}\left(\beta^3\right)
     \end{gathered}
\]
Substituting $\beta=\frac{3}{2 n j(j+1)(2 j+1)}$ we finally get
\[
     F(\widehat{\mathcal{N} \circ \mathcal{E}}, \Lambda_0)=F_{I/2}(\widehat{\mathcal{N} \circ \mathcal{E}}, \Lambda_0)=f(\eta_{1/2},\sigma_{1/2})=1-\frac{9}{32 n^2 j(j+1)}+O_{j}\left(n^{-3}\right),
\]
which completes the proof.
\end{proof}

\section{\texorpdfstring{$SU(d)$}{SU(d)}-covariant codes}\label{sec:supple_sud}

As an auxiliary object for proof of supporting lemmas, we will use the Lie group $GL_d(\mathbb{C})$ and its representations. 
There is a well-known one-to-one correspondence between irreducible representations of $GL_d(\mathbb{C})$ and partitions $\lambda$~\cite{FultonHarris1991}; we will
denote the representation corresponding to partition $\lambda$ as $S_\lambda V$, where $V\cong \mathbb{C}^d$.
Let \(\{e_1,\ldots,e_d\}\) be the standard basis of \(\mathbb C^d\). Consider a vector space
$
    \bigwedge^r \mathbb C^d.
$
defined by the basis vectors
$
    e_I:=e_{i_1}\wedge\cdots\wedge e_{i_r},
    \qquad
    1\le i_1<\cdots<i_r\le d,
$
so $\dim \bigwedge^r \mathbb C^d=\binom dr$. This vector space can be turned into a representation of $GL_d(\mathbb{C})$ in the following way.
The action of \(g\in GL_d(\mathbb{C})\) is induced from the defining action on
\(\mathbb C^d\):
\begin{equation}\label{eq:action_on_extirior_power}
    g\cdot(v_1\wedge\cdots\wedge v_r)
    =
    (gv_1)\wedge\cdots\wedge(gv_r).
\end{equation}
Equivalently, for \(X\in\mathfrak{gl}_d(\mathbb{C})\), the infinitesimal action is
\[
    X\cdot(v_1\wedge\cdots\wedge v_r)
    =
    \sum_{m=1}^r
    v_1\wedge\cdots\wedge Xv_m\wedge\cdots\wedge v_r.
\]
In fact, one can check that $\bigwedge^r \mathbb C^d\cong S_{\left(1^r\right)}\left(\mathbb{C}^d\right)$, 
which means that $\bigwedge^r \mathbb C^d$ is an irreducible representation.
Notice that $\bigwedge^r \mathbb C^d$ can also be considered as a representation of $GL_d(\mathbb{C})$ -- it is irreducible and has dimension $\binom{d}{r}$.
We use the following notation
\[
    \operatorname{det} V:=\Lambda^d V,
\]
to emphasize the fact that $\Lambda^d V$ is the one-dimensional representation of $GL_d(\mathbb{C})$.

\begin{smremark}
    From now on, we will use $\mathcal{V}$ and $\C^d$ interchangeably.
\end{smremark}

\begin{smdefinition}\label{def:rfundSUd}
    We define $r$-fundamental representation of $SU(d)$ as following
    \[
        V_{\omega_r}:=\bigwedge^r \mathbb C^d
    \]
    where action of $SU(d)$ is defined by restriction of the action of $GL_d(\mathbb{C})$ defined in (\ref{eq:action_on_extirior_power}).
\end{smdefinition}

\begin{smremark}\label{rem:fundamental_weights}
    Notation $V_{\omega_r}$ is not coincidental. It is known that irreducible representations of $SU(d)$ are in one-to-one correspondence with \textit{dominant integral weights} $\mu=a_1 \omega_1+\cdots+a_{d-1} \omega_{d-1}, \quad a_i \in \mathbb{Z}_{\geq 0}$,
    where $\omega_i$ are the so-called \textit{fundamental weights} of $SU(d)$~\cite{Hall2015}.
    We will not go into details of this correspondence and the nature of fundamental weights. So, strictly speaking, $V_{\omega_r}\cong \bigwedge^r \mathbb C^d$.  
    But we will use the following fact about the correspondence between $S_\lambda V$ -- irreducible representations of $GL_d(\mathbb{C})$ -- and $V_\lambda$ -- irreducible representations of $SU(d)$:
    if $\left(a_1, \ldots, a_{d-1}\right)=\left(\lambda_1-\lambda_2, \lambda_2-\lambda_3, \ldots, \lambda_{d-1}-\lambda_d\right)$, then $S_\lambda V$ restricts to $V_\mu$ as representation of $SU(d)$, where $\mu=a_1 \omega_1+\cdots+a_{d-1} \omega_{d-1}$.
    This means that, for example, $S_{(1, 0^{d-1})}$ and $S_{(a+1, a^{d-1})}$ will both restrict to $V_{\omega_1}$, and $\omega_r$ generally corresponds to partition $(1^r, 0^{d-r})$.
    Finally, one should not confuse $\lambda$ -- partition, used to label representations of $GL_d(\mathbb{C})$, namely $S_\lambda V$ -- with $\lambda$ -- dominant integral weight, used to label representations of $SU(d)$, namely $V_\lambda$. It will be clear from the context which one we are talking about.
\end{smremark}

For clarification of notation, see the Nomenclature section.

\begin{smremark}
    For more information about $GL_d(\mathbb{C})$ representations see, for example,~\cite{FultonHarris1991}.
    All standard symmetric functions facts and objects -- symmetric functions definitions, Kostka numbers, involution $\omega$, Hall inner product, Pieri's formula, etc. -- can be found in~\cite{Macdonald1995}.
\end{smremark}

\begin{smlemma}\label{lem:fundamental_in_extirior_power}
    $SU(d)$ representation $\left(V_{\omega_r}\right)^{\otimes n}$, contains 
    at least one copy of the fundamental representation, i. e. $V_{\omega_1}$, iff $nr\equiv 1 \pmod{d}$.
\end{smlemma}

\begin{proof}
    We first work in the context of $GL_d(\mathbb{C})$ representation theory and then 
    restrict ourselves to representations of $SU(d)$. The reason is the existence of
    convenient tools for $GL_d(\mathbb{C})$ representations, namely, Schur polynomials, which have an established combinatorial framework. 
    
    We will first consider multiplicity of representation $S_{\left(q+1, q^{d-1}\right)} V$ in $\left(S_{\left(1^r\right)}V\right)^{\otimes n}$, where $q$ is a fixed natural number. 
    The reason we consider $S_{\left(q+1, q^{d-1}\right)} V$ is that this representation will restrict to the fundamental representation of $SU(d)$, namely $V_{\omega_1}$, and $\left(S_{\left(1^r\right)}V\right)^{\otimes n}$ -- to the physical representation of $SU(d)$, namely $(V_{\omega_r})^{\otimes n}$, but we 
    will get to this point later in the proof. We first prove that this multiplicity can be described by the Kostka number, i. e.
    \begin{equation}\label{eq:KostkaMultiplicity}
        \operatorname{dim} \operatorname{Hom}_{GL_d(\mathbb{C})}\left(S_{\left(q+1, q^{d-1}\right)} V,\left(S_{\left(1^r\right)}V\right)^{\otimes n}\right)=K_{\nu', \mu}=K_{\left(d^q, 1\right),\left(r^n\right)},
    \end{equation}
    where $\nu=\left(q+1, q^{d-1}\right)$ is the Young diagram of the representation $S_{\left(q+1, q^{d-1}\right)} V$, and $\mu=\left(r^n\right)$ is the Young diagram of $\left(S_{\left(1^r\right)}V\right)^{\otimes n}$, and prime
    denotes transpose of the Young diagram (one can easily check that $\left(q+1, q^{d-1}\right)' = \left(d^q, 1\right)$), for some natural number $q$.

    Indeed, characters of $GL_d(\mathbb{C})$ representations are given by Schur polynomials. 
    In our case since $S_{\left(1^r\right)}V \leftrightarrow s_{\left(1^r\right)}=e_r$, then
    $
        \left(S_{\left(1^r\right)}V\right)^{\otimes n} \leftrightarrow\left(e_r\right)^n=s_{\left(1^r\right)}^n,
    $
    where $(\cdot)^n$ corresponds to product of functions. It is hard to expand $\left(e_r\right)^n$
    directly in Schur basis, but we can use Pieri's famous formula from which we get that
    $
        \mathrm{h}_\lambda=\sum_\mu K_{\mu \lambda} \mathrm{s}_\mu,
    $
    and the fact that 
    $
        \omega\left(h_r\right)=e_r,
    $
    where $\omega$ is the involution on symmetric functions defined by $\omega\left(s_\lambda\right)=s_{\lambda'}$. Finally, it is easy to check
    $
        h_{\left(r^n\right)}=h_r h_r \cdots h_r=h_r^n,
    $
    which leads to
    \[
        h_{\left(r^n\right)}=h_r^n=\sum_\lambda K_{\lambda,\left(r^n\right)} s_\lambda.
    \]
    Since multiplicity of $S_\nu V$ in $\left(S_{\left(1^r\right)}V\right)^{\otimes n}$ is given by
    $
        \left\langle e_r^n, s_\nu\right\rangle,
    $
    where $\langle\cdot, \cdot\rangle$ is the Hall inner product with respect to which Schur polynomials form an orthonormal basis $\left\langle s_\lambda, s_\mu\right\rangle=\delta_{\lambda \mu}$, we have
    \[
        \left\langle e_r^n, s_\nu\right\rangle=\left\langle\omega\left(h_r^n\right), s_\nu\right\rangle=\left\langle h_r^n, \omega\left(s_\nu\right)\right\rangle=\left\langle h_r^n, s_{\nu^{\prime}}\right\rangle=K_{\nu', (r^n)},
    \]
    where second equality uses the fact that $\omega$ is self-adjoint with respect to the Hall inner product (indeed, 
    $\left\langle\omega\left(s_\lambda\right), \omega\left(s_\mu\right)\right\rangle=\left\langle s_{\lambda^{\prime}}, s_{\mu^{\prime}}\right\rangle=\delta_{\lambda^{\prime}, \mu^{\prime}}=\delta_{\lambda, \mu}=\left\langle s_\lambda, s_\mu\right\rangle$
    and $\langle\omega(f), g\rangle=\left\langle\omega^2(f), \omega(g)\right\rangle=\langle f, \omega(g)\rangle$).
    
    Now if we take $\nu=\left(q+1, q^{d-1}\right)$, then we obtain desired equality (\ref{eq:KostkaMultiplicity}).
    
    Now let's prove that $K_{\left(d^q, 1\right),\left(r^n\right)} > 0$ iff $n r=d q+1$.
    
    First recall standard fact about Kostka numbers: $K_{\lambda, \mu} \neq 0$ iff $\lambda$ dominates $\mu$, i. e. $\lambda \trianglerighteq \mu$.  
    
    Let's prove necessity. If $S_{\left(q+1, q^{d-1}\right)} V$ occurs inside $\left(S_{\left(1^r\right)}V\right)^{\otimes n}$, then with necessity 
    $
        |\nu|=nr,
    $
    and since $\nu=\left(q+1, q^{d-1}\right)$, then $|\nu|=(q+1)+(d-1) q=d q+1$, from which we get
    $
        n r=d q+1,
    $
    so if $K_{\left(d^q, 1\right),\left(r^n\right)} > 0$, then $n r=d q+1$. Necessity is proven.
    
    Let's prove sufficiency. 
    For $\lambda=\left(d^q, 1\right)$ and $\mu=\left(r^n\right)$, the dominance condition is satisfied if $n r=d q+1$. 
    Indeed, for $t \leq q$, we have
    $
    t d \geq t r,
    $
    because $r \leq d-1$; if $q+1 \leq n$, then
    $
    q d+1=r n \geq(q+1) r,
    $
    since $(q+1) r \leq n r$. So, if $n r=d q+1$, the dominance condition is satisfied, meaning that $K_{\left(d^q, 1\right),\left(r^n\right)} > 0$. Sufficiency is proven.

    Thus, we've proven that $S_{\left(q+1, q^{d-1}\right)} V$ is contained in $\left(S_{\left(1^r\right)}\left(\mathcal{V}\right)\right)^{\otimes n}$ with multiplicity $K_{\left(d^q, 1\right),\left(r^n\right)}>0$ iff 
    $n r=d q+1$ -- all as representations of $GL_d(\mathbb{C})$.
    
    Finally, let's restrict ourselves to $SU(d)$ representations. To be precise, 
    \[
        S_{\left(1^r\right)} V\downarrow_{SU(d)}\cong V_{\omega_r}, \quad S_{\left(q+1, q^{d-1}\right)} V\downarrow_{SU(d)} \cong V_{\omega_1},
    \]
    from which we conclude that $V_{\omega_1}$ is contained in $\left(V_{\omega_r}\right)^{\otimes n}$ with non-zero multiplicity iff $nr\equiv 1 \pmod{d}$. 
    This completes our proof.
    
\end{proof}

\begin{figure}[htbp]
    \centering
    \IfFileExists{figs/Young_diagrams.pdf}{%
        \includegraphics[width=0.4\textwidth]{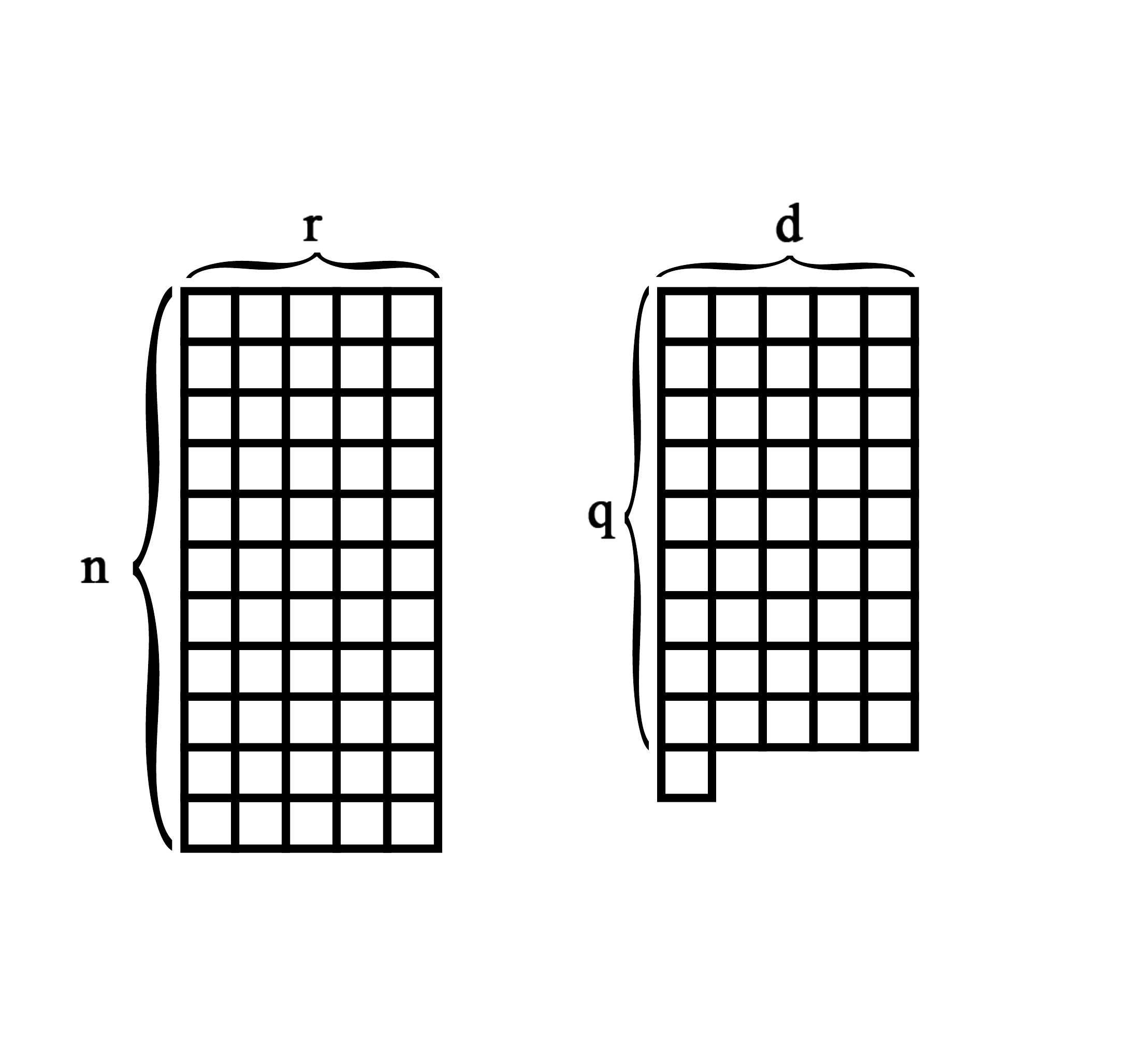}%
    }{%
        \fbox{\rule{0pt}{2in}\rule{0.7\textwidth}{0pt}}
    }
    \caption{Young diagrams of the representation $S_{\left(q+1, q^{d-1}\right)} V$ and the representation $\left(S_{\left(1^r\right)}V\right)^{\otimes n}$, where $n r=d q+1$.}
    \label{fig:Young_diagrams}
\end{figure}

\begin{smlemma}\label{lem:adjoint_in_end}
    $SU(d)$ representation $\operatorname{End}(V_{\omega_r})$, contains 
    exactly one copy of the adjoint and trivial representations.
\end{smlemma}

\begin{proof}
As in previous proof, we first work in the context of $GL_d(\mathbb{C})$ representations and then restrict ourselves to $SU(d)$ representations.
Since $\Lambda^r V \cong V_{\omega_r}$ as $SU(d)$ representations, it is reasonable to consider $\Lambda^r V \cong S_{\left(1^r\right)}(V)$ as $GL_d(\mathbb{C})$ representation.
First, notice that
$$
    \operatorname{End}(S_{\left(1^r\right)}(V))\cong S_{\left(1^r\right)}(V) \otimes (S_{\left(1^r\right)}(V))^* \cong (\det V)^{-1} \otimes S_{\left(1^r\right)} V \otimes S_{\left(1^{d-r}\right)} V.
$$
Now we apply Pieri's formula in the following form
\[
S_{\left(1^r\right)} V \otimes S_{\left(1^{d-r}\right)} V \cong \bigoplus_{i=0}^{\min (r, d-r)} S_{\left(2^i, 1^{d-2 i}\right)} V
\]
Indeed, starting from $(1^r)$ and multiplying by $(1^{d-r})$ we add a vertical strip of size $d-r$ to the first diagram (according to Pieri's formula). 
Since the first diagram is a single column itself, we have two options: add to an existing row, producing a row of length $2$, or
create a new row, producing a row of length $1$.
The resulting partition is 
$
\left(2^i, 1^{(r-i)+((d-r)-i)}\right)=\left(2^i, 1^{d-2 i}\right) .
$
Multiplicity representation, corresponding to each partition, is one -- again, according to Pieri's formula.
Thus,
\begin{equation}\label{eq:decompositionEnd}
    \operatorname{End}\left(S_{\left(1^r\right)}(V)\right) \cong \bigoplus_{i=0}^{\min (r, d-r)}(\operatorname{det} V)^{-1} \otimes S_{\left(2^i, 1^{d-2 i}\right)} V.
\end{equation}

Now let's find the expression of the adjoint representation as a polynomial $GL_d(\mathbb{C})$-module. On the one side, we have
\[
    \operatorname{End}(V)=V \otimes V^* \cong \mathbf{1} \oplus \mathrm{Ad}
\]
On the other side
\[
    V \otimes V^* \cong(\operatorname{det} V)^{-1} \otimes S_{\left(1\right)} V \otimes S_{\left(1^{d-1}\right)} V
\]
and using Pieri's formula again, we get
\[
    V \otimes V^* \cong (\operatorname{det} V)^{-1} \otimes \left(S_{\left(1^d\right)} V \oplus S_{\left(2,1^{d-2}\right)} V\right)=\mathbf{1} \oplus(\operatorname{det} V)^{-1} \otimes S_{\left(2,1^{d-2}\right)} V
\]
from which we derive
\[
    \mathrm{Ad} \cong(\operatorname{det} V)^{-1} \otimes S_{\left(2,1^{d-2}\right)} V
\]
and from decomposition (\ref{eq:decompositionEnd}) we conclude that
\[
    \operatorname{dim} \operatorname{Hom}_{GL_d(\mathbb{C})}\left(\mathrm{Ad}, \operatorname{End}\left(S_{\left(1^r\right)}(V)\right)\right)=1 .
\]
Finally, let's restrict ourselves to $SU(d)$ representations. To be precise, 
    \[
        S_{\left(1^r\right)}(V)\downarrow_{SU(d)}\cong V_{\omega_r}, \quad \mathrm{Ad}\downarrow_{SU(d)} \cong \mathrm{Ad},
    \]
From \ref{eq:decompositionEnd} we infer that there is exactly one copy of the trivial representation -- term $i=0$.
This completes our proof.
\end{proof}

We define generators of $\mathfrak{su}(d)$ (fundamental representation) in a standard way:
choose $t_a$ that are traceless Hermitian complex $d \times d$ matrices, where:
\begin{equation}\label{eq:sudgenerators}
t_a t_b=\frac{1}{2 d} \delta_{a b} I_d+\frac{1}{2} \sum_{c=1}^{d^2-1}\left(i f_{a b c}+d_{a b c}\right) t_c
\end{equation}
where the $f$ are the structure constants and are antisymmetric in all indices, while the $d$-coefficients are symmetric in all indices.
As a consequence,
$$
\left[t_a, t_b\right]=i \sum_{c=1}^{d^2-1} f_{a b c} t_c, \quad \left\{t_a, t_b\right\}=\frac{1}{d} \delta_{a b} I_d+\sum_{c=1}^{d^2-1} d_{a b c} t_c.
$$
With this definition conventional normalizations are
\begin{equation}\label{eq:sudnormalisation}
\operatorname{Tr}\left(t_a t_b\right)=\frac{1}{2} \delta_{a b}, \quad \sum_{c, e=1}^{d^2-1} d_{a c e} d_{b c e}=\frac{d^2-4}{d} \delta_{a b}.
\end{equation}
We will also use the following identities, which can be easily derived from the definition of generators:
\begin{equation}\label{eq:sudidentities}
    \sum_{m, n=1}^{d^2-1}f_{a m n} f_{b m n}=d \delta_{a b}, \quad \sum_{m, n=1}^{d^2-1}f_{a m n} d_{b m n}=0 , \quad \sum_{m=1}^{d^2-1}d_{a m m}=0 .
\end{equation}

We denote by $\{T_a\}$ the set of operators on $V_{\omega_r}$ that are representation of generators $\{t_a\}$ of $\mathfrak{su}(d)$, omitting the index $r$ for brevity.
Notice that $T_a$ are also traceless Hermitian complex matrices.
We will later need the following normalization 
\begin{equation}\label{eq:sudnormalisation_2}
\operatorname{Tr}_{V_{\omega_r}}\left(T_a T_b\right)=\kappa_{\omega_r} \delta_{a b}, \quad \kappa_{\omega_r}=\frac{1}{2}\binom{d-2}{r-1}.
\end{equation}

\begin{smproposition}\label{prop:SUdcov}
    Let
    $
        \mathcal H_L=V_{\omega_1},\;
        \mathcal H_P=(V_{\omega_r})^{\otimes n}
    $
    and $nr\equiv 1\;(\operatorname{mod} d)$.
    Equip \(\mathcal H_L\) with the fundamental representation of
    \(\mathfrak{su}(d)\), and equip \(\mathcal H_P\) with the transversal
    representation given by
    \begin{equation}
    \begin{gathered}
    t_a \rightarrow T_a \equiv \sum_{i=1}^n T_a^{(i)}, \quad a=1,\ldots,d^2-1
    \end{gathered}
    \end{equation}
    where $T_a^{(i)}$ is a representation of $t_a$ on the $i$-th physical space $V_{\omega_r}$, acting as in Eq.~\eqref{eq:action_on_extirior_power}.
    Then there exist a code space $\mathfrak{su}(d)$-covariant encoding $\mathcal{E}$ with respect to defined physical and logical representations, 
    such that the state seen by the environment after erasure of any single physical qudit is of the form   
    \[
        \rho^{(i)}\left(\rho_L\right)=\frac{I}{\operatorname{dim} V_{\omega_r}}+\frac{1}{n} \frac{1}{2\kappa_{\omega_r}} \sum_a r_a T_a^{(i)} .    
    \]
    Namely, this code space is one of the copies of the fundamental representation $V_{\omega_1}$ in $\left(V_{\omega_r}\right)^{\otimes n}$ 
    which is invariant under cyclic permutations of physical qudits.
\end{smproposition}

\begin{proof}
We recall that from Schur lemma $\left(V_{\omega_r}\right)^{\otimes n}\cong \bigoplus_{\lambda\in \Lambda} V_\lambda\otimes M_{\lambda}$, 
where $\Lambda$ is a set of irreducible representations of $SU(d)$, 
$V_\lambda$ is the irreducible representation space and $M_\lambda$ is the multiplicity space.

In Lemma~\ref{lem:fundamental_in_extirior_power} we prove that for the choice of individual physical space to be $V_{\omega_r}$ and $n r \equiv 1 (\bmod d)$, it is guaranteed 
that $\operatorname{dim} M_{V_{\omega_1}}>0$.

The linear extension $\Phi^{(i)}:\operatorname{End}(V_{\omega_1})\rightarrow \operatorname{End}(V_{\omega_r})$ of the reduced state map $\rho_L \rightarrow \rho^{(i)}\left(\rho_L\right)$ is an intertwiner of $SU(d)$ representations. 
Logical space as $SU(d)$ representation is
\[
     \text{End}(V_{\omega_1})\cong V_{\omega_1} \otimes (V_{\omega_1})^* \cong \mathbf{1} \oplus \mathrm{Ad}.
\]
For individual system representation $V_{\omega_r}$ we have the following decomposition
\[
    \text{End}(V_{\omega_r})\cong V_{\omega_r} \otimes V_{\omega_r}^* \cong \mathbf{1} \oplus \mathrm{Ad} \oplus \cdots.
\]
From Lemma~\ref{lem:adjoint_in_end} we infer that there is always exactly one copy of the trivial representation and also exactly one copy of the adjoint representation.
As an intertwiner $\Phi^{(i)}$ should have a form
\[
\Phi^{(i)}\cong I_{\mathbf{1}}\otimes A_{M_\mathbf{1}\rightarrow M_\mathbf{1}}\oplus I_{\mathrm{Ad}}\otimes A_{M_\mathrm{Ad}\rightarrow M_\mathrm{Ad}}
\]
and since $\dim M_\mathbf{1}=\dim M_\mathrm{Ad}=1$ we get
\begin{equation}
\begin{gathered}
        \Phi^{(i)}\left(\frac{I_d}{d}\right)=\frac{I_{V_{\omega_r}}}{\operatorname{dim} V_{\omega_r}},\\
        \Phi^{(i)}\left(\sum_a r_a t_a\right)=\beta^{(i)}\sum_a r_a T_a^{(i)}.
\end{gathered}
\end{equation}

Going back to reduced state map $\rho^{(i)}\left(\rho_L\right)$, since logical state has the following form
\[
     \rho_L=\frac{I_d}{d}+\left(\rho_L-\frac{I_d}{d}\right)=\frac{I_d}{d}+\sum_a r_a t_a,
\]
where $\frac{I}{d}$ is in trivial representation and $\rho_L-\frac{I}{d}$ is in adjoint representation, we have that $\rho^{(i)}\left(\rho_L\right)$ has the following form
\[
    \rho^{(i)}\left(\rho_L\right) = \frac{I_{V_{\omega_r}}}{\operatorname{dim} V_{\omega_r}}+\beta^{(i)} \sum_a r_a T_a^{(i)}.
\]
where $\beta^{(i)}$ is some constant that depends on the choice of code space.
Indeed, $\rho^{(i)}\left(\rho_L\right)$ is a covariant map, so it takes a trivial representation to a trivial representation and an adjoint representation.
And since we have chosen $V_{\omega_r}$ such that it contains exactly one copy of the adjoint representation, we have the desired form. Let

\[
    \mathcal{V}: V_{\omega_1} \rightarrow V_{\omega_r}^{\otimes n}
\]
our encoding isometry, where $V_{\omega_1}$ is the fundamental representation of $SU(d)$. 
We want to show

\[
\mathcal{V}^{\dagger} T_a^{(i)} \mathcal{V}=\alpha^{(i)} t_a.
\]
Equivalently, if we define
\begin{equation}\label{eq:Vt_aV}
     K_a^{(i)}:=\mathcal{V}^{\dagger} T_a^{(i)} \mathcal{V} \in \operatorname{End}\left(V_{\omega_1}\right),
\end{equation}
we need to show
\[
    K_a^{(i)}=\alpha^{(i)} t_a.
\]
From the fact that local generators transform in the adjoint representation, i. e.
\[
     \left(\bigotimes_{k=1}^n U^{(k)}(g)\right) T_a^{(i)} \left(\bigotimes_{k=1}^n U^{(k)}(g)\right)^{\dagger}=\sum_b \operatorname{Ad}_{b a}(g) T_a^{(i)}
\]
and the covariance condition of the encoding map we have
\[
     u(g) K_a^{(i)} u(g)^{\dagger}=\sum_b \operatorname{Ad}_{b a}(g) K_b^{(i)}.
\]
$\operatorname{End}\left(V_{\omega_1}\right)$ contains exactly one copy of adjoint representation
\[
     \operatorname{End}\left(V_{\omega_1}\right) \cong V_{\omega_1} \otimes V_{\omega_1}^* \cong \mathbf{1} \oplus \operatorname{Ad}
\]
which means that by Shur's lemma
\[
     \operatorname{dim} \operatorname{Hom}_{S U(d)}\left(\operatorname{Ad}, \operatorname{End}\left(V_{\omega_1}\right)\right)=1
\]
so we have that $K_a^{(i)}$ is proportional to $t_a$. Now we are to obtain $\beta^{(i)}$ in terms of $\alpha^{(i)}$. We have

\[
     \operatorname{Tr}\left(\rho^{(i)}\left(\rho_L\right) T_a^{(i)}\right)=\beta^{(i)} \sum_b r_b \operatorname{Tr}\left(T_b^{(i)} T_a^{(i)}\right) .
\]
From normalization relations~\eqref{eq:sudnormalisation} and~\eqref{eq:sudnormalisation_2} we get
\[
     \operatorname{Tr}\left(\rho^{(i)}\left(\rho_L\right) T_a^{(i)}\right)=\beta^{(i)} \kappa_{\omega_r} r_a .
\]
and on the other hand, using properties of $\Tr$ and covariance condition of encoding map we get
\[
     \operatorname{Tr}\left(\rho^{(i)}\left(\rho_L\right) T_a^{(i)}\right)=\operatorname{Tr}\left(\mathcal{V} \rho_L \mathcal{V}^{\dagger} T_a^{(i)}\right)=
     \operatorname{Tr}\left(\rho_L \mathcal{V}^{\dagger} T_a^{(i)} \mathcal{V}\right)=\alpha^{(i)} \operatorname{Tr}\left(\rho_L t_a\right) .
\]
from $\rho_L=\frac{I}{d}+\sum_a r_a t_a$, we have
$
     \operatorname{Tr}\left(\rho_L t_a\right)=\frac{1}{2} r_a
$
so on the other hand we have
$
     \operatorname{Tr}\left(\rho^{(i)}\left(\rho_L\right) T_a^{(i)}\right)=\frac{\alpha^{(i)}}{2} r_a .
$
from which we finally obtain the following expression for $\beta^{(i)}$:
\[
     \beta^{(i)}=\frac{1}{2\kappa_{\omega_r}}\alpha^{(i)}  .
\]

Finally, the reduced state on the $i$-th physical qudit has the following form
\[
    \rho^{(i)}\left(\rho_L\right)=\frac{I}{\operatorname{dim} V_{\omega_r}}+\alpha^{(i)} \frac{1}{2\kappa_{\omega_r}} \sum_a r_a T_a^{(i)} .
\]
Using the same code construction for symmetric code as we did for $SU(2)$ case (see the proof of the Proposition~\ref{prop:SU2cov}), namely, making code space invariant under cyclic shifts, 
our code space is such that $\alpha^{(i)}=\frac{1}{n}$ for all $i$, which gives the desired form of reduced state.

\end{proof}

\begin{smproposition}\label{prop:SUdcovfidelity}
    For the encoding $\mathcal{E}$ defined in the Proposition~\ref{prop:SUdcov} and noise channel $\mathcal{N}(\sigma)=\sum_{i=1}^n p_i \ket{i}\bra{i}_F \otimes \ket{e}\bra{e}_{A_i} \otimes \operatorname{Tr}_i(\sigma)$ the worst-case entanglement fidelity is
    \[
        F(\widehat{\mathcal{N} \circ \mathcal{E}}, \Lambda_0)=1-\frac{(d-1)^2(d+1)}{8 r(d-r)}\frac{1}{n^2}+O_{d}\left(n^{-3}\right),
    \]
    where we take $\Lambda_0(\rho)=\operatorname{Tr}(\rho)\omega_{\text{flag}}\otimes\frac{I_{V_{\omega_r}}}{\operatorname{dim} V_{\omega_r}}$, where $\omega_{\text {flag}}=\sum p_i|i\rangle\langle i|_{F_E}$.
\end{smproposition}

\begin{proof}
In this section, we will calculate the worst-case fidelity of $SU(d)$-covariant codes with flagged erasure error model
in the same way as we did for $SU(2)$-covariant codes. 
For the later calculations we will use notation $\beta:=\frac{1}{n} \frac{1}{2\kappa_{\omega_r}}$, $D=\operatorname{dim} V_{\omega_r}$.
Since reduced states are identical on each site, the complementary channel becomes
\[
        \widehat{\mathcal{N} \circ \mathcal{E}}(\rho)=\omega_{\text {flag }} \otimes \left(\frac{I_D}{D}+ \beta \sum_a r_a T_a^{(1)}\right),
\]
where $\omega_{\text {flag }}=\sum p_i|i\rangle\langle i|_{F_E}$ is the flag state. We chose the fixed channel to be 
$$
\Lambda_0(\rho)=\operatorname{Tr}(\rho)\cdot \omega_{\text {flag }}\otimes \frac{I_{D}}{D} \equiv \operatorname{Tr}(\rho)\cdot \omega_{\text {flag }}\otimes \tau.
$$ 
We drop flag state $\omega_{\text {flag }}$, since fidelity is invariant under 
tensor multiplication on the same state; we also omit the index of a site $(1)$.

As we know from \ref{fidelity_opt} the optimum of $F_{\rho}(\widehat{\mathcal{N} \circ \mathcal{E}}, \Lambda_0)$ 
is attained on the chaotic state $\frac{I_d}{d}$, its purification is maximally entangled state $\ket{\Psi}=\frac{1}{\sqrt{d}}\sum |r\rangle \otimes|r\rangle $. 
We have to compute $(\widehat{\mathcal{N} \circ \mathcal{E}} \otimes \mathrm{id})(\ket{\Psi}\bra{\Psi})$ and $(\Lambda_0 \otimes \mathrm{id})(\ket{\Psi}\bra{\Psi})$.
We define the linear extension of the reduced state channel to all complex matrices as follows
$$
\Phi: \operatorname{End}\left(V_{\omega_1}\right) \rightarrow \operatorname{End}\left(V_{\omega_r}\right), \quad \Phi\left(t_a\right):=T_a, \quad \Phi\left(\frac{I_d}{d}\right):=\tau
$$
and we will use the following notation for our derivations
\begin{equation}\label{eq:matr_units}
     E_{r s}:=|r\rangle\langle s|, \quad F_{r s}:=E_{r s}-\delta_{r s} \frac{I_d}{d}.
\end{equation}
For the constant channel, we have
\[
     \eta_{1/d}:=(\Lambda_0 \otimes \mathrm{Id})\left(\ket{\Psi}\bra{\Psi}\right)=\tau \otimes \frac{I_d}{d}.
\]  
For the complementary channel using 
$
\widehat{\mathcal{N} \circ \mathcal{E}}\left(E_{r s}\right)=\delta_{r s} \tau+\beta \Phi\left(F_{r s}\right) .
$
we define 
\[
\begin{gathered}
     \sigma_{1/d}:=\left(\widehat{\mathcal{N} \circ \mathcal{E}} \otimes \mathrm{Id}\right)\left(\ket{\Psi}\bra{\Psi}\right)=
     \frac{1}{d}\sum_{r, s=1}^d \widehat{\mathcal{N} \circ \mathcal{E}}\left(E_{r s}\right) \otimes|r\rangle\langle s| =\\
     =\frac{1}{d}\sum_{r, s=1}^d \left(\delta_{r s} \tau+\beta \Phi\left(F_{r s}\right)\right) \otimes|r\rangle\langle s| .
\end{gathered}
\]
We obtain the following expression for fidelity
\[
     f\left(\sigma_{1/d}, \eta_{1/d}\right)=\operatorname{Tr}\sqrt{\sqrt{\eta_{1/d}} \sigma_{1/d} \sqrt{\eta_{1/d}}}=\operatorname{Tr} \sqrt{\frac{1}{Dd^2}\left[\delta_{r s} \tau+\beta \Phi\left(F_{r s}\right)\right]_{r, s=1}^d}.
\]
Now, we recall that $\Phi$ actually acts as follows
\[
        \Phi\left(F_{a b}\right)=T_{a b}^{(r)}-\delta_{a b} \frac{r}{d} I_D,
\]
where $T_{a b}^{(r)}$ is the representation of $F_{a b}$ in the adjoint representation of $SU(d)$ and defined as follows
\[
    T_{a b}^{(r)}\left(v_1 \wedge \cdots \wedge v_r\right)=\sum_{m=1}^r v_1 \wedge \cdots \wedge E_{a b} v_m \wedge \cdots \wedge v_r
\]
so we can rewrite the expression for fidelity as follows
\[
\begin{aligned}
    f\left(\sigma_{1/d}, \eta_{1/d}\right)&=\operatorname{Tr} \sqrt{\frac{1}{Dd^2}\left[\delta_{a b} \tau+\beta \left(T_{a b}^{(r)}-
    \delta_{a b} \frac{r}{d} I_D\right)\right]_{a, b=1}^d}\\&\equiv\frac{1}{d \sqrt{D}} \operatorname{Tr} \sqrt{\frac{I_{D d}}{D}+\beta B} .
\end{aligned}
\]
where 
\[
    B=\left(\begin{array}{cccc}
    T_{11}^{(r)}-\frac{r}{d} I_D & T_{12}^{(r)} & \cdots & T_{1 d}^{(r)} \\
    T_{21}^{(r)} & T_{22}^{(r)}-\frac{r}{d} I_D & \cdots & T_{2 d}^{(r)} \\
    \vdots & \vdots & \ddots & \vdots \\
    T_{d 1}^{(r)} & T_{d 2}^{(r)} & \cdots & T_{d d}^{(r)}-\frac{r}{d} I_D
    \end{array}\right).
\]
We are now to expand this expression up to a leading order in $\beta$
\[
\begin{aligned}
        f\left(\sigma_{1/d}, \eta_{1/d}\right)&=1+\frac{\beta}{2 d} \operatorname{Tr} B-\frac{D \beta^2}{8 d} \operatorname{Tr} B^2+O_{d}\left(\beta^3\right)
        \\&=1-\frac{D^2 \beta^2}{8 d}\left(r(d-r+1)-\frac{r^2}{d}\right)+O_{d}\left(\beta^3\right)
\end{aligned}
\]
where we use the fact that $\operatorname{Tr} B=0$ and $\operatorname{Tr} B^2=D\left(r(d-r+1)-\frac{r^2}{d}\right)$. Indeed, 
$$
\begin{aligned}
    \operatorname{Tr} B^2&=\sum_{a, b=1}^d \operatorname{Tr}_{V_{\omega_r}}\left(\Phi\left(F_{a b}\right) \Phi\left(F_{b a}\right)\right)=\sum_{a\neq b}^d\operatorname{Tr}\left(T_{a b}^{(r)} T_{b a}^{(r)}\right)+\sum_{a=1}^d \operatorname{Tr}\left(T_{a a}^{(r)}-\frac{r}{d} I_D\right)^2\\&=d(d-1)\binom{d-2}{r-1}+d\frac{Dr(d-r)}{d^2}=D\left(r(d-r+1)-\frac{r^2}{d}\right)
\end{aligned}
$$
Finally, if we substitute $\beta=\frac{1}{n} \frac{1}{2\kappa_{\omega_r}}$, $D=\operatorname{dim} V_{\omega_r}$ and 
$\kappa_{\omega_r}=\frac{1}{2}\binom{d-2}{r-1}$ into $F(\widehat{\mathcal{N} \circ \mathcal{E}}, \Lambda_0)=f(\sigma_{1/d}, \eta_{1/d
})$ we get a desired expression for fidelity.
\end{proof}

\subsection{Explicit Decoder construction}\label{subsec:supple_decoder}

In this section, we will construct an explicit decoder for $SU(d)$-covariant codes with the erasure error model.
For this, we will use the Petz recovery map, which is defined as follows
\begin{smdefinition}
    For a given channel $\mathcal{N}$ and reference state $\omega$, the Petz map $\mathcal{P}_{\omega, \mathcal{N}}$ is defined as:
    \[
        \mathcal{P}_{\omega, \mathcal{N}}(\sigma)=\omega^{1 / 2} \mathcal{N}^*\left(\mathcal{N}(\omega)^{-1 / 2} \sigma \mathcal{N}(\omega)^{-1 / 2}\right) \omega^{1 / 2},
    \]
    where $\mathcal{N}^*$ is the adjoint of $\mathcal{N}$ with respect to the Hilbert-Schmidt inner product, and $\omega$ is so-called reference state.
\end{smdefinition}

\begin{smtheorem}\label{thm:petz_decoder_explicit}
For the encoding $\mathcal E$ and the flagged single-qudit erasure channel $\mathcal N$ defined in Theorem~\ref{thrm:sudscale}, 
the Petz recovery map with reference state $\omega=I_d/d$ has the form
$$
\mathcal R_{\mathrm{Petz}}(X)
=
\sum_{i:,p_i>0}
\mathcal{V}^\dagger
\left(
S_i^{-1/2}
\langle i,e|X|i,e\rangle_{F,A_i}
S_i^{-1/2}
\otimes I_{A_i}
\right)
\mathcal{V},
$$
where $\mathcal{V}$ is the encoding isometry and
$$
S_i
:=
\operatorname{Tr}_i\left(\mathcal{V}\mathcal{V}^\dagger\right).
$$
Moreover, this recovery is near-optimal in the sense that
$$
    d\left(
    \mathcal R_{\mathrm{Petz}}
    \circ
    \mathcal N
    \circ
    \mathcal E,
    \mathrm{id}
    \right)
    =
    O_{d}(n^{-1}) .
$$
\end{smtheorem}

\begin{proof}
Define the erased channel for site $i$ by
$
    \mathcal{C}_i\left(\rho_L\right):=\operatorname{Tr}_{i}\left(\mathcal{V} \rho_L \mathcal{V}^{\dagger}\right).
$
Our noise channel in this notation can be written as
$
    \mathcal{N}(\sigma)=\sum_{i=1}^n p_i\ket{i}\bra{i}_F \otimes \ket{e}\bra{e}_{A_i}\otimes \mathcal{C}_i.
$
We take a reference state $\omega$ to be the chaotic state $\frac{I}{d}$, which is the only $SU(d)$-invariant state on logical space. We get
\[
    \mathcal{C}_i\left(\omega\right)=\frac{1}{d} \operatorname{Tr}_i\left(\mathcal{V} \mathcal{V}^{\dagger}\right)=\frac{S_i}{d},
\]
where we defined $S_i:=\operatorname{Tr}_i\left(\mathcal{V} \mathcal{V}^{\dagger}\right)$. One can easily check that
\[
    \mathcal{C}_i^{\dagger}(Y)=\mathcal{V}^{\dagger}\left(Y \otimes I_{W_i}\right) \mathcal{V}
\]
where we denoted $W_i$ as the Hilbert space associated with site $i$. Thus, the recovery map is given by
\[
    \mathcal{R}_{\mathrm{Petz}}(X)=\sum_{i: p_i>0} \mathcal{V}^{\dagger}\left(S_i^{-1 / 2} \langle i, e| X|i,e\rangle_{F, A_i} S_i^{-1 / 2} \otimes I_{W_i}\right) \mathcal{V},
\]
where the inverse should always be understood as the inverse on the support of $S_i$. Now, let's explore the performance of this decoder, 
i. e. find $F\left(\mathcal{R}_{\mathrm{Petz}} \circ \mathcal{N}\circ \mathcal{E}, \mathrm{id}\right)$, since we 
are in AQECC setting. We denote
\[
    \Phi:=\mathcal{R}_{\mathrm{Petz}} \circ \mathcal{N}\circ \mathcal{E}
\]
and let's analyze the structure of $\Phi$. First of all, $\Phi: \operatorname{End}(\mathcal{H}_L)\rightarrow \operatorname{End}(\mathcal{H}_L)$ is a 
$SU(d)$-covariant map, i. e.
\[
    \Phi\left(u(g) \rho u(g)^{\dagger}\right)=u(g) \Phi(\rho) u(g)^{\dagger}, \quad \forall g \in SU(d).
\]
Indeed, $\mathcal{C}_i\left(u(g) \rho u(g)^{\dagger}\right)=U^{\hat{i}}(g) \mathcal{C}_i(\rho) \left(U^{\hat{i}}(g)\right)^{\dagger}$ and $S_i^{-1 / 2}$ 
commutes with $U^{\hat{i}}(g)$ on the support of $S_i$. So it has to take the following form
\[
    \Phi(\rho)=\lambda_{\mathrm{Petz}} \rho+\left(1-\lambda_{\mathrm{Petz}}\right) \frac{I_d}{d} \operatorname{Tr} \rho,
\]
that is basically a depolarizing channel with depolarizing parameter $\lambda_{\mathrm{Petz}}$. Let's show that this depolarizing parameter is given by
\begin{equation}\label{eq:lambdaPetz}
    \lambda_{\text{Petz}}=\frac{2}{d^2-1}\sum_{i=1}^n p_i \sum_{a=1}^{d^2-1} \operatorname{Tr}\left[S_i^{-1 / 2} \mathcal{C}_i\left(t_a\right) 
    S_i^{-1 / 2} \mathcal{C}_i\left(t_a\right)\right] .
\end{equation}
Indeed, since $\Phi\left(t_a\right)=\lambda_{\text{Petz}} t_a$, we have
$$
\operatorname{Tr}\left(t_a \Phi\left(t_a\right)\right)=\frac{1}{2} \lambda_{\text{Petz}},
$$
where we used $\operatorname{Tr}\left(t_a t_b\right)=\frac{1}{2} \delta^{a b}$. Summing over $a$,
$$
\lambda_{\text{Petz}}=\frac{2}{\left(d^2-1\right)} \sum_{a=1}^{d^2-1} \operatorname{Tr}\left(t_a \Phi\left(t_a\right)\right)
$$
Now substitute the Petz formula:
$$
\Phi_{\mathrm{Petz}}\left(t_a\right)=\sum_i p_i \mathcal{V}^{\dagger}\left(S_i^{-1 / 2} \mathcal{C}_i\left(t_a\right) S_i^{-1 / 2} \otimes I_{A_i}\right) \mathcal{V}.
$$
Hence
$$
\operatorname{Tr}\left(t_a \Phi_{\text{Petz}}\left(t_a\right)\right)=\sum_i p_i \operatorname{Tr}\left[S_i^{-1 / 2} 
\mathcal{C}_i\left(t_a\right) S_i^{-1 / 2} \mathcal{C}_i\left(t_a\right)\right]
$$
where we used the partial-trace identity
$
\operatorname{Tr}\left(X\left(Y \otimes I_W\right)\right)=\operatorname{Tr}\left(\operatorname{Tr}_W(X) Y\right),
$
with $X=\mathcal{V} t^a \mathcal{V}^{\dagger}$ and $Y=A^{-1 / 2} \mathcal{C}_i\left(t^a\right) A^{-1 / 2}$, gives desired expression for $\lambda_{\mathrm{Petz}}$.

For a pure state \(\psi=|\psi\rangle\langle\psi|\), we have \(\Phi_{\mathrm{Petz}}(\psi)=\lambda_{\mathrm{Petz}}\psi+(1-\lambda_{\mathrm{Petz}})I_d/d\). 
Since the target state is pure, the root fidelity is
\[
f(\Phi_{\mathrm{Petz}}(\psi),\psi)
=
\sqrt{\langle\psi|\Phi_{\mathrm{Petz}}(\psi)|\psi\rangle}
=
\sqrt{\lambda_{\mathrm{Petz}}+\frac{1-\lambda_{\mathrm{Petz}}}{d}} .
\]
This expression is independent of \(\psi\). Hence
\[
\eta_{\mathrm{Petz}}^{\mathrm{st}}
:=
1-\min_{\psi}f^2(\Phi_{\mathrm{Petz}}(\psi),\psi)
=
\frac{d-1}{d}(1-\lambda_{\mathrm{Petz}}).
\]

Our channel fidelity is defined as the minimum over \(\rho\) of \(F_\rho(\Phi_{\mathrm{Petz}},\mathrm{id})=f((\Phi_{\mathrm{Petz}}\otimes\mathrm{id})(|\psi_\rho\rangle\langle\psi_\rho|),|\psi_\rho\rangle\langle\psi_\rho|)\), where \(|\psi_\rho\rangle\) purifies \(\rho\). 
For the depolarizing channel \(\Phi_{\mathrm{Petz}}\), one finds
\[
(\Phi_{\mathrm{Petz}}\otimes \mathrm{id})(|\psi_\rho\rangle\langle\psi_\rho|)
=
\lambda_{\mathrm{Petz}}|\psi_\rho\rangle\langle\psi_\rho|
+(1-\lambda_{\mathrm{Petz}})\frac{I_d}{d}\otimes \rho_R .
\]
Since the target is pure, this gives \(F_\rho(\Phi_{\mathrm{Petz}},\mathrm{id})=\sqrt{\lambda_{\mathrm{Petz}}+\frac{1-\lambda_{\mathrm{Petz}}}{d}\operatorname{Tr}(\rho^2)}\). 
This quantity is minimized at \(\rho=I_d/d\), i.e. when \(\operatorname{Tr}(\rho^2)=1/d\). Therefore
\[
F(\Phi_{\mathrm{Petz}},\mathrm{id})
=
\sqrt{\lambda_{\mathrm{Petz}}+\frac{1-\lambda_{\mathrm{Petz}}}{d^2}} .
\]
Equivalently,
\[
\eta_{\mathrm{Petz}}^{\mathrm{ent}}
:=
1-F^2(\Phi_{\mathrm{Petz}},\mathrm{id})
=
\frac{d^2-1}{d^2}(1-\lambda_{\mathrm{Petz}})
=
\frac{d+1}{d}\eta_{\mathrm{Petz}}^{\mathrm{st}} .
\]

Thus, the exact worst-case channel fidelity of the Petz-recovered channel is determined by \(\lambda_{\mathrm{Petz}}\). 
However, analyzing this expression directly would require obtaining the scaling of \(\lambda_{\mathrm{Petz}}\) in terms of \(n\). 
This, in turn, requires a careful analysis of Eq.~\eqref{eq:lambdaPetz}, namely of the scaling of \(A\) and \(\mathcal C_i(\bar t_a)\) with \(n\). 
This is highly nontrivial, because these operators lie in the adjoint representation inside \((V_{\omega_r})^{\otimes(n-1)}\), 
where the adjoint component appears with large multiplicity, in contrast to the single-qudit case. 
Similar complications will appear below in the analysis of multi-qudit erasure errors.

We leave a direct analysis of \(\lambda_{\mathrm{Petz}}\) for future work. In this work we use the near-optimality theorem of Ng and Mandayam for the transpose channel, 
equivalently the Petz recovery map with the maximally mixed code state as reference state~\cite{Ng2010}. Their result states that
\[
\eta_{\mathrm{Petz}}^{\mathrm{st}}
\le
\eta_{\mathrm{op}}^{\mathrm{st}}
\frac{(d+1)-\eta_{\mathrm{op}}^{\mathrm{st}}}
{1+(d-1)\eta_{\mathrm{op}}^{\mathrm{st}}},
\]
where \(\eta_{\mathrm{op}}^{\mathrm{st}}:=\min_{\mathcal R}[1-\min_{\psi}f^2(\psi,(\mathcal R\circ\mathcal C)(\psi))]\). 
Moreover, the optimal state-fidelity error is bounded by the optimal entanglement-fidelity error, \(\eta_{\mathrm{op}}^{\mathrm{st}}\leq \eta_{\mathrm{op}}^{\mathrm{ent}}\). 
Therefore, since Theorem~$2$ from the main part of the paper gives \(\eta_{\mathrm{op}}^{\mathrm{ent}}=O_{d}(n^{-2})\), we obtain \(\eta_{\mathrm{Petz}}^{\mathrm{st}}=O_{d}(n^{-2})\). 
Finally, using \(\eta_{\mathrm{Petz}}^{\mathrm{ent}}=\frac{d+1}{d}\eta_{\mathrm{Petz}}^{\mathrm{st}}\), we conclude that
\[
\eta_{\mathrm{Petz}}^{\mathrm{ent}}=O_{d}(n^{-2}).
\]
This gives the desired scaling.

\end{proof}

\subsection{Flagged two- and three-qudit erasure}\label{subsec:supple_multi_erasure_noise}

The following Lemmas~\ref{lem:M_module_irreducibility}-\ref{lem:3-body_action} are axiliary results that we will need for our analysis of multi-qudit 
erasure errors and thus can be skipped on first reading.
For notation see Section~\ref{sec:supple_sud} and Nomenclature section.

\begin{smlemma}\label{lem:M_module_irreducibility}
    Let $M_{\omega_1}$ be a multiplicity space of the fundamental $SU(d)$ representation $V_{\omega_1}$ in $\left(V_{\omega_1}\right)^{\otimes n}$, where $n\equiv 1 \pmod{d}$. 
    Then $M_{\omega_1}$ is an irreducible $S_n$-module, namely the Specht module $[q+1, q^{d-1}]$, where $q=\frac{n-1}{d}$.
\end{smlemma}

\begin{proof}
    We recall the notation $\mathcal{V}=\C^d$. As always, we first work in the context of $GL_d(\mathbb{C})$ representations and then restrict ourselves to $SU(d)$ representations. 
    From the Schur-Weil duality, we have the following decomposition
    \[
        V^{\otimes n} \cong \bigoplus_{\lambda \vdash n, \ell(\lambda) \leq d} S_\lambda(V) \otimes[\lambda],
    \]
    where $S_\lambda(V)$ is $GL_d(\mathbb{C})$-module and $[\lambda]$ is irreducible $S_n$-module, known as Specht module.
    If we restrict this decomposition to $SU(d)$, then we get
    \[
        V^{\otimes n}\downarrow_{SU(d)}\cong V_{\omega_1}^{\otimes n} \cong \bigoplus_{\lambda \vdash n, \ell(\lambda) \leq d}\left(S_\lambda(V) \downarrow_{S U(d)}\right) \otimes[\lambda],
    \]
    where $S_\lambda(V) \downarrow_{S U(d)}$ is the restriction of $GL_d(\mathbb{C})$-module $S_\lambda(V)$ to $SU(d)$. 
    The decomposition is basically the same as for $GL_d(\mathbb{C})$ case, but the main difference
    is that $S_\lambda(V) \downarrow_{S U(d)}$ may coincide with $S_\mu(V) \downarrow_{S U(d)}$ for $\lambda \neq \mu$, 
    so multiplicity of $V_{\omega_1}$ as a representation of $SU(d)$ requires a careful analysis.
    Multiplicity space $M_{\omega_1}$ is described by the following formula
    \[
        M_{\omega_1}\cong \operatorname{Hom}_{S U(d)}\left(V_{\omega_1}, V_{\omega_1}^{\otimes n}\right) 
        \cong \bigoplus_{\lambda \vdash n, \ell(\lambda) \leq d} \operatorname{Hom}_{S U(d)}\left(V_{\omega_1}, S_\lambda(V) \downarrow_{S U(d)}\right) \otimes[\lambda],
    \]
    so we can write the following
    \[
        M_{\omega_1} \cong \bigoplus_{\lambda \vdash n, \ell(\lambda) \leq d} m_\lambda[\lambda],
    \]
    where since $S_\lambda(V)$ is still irreducible as a representation of $SU(d)$, we have $m_{\lambda}\in \{0,1\}$ for every
    $\lambda$.
    The restriction $S_\lambda(V)\downarrow_{SU(d)}$ is indeed irreducible. By the unitarian trick, $S_\lambda(V)$ is irreducible as a representation 
    of the compact group $U(d)$, because $\mathfrak{gl}_d(\C)$ is the complexification of $\mathfrak u(d)$ and the two therefore have the same invariant subspaces. 
    Every element of $U(d)$ is the product of an element of $SU(d)$ with a central phase $e^{i\theta}I$, and the latter acts on $S_\lambda(V)\subseteq V^{\otimes|\lambda|}$ 
    by the scalar $e^{i|\lambda|\theta}$. A subspace invariant under $SU(d)$ is therefore invariant under all of $U(d)$, hence is $0$ or $S_\lambda(V)$. 
    Irreducibility of the restriction then makes $\operatorname{Hom}_{SU(d)}(V_{\omega_1},S_\lambda(V)\downarrow_{SU(d)})$ one-dimensional when $S_\lambda(V)\downarrow_{SU(d)}\cong V_{\omega_1}$, and zero otherwise.
    We need to show that $m_\lambda=1$ for exactly one $\lambda$. Indeed,
    $S_\lambda(V) \downarrow_{S U(d)} \cong V_{\omega_1}$ happens exactly iff $\lambda=\left(q+1, q^{d-1}\right)$ for some $q\geq 0$.
    Since also $|\lambda|=n$, we have $n=(q+1)+q(d-1)=d q+1$, so $q=\frac{n-1}{d}$ is an integer. 
    Therefore,
    \[
        M_{\omega_1} \cong\left[q+1, q^{d-1}\right],
    \]
    which proves the lemma.
\end{proof}

\begin{smlemma}\label{lem:dimension_formula}
    Let $[q+1, q^{d-1}]$ be the Specht module corresponding to the partition $(q+1, q^{d-1})$. Then
    \[
        \operatorname{dim} [q+1, q^{d-1}] = \frac{(d q+1)!d!\prod_{i=0}^{d-2} i!}{(q+d)!\prod_{i=0}^{d-2}(q+i)!}.
    \]
\end{smlemma}

\begin{proof}
    The general formula for the dimension of Specht module $[\lambda]$ corresponding to the partition $\lambda$ is given by the hook-length formula, i. e.
    \[
        \operatorname{dim} [\lambda] = \frac{n!}{\prod_{(i, j) \in \lambda} h_{i j}} .
    \]
    where $h_{ij}$ is the hook length of the box in the $i$-th row and $j$-th column, precisely
    \[
        h_{i j}=\#\{\text { boxes to the right of }(i, j)\}+\#\{\text { boxes below }(i, j)\}+1 .
    \]
    It can also be expressed as
    $
        h_{i j}=\lambda_i-j+\lambda_j^{\prime}-i+1 .
    $
    Our Young diagram has the shape of a rectangle with an extra box in the top right corner, and its transpose 
    has the shape of a rectangle with an extra box in the bottom left corner, i. e.
    $$
        \lambda=(q+1, \underbrace{q, \ldots, q}_{d-1 \text { times }}), \quad \lambda'=(\underbrace{d, d, \ldots, d}_{q \text { times }}, 1).
    $$
    Now let's compute the hook lengths for each box in the Young diagram $\lambda$ -- row by row.
    First row, columns $1 \leq j \leq q$. Here $\lambda_1 = q+1$ and $\lambda_j'=d$, 
    $
        h_{1 j}=(q+1-j)+(d-1)+1=q+d+1-j.
    $
    The extra box in the first row has length 1, i. e. 
    $
        h_{1, q+1}=1.
    $
    Thus, their product is
    \[
        \prod_{j=1}^q h_{1 j}=\prod_{j=1}^q(q+d+1-j)=\frac{(q+d)!}{d!} .
    \]
    Rows $2 \leq i \leq d$, columns $1 \leq j \leq q$. Here $\lambda_i=q$ and $\lambda_j'=d$, so
    $
        h_{i j}=(q-j)+(d-i)+1=q+d-i-j+1 .
    $
    Thus, their product is
    \[
       \prod_{i=2}^d \prod_{j=1}^q h_{i j}=\frac{(q+d-i)!}{(d-i)!}.
    \]
    Combining all together, we have
    \[
        \prod_{(i, j) \in \lambda} h_{i j}=\frac{(q+d)!}{d!} \prod_{i=2}^d \frac{(q+d-i)!}{(d-i)!}
    \]
    and the formula in the formulation of the Lemma follows immediately.

\end{proof}

\begin{smlemma}\label{lem:dimension_bound}

Let $[q+1, q^{d-1}]$ be the Specht module corresponding to the partition $(q+1, q^{d-1})$. 
Then, for $q\geq d$ the following lower bound holds
\[
    \operatorname{dim} [q+1, q^{d-1}] \geq \left(\frac{d!\sqrt{2 \pi d}\prod_{i=0}^{d-2} i!}{e^d 2^{\frac{d^2-d+2}{2}}} \right) d^{d q} q^{-\frac{d^2+1}{2}} := C_d d^{d q} q^{-\frac{d^2+1}{2}}.
\]

\end{smlemma}

\begin{proof}
    From the previous lemma, we have
\[
    \operatorname{dim} [q+1, q^{d-1}]=\frac{(d q+1)!d!\prod_{i=0}^{d-2} i!}{(q+d)!\prod_{i=0}^{d-2}(q+i)!}
\]
and we need to estimate this quantity from below. Using obvious bounds
\[
    (q+d)!\leq q!(q+d)^d, \quad(q+i)!\leq q!(q+d)^i \quad(0 \leq i \leq d-2),
\]
we can express the denominator as
\[
    (q+d)!\prod_{i=0}^{d-2}(q+i)!\leq(q!)^d(q+d)^{\frac{d^2-d+2}{2}},
\]
so we have
\begin{equation}\label{eq:specht_dim_lb}
    \operatorname{dim} [q+1, q^{d-1}] \geq d!\prod_{i=0}^{d-2} i! \cdot \frac{(dq)!}{(q!)^d(q+d)^{\frac{d^2-d+2}{2}}}.
\end{equation}
Notice, that we are interested in scaling of $\operatorname{dim} [q+1, q^{d-1}]$ as $q$ grows, because in our setting
$d$ is fixed and $q$ depends linearly on $n$ (recall that $n=d q+1$). Using a simplified version of Robbins' bounds 
(which are derived from famous Stirling's formula, see~\cite{Robbins1955}), we get 
\[
   \sqrt{2 \pi N}(N / e)^N \leq \sqrt{2 \pi n}\left(\frac{n}{e}\right)^n e^{\frac{1}{12 n+1}} 
   \leq N!\leq \sqrt{2 \pi n}\left(\frac{n}{e}\right)^n e^{\frac{1}{12 n}} \leq e \sqrt{N}(N / e)^N.
\]
Then
\[
    \frac{(d q)!}{(q!)^d} \geq \frac{\sqrt{2 \pi d q}(d q / e)^{d q}}{\left(e \sqrt{q}(q / e)^q\right)^d}=\frac{\sqrt{2 \pi d}}{e^d} d^{d q} q^{-(d-1) / 2}
\]
Substituting into lower bound~\eqref{eq:specht_dim_lb} we get
\[
    \operatorname{dim} [q+1, q^{d-1}] \geq d!\prod_{i=0}^{d-2} i! \cdot \frac{\sqrt{2 \pi d}}{e^d} d^{d q} q^{-(d-1) / 2}(q+d)^{-\frac{d^2-d+2}{2}} 
    \geq \left(\frac{d!\sqrt{2 \pi d}\prod_{i=0}^{d-2} i!}{e^d 2^{\frac{d^2-d+2}{2}}} \right) d^{d q} q^{-\frac{d^2+1}{2}},
\]
where we used the condition $q\geq d$ in the second inequality. This completes our proof.

\end{proof}

The following lemma was proven in~\cite{LarsenShalev2008}. We will need it as a tool for bounding characters in Lemmas~\ref{lem:2_transitive_subgroup} and~\ref{lem:3_transitive_subgroup}.

\begin{smlemma}\label{lem:character_bound}
    Let $\sigma \in S_n$ and $\chi \in \text{Irr}(S_n)$. If $\sigma$ is fixed-point-free, or has $n^{o(1)}$ fixed points, then
    \[
        |\chi(\sigma)| \leq \chi(1)^{1 / 2+o(1)}.
    \]
\end{smlemma}

\begin{smlemma} \label{lem:2_transitive_subgroup}
    There exists infinitely many primes $p$ such that $n$ is a prime power, i. e. $n=p^k$ for any $k\geq 1$, and $n\equiv 1 \pmod{d}$, there is a family of $2$-transitive subgroups of $S_n$, namely
    \begin{equation}\label{eq:AGL}
        G\equiv\operatorname{AGL}(1, n)=\left\{x \mapsto a x+b: a \in \mathbb{F}_n^{\times}, b \in \mathbb{F}_n\right\}\leq S_n,
    \end{equation}
    acting on $\mathbb{F}_n$, such that for sufficiently large $q:=\frac{n-1}{d}$ (i. e. for sufficiently large $n$), we have 
    \[\dim [q+1, q^{d-1}]^G > 0.\]
\end{smlemma}

\begin{smremark} \label{rem:infinitely_many_primes}
    First of all, $n$ has to be a prime power, since $n$ is the size of the field $\mathbb{F}_{n}$. 
    The fact that an infinite family of such $p$ (i.e., infinitely many $p$ such that $n=p^k$ for any $k\geq 1$, and $n\equiv 1 \pmod{d}$) exists is a direct consequence of one of Dirichlet's 
    Theorems on Arithmetic Progressions.
    Indeed, the theorem directly states that there are infinitely many primes $p$ such that $p \equiv 1 \pmod{d}$.
    To provide proof for arbitrary $k\geq 1$, we write $p=1+d q$ for some $q\geq 1$, and then we have
    \[
        p^k=(1+d q)^k=\sum_{j=0}^k\binom{k}{j}(d q)^j=1+\sum_{j=1}^k\binom{k}{j}(d q)^j := 1+d \cdot t,
    \]
    thereby $p^k \equiv 1 \pmod{d}$ for any $k\geq 1$.
\end{smremark}

\begin{proof}
The fact that given $G$ is $2$-transitive is easy to see. Indeed, given any two pairs of distinct elements $(x, y)$ and $(x', y')$ in $\mathbb{F}_n$, we can do the following mapping
\[
    z \mapsto a z+b, \quad a=\frac{y^{\prime}-x^{\prime}}{y-x}, \quad b=x^{\prime}-a x
\]
that sends $(x, y)$ to $(x', y')$, so it is $2$-transitive.

Now we need to show that $\dim [q+1, q^{d-1}]^G > 0$. For that, we recall the following formula from the character theory of finite groups:
\[
    \operatorname{dim}\left([\lambda]\right)^G=\frac{1}{|G|} \sum_{g \in G} \chi^\lambda(g)=\frac{1}{|G|}\left(\chi^\lambda(1)+\sum_{g \neq 1}\chi^\lambda(g) \right),
\]
so proving that $\dim [q+1, q^{d-1}]^G > 0$ is equivalent to proving that
\[
    \sum_{1 \neq g \in G}\left|\chi^\lambda(g)\right|<\chi^\lambda(1).
\]
Now we are to apply the bound from Lemma \ref{lem:character_bound}. Let's
prove that every non-identity element of $G$ has at most $1$ fixed point (which is definitely $n^{o(1)}$ fixed points).
Consider $g\neq 1 \in G$, the general form would be $x \mapsto a x+b, \quad a \in \mathbb{F}_n^{\times}, b \in \mathbb{F}_n$. 
If $a=1$, then $g$ is a translation, and it has no fixed points. If $a \neq 1$, then the fixed-point equation $x=a x+b$ has exactly one solution $x=\frac{b}{1-a}$, so $g$ has exactly one fixed point.
Therefore, we can apply the bound from Lemma \ref{lem:character_bound}, and we have
\[
    \left|\chi^\lambda(g)\right| \leq\left(\chi^\lambda(1)\right)^{1 / 2+\varepsilon_n} \quad(1 \neq g \in G) .
\]
where $\varepsilon_n \to 0$ as $n \to \infty$. Group $G$ size can be estimated as follows
$
    |G|=n(n-1)<n^2 .
$
Combining these two inequalities we have
\[
    \sum_{1 \neq g \in G}\left|\chi^\lambda(g)\right| \leq(|G|-1)\left(\chi^\lambda(1)\right)^{1 / 2+\varepsilon_n}<n^2\left(\chi^\lambda(1)\right)^{1 / 2+\varepsilon_n} .
\]
We choose $n$ sufficiently large, we can make $\varepsilon_n<\frac{1}{3}$, and recalling exponential scaling with $n$ 
of $\chi^{\lambda}(1)$ from Lemma~\ref{lem:dimension_bound}, i. e. $\dim \left[q+1, q^{d-1}\right]= \chi^\lambda(1) \geq C_d d^{d q} q^{-\frac{d^2+1}{2}}$ 
($d>1$ in our case) we finally get for sufficiently large $n$ that
\[
    n^2\left(\chi^\lambda(1)\right)^{5 / 6}<\chi^\lambda(1),
\]
which completes our proof.
\end{proof}

\begin{smtheorem}\label{thrm:2_transitive_covariant_code}
        Let
    $
        \mathcal H_L=V_{\omega_1},\;
        \mathcal H_P=(V_{\omega_1})^{\otimes n}
    $
    and $n$ is such that $n \equiv 1(\bmod d)$ and $n=p^k$, where $p$ is a prime number and $k$ is a positive integer.
    Equip \(\mathcal H_L\) with the fundamental representation of
    \(\mathfrak{su}(d)\), and equip \(\mathcal H_P\) with the transversal
    representation given by
    \begin{equation}
    \begin{gathered}
    t_a \rightarrow \sum_{i=1}^n t_a^{(i)}, \quad a=1,\ldots,d^2-1,
    \end{gathered}
    \end{equation} 
    where $t_a^{(i)}$ is a fundamental of $t_a$ on the $i$-th physical space $V_{\omega_1}$.
    For sufficiently
    large $n$ there exists an $\mathfrak{su}(d)$-covariant encoding $\mathcal{E}$ with respect to these representations, 
    such that the code space is invariant under the action of $\mathrm{AGL(1, n)}$. 
\end{smtheorem}
\begin{proof}
    As we pointed out before, choosing a code space comes to choosing a multiplicity vector in the multiplicity space $M_{\omega_1}$,
    corresponding to the fundamental representation $V_{\omega_1}$ in the decomposition of physical representation $(V_{\omega_1})^{\otimes n}$ -- that is our only "degree of freedom" in code construction. 
    For this choice of individual physical space, we prove in Lemma~\ref{lem:M_module_irreducibility} that the multiplicity space is isomorphic to the Specht module $M_{\omega_1}\cong [q+1, q^{d-1}]$, 
    where $q=\frac{n-1}{d}$. According to Lemma~\ref{lem:2_transitive_subgroup}, there is a subgroup of $S_n$ that acts 2-transitively on $\{1, \dots, n\}$ and has the invariant subspace in the multiplicity space $[q+1, q^{d-1}]$ 
    (for sufficiently large $n$). We don't care about the structure of this subgroup (the detailed description can be found in the proof of Lemma~\ref{lem:2_transitive_subgroup}), 
    and we will denote it as $G_2$. We denote the invariant subspace in the multiplicity space as $[q+1, q^{d-1}]^{G_2}$. 
    So we choose our multiplicity vector to be one of the vectors from $[q+1, q^{d-1}]^{G_2}$, so our code space is invariant under the action of $G_2$. 
    (That is actually the reason why we don't care about the structure of $G_2$ -- we will only leverage $G_2$ being 2-transitive and the chosen code space being invariant under its action).
\end{proof}

\begin{smlemma} \label{lem:3_transitive_subgroup}
    There exist infinitely many primes $p$ such that $n-1$ is a prime power, i. e. $n-1=p^k$ for any $k\geq 1$, and $n\equiv 1 \pmod{d}$, there is a family of $3$-transitive subgroups of $S_n$, namely
    \begin{equation}\label{eq:3_transitive_subgroup}
        G\equiv\operatorname{PGL}(2, n-1)=\left\{A \in M_{2 \times 2}\left(\mathbb{F}_{n-1}\right): \operatorname{det} A \neq 0\right\}/\left\{\lambda I_2: \lambda \in \mathbb{F}_{n-1}^{\times}\right\}\leq S_n,
    \end{equation}
    acting on $\mathbb{P}^1\left(\mathbb{F}_{n-1}\right)=\mathbb{F}_{n-1} \cup\{\infty\}$, such that for sufficiently large $q:=\frac{n-1}{d}$ (i. e. for sufficiently large $n$), we have 
    \[\dim [q+1, q^{d-1}]^G > 0.\]
\end{smlemma}

\begin{smremark}
    Unfortunately, we have to impose $d$ being a prime power. Indeed, $n-1$ is the size of the field $\mathbb{F}_{n-1}$, so $n-1$ must be a prime power,
    and since $n-1=d q$, then $d$ must be a prime power as well. It is not restrictive, because if $d=p^r$ it means that we have a system of $r$ particles
    of dimension $p$, which we are aiming to encode -- this is quite a common scenario.
    Nevertheless, as it was pointed out before, it doesn't fundamentally restrict the absence of a 3-transitive group, which will serve our purposes
    for arbitrary $d$.
\end{smremark}

\begin{proof}
    The idea for this proof is the same as for Lemma~\ref{lem:2_transitive_subgroup}. 
    First, we need to show that given $G$ is $3$-transitive. It is enough to show that every triple $(x_1, x_2, x_3)$ can be mapped to $(\infty, 0, 1)$. One can easily check that such a mapping is given by the following formula
    \[
        T_x(t)=\frac{\left(t-x_2\right)\left(x_3-x_1\right)}{\left(t-x_1\right)\left(x_3-x_2\right)},
    \]
    so it is indeed $3$-transitive, and, hence, $3$-transitive.

    In Lemma~\ref{lem:2_transitive_subgroup}, we showed that showing $\dim [q+1, q^{d-1}]^G > 0$ is equivalent to showing that $\sum_{1 \neq g \in G}\left|\chi^\lambda(g)\right|<\chi^\lambda(1)$. 
    We can apply the same bound from Lemma~\ref{lem:character_bound} as in Lemma~\ref{lem:2_transitive_subgroup}, because every non-identity element of $G$ has at most $2$ fixed points. 
    Indeed, a non-trivial element of $G$ is described by a Mobius transformation (in projective coordinates) 
    $
        x \mapsto \frac{a x+b}{c x+d},
    $
    for which fixed point equation is given by
    $
        x = \frac{a x+b}{c x+d},
    $
    equivalent to
    $
        c x^2+(d-a) x-b=0,
    $
    that is a quadratic equation, so it has at most $2$ solutions. Therefore, every non-identity element of $G$ has at most $2$ fixed points.

    Group $G$ size can be estimated as follows
    $
        |G|=(n-1) n(n-2)<n^3,
    $
    so, as before, we have
    \[
        \sum_{1 \neq g \in G}\left|\chi^\lambda(g)\right| \leq n^3\left(\chi^\lambda(1)\right)^{1 / 2+\varepsilon_n}
    \]
    and we again choose $n$ sufficiently large, we can make $\varepsilon_n<\frac{1}{3}$, so
    \[
        \sum_{1 \neq g \in G}\left|\chi^\lambda(g)\right| \leq n^3\left(\chi^\lambda(1)\right)^{5 / 6}
    \]
    the rest of the proof is the same as in Lemma~\ref{lem:2_transitive_subgroup}, which completes our proof.
\end{proof}

\begin{smtheorem}\label{thrm:3_transitive_covariant_code}
        Let
    $
        \mathcal H_L=V_{\omega_1},\;
        \mathcal H_P=(V_{\omega_1})^{\otimes n}
    $
    with $d=p^r$ where $p$ is prime and $r$ is a positive integer, 
    and $n$ is such that $n-1=p^k$, where $k \geq r$ is a positive integer. 
    Equip \(\mathcal H_L\) with the fundamental representation of
    \(\mathfrak{su}(d)\), and equip \(\mathcal H_P\) with the transversal
    representation given by
    \begin{equation}
    \begin{gathered}
    t_a \rightarrow \sum_{i=1}^n t_a^{(i)}, \quad a=1,\ldots,d^2-1,
    \end{gathered}
    \end{equation} 
    where $t_a^{(i)}$ is a fundamental of $t_a$ on the $i$-th physical space $V_{\omega_1}$.
    For sufficiently
    large $n$ there exists an $\mathfrak{su}(d)$-covariant encoding $\mathcal{E}$ with respect to these representations, 
    such that the code space is invariant under the action of $\mathrm{PGL}(2, n-1)$. 
\end{smtheorem}

\begin{proof}
    The proof is the same as for Theorem~\ref{thrm:2_transitive_covariant_code}, but instead of a 2-transitive subgroup, 
    we use a 3-transitive subgroup defined in Lemma~\ref{lem:3_transitive_subgroup}.
\end{proof}

\begin{smlemma}\label{lem:generalisation_of_Casimir_identity}
Consider $SU(d)$ representation of the form $(V_{\omega_1})^{\otimes n}$, where $V_{\omega_1}$ is 
a fundamental representation of $SU(d)$. Let $\{t_a^{(i)}\}$ be representation of generators of $\mathfrak{su}(d)$ in the $i$-th copy of $V_{\omega_1}$, where $a=1, \ldots, d^2-1$ and $i=1, \ldots, n$. Then the following identity holds
\[
    \sum_a\left(\sum_{j=1}^n t^{(j)}_a\right)^2=\frac{n\left(d^2-1\right)}{2 d} I+2 \sum_{i<j} \sum_a t^{(i)}_a t^{(j)}_a.
\]
\end{smlemma}

\begin{proof}
    We recall that the global generator operator is
    $
        T_a=\sum_{j=1}^n t^{(j)}_a,
    $
    and the total quadratic Casimir operator
    $
        C_2^{\text {total }}=\sum_a\left(T_a\right)^2.
    $
    Expanding the square of the sum over the sites we get
    \[
        \sum_a\left(T_a\right)^2=\sum_a\left(\sum_{j=1}^n t^{(j)}_a\right)\left(\sum_{k=1}^n t^{(k)}_a\right)=\sum_a \sum_{j=1}^n \sum_{k=1}^n t^{(j)}_a t^{(k)}_a,
    \]
    which can be split into a sum of two terms: "diagonal" term, where $j=k$, and "cross" term, where $j \neq k$, namely
    \[
        \sum_a\left(T_a\right)^2=\sum_{j=1}^n\left(\sum_a\left(t^{(j)}_a\right)^2\right)+\sum_{j \neq k} \sum_a t^{(j)}_a t^{(k)}_a.
    \]
    $\sum_a\left(t^{(j)}_a\right)^2$ is a quadratic Casimir operator for the $j$-th copy of $V_{\omega_1}$, and from the general 
    representation theory of $SU(d)$, we know that it is proportional to the identity operator, with a constant equal to $\frac{d^2-1}{2 d}$, so
    \[
        \sum_{j=1}^n\left(\sum_a\left(t^{(j)}_a\right)^2\right)=\sum_{j=1}^n \frac{d^2-1}{2 d} I=\frac{n\left(d^2-1\right)}{2 d} I.
    \]
    As for the "cross" term, we can rewrite it as follows
    \[
        \sum_{j \neq k} \sum_a t^{(j)}_a t^{(k)}_a=2 \sum_{j<k} \sum_a t^{(j)}_a t^{(k)}_a.
    \]
    Combining the two terms together, we get the desired identity
    \[
        \sum_a\left(T_a\right)^2=\frac{n\left(d^2-1\right)}{2 d} I+2 \sum_{j<k} \sum_a t^{(j)}_a t^{(k)}_a.
    \]
\end{proof}

From now on we will mostly use Einstein summation convention, so repeated indices are summed over.
We use notation $\bar{t}_a$ instead of $t_a$ to emphasize that $\bar{t}_a$ act on logical space.

\begin{smlemma}\label{lem:2-body_action}
Let $\mathcal{V}$ be an $SU(d)$-covariant encoding isometry. For distinct indices $i, j$, the following holds
\[
    \mathcal{V}^{\dagger} t_a^{(i)} t_b^{(j)} \mathcal{V}=\alpha^{(i j)} \delta_{a b} I_L+\beta^{(i j)}_{f} f_{abc} \bar{t}_c+\beta^{(i j)}_{d} d_{abc} \bar{t}_c .
\]
where $\alpha^{(i j)}, \beta^{(i j)}_{f}, \beta^{(i j)}_{d}$ are some constants.
\end{smlemma}
\begin{proof}
    Let's denote $X_{a b}^{(i j)}:=\mathcal{V}^{\dagger} t_a^{(i)} t_b^{(j)} \mathcal{V}$. Since $\operatorname{End}\left(V_{\omega_1}\right) \cong \mathbf{1} \oplus \operatorname{Ad}$,
    the operator can be decomposed as
    \[
        X_{a b}^{(i j)}=S_{a b}^{(i j)} I_L+V_{abc}^{(i j)} \bar{t}_c,
    \]
    where $S_{a b}^{(i j)}$ and $V_{abc}^{(i j)}$ are some constants.
    First of all, let's understand the structure of $S_{a b}^{(i j)}$ and $V_{abc}^{(i j)}$. 
    Under a global $SU(d)$ transformation,
    $$
    \left(\bigotimes_{k=1}^n U^{(k)}(g)\right) t_a^{(i)} \left(\bigotimes_{k=1}^n U^{(k)}(g)\right)^{-1}=\mathrm{Ad}_{ab}(g) t_b^{(i)}, \quad \forall g \in SU(d).
    $$
    Therefore
    $$
    X_{a b}^{(i j)}=\mathrm{Ad}_{a a^{\prime}}(g) \mathrm{Ad}_{b b^{\prime}}(g) X_{a^{\prime} b^{\prime}}^{(i j)} .
    $$
    Then $S_{ab}^{(i j)}$ must be an invariant tensor in $\operatorname{Hom}_{SU(d)}(\operatorname{Ad} \otimes \operatorname{Ad}, \mathbf{1})$, and
    $V_{abc}^{(i j)}$ must be an invariant tensor in $\operatorname{Hom}_{SU(d)}(\operatorname{Ad} \otimes \operatorname{Ad}, \operatorname{Ad})$, where $i, j$ are distinct and fixed and tensor indexes are $a, b, c$.
    Obviously, 
    $$
    \operatorname{Hom}_{S U(d)}(\operatorname{Ad} \otimes \operatorname{Ad}, \mathbf{1})=\mathbb{C} \delta_{a b},
    $$
    since it can be interpreted as $SU(d)$ invariant bilinear form, therefore $S_{a b}^{(i j)}=\alpha^{(i j)}\delta_{ab}$. 
    As for $\operatorname{Hom}_{SU(d)}(\operatorname{Ad} \otimes \operatorname{Ad}$, $\operatorname{Ad})$, one can notice that maps $(X, Y) \mapsto[X, Y]$ -- $\left[t_a, t_b\right]=i f_{abc} t_c$ on basis, -- and $(X, Y) \mapsto\{X, Y\}-\frac{2}{d} \operatorname{tr}(X Y) I$ -- 
    $\left\{t_a, t_b\right\}-\frac{1}{d} \delta_{a b} I=d_{abc} t_c$ on basis, -- are exactly two linear independent $SU(d)$-covariant bilinear maps. And any other map should be a linear combination of these two maps,
    so 
    $$
    \operatorname{Hom}_{SU(d)}(\operatorname{Ad} \otimes \operatorname{Ad}, \operatorname{Ad})=\mathbb{C} f_{abc} \oplus \mathbb{C} d_{abc},
    $$ 
    from which we conclude that $V_{abc}^{(i j)}=\beta^{(i j)}_{f} f_{abc}+\beta^{(i j)}_{d} d_{abc}$, which completes our proof.
\end{proof}

\begin{smlemma}\label{lem:3-body_action}
Let $\mathcal{V}$ be an $SU(d)$-covariant encoding isometry. For distinct indices $i, j, k$, the following holds
\[
    \mathcal{V}^{\dagger} t_a^{(i)} t_b^{(j)} t_c^{(k)} \mathcal{V}=\mu_{f}^{(ijk)} f_{a b c} I_L+\mu_{d}^{(ijk)} d_{a b c} I_L+B_{a b c m}^{(ijk)} \bar{t}_m
\]
where $\mu_{f}^{(ijk)}, \mu_{d}^{(ijk)}$ are some constants and $B_{a b c m}^{(ijk)}$ is some tensor.
\end{smlemma}

\begin{proof}
    We will follow the same logic as in the proof of Lemma~\ref{lem:2-body_action}. Let's denote $X_{a b c}^{(i j k)}:=\mathcal{V}^{\dagger} t_a^{(i)} t_b^{(j)} t_c^{(k)} \mathcal{V}$. Since $\operatorname{End}\left(V_{\omega_1}\right) \cong \mathbf{1} \oplus \operatorname{Ad}$,
    the operator can be decomposed as
    \[
        X_{abc}^{(i j k)}=A_{abc}^{(i j k)} I_L+B_{abcm}^{(i j k)} \bar{t}_m,
    \]
    Under $SU(d)$ transformations, we have
    \[
        X_{abc}^{(i j k)}=\operatorname{Ad}_{a a^{\prime}}(g) \operatorname{Ad}_{b b^{\prime}}(g) \operatorname{Ad}_{c c^{\prime}}(g) X_{a^{\prime} b^{\prime} c^{\prime}}^{(i j k)},
    \]
    which means that $A_{abc}$ must be an invariant tensor in $\operatorname{Hom}_{SU(d)}(\operatorname{Ad}^{\otimes 3}, \mathbf{1})$, and
    $B_{abcm}$ must be an invariant tensor in $\operatorname{Hom}_{SU(d)}(\operatorname{Ad}^{\otimes 3}, \operatorname{Ad})$.
    One can check that 
    $$
    \operatorname{Hom}_{S U(d)}\left(\operatorname{Adj}^{\otimes 3}, \mathbf{1}\right)=\mathbb{C} f_{a b c} \oplus \mathbb{C} d_{a b c},
    $$ 
    so $A_{abc}^{(i j k)}=\mu_{f}^{(ijk)} f_{a b c}+\mu_{d}^{(ijk)} d_{a b c}$. As for $\operatorname{Hom}_{SU(d)}\left(\operatorname{Adj}^{\otimes 3}, \operatorname{Adj}\right)$, it is more complicated, so we will leave $B_{abcm}^{(i j k)}$ as it is, which completes our proof.
\end{proof}

\begin{smproposition}\label{prop:1_and_2_body_action}
    Let $\mathcal{V}$ be encoding isometry of the code defined in the Theorem~\ref{thrm:2_transitive_covariant_code}, then for sufficiently large $n$
    \[
        \mathcal{V}^{\dagger}t^{i}_a \mathcal{V}=\frac{1}{n} \bar{t}_a, \quad \mathcal{V}^{\dagger} t_{a_1}^{\left(i_1\right)} t_{a_2}^{\left(i_2\right)} \mathcal{V}=-\frac{1}{2dn} \delta_{a_1 a_2} I_L,
    \]
    where $i_1, i_2$ are distinct.
\end{smproposition}

\begin{smremark}
    The existence of an infinite number of $n$ that satisfy the conditions of the proposition is shown in Remark~\ref{rem:infinitely_many_primes}.
\end{smremark}

\begin{proof}
The proof will consist of two steps. First, we will choose a code space, so $t_a^{(i)}$ acts on it independently of $i$, and $t_a^{(i)}t_b^{(j)}$ act the same way on it independent of $(i,j)$ pairs. Second, we will find their absolute values.

In Lemma~\ref{lem:2-body_action}, we prove that the action of two-body operators on the code space has the following form
\[
    \mathcal{V}^{\dagger} t_a^{(i)} t_b^{(j)} \mathcal{V}=\alpha^{(i j)} \delta_{a b} I_L+\beta^{(i j)}_{f} f_{a b c} \bar{t}_c +\beta^{(i j)}_{d} d_{a b c} \bar{t}_c .
\]
where $d_{a b c}$ and $f_{a b c}$ are constants symmetric and antisymmetric on all indexes correspondingly (the action of one-body operators we already know from (\ref{eq:Vt_aV}): $\mathcal{V}^{\dagger} t_a^{(i)} \mathcal{V}=\alpha^{(i)} \bar{t}_a$). 
First of all, because $G_2$ is 2-transitive, it is also 1-transitive, so for any $i$ and $j$ there exists $g \in G_2$ such that
$
    g(i)=j.
$
Since $\mathcal{V}$ is $G_2$-invariant,
\[
    \mathcal{V}^{\dagger} t_a^{(i)} \mathcal{V}=\mathcal{V}^{\dagger} U(g)^{\dagger} t_a^{(i)} U(g) \mathcal{V}=\mathcal{V}^{\dagger} t_a^{(j)} \mathcal{V} .
\]
Therefore
$
\alpha^{(i)}=\alpha^{(j)}.
$
So there exists a constant $\alpha^{\text{I}}$ such that
\[
\mathcal{V}^{\dagger} t_a^{(i)} \mathcal{V}=\alpha^{\text{I}} \bar{t}_a .
\]
Analogously, because $G_2$ is 2-transitive, for any ordered distinct pairs $(i, j)$ and $(k, l)$, there exists $g \in G_2$ with
$
g(i)=k, \quad g(j)=l .
$
Since $\mathcal{V}$ is $G_2$-invariant,
\[
\mathcal{V}^{\dagger} t_a^{(i)} t_b^{(j)} \mathcal{V}=\mathcal{V}^{\dagger} U(g)^{\dagger} t_a^{(i)} t_b^{(j)} U(g) \mathcal{V}=\mathcal{V}^{\dagger} t_a^{(k)} t_b^{(l)} \mathcal{V} .
\]
Therefore
$$
\alpha^{(i j)}=\alpha^{(k l)}=\alpha^{\text{II}}, \quad \beta^{(i j)}_{f}=\beta^{(k l)}_{f}=\beta_{f}, \quad \beta^{(i j)}_{d}=\beta^{(k l)}_{d}=\beta_{d}
$$

Now we are to find the absolute values of $\alpha^{\text{I}}$ and $\alpha^{\text{II}}$. 
From the definition of $\bar{t}_a$ we have
$
\bar{t}_a=\sum_{i=1}^n \mathcal{V}^{\dagger} t_a^{(i)} \mathcal{V}=\sum_{i=1}^n \alpha^{\text{I}} \bar{t}_a=n \alpha^{\text{I}} \bar{t}_a.
$
Since $\bar{t}_a \neq 0$, we have $\alpha^{\text{I}}=\frac{1}{n}$.
Moving to two-body identities
\begin{equation}\label{eq:two_body_identity}
    \mathcal{V}^{\dagger} t_b^{(j)} t_a^{(i)} \mathcal{V}=\alpha^{\text{II}} \delta_{b a} I_L+\beta_{f} f_{b ac} \bar{t}_c +\beta_{d} d_{b ac} \bar{t}_c .
\end{equation}
using tensor index permutation identities
$
    \delta_{b a}=\delta_{a b}, f_{b ac}=-f_{a bc},  d_{b ac}=d_{a bc},
$
we get
\begin{equation}\label{eq:two_body_identity_2}
    \mathcal{V}^{\dagger} t_b^{(j)} t_a^{(i)} \mathcal{V}=\alpha^{\text{II}} \delta_{a b} I_L-\beta_{f} f_{a bc} \bar{t}_c +\beta_{d} d_{a bc} \bar{t}_c ,
\end{equation}
and  from comparing Eq.~\eqref{eq:two_body_identity} and Eq.~\eqref{eq:two_body_identity_2} we have $\beta_{f}=-\beta_{f}$, so $\beta_{f}=0$. Let's find $\alpha^{\text{II}}$ first.
We use the following identity proven in \ref{lem:generalisation_of_Casimir_identity}
\[
    \sum_a\left(\sum_{j=1}^n t^{(j)}_a\right)^2=\frac{n\left(d^2-1\right)}{2 d} I+2 \sum_{i<j} \sum_a t^{(i)}_a t^{(j)}_a.
\]
and act on it with $P=\mathcal{V} \mathcal{V}^{\dagger}$ from both sides. 
Noticing that
\begin{align}
    &P\left(\sum_{j=1}^n t^{(j)}_a\right)^2 P=P\left(\sum_{j=1}^n t^{(j)}_a\right)(\mathbb{1}-P+P)\left(\sum_{j=1}^n t^{(j)}_a\right)P=\\&=P\left(\sum_{j=1}^n t^{(j)}_a\right) 
    (\mathbb{1}-P)\left(\sum_{j=1}^n t^{(j)}_a\right)P+P\left(\sum_{j=1}^n t^{(j)}_a\right) P\left(\sum_{j=1}^n t^{(j)}_a\right)P=\\&=P\left(\sum_{j=1}^n t^{(j)}_a\right) 
    P\left(\sum_{j=1}^n t^{(j)}_a\right)P=\left(P\left(\sum_{j=1}^n t^{(j)}_a\right) P\right)\left(P\left(\sum_{j=1}^n t^{(j)}_a\right)P\right)=\\&=\mathcal{V}\bar{t}_a^2 \mathcal{V}^{\dagger},
\end{align}
left side gives
\[
    P \sum_a\left(\sum_{j=1}^n t^{(j)}_a\right)^2 P=\sum_a \mathcal{V}\vec{t}_a^2 \mathcal{V}^{\dagger} =\frac{d^2-1}{2 d} P,
\]
As for the right side, noticing that
\[
    P \sum_a t_a^{(i)} t_a^{(j)} P=\alpha^{\text{II}} \sum_a\delta_{a a} P+\beta_{d} \sum_a d_{c a a} \mathcal{V}\bar{t}_c \mathcal{V}^{\dagger}=\alpha^{\text{II}}(d^2-1) P,
\]
where we used the fact that $\sum_a d_{c a a}=0$ for all $c$, we have
\begin{align}
    &P \sum_{i<j} \sum_a t^{(i)}_a t^{(j)}_a P=\alpha^{\text{II}}(d^2-1) \sum_{i<j} P=\alpha^{\text{II}} \frac{n(n-1)(d^2-1)}{2} P.
\end{align}
Thus, substituting them back into the identity, we have
\[
    \frac{d^2-1}{2 d} P=n \frac{d^2-1}{2 d} P+n(n-1)\left(d^2-1\right) \alpha^{\text{II}} P,
\]
from which we conclude that $\alpha^{\text{II}}=-\frac{1}{2dn}$.

Now let's find $\beta_{d}$. We first notice that
$
    P d_{a b m} T_a T_b P=d_{a b m} \mathcal{V} \bar{t}_a \bar{t}_b \mathcal{V}^{\dagger},
$
and from Eq.~\eqref{eq:sudgenerators} and tensor index permutation identities, we have
\begin{equation}\label{eq:two_body_identity_log}
    d_{a b m} \bar{t}_a \bar{t}_b=\frac{1}{2} d_{a b m} d_{a b c} \bar{t}_c=\frac{d^2-4}{2 d} \bar{t}_m,
\end{equation}
which gives
\begin{equation}\label{eq:d_tensor_identity}
    P d_{a b m} T_a T_b P=\frac{d^2-4}{2 d} \mathcal{V}\bar{t}_m \mathcal{V}^{\dagger}.
\end{equation}
Now, we decompose $T_a T_b$ into single site operators 
\[
    T_a T_b=\sum_i t_a^{(i)} t_b^{(i)}+\sum_{i \neq j} t_a^{(i)} t_b^{(j)}
\]
First, consider the same site contribution. Using the same derivations as in Eq.~\eqref{eq:two_body_identity_log} we have
\[
    d_{a b m} t_a^{(i)} t_b^{(i)}=\frac{d^2-4}{2 d} t_m^{(i)},
\]
from which we get
\[
    P d_{a b m} t_a^{(i)} t_b^{(i)} P=\frac{d^2-4}{2 d} \cdot \frac{1}{n} \mathcal{V}\bar{t}_m \mathcal{V}^{\dagger}.
\]
As for the different site contributions, we have
\[
    P d_{a b m} t_a^{(i)} t_b^{(j)} P=\alpha^{\text{II}} d_{a b m} \delta_{a b} P+\beta_{d} d_{a b m} d_{a bc} \mathcal{V} \bar{t}_c \mathcal{V}^{\dagger}=\beta_{d} \frac{d^2-4}{d} \mathcal{V} \bar{t}_m \mathcal{V}^{\dagger},
\]
where we used $d_{a b m} \delta_{a b}=0$ and $d_{a b m} d_{a b c}=\frac{\left(d^2-4\right)}{d} \delta_{cm}$.

Thus, the left side of (\ref{eq:d_tensor_identity}) becomes
\[
    P d_{a b m} T_a T_b P=\sum_i P d_{a b n n} t_a^{(i)} t_b^{(i)} P+\sum_{i \neq j} P d_{a b m} t_a^{(i)} t_b^{(j)} P=\frac{d^2-4}{2 d} 
    \mathcal{V}\bar{t}_m \mathcal{V}^{\dagger}+n(n-1) \beta_{d} \frac{d^2-4}{d} \mathcal{V} \bar{t}_m \mathcal{V}^{\dagger},
\]
comparing it to the right side we have
\[
\frac{d^2-4}{2 d} \mathcal{V}\bar{t}_m \mathcal{V}^{\dagger}=\frac{d^2-4}{2 d} \mathcal{V}\bar{t}_m \mathcal{V}^{\dagger}+n(n-1) \beta_{d} \frac{d^2-4}{d} \mathcal{V} \bar{t}_m \mathcal{V}^{\dagger},
\]
which clearly gives $\beta_{d}=0$. Thus, we have shown that
\[
    \mathcal{V}^{\dagger} t_a^{(i)} t_b^{(j)} \mathcal{V}=-\frac{1}{2 d n} \delta_{a b} I_L,
\]
which concludes the proof.
\end{proof}

\begin{smproposition}\label{prop:1_2_and_3_body_action}
    Let $\mathcal{V}$ be encoding isometry of the code defined in the Theorem \ref{thrm:3_transitive_covariant_code}, then for sufficiently large $n$
    \begin{align}
    &\mathcal{V}^{\dagger} t^{(i)}_a \mathcal{V}=\frac{1}{n} \bar{t}_a, \quad \mathcal{V}^{\dagger} t_{a_1}^{(i)} t_{a_2}^{(j)} \mathcal{V}=-\frac{1}{2 d n} \delta_{a_1 a_2} I_L, \\ &\mathcal{V}^{\dagger} t_{a_1}^{(i)} t_{a_2}^{(j)} t_{a_3}^{(k)} \mathcal{V}
    =\frac{1}{2 dn(n-2)} d_{a_1 a_2 a_3} I_L-\frac{1}{2 dn(n-2)}\left(\delta_{a_1 a_2} \bar{t}_{a_3}+\delta_{a_1 a_3} \bar{t}_{a_2}+\delta_{a_2 a_3} \bar{t}_{a_1}\right),
    \end{align}
    where $i, j, k$ are distinct.
\end{smproposition}

\begin{proof}
    Everything about one-body and two-body operators is the same as in the previous proposition, 
    because every 3-transitive group is also 2-transitive, so from Proposition~\ref{prop:1_and_2_body_action} we have
    \begin{equation}\label{eq:VtaV_and_Vta_tbV}
        \mathcal{V}^{\dagger} t^{(i)}_a \mathcal{V}=\frac{1}{n} \bar{t}_a, \quad \mathcal{V}^{\dagger} t_{a_1}^{(i)} t_{a_2}^{(j)} \mathcal{V}=-\frac{1}{2 d n} \delta_{a_1 a_2} I_L.
    \end{equation}

    As we show in Lemma~\ref{lem:3-body_action}, for three-body operators, the general form is given by
    \[
        \mathcal{V}^{\dagger} t_a^{(i)} t_b^{(j)} t_c^{(k)} \mathcal{V}=\mu_f^{(i j k)} f_{a b c} I_L+\mu_d^{(i j k)} d_{a b c} I_L+B_{a b c m}^{(i j k)} \bar{t}_m
    \]
    Using $G_3$ invariance of $\mathcal{V}$ and $3$-transitivity of $G_3$ we get
    \[
        \mathcal{V}^{\dagger} t_a^{(i)} t_b^{(j)} t_c^{(k)}  \mathcal{V}=\mathcal{V}^{\dagger} U(g)^{\dagger} t_a^{(k)} t_b^{(l)} t_c^{(m)} U(g) \mathcal{V}
        =\mathcal{V}^{\dagger} t_a^{(k)} t_b^{(l)} t_c^{(m)} \mathcal{V} .
    \]
    from which we conclude that
    \[
        \mu_f(i, j, k)=\mu_f, \quad \mu_d(i, j, k)=\mu_d, \quad B_{a b c m}(i, j, k)=B_{a b c m},
    \]
    Exploiting the fact that operators on different sites commute, we conclude that $\mu_f=0$ and $B_{a b c m}$ must be symmetric in $a, b, c$,
    so reduced form is
    \[
    \mathcal{V}^{\dagger} t_a^{(i)} t_b^{(j)} t_c^{(k)} \mathcal{V}=\mu_d d_{a b c} I_L+B_{a b c m}\bar{t}_m.
    \]
    Now we are to find $\mu_d$ and $B_{a b c m}$. We will express them in terms of $\alpha^{\text{I}}$ and $\alpha^{\text{II}}$, which we already know. 
    Let
    $
        T_c:=\sum_{i=1}^n t_c^{(i)}
    $
    and recall that $\mathcal{V}^{\dagger} T_c \mathcal{V}=\bar{t}_c$. Hence, using $P=\mathcal{V}\mathcal{V}^{\dagger}$, we get
    \begin{equation}\label{eq:Vta_tbT_cVR}
        \mathcal{V}^{\dagger} t_a^{(i)} t_b^{(j)} T_c \mathcal{V}=\mathcal{V}^{\dagger} t_a^{(i)} t_b^{(j)} \mathcal{V} \bar{t}_c =-\frac{1}{2 d n} \delta_{a b} \bar{t}_c,
    \end{equation}
    On the other hand, expanding
    \[
        T_c=t_c^{(i)}+t_c^{(j)}+\sum_{k \neq i, j} t_c^{(k)},
    \]
    we get
    \begin{equation}\label{eq:Vta_tbT_cVL}
        \mathcal{V}^{\dagger} t_a^{(i)} t_b^{(j)} T_c \mathcal{V}=\mathcal{V}^{\dagger} t_a^{(i)} t_b^{(j)} t_c^{(i)} \mathcal{V}+\mathcal{V}^{\dagger} t_a^{(i)} t_b^{(j)} t_c^{(j)} \mathcal{V}+
        \sum_{k \neq i, j} \mathcal{V}^{\dagger} t_a^{(i)} t_b^{(j)} t_c^{(k)} \mathcal{V} .
    \end{equation}
    For the first and second terms, we write
    \[
        \begin{aligned}
        &\mathcal{V}^{\dagger} t_a^{(i)} t_b^{(j)} t_c^{(i)} \mathcal{V}=\frac{1}{2 d n} \delta_{a c} \bar{t}_b-\frac{1}{4 d n}\left(d_{a c b}+i f_{a c b}\right) I_L\\
        &\mathcal{V}^{\dagger} t_a^{(i)} t_b^{(j)} t_c^{(j)} \mathcal{V}=\frac{1}{2 d n} \delta_{b c} \bar{t}_a-\frac{1}{4 d n}\left(d_{b c a}+i f_{b c a}\right) I_L .
        \end{aligned},
    \]
    where we used (\ref{eq:VtaV_and_Vta_tbV}) and (\ref{eq:sudnormalisation}). As for the third term, we have
    \[
        \sum_{k \neq i, j} \mathcal{V}^{\dagger} t_a^{(i)} t_b^{(j)} t_c^{(k)} \mathcal{V}=(n-2) \mu_d d_{a b c} I_L+(n-2) B_{a b c m} \bar{t}_m,
    \]
    and substituting everything back into (\ref{eq:Vta_tbT_cVL}) and comparing to the right side of (\ref{eq:Vta_tbT_cVR}) we have
    \[
        \begin{aligned}
        -\frac{1}{2 d n} \delta_{a b} \bar{t}_c = & \frac{1}{2 d n} \delta_{a c} \bar{t}_b +\frac{1}{2 d n} \delta_{b c} \bar{t}_a \\
        & +\left((n-2) \mu_d-\frac{1}{2 d n}\right) d_{a b c} +(n-2) B_{a b c m} \bar{t}_m
        \end{aligned}
    \]
    comparison of the corresponding vector and scalar parts gives
    \[
        \begin{gathered}
        \mu_d=\frac{1}{2 d n(n-2)}, \\
        B_{a b c m}=-\frac{1}{2 d n(n-2)}\left(\delta_{a b} \delta_{c m}+\delta_{a c} \delta_{b m}+\delta_{b c} \delta_{a m}\right),
        \end{gathered}
    \]
    which concludes the proof.
\end{proof}

\begin{smremark}
    We conjecture that there exists a proper choice of a code space such that the odd-body coefficients will scale as $O_{d,k}\left(\frac{1}{n^{r+1}}\right)$, 
    where $k=2r+1$ and act on a code space as vector part only, and even-body coefficients will scale as $O_{d,k}\left(\frac{1}{n^r}\right)$, 
    where $k=2r$ and act on a code space as scalar part only. It can be advantageous to use induction to express $k$-body coefficients in terms of lower-body coefficients. 
\end{smremark}

It is obvious, that every operator from $\operatorname{End}(V_{\omega_1}^{\otimes k})$ 
can be expressed as a linear combination of fixed-body operators, i. e. operators of the form
$t_a^{(i)}, t_a^{(i)} t_b^{(j)}, ..., t_{a_1}^{(i_1)} t_{a_2}^{(i_2)} ... t_{a_k}^{(i_k)}$. 
In the lemmas below, we show that if an operator lies entirely in the adjoint isotypic component, then it can be expressed as 
a linear combination of some specific fixed-body operators, which we call covariants.

\begin{smlemma}\label{lem:adjoint-3site}
    Let $V_{\omega_1}$ be a fundamental representation of $SU(d)$. Consider
    $\operatorname{End}(V_{\omega_1}^{\otimes 3})$ as a representation of $SU(d)$. If $X\in \operatorname{End}(V_{\omega_1}^{\otimes 3})$ is an operator
    that lies entirely in the adjoint isotypic component, then it can be expressed as
    \[
        \left.X=\sum_{a=1}^{d^2-1}\left[\sum_{p=1}^3 \alpha_p^a A_a^{(p)}+\sum_{1 \leq p<q \leq 3}\left(\beta_{p q}^a D_a^{(p q)}+\gamma_{p q}^a F_a^{(p q)}\right)+\sum_{r=1}^3\left(\rho_r^a U_a^{(r)}+\sigma_r^a V_a^{(r)}+\tau_r^a W_a^{(r)}\right)\right)\right],
    \]
    where $\alpha_p^a, \beta_{p q}^a, \gamma_{p q}^a, \rho_r^a, \sigma_r^a, \tau_r^a$ are some constants, and we defined the following operators
    \begin{itemize}
    \item  One body covariants define 3 copies of the adjoint representation
        \[
            A_a^{(1)}:=t_a^{(1)}, \quad A_a^{(2)}:=t_a^{(2)}, \quad A_a^{(3)}:=t_a^{(3)}
        \]
    \item  Two body covariants. For each pair of sites $1 \leq p<q \leq 3$ we define
        \[
            D_a^{(p q)}:=d_{a b c} t_b^{(p)} t_c^{(q)}, \quad F_a^{(p q)}:=f_{a b c} t_b^{(p)} t_c^{(q)}.
        \]
        For $d\geq 3$ these are linearly independent, but for $d=2$ we have $D_a^{(p q)}=0$ so only $F_a^{(p q)}$ remains.
        Thus, for $d\geq 3$ we have 6 two-body copies, while for $d=2$ we have only 3 two-body copies of the adjoint representation.
    \item  Three body covariants. We define the following invariant two-body operators
        \[
            \Omega_{12}:=\sum_b t_b^{(1)} t_b^{(2)}, \quad \Omega_{13}:=\sum_b t_b^{(1)} t_b^{(3)}, \quad \Omega_{23}:=\sum_b t_b^{(2)} t_b^{(3)},
        \]
        so covariants that form three-body copies of the adjoint representation are defined as follows
        \[
            U_a^{(1)}:=\Omega_{23} t_a^{(1)}, \quad U_a^{(2)}:=\Omega_{13} t_a^{(2)}, \quad U_a^{(3)}:=\Omega_{12} t_a^{(3)},
        \]
        \[
            \begin{aligned}
            & V_a^{(1)}:=d_{a b e} d_{e c d} t_b^{(1)} t_c^{(2)} t_d^{(3)}, \\
            & V_a^{(2)}:=d_{a c e} d_{e d b} t_b^{(1)} t_c^{(2)} t_d^{(3)}, \\
            & V_a^{(3)}:=d_{a d e} d_{e b c} t_b^{(1)} t_c^{(2)} t_d^{(3)},
            \end{aligned}\quad
            \begin{aligned}
            W_a^{(1)} & :=d_{a b e} f_{e c d} t_b^{(1)} t_c^{(2)} t_d^{(3)}, \\
            W_a^{(2)} & :=d_{a c e} f_{e d b} t_b^{(1)} t_c^{(2)} t_d^{(3)}, \\
            W_a^{(3)} & :=d_{a d e} f_{e b c} t_b^{(1)} t_c^{(2)} t_d^{(3)} .
            \end{aligned}
        \]
        For $d\geq 4$, these are linearly independent, so we have 9 three-body copies of the adjoint representation; for $d=2$, all $d_{a b e}$ vanish, so only $U_a^{(r)}$ remain, 
        and we have 3 three-body copies of the adjoint representation;
        for $d=3$, we have $V_a^{(1)}+V_a^{(2)}+V_a^{(3)}=\frac{1}{3}\left(U_a^{(1)}+U_a^{(2)}+U_a^{(3)}\right)$, so we have 8 three-body copies of the adjoint representation.
    \end{itemize}
\end{smlemma}

\begin{proof}
    Let's count the number of copies of the adjoint representation in $\operatorname{End}(V_{\omega_1}^{\otimes 3})$. We will 
    find it using character theory. Indeed, multiplicity $m$ of adjoint representation in $\operatorname{End}(V_{\omega_1}^{\otimes 3})$ is given by the following scalar product of characters
    \[
        m=\left\langle\chi_{A d}, \chi_{V^{\otimes 3} \otimes\left(V^*\right)^{\otimes 3}}\right\rangle
    \]
    and using the fact that $\chi_{Ad}=\chi_V \chi_{V^*}-1$(the character of $V_{\omega_1} \otimes V_{\omega_1}^*$ minus the trivial character), and noting that the character of the tensor product is the product of characters, we have:
    \[
        m=\left\langle\chi_V \chi_{V^*}-1,\left(\chi_V \chi_{V^*}\right)^3\right\rangle=\left\langle \chi_V \chi_{V^*},\left(\chi_V \chi_{V^*}\right)^3\right\rangle-\left\langle 1, \left(\chi_V \chi_{V^*}\right)^3\right\rangle.
    \]
    One can easily notice that
    \[
        \left\langle 1, \left(\chi_V \chi_{V^*}\right)^3\right\rangle=\operatorname{dim} \operatorname{End}_{S U(d)}\left(V^{\otimes 3}\right), \quad \left\langle \chi_V \chi_{V^*},\left(\chi_V \chi_{V^*}\right)^3\right\rangle=\operatorname{dim} \operatorname{End}_{S U(d)}\left(V^{\otimes 4}\right),
    \]
    thus, we obtain a convenient formula 
    \[
        m=\operatorname{dim} \operatorname{End}_{S U(d)}\left(V^{\otimes 4}\right)-\operatorname{dim} \operatorname{End}_{S U(d)}\left(V^{\otimes 3}\right).
    \]
    From Shur-Weyl duality
    \[
        \operatorname{dim} \operatorname{End}_{S U(d)}\left(V^{\otimes n}\right)=\sum_{\lambda \vdash n, \ell(\lambda) \leq d}\left(\operatorname{dim}[\lambda]\right)^2.
    \]
    For $n=3$: partitions are $(3),(2,1),(1,1,1)$ and corresponding $\operatorname{dim}[\lambda]$ -- $1,2,1$.
    Hence
    $$
    \operatorname{dim} \operatorname{End}_{S U(d)}\left(V^{\otimes 3}\right)= \begin{cases}6, & d \geq 3 \\ 5, & d=2\end{cases}
    $$
    For $n=4$: partitions are $(4),(3,1),(2,2),(2,1,1),(1,1,1,1)$ and corresponding $\operatorname{dim}[\lambda]$ -- $1,3,2,3,1$.
    Hence
    $$
    \operatorname{dim} \operatorname{End}_{S U(d)}\left(V^{\otimes 4}\right)= \begin{cases}24, & d \geq 4 \\ 23, & d=3 \\ 14, & d=2\end{cases}
    $$
    Therefore the multiplicity of the adjoint in $\operatorname{End}\left(V^{\otimes 3}\right)$ is given by
    \[
        \begin{cases}18, & d \geq 4 \\ 17, & d=3 \\ 9, & d=2\end{cases}.
    \]
    One can easily check that the operators defined in the statement of the lemma are linearly independent, and there are exactly $18$ of them for $d \geq 4$, $17$ for $d=3$ and $9$ for $d=2$, 
    so they form a basis in the subspace of $\operatorname{End}(V_{\omega_1}^{\otimes 3})$ that corresponds to the adjoint isotypic component, which completes our proof.
\end{proof}

\begin{smlemma}\label{lem:trivial-3site}
    Let $V_{\omega_1}$ be a fundamental representation of $SU(d)$. Consider $\operatorname{End}(V_{\omega_1}^{\otimes 3})$ as a representation of $SU(d)$. 
    If $Y\in \operatorname{End}(V_{\omega_1}^{\otimes 3})$ is an operator that lies entirely in the trivial isotypic component, then it can be expressed as
    \[
        Y=\lambda \mathbf{1}+\mu_{12} \Omega_{12}+\mu_{13} \Omega_{13}+\mu_{23} \Omega_{23}+\nu \Delta+\xi \Phi,
    \]
    where $\lambda, \mu_{12}, \mu_{13}, \mu_{23}, \nu, \xi$ are some constants, and we defined the following operators
    \[
        \begin{gathered}
        \Omega_{i j}:=\sum_{m=1}^{d^2-1} t_m^{(i)} t_m^{(j)} \\
        \Delta:=d_{m n p} t_m^{(1)} t_n^{(2)} t_p^{(3)}, \quad \Phi:=f_{m n p} t_m^{(1)} t_n^{(2)} t_p^{(3)}.
        \end{gathered}
    \]
    (Note that for the case $d=2$ we have $\Delta=0$, so only $5$ linearly independent operators remain).
\end{smlemma}

\begin{proof}
    As in the previous lemma, we can find the number of copies of the trivial representation in $\operatorname{End}(V_{\omega_1}^{\otimes 3})$ using character theory. Indeed, in this case
    \[
        \operatorname{dim} \operatorname{End}_{S U(d)}\left(V^{\otimes 3}\right)=\sum_{\lambda \vdash 3, \ell(\lambda) \leq d}(\operatorname{dim}[\lambda])^2.
    \]
    Partitions of 3 are $(3),(2,1),(1,1,1)$ and corresponding $\operatorname{dim}[\lambda]$ -- $1,2,1$. therefore
    \[
        \operatorname{dim} \operatorname{Hom}_{S U(d)}\left(\mathbf{1}, \operatorname{End}\left(V_{\omega_1}^{\otimes 3}\right)\right)= \begin{cases}6, & d \geq 3, \\ 5, & d=2,\end{cases}
    \]
    since partition $(1,1,1)$ is not allowed for $d=2$. 
    One can easily check that the operators defined in the statement of the lemma are linearly independent (for $d\geq 3$ there are 6 of them, for $d=2$ there are 5 of them), 
    and they lie in the trivial isotypic component, so they form a basis.
    This completes our proof.
\end{proof}

\begin{smlemma}\label{lem:covar_3-site_vector_operator}
    If $X_b$ satisfies conditions of Lemma \ref{lem:adjoint-3site}, and the following scalar product identities hold for 
    \[
        \begin{gathered}
            \operatorname{Tr}\left(X_b t_m^{(i)}\right)=\frac{1}{2 n} \delta_{b m},\\
            \operatorname{Tr}\left(X_b t_a^{(i)} t_c^{(j)}\right)=0,\\
            \operatorname{Tr}\left(X_b t_a^{(i)} t_c^{(j)} t_f^{(k)}\right)=-\frac{1}{4 d n (n-2)}\left(\delta_{a c} \delta_{f b} + \delta_{a f} \delta_{c b} + \delta_{c f} \delta_{a b}\right),
        \end{gathered}
    \]
    where $i, j, k$ are distinct site indexes and $a, b, c, f$ are distinct generator indexes,
    then it can be expressed as
    \[
        X_b=\frac{1}{n d^2}\left(t_b^{(1)}+t_b^{(2)}+t_b^{(3)}\right)-\frac{2}{d n(n-2)} \sum_{m=1}^{d^2-1}\left(t_b^{(1)} t_m^{(2)} t_m^{(3)}+t_m^{(1)} t_b^{(2)} t_m^{(3)}+t_m^{(1)} t_m^{(2)} t_b^{(3)}\right).
    \]
\end{smlemma}

\begin{proof}
Let's start with one-body coefficients. On the one hand,
\[
    \operatorname{Tr}\left(X_b t_a^{(i)}\right)=\frac{d^2}{2} \alpha_i \delta_{a b}
\]
from which the first scalar product identity gives us $\alpha_i=\frac{1}{n d^2}$ for all $i=1,2,3$. 

As for two-body coefficients, we have
\[
    \operatorname{Tr}\left(D_b^{(i j)} t_a^{(i)} t_c^{(j)}\right)=\frac{d}{4} d_{b a c}, \quad \operatorname{Tr}\left(F_b^{(i j)} t_a^{(i)} t_c^{(j)}\right)=\frac{d}{4} f_{b a c}
\]
for distinct $i, j$ from which we get
\[
    \operatorname{Tr}\left(X_b t_a^{(i)} t_c^{(j)}\right)=\frac{d}{4} \beta_{i j} d_{b a c}+\frac{d}{4} \gamma_{i j} f_{b a c}
\]
and since for all $a, b, c$ our two-body scalar product is $0$ we get $\beta_{12}=\beta_{13}=\beta_{23}=0$, and $\gamma_{12}=\gamma_{13}=\gamma_{23}=0$. 

Finally, for three-body coefficients we have
\[
    \begin{aligned}
    \operatorname{Tr}\left(U_b^{(1)} t_a^{(1)} t_c^{(2)} t_f^{(3)}\right) & =\frac{1}{8} \delta_{a b} \delta_{c f}, \\
    \operatorname{Tr}\left(U_b^{(2)} t_a^{(1)} t_c^{(2)} t_f^{(3)}\right) & =\frac{1}{8} \delta_{a f} \delta_{c b}, \\
    \operatorname{Tr}\left(U_b^{(3)} t_a^{(1)} t_c^{(2)} t_f^{(3)}\right) & =\frac{1}{8} \delta_{a c} \delta_{f b} .
    \end{aligned}
\]
so combining these three equations together we get
\[
    \operatorname{Tr}\left(\sum_{r=1}^3 \rho_r U_b^{(r)} t_a^{(1)} t_c^{(2)} t_f^{(3)}\right)=\frac{1}{8}\left(\rho_1 \delta_{a b} \delta_{c f}+\rho_2 \delta_{a f} \delta_{c b}+\rho_3 \delta_{a c} \delta_{f b}\right)
\]
which in comparison with third scalar product identity gives us $\rho_1=\rho_2=\rho_3=-\frac{2}{d n(n-2)}$. 
As for $V_b^{(r)}$ terms
\[
    \begin{aligned}
    \operatorname{Tr}\left(V_b^{(1)} t_a^{(1)} t_c^{(2)} t_f^{(3)}\right) & =\frac{1}{8} d_{b a e} d_{e c f}, \\
    \operatorname{Tr}\left(V_b^{(2)} t_a^{(1)} t_c^{(2)} t_f^{(3)}\right) & =\frac{1}{8} d_{b c e} d_{e f a}, \\
    \operatorname{Tr}\left(V_b^{(3)} t_a^{(1)} t_c^{(2)} t_f^{(3)}\right) & =\frac{1}{8} d_{b f e} d_{e a c},
    \end{aligned}
\]
and the same is for $W_b^{(r)}$ terms, where we have to substitute $d_{b a e} d_{e c f}$ with $d_{b a e} f_{e c f}$, etc.
Comparing these equations with the third scalar product identity, we get $\sigma_1=\sigma_2=\sigma_3=0$ and $\quad \tau_1=\tau_2=\tau_3=0$.
From which we get the final formula for $X_b$ as stated in the lemma. 
\end{proof}

\begin{smlemma}\label{lem:covar_3-site_scalar_operator}
    If $Y$ satisfies conditions of Lemma \ref{lem:trivial-3site}, and the following scalar product identities hold for 
    \[
        \begin{gathered}
            \operatorname{Tr}\left(Y t_m^{(i)}\right)=0,\\
            \operatorname{Tr}\left(Y t_a^{(i)} t_c^{(j)}\right)=-\frac{1}{2 d n} \delta_{a c},\\
            \operatorname{Tr}\left(Y t_a^{(i)} t_c^{(j)} t_f^{(k)}\right)=\frac{1}{2 d n(n-2)} d_{a c f},
        \end{gathered}
    \]
    where $i, j, k$ are distinct site indexes and $a, c, f$ are distinct generator indexes, and in addition $\operatorname{Tr}(Y)=1$, then it can be expressed as
    \[
        Y=\frac{1}{d^3} \mathbf{1}-\frac{2}{d^2 n} \sum_{m=1}^{d^2-1}\left(t_m^{(1)} t_m^{(2)}+t_m^{(1)} t_m^{(3)}+t_m^{(2)} t_m^{(3)}\right)+\frac{4}{d n(n-2)} \sum_{m, n, p=1}^{d^2-1} d_{m n p} t_m^{(1)} t_n^{(2)} t_p^{(3)} .
    \]
\end{smlemma}

\begin{proof}
    First of all, notice that since the problem is symmetric under permutation of the three sites, we must have $\mu_{12}=\mu_{13}=\mu_{23}=: \mu$ 
    and $\zeta=0$ because $\Phi$ is antisymmetric under permutation of the sites. So the reduced formula for $Y$ is given by
    \[
        Y=\lambda \mathbf{1}+\mu\left(\Omega_{12}+\Omega_{13}+\Omega_{23}\right)+\nu \Delta.
    \]
    From the trace identity, we get
    \[
        \operatorname{Tr}(Y)=\lambda \operatorname{Tr}(\mathbf{1})=\lambda d^3
    \]
    which gives us $\lambda=\frac{1}{d^3}$.
    As for two-body coefficients, we have
    \[
        \operatorname{Tr}\left(Y t_a^{(1)} t_c^{(2)}\right)=\mu \frac{d}{4} \delta_{a c},
    \]
    and matching second scalar product identity gives us $\mu=-\frac{2}{d^2 n}$.
    Finally, for three-body coefficient we have
    \[
        \operatorname{Tr}\left(Y t_a^{(1)} t_c^{(2)} t_f^{(3)}\right)=\nu \frac{1}{8} d_{a c f},
    \]
    and matching third scalar product identity gives us $\nu=\frac{4}{d n(n-2)}$.
    This completes our proof.
\end{proof}

\begin{smproposition}\label{prop:3site-fundamental}
Let
    $
            \mathcal H_L=V_{\omega_1},\;
                    \mathcal H_P=(V_{\omega_1})^{\otimes n}
                        $
                            with $d=p^r$ where $p$ is prime and $r$ is a positive integer, 
                                and $n$ is such that $n-1=p^k$, where $k \geq r$ is a positive integer. 
                                    Equip \(\mathcal H_L\) with the fundamental representation of
                                        \(\mathfrak{su}(d)\), and equip \(\mathcal H_P\) with the transversal
                                            representation given by
                                                \begin{equation}
                                                    \begin{gathered}
                                                        t_a \rightarrow \sum_{i=1}^n t_a^{(i)}, \quad a=1,\ldots,d^2-1,
                                                            \end{gathered}
                                                                \end{equation} 
                                                                    where $t_a^{(i)}$ is a fundamental of $t_a$ on the $i$-th physical space $V_{\omega_1}$.
    For sufficiently
    large $n$ there exists an $\mathfrak{su}(d)$-covariant encoding $\mathcal{E}$ with respect to these representations, such that the three-site reduced state has the following form:
\begin{equation}\label{eq:main-3site-difference}
\begin{aligned}
    &\rho^{(ijk)}(\rho_L)=
    \\&=\left(\frac{1}{d^3} \mathbf{1}-\frac{2}{d^2 n} \sum_{m=1}^{d^2-1}\left(t_m^{(i)} t_m^{(j)}+t_m^{(i)} t_m^{(k)}+t_m^{(j)} t_m^{(k)}\right)+\frac{4}{d n(n-2)} \sum_{m, n, p=1}^{d^2-1} d_{m n p} t_m^{(i)} t_n^{(j)} t_p^{(k)}\right) +
    \\&+\sum_b r_b\left(\frac{1}{n d^2}\left(t_b^{(i)}+t_b^{(j)}+t_b^{(k)}\right)-\frac{2}{d n(n-2)} \sum_{m=1}^{d^2-1}\left(t_b^{(i)} t_m^{(j)} t_m^{(k)}+t_m^{(i)} t_b^{(j)} t_m^{(k)}+t_m^{(i)} t_m^{(j)} t_b^{(k)}\right)\right).
\end{aligned}
\end{equation}
so the leading logical-state-dependent contribution is purely one-body and scales as \(n^{-1}\). Namely, this code space is one of the copies of the $V_{\omega_1}$ which is invariant under the action of $\mathrm{PGL}(2, n-1)$.
\end{smproposition}

\begin{smremark}\label{rem:output_3-site_state}
    Notice that the first term in the formula (\ref{eq:main-3site-difference}) is a fixed state that doesn't depend on logical state.
\end{smremark}

\begin{proof}
By linearity of the reduced channel,
$
    \rho^{(ijk)}(\rho_L)
    =
    \rho^{(ijk)}(I_d/d)+\sum_a r_a\, \Phi^{(ijk)}(\bar t_a),
$
where $\Phi^{(ijk)}$ is an extension by linearity of $\rho^{(ijk)}$ to all operators on logical space.
We set
$$
    \tau^{(ijk)}:=\rho^{(ijk)}(I_d/d),\quad \Delta^{(ijk)}_a:=\Phi^{(ijk)}(\bar t_a).
$$
Then
$
    \rho^{(ijk)}(\rho_L)=\tau^{(ijk)}+\sum_a r_a\Delta^{(ijk)}_a.
$
Notice that $\rho^{(ijk)}(\rho_L)$ (and thus $\Phi^{(ijk)}$) is a covariant channel, so $\Delta^{(ijk)}_a$ lies in the adjoint representation and $\tau^{(ijk)}$ 
lies in trivial isotypic components of representation $\operatorname{End}(V_{\omega_1}^{\otimes 3})$.
Using identities from Proposition~\ref{prop:1_2_and_3_body_action}, one can easily derive the following scalar product identities for $\tau^{(ijk)}$ and $\Delta^{(ijk)}_a$. 
\[
\left\{
\begin{aligned}
&\operatorname{Tr}\left(\Delta^{(ijk)}_a t_m^{(i)}\right)
=\frac{1}{2 n} \delta_{a m}, \\
&\operatorname{Tr}\left(\Delta^{(ijk)}_a t_b^{(i)} t_c^{(i)}\right)
=0, \\
&\operatorname{Tr}\left(\Delta^{(ijk)}_a t_b^{(i)} t_c^{(j)} t_f^{(k)}\right)
=-\frac{1}{4 d n(n-2)}
\left(\delta_{a c} \delta_{f b}+\delta_{a f} \delta_{c b}+\delta_{c f} \delta_{a b}\right),
\end{aligned}
\right.
\quad
\left\{
\begin{aligned}
&\operatorname{Tr}\left(\tau^{(ijk)} t_m^{(i)}\right)
=0, \\
&\operatorname{Tr}\left(\tau^{(ijk)} t_a^{(i)} t_c^{(j)}\right)
=-\frac{1}{2 d n} \delta_{a c}, \\
&\operatorname{Tr}\left(\tau^{(ijk)} t_a^{(i)} t_c^{(j)} t_f^{(k)}\right)
=\frac{1}{2 d n(n-2)} d_{a c f},
\end{aligned}
\right.
\]
Thus, applying Lemmas~\ref{lem:covar_3-site_vector_operator} and~\ref{lem:covar_3-site_scalar_operator} to $\Delta^{(ijk)}_a$ and $\tau^{(ijk)}$ correspondingly, 
we get desired form of reduced state. This concludes the proof.
\end{proof}

Our final step in this section is the calculation of the worst-case fidelity of proposed $SU(d)$-covariant codes 
for erasure noise on three sites. We will show that for any fixed $d$ and sufficiently large $n$, worst-case fidelity 
scales as $1-O_{d}\left(\frac{1}{n}\right)$.

\begin{smproposition}\label{prop:3_transitive_fidelity}
Let $\mathcal{V}$ be the encoding isometry of the code defined in Theorem~\ref{thrm:3_transitive_covariant_code}. 
For noise channel $\mathcal{N}(\sigma)=\sum_{1 \leq i<j<k \leq n} p_{i j k} \left|i, j, k\right\rangle\left\langle i, j, k\right|_F \otimes \left| e\right\rangle\left\langle e\right|_{i, j, k} \otimes \operatorname{Tr}_{i j k}(\sigma)$ and sufficiently large $n$
\[
    F\left(\widehat{\mathcal{N} \circ \mathcal{E}}, \Lambda_0\right)=1-\frac{3\left(d^2-1\right)}{8n^2}+O_{d}\left(n^{-3}\right),
\]
where $\Lambda_0$ is defined in (\eqref{eq:constant_channel_3site}).
\end{smproposition}

\begin{proof}
    We will follow the same procedure for calculating fidelity as we did in the case of the single-site erasure noise. 
    Turns out that the main difference lies in the choice of constant channel $\Lambda$, because fixed state, as can be seen from formula (\ref{eq:main-3site-difference}), 
    is no longer the maximally mixed state, but rather some state that approaches the maximally mixed state as $n$ grows. 
    
    For the considered noise channel, since all $3$-qudit reduced states are identical by the choice of the code space, the complementary channel is given by
    \[
        \widehat{\mathcal{N} \circ \mathcal{E}}\left(\rho_L\right)= \sum_{1 \leq i<j<k \leq n} p_{i j k}\,|i j k\rangle\left\langle i j k\right|_{\text{flag}} \otimes \rho^{(ijk)}\left(\rho_L\right)\cong \omega_{\text{flag }} \otimes \rho^{(i_0j_0k_0)}\left(\rho_L\right).
    \]
    where $\omega_{\text {flag }}=\sum_{i<j<k} p_{i j k}|i j k\rangle\langle i j k|_{F_E}$ and $i_0, j_0, k_0$ are fixed indices. We will use notation $(i j k)$ for these fixed indices instead of $(i_0 j_0 k_0)$ in the following calculations.
    We define the constant channel as
    \begin{equation}\label{eq:constant_channel_3site}
        \begin{gathered}
            \Lambda_0(\rho_L)=\operatorname{Tr}(\rho_L)\cdot\omega_{\text{flag }} \otimes \left(\frac{1}{d^3} \mathbf{1}-\frac{2}{d^2 n} \sum_{m=1}^{d^2-1}\left(t_m^{(i)} t_m^{(j)}+t_m^{(i)} t_m^{(k)}+t_m^{(j)} t_m^{(k)}\right)+\frac{4}{d n(n-2)} \sum_{m, n, p=1}^{d^2-1} d_{m n p} t_m^{(i)} t_n^{(j)} t_p^{(k)}\right)\\\equiv \operatorname{Tr}(\rho_L)\cdot \omega_{\text{flag }} \otimes \tau.
        \end{gathered}
    \end{equation}
    (notice that $\tau$ is a part of the reduced state, which is independent of the input state, see Eq.~\eqref{eq:main-3site-difference} and Remark~\ref{rem:output_3-site_state}). We drop the flag state $\omega_{\text{flag }}$ in the following calculations, since it doesn't affect the fidelity.
    As we know from \ref{fidelity_opt} the optimum of $F_{\rho}(\widehat{\mathcal{N} \circ \mathcal{E}}, \Lambda_0)$ is attained on chaotic state $\frac{I_d}{d}$, its purification is maximally entangled state $\ket{\Psi}=\frac{1}{\sqrt{d}}\sum |r\rangle \otimes|r\rangle $. We have to compute $(\widehat{\mathcal{N} \circ \mathcal{E}} \otimes \mathrm{id})(\ket{\Psi}\bra{\Psi})$ and $(\Lambda_0 \otimes \mathrm{id})(\ket{\Psi}\bra{\Psi})$.
    We define the traceless linear map, which is an extension by linearity of the reduced state map, i. e.
    \[
        \Phi:  \operatorname{End}\left(V_{\omega_1}\right) \rightarrow \operatorname{End}\left(\left(V_{\omega_1}\right)^{\otimes 3}\right), \quad \Phi\left(\bar{t}_a\right):=K_a, \quad\Phi\left(\frac{I}{d}\right):=\tau
    \]
    where
    \[
        K_a:=\frac{1}{n d^2}\left(t_a^{(i)}+t_a^{(j)}+t_a^{(k)}\right)-\frac{2}{d n(n-2)} \sum_{m=1}^{d^2-1}\left(t_a^{(i)} t_m^{(j)} t_m^{(k)}+t_m^{(i)} t_a^{(j)} t_m^{(k)}+t_m^{(i)} t_m^{(j)} t_a^{(k)}\right).
    \]
    In this notation 
    \[
        \sigma_{1/d}:=\left(\widehat{\mathcal{N} \circ \mathcal{E}} \otimes \operatorname{id}\right)\left(\ket{\Psi}\bra{\Psi}\right)=\frac{1}{d}\sum_{r,s=1}^d \left(\delta_{r s} \tau+\Phi\left(F_{r s}\right)\right) \otimes|r\rangle\langle s| .
    \]
    (see Eq.~\eqref{eq:matr_units} for the definition of $F_{r s}$). As for the constant channel
    \[
        \eta_{1/d}:=\left(\Lambda_0 \otimes \mathrm{id}\right)\left(\ket{\Psi}\bra{\Psi}\right)=\tau \otimes \frac{I}{d}.
    \]
    Therefore, fidelity is given by
    \[
        F(\widehat{\mathcal{N} \circ \mathcal{E}}, \Lambda_0)=f(\sigma_{1/d}, \eta_{1/d})=\operatorname{Tr} \frac{1}{d}\sqrt{\left[\sqrt{\tau}\left(\delta_{r s} \tau+\Phi\left(F_{r s}\right)\right) \sqrt{\tau} \right]_{r, s=1}^d}\equiv \operatorname{Tr} \sqrt{X},
    \] 
    where we define
    \[
        X=\frac{1}{d^2} \sum_{r, s=1}^d \sqrt{\tau}\left(\delta_{r s} \tau+\Phi\left(F_{r s}\right)\right) \sqrt{\tau} \otimes E_{r s}.
    \]
    (see Eq.~\eqref{eq:matr_units} for the definition of $E_{r s}$). We split it as
    \[
        X=\eta^2+Y, \quad \eta^2=\frac{1}{d^2} \tau^2 \otimes I, \quad Y:=\frac{1}{d^2} \sum_{r, s=1}^d \sqrt{\tau} \Phi\left(F_{r s}\right) \sqrt{\tau} \otimes E_{r s}
    \]
    Now, let's transform these operators into a more convenient form. First, define
    \[
        \Delta=\frac{1}{d} \sum_{r, s=1}^d \Phi\left(F_{r s}\right) \otimes E_{r s}=\frac{1}{n d^3} \sum_{r, s=1}^d\left(F_{r s}^{(i)}+F_{r s}^{(j)}+F_{r s}^{(k)}\right) \otimes E_{r s}+O_{d}\left(n^{-2}\right),
    \]
    and notice that $\operatorname{Tr} \Delta=0$. Since
    $
        \sum_{r, s} E_{r s} \otimes E_{r s}=d\ket{\Phi}\bra{\Phi}, F_{r s}=E_{r s}-\delta_{r s} \frac{I}{d},
    $
    we rewrite this operator in a simplified form
    \[
        \Delta=\frac{1}{n d^3}\left(d\ket{\Phi}\bra{\Phi}_{i R}+d\ket{\Phi}\bra{\Phi}_{j R}+d\ket{\Phi}\bra{\Phi}_{k R}-\frac{3}{d} I\right)+O_{d}\left(n^{-2}\right),
    \]
    where $d\ket{\Phi}\bra{\Phi}_{a R}$ acts at site $a$ with reference system $R$. 
    Since
    \[
        \eta=\tau \otimes \frac{I}{d}=\frac{I_{\mathcal{H}}}{d^4}+O_{d}\left(n^{-1}\right):=\eta_0+O_{d}\left(n^{-1}\right),
    \]
    we can rewrite
    \[
        Y=(\eta_0^{1 / 2}+O_{d}\left(n^{-1}\right))\Delta(\eta_0^{1 / 2}+O_{d}\left(n^{-1}\right))=\frac{1}{d^4} \Delta+O_{d}\left(n^{-2}\right),
    \]
    where we used the fact that $\Delta=\Theta(\frac{1}{n})$.
    Let's denote
    $
        G(X):=\operatorname{Tr} \sqrt{X},
    $
    for which Taylor expansion formula looks as follows
    \[
        G\left(X+H\right)=G\left(X\right)+\left.\frac{d}{d t}\right|_{t=0} G(X+t H)+\frac{1}{2} \left.\frac{d^2}{d t^2}\right|_{t=0} G\left(X+t H\right)+O_{d}\left(\left\|H\right\|^3\right),
    \]
    Now, we are to compute each summand separately substituting $X=\eta^2$ and $H=Y$ into this formula
    \[
        G\left(\eta^2+H\right)=G\left(\eta^2\right)+\left.\frac{d}{d t}\right|_{t=0} G(\eta^2+t Y)+\frac{1}{2} \left.\frac{d^2}{d t^2}\right|_{t=0} G\left(\eta^2+t Y\right)+O_{d}\left(\left\|Y\right\|^3\right)
    \]
    First term is
    \[
        G\left(\eta^2\right) = \operatorname{Tr} \eta= \operatorname{Tr}\left(\tau\right) \operatorname{Tr}\left(\frac{I_R}{d}\right)=1 \cdot 1=1.
    \]
    As for the second term, we have
    \[
        \left.\frac{d}{d t}\right|_{t=0} G(\eta^2+t Y)=\frac{1}{2} \operatorname{Tr}\left(\eta^{-1} \eta^{1 / 2} \Delta \eta^{1 / 2}\right)=\frac{1}{2} \operatorname{Tr}\left(\Delta\right)=0
    \]
    Finally, for the third term, since $\eta^2=\eta_0^2+O_{d}\left(n^{-1}\right)$ and $Y=O_{d}\left(n^{-1}\right)$, we can write
    \[
        \left.\frac{d^2}{d t^2}\right|_{t=0} G\left(\eta^2+t Y\right)=\left.\frac{d^2}{d t^2}\right|_{t=0} G\left(\eta_0^2+t Y\right)+O_{d}\left(n^{-3}\right)
    \]
    and since $\eta_0^2$ is proportional to identity, we can use the following formula for the second derivative of $G$ 
    \[
        \left.\frac{d^2}{d t^2}\right|_{t=0} G\left(\eta_0^2+t Y\right)=-\frac{1}{4}  \operatorname{Tr}\left(\eta_0^{-3} Y^2\right)=-\frac{d^4}{4} \operatorname{Tr}\left(\Delta^2\right)+O_{d}\left(n^{-3}\right).
    \]
    Computing $\operatorname{Tr}(\Delta^2)$, we get
    \[
        \operatorname{Tr}\left(\Delta^2\right)=\frac{1}{n^2 d^6} \operatorname{Tr}\left(d\ket{\Phi}\bra{\Phi}_{i R}+d\ket{\Phi}\bra{\Phi}_{j R}+d\ket{\Phi}\bra{\Phi}_{k R}-\frac{3}{d} I\right)^2+O_{d}\left(n^{-3}\right)=\frac{3\left(d^2-1\right)}{d^4n^2}+O_{d}\left(n^{-3}\right).
    \]
    Putting everything together, we get
    \[
                F(\widehat{\mathcal{N} \circ \mathcal{E}}, \Lambda_0)=1-\frac{3\left(d^2-1\right)}{8n^2}+O_{d}\left(n^{-3}\right),
    \]
    which completes our proof.
\end{proof}

\subsection{Flagged erasure on an arbitrary fixed number of qudits}\label{subsec:supple_multiqudit_erasure}

Let $\{e_1,\ldots,e_d\}$ be the standard basis of $\C^d$. Define the normalized totally antisymmetric $d$-qudit state
\begin{equation}
\label{eq:singlet}
  \ket{\zeta_d}
  :=\frac1{\sqrt{d!}}
  \sum_{\pi\in S_d}\operatorname{sgn}(\pi)
  e_{\pi(1)}\otimes\cdots\otimes e_{\pi(d)}.
\end{equation}

\begin{smlemma}
\label{lem:singlet-red-states}
For every $g\in SU(d)$,
$
  g^{\otimes d}\ket{\zeta_d}=\ket{\zeta_d}.
$
For $0\le s\le d$, the reduced density operator of any $s$ tensor factors of $\ket{\zeta_d}\bra{\zeta_d}$ is
\begin{equation}
\label{eq:tau-s}
  \tau_s
  :=\frac{\Pi_{\wedge^sV_{\omega_1}}}{\binom ds},
\end{equation}
where $\Pi_{\wedge^sV_{\omega_1}}$ is the orthogonal projector onto $\wedge^sV_{\omega_1}$; by convention $\tau_0=1$. In particular,
$
  \tau_1=\tau:=\frac{I_d}{d}.
$
\end{smlemma}
\begin{proof}
The vector in~\eqref{eq:singlet} spans $\wedge^dV_{\omega_1}$.  
The action of $g^{\otimes d}$ on this one-dimensional space is multiplication by $\det g=1$, proving invariance. 
An $s$-site reduced state is $SU(d)$-invariant and is supported on $\wedge^sV_{\omega_1}$, because the global vector is antisymmetric under interchange of any two retained sites.  
The representation $\wedge^sV_{\omega_1}$ is irreducible for $0\le s\le d$.  
Schur's lemma therefore makes the reduced state a scalar multiple of $\Pi_{\wedge^sV_{\omega_1}}$, and unit trace fixes the scalar to $1/\binom ds$.
\end{proof}

We partition the $n=qd+1$ physical positions into a distinguished position $0$ and $q$ ordered blocks $B_1,\ldots,B_q$, each of size $d$.  
Define the seed isometry
\begin{equation}
\label{eq:seed-isometry}
  \mathcal{V}_0\ket\psi
  :=\ket\psi_0\otimes\ket{\zeta_d}_{B_1}\otimes\cdots\otimes\ket{\zeta_d}_{B_q}.
\end{equation}

\begin{smproposition}
\label{prop:seed-covariance}
The map $\mathcal{V}_0:V_{\omega_1}\to V_{\omega_1}^{\otimes n}$ is an isometry and $SU(d)$-covariant, i. e.
\[
  \mathcal{V}_0U(g)=U(g)^{\otimes n}\mathcal{V}_0
  \qquad(g\in SU(d)).
\]
\end{smproposition}

\begin{proof}
Isometry is immediate from normalization of the singlets.  
Covariance follows by applying Lemma~\ref{lem:singlet-red-states} to each of the $q$ background blocks.
\end{proof}
We fix an isometric Schur--Weyl inclusion
\begin{equation}
\label{eq:W-inclusion}
  W_{\omega_1}:V_{\omega_1}\otimes M_{\omega_1}\longrightarrow V_{\omega_1}^{\otimes n},
\end{equation}
which is obviously an $SU(d)$- and $S_n$-intertwiner.
Such an inclusion exists and is unique up to a phase. Schur--Weyl duality decomposes $V_{\omega_1}^{\otimes n}$ into the summands $\left(S_\lambda(V)\downarrow_{SU(d)}\right)\otimes[\lambda]$, 
and the Specht modules $[\lambda]$ are pairwise non-isomorphic, so these summands are pairwise non-isomorphic as $SU(d)\times S_n$-modules. 
By Lemma~\ref{lem:M_module_irreducibility} exactly one of them has first factor $V_{\omega_1}$, namely $\lambda_q=(q+1,q^{d-1})$ with $M_{\omega_1}\cong[\lambda_q]$. 
Hence $V_{\omega_1}\otimes M_{\omega_1}$ is irreducible as an $SU(d)\times S_n$-module, being an external tensor product of irreducibles, and occurs in $V_{\omega_1}^{\otimes n}$ with multiplicity one. 
Schur's lemma then makes any two isometric intertwiners of the form~\eqref{eq:W-inclusion} differ by a unimodular scalar. 
Consequently $\mathcal{V}_m$ below is determined by $\ket m$ up to a global phase, which leaves $\mathcal E_m$ unchanged, and the amplitudes of $\ket m$ in a fixed Young--Yamanouchi basis of $M_{\omega_1}$ are well defined.
Since $\mathcal{V}_0$ is an $SU(d)$-intertwiner into the fundamental isotypic component, Schur's lemma gives a unit vector $\ket{m_0}\in M_{\omega_1}$ such that, after fixing the identification of the irrep factor,
\begin{equation}
\label{eq:seed-multiplicity-vector}
  \mathcal{V}_0=W_{\omega_1}(I_{V_{\omega_1}}\otimes\ket{m_0}).
\end{equation}

For every unit vector $\ket m\in M_{\omega_1}$, define
\begin{equation}
\label{eq:Vm}
  \mathcal{V}_m:=W_{\omega_1}(I_{V_{\omega_1}}\otimes\ket m):V_{\omega_1}\longrightarrow V_{\omega_1}^{\otimes n},
  \qquad
  \mathcal{E}_m(X):=\mathcal{V}_mX\mathcal{V}_m^\dagger.
\end{equation}
Every $\mathcal{V}_m$ is an isometry and is exactly $SU(d)$-covariant.

\begin{smproposition}
\label{prop:orbit-twirl}
For every $X\in\End(V_{\omega_1})$,
\begin{equation}
\label{eq:orbit-twirl}
  \mathbb{E}_{m\sim\mathrm{Haar}(M_{\omega_1})}\bigl[\mathcal{V}_mX\mathcal{V}_m^\dagger\bigr]
  =\frac1{n!}\sum_{\pi\in S_n}
  P_\pi\,\mathcal{V}_0X\mathcal{V}_0^\dagger P_\pi^\dagger,
\end{equation}
where $P_\pi$ is the permutation operator on $V_{\omega_1}^{\otimes n}$ corresponding to $\pi\in S_n$; see~\eqref{eq:group_act_on_phys_space}.
\end{smproposition}

\begin{proof}
The Haar first moment is
\[
  \mathbb{E}_m\ket m\bra m=\frac{I_{M_{\omega_1}}}{\dim M_{\omega_1}},
\]
hence the left-hand side is $W_{\omega_1}(X\otimes I_{M_{\omega_1}}/\dim M_{\omega_1})W_{\omega_1}^\dagger$.  
On the other hand, by~\eqref{eq:group_act_on_code} and~\eqref{eq:seed-multiplicity-vector}, the right-hand side is
\[
  W_{\omega_1}\!\left(
    X\otimes\frac1{n!}\sum_{\pi\in S_n}
    R_{\omega_1}(\pi)\ket{m_0}\bra{m_0}R_{\omega_1}(\pi)^\dagger
  \right)W_{\omega_1}^\dagger.
\]
The operator in parentheses commutes with the irreducible representation $R_{\omega_1}$.  
Schur's lemma and unit trace make it $I_{M_{\omega_1}}/\dim M_{\omega_1}$.
\end{proof}

Set
\[
  N:=qd=n-1.
\]
Let $\beta_{\ell,N}$ denote the state of $\ell$ labeled positions obtained by choosing an ordered $\ell$-tuple of distinct sites uniformly from the $N$ sites 
of the background state $\ket{\zeta_d}\bra{\zeta_d}^{\otimes q}$ and retaining the chosen sites in their sampled order.  
Equivalently, it is the $\ell$-site reduced state of the permutation twirl of the background. We set $\beta_{0,N}=1$.

For $x\in\C$ and $r\in\mathbb N$, write
\[
  (x)_r=x(x-1)\cdots(x-r+1)
\]
for the falling factorial. If $\pi$ is a set partition of $[\ell]$ and $B\in\pi$ is a block, let $\tau_{|B|}^{(B)}$ denote the state~\eqref{eq:tau-s} acting on the tensor positions indexed by $B$.

\begin{smproposition}
\label{prop:beta-partition}
For every $0\le\ell\le N$,
\begin{equation}
\label{eq:beta-partition}
  \beta_{\ell,N}
  =
  \sum_{\substack{\pi\in\operatorname{Part}([\ell])\\ |B|\le d\ \forall B\in\pi}}
  \frac{(q)_{|\pi|}\prod_{B\in\pi}(d)_{|B|}}
       {(qd)_\ell}
  \bigotimes_{B\in\pi}\tau_{|B|}^{(B)}.
\end{equation}
\end{smproposition}

\begin{proof}
An ordered sample of $\ell$ distinct background sites determines a set partition $\pi$ of $[\ell]$: two sample labels lie in the same block precisely when the corresponding 
sites came from the same singlet block. If $\pi$ has $r$ blocks, there are $(q)_r$ injective assignments of its blocks to the $q$ singlets. 
Once a block $B$ has been assigned to a singlet, there are $(d)_{|B|}$ ordered choices of distinct sites within that singlet.  
The total number of ordered samples is $(qd)_\ell$. Conditional on the occupancy partition, the retained state factorizes across distinct singlets, and the factor belonging to $B$ is $\tau_{|B|}^{(B)}$ by 
Lemma~\ref{lem:singlet-red-states}. Summing the conditional states with their exact probabilities gives~\eqref{eq:beta-partition}.
\end{proof}

\begin{smlemma}
\label{lem:collision}
For fixed $\ell$,
\begin{equation}
\label{eq:collision-bound}
  \left\|\beta_{\ell,N}-\tau^{\otimes\ell}\right\|_1
  \le
  2\binom\ell2\frac{d-1}{N-1}.
\end{equation}
Consequently $\beta_{\ell,N}=\tau^{\otimes\ell}+O_{d,\ell}(n^{-1})$ in trace norm.
\end{smlemma}

\begin{proof}
For any fixed pair of sample labels, after the first site has been chosen, exactly $d-1$ of the remaining $N-1$ sites belong to the same singlet. 
Hence the probability that this pair collides is $(d-1)/(N-1)$. By the union bound, the probability $p_{\mathrm{coll}}$ that at least one pair collides obeys
\[
  p_{\mathrm{coll}}
  \le\binom\ell2\frac{d-1}{N-1}.
\]
Conditioned on no collision, the $\ell$ sites come from distinct singlets. Their state is exactly $\tau^{\otimes\ell}$ because different singlet blocks are independent and every one-site singlet reduced state is $\tau$.
Thus
\[
  \beta_{\ell,N}
  =(1-p_{\mathrm{coll}})\tau^{\otimes\ell}
    +p_{\mathrm{coll}}\sigma_{\mathrm{coll}}
\]
for a density operator $\sigma_{\mathrm{coll}}$, and~\eqref{eq:collision-bound} follows from $\|\sigma_{\mathrm{coll}}-\tau^{\otimes\ell}\|_1\le2$.
\end{proof}

For a subset $S\subseteq\{1,\ldots,n\}$ of size $k$, use increasing order to identify $V_{\omega_1}^{\otimes S}$ with $V_{\omega_1}^{\otimes k}$ and define \textit{the erased-subsystem channel}
\begin{equation}
\label{eq:local-channel}
  \Phi_{S,m}(X)
  :=\Tr_{\hat S}(\mathcal{V}_mX\mathcal{V}_m^\dagger).
\end{equation}
and define the \textit{mean $k$-site channel}
\[
  \overline\Phi_{k,n}
  :=\mathbb{E}_m\Phi_{S,m},
\]
where $S$ is any fixed $k$-subset. By permutation invariance of the average, this is independent of $S$ after the canonical identification of the output sites.

\begin{smproposition}
\label{prop:exact-mean-channel}
Let
    $
        \mathcal H_L=V_{\omega_1},\;
        \mathcal H_P=(V_{\omega_1})^{\otimes n}
    $
    with $n=qd+1$.
    Equip \(\mathcal H_L\) with the fundamental representation of
    \(\mathfrak{su}(d)\), and equip \(\mathcal H_P\) with the transversal
    representation given by
    \begin{equation}
    \begin{gathered}
    t_a \rightarrow \sum_{i=1}^n t_a^{(i)}, \quad a=1,\ldots,d^2-1,
    \end{gathered}
    \end{equation}
    where $t_a^{(i)}$ is a fundamental of $t_a$ on the $i$-th physical space $V_{\omega_1}$. Then, for every $X\in\End(V_{\omega_1})$,
\begin{equation}
\label{eq:exact-mean-channel}
    \overline{\Phi}_{k,n}(X)
    =
    \frac{n-k}{n}\Tr(X)\,\beta_{k,N}
    +
    \frac{1}{n}\sum_{j=1}^k
    X^{(j)}\otimes
    \beta_{k-1,N}^{(\widehat j)}
    .
\end{equation}
Here $X^{(j)}$ acts on the $j$-th retained tensor factor, while
$\beta_{k-1,N}^{(\widehat j)}$ acts on all retained tensor factors
except $j$.
\end{smproposition}

\begin{proof}
Use the permutation-twirl description in Proposition~\ref{prop:orbit-twirl}. Under a uniformly random permutation, the distinguished logical seed position is uniformly 
distributed among the $n$ physical positions.
If it is outside the retained $k$-set, which occurs with probability $(n-k)/n$, tracing it out produces the scalar $\Tr X$, 
and the retained sites form a uniform ordered $k$-sample from the $N$ background sites. This gives the first term of~\eqref{eq:exact-mean-channel}.
If the distinguished position is the $j$th retained position, which occurs with probability $1/n$, that factor contains $X$, while the remaining $k-1$ factors form a uniform ordered sample from the background. 
Summing over $j$ gives the second term.  
\end{proof}

\begin{smlemma}
\label{lem:bures-expansion}
Let $r$ be fixed.  Let $\eta>0$ be an $r\times r$ density operator with eigenvalues $\lambda_1,\ldots,\lambda_r$, and let $\Delta=\Delta^\dagger$ satisfy $\Tr\Delta=0$ and $\eta+\Delta\ge0$. 
If $\lambda_{\min}(\eta)\ge c>0$ and $\|\Delta\|_2$ is sufficiently small, then, in an eigenbasis of $\eta$,
\begin{equation}
\label{eq:general-bures-expansion}
  f(\eta+\Delta,\eta)
  =1-\frac14\sum_{a,b=1}^r
      \frac{|\Delta_{ab}|^2}{\lambda_a+\lambda_b}
      +O_{r,c}(\|\Delta\|_2^3).
\end{equation}
If in addition
\[
  \eta=\frac{I_r}{r}+O(\epsilon),
  \qquad
  \Delta=O(\epsilon),
\]
then
\begin{equation}
\label{eq:mixed-bures-expansion}
  f(\eta+\Delta,\eta)
  =1-\frac r8\Tr(\Delta^2)+O_r(\epsilon^3).
\end{equation}
\end{smlemma}

\begin{proof}
Diagonalize $\eta=\operatorname{diag}(\lambda_1,\ldots,\lambda_r)$ and set
\[
  B:=\sqrt\eta\,\Delta\sqrt\eta.
\]
Then
\[
  f(\eta+\Delta,\eta)
  =\Tr\sqrt{\eta^2+B}.
\]
For $f(x)=\sqrt x$, the second derivative of the trace functional along a Hermitian perturbation is
\[
  \left.\frac{\dd^2}{\dd t^2}\Tr f(A+tB)\right|_{t=0}
  =\sum_{a,b}
  \frac{f'(\mu_a)-f'(\mu_b)}{\mu_a-\mu_b}|B_{ab}|^2,
\]
with the diagonal terms interpreted by continuity, where $\mu_a$ are the eigenvalues of $A$. 
Here $A=\eta^2$, so $\mu_a=\lambda_a^2$, $B_{ab}=\sqrt{\lambda_a\lambda_b}\,\Delta_{ab}$, and
\[
  \frac{f'(\lambda_a^2)-f'(\lambda_b^2)}
       {\lambda_a^2-\lambda_b^2}
  =-\frac1{2\lambda_a\lambda_b(\lambda_a+\lambda_b)}.
\]
The first derivative is
\[
  \sum_a\frac1{2\lambda_a}B_{aa}
  =\frac12\Tr\Delta=0.
\]
Multiplying the second derivative by the Taylor factor $1/2$ gives the quadratic term in~\eqref{eq:general-bures-expansion}. 
Because the spectrum remains in a compact subset of $(0,\infty)$, the trace square-root functional has a uniformly bounded third Fr\'echet derivative in a neighborhood of $\eta^2$, 
yielding the stated remainder.

If $\lambda_a=1/r+O(\epsilon)$, then
\[
  \frac1{\lambda_a+\lambda_b}=\frac r2+O_r(\epsilon).
\]
Since $\sum_{a,b}|\Delta_{ab}|^2=\Tr(\Delta^2)=O(\epsilon^2)$, substituting this into~\eqref{eq:general-bures-expansion} gives~\eqref{eq:mixed-bures-expansion}.
\end{proof}

\begin{smproposition}\label{prop:k-loc_fidelity}
For fixed $k, d$
\begin{equation}
\label{eq:mean-fidelity}
  F(\overline\Phi_{k,n},\Lambda_{k,n})
  =1-\frac{k(d^2-1)}{8n^2}+O_{d,k}(n^{-3}).
\end{equation}
where 
\begin{equation}
\label{eq:lambda-kn}
  \lambda_{k,n}:=\overline\Phi_{k,n}(\tau),
  \qquad
  \Lambda_{k,n}(X):=\Tr(X)\,\lambda_{k,n}.
\end{equation}
\end{smproposition}

\begin{proof}
    From Proposition~\ref{prop:exact-mean-channel} and Lemma~\ref{lem:collision} it follows that
    \begin{equation}
    \label{eq:exact-leakage}
    (\overline\Phi_{k,n}-\Lambda_{k,n})(X)
    =\frac1n\sum_{j=1}^k
    \left(X-\Tr(X)\tau\right)^{(j)}
    \otimes\beta_{k-1,N}^{(\widehat j)}=\frac1n\sum_{j=1}^k
  \left(X-\Tr(X)\tau\right)^{(j)}
  \otimes\tau^{\otimes(k-1)}_{\widehat j}
  +O_{d,k}(n^{-2})\,\|X\|_1.
    \end{equation}
    Both $\overline\Phi_{k,n}$ and $\Lambda_{k,n}$ are $SU(d)$-covariant under the fundamental input action and the collective output action.  
    Indeed, every $\beta_{\ell,N}$ is $SU(d)$-invariant.  Therefore Lemma~\ref{fidelity_opt} applies,
    meaning that optimum is attained on maximally mixed state $I_d/d$, i. e.
    \[
        F(\overline\Phi_{k,n},\Lambda_{k,n})=f(\sigma_{k,n},\eta_{k,n})=f(\eta_{k,n}+\Delta_{k,n},\eta_{k,n}),
    \]
    where we set
    \begin{align}
    \sigma_{k,n}
    &:=(\overline\Phi_{k,n}\otimes\id_R)(\Psi_{LR}),\\
    \eta_{k,n}
    &:=(\Lambda_{k,n}\otimes\id_R)(\Psi_{LR})
        =\lambda_{k,n}\otimes\tau_R,\\
    \Delta_{k,n}&:=\sigma_{k,n}-\eta_{k,n}.
    \end{align}
    Let $R\cong V_{\omega_1}$ be a reference and let
    \[
    \ket{\Psi_d}_{LR}=\frac1{\sqrt d}\sum_{a=1}^d\ket a_L\ket a_R,
    \qquad
    \Psi_{LR}:=\ket{\Psi_d}\bra{\Psi_d}
    \]
    be purification of maximally mixed state.
    For a retained output site $j$, write $\Psi_{jR}$ for the same maximally entangled state carried by site $j$ and the reference.
    Define the traceless two-body operator
    \begin{equation}
    \label{eq:Gamma}
    \Gamma_{jR}:=\Psi_{jR}-\tau_j\otimes\tau_R.
    \end{equation}
    Notice, that
    \begin{equation}
    \label{eq:exact-choi-difference}
    \Delta_{k,n}
    =\frac1n\sum_{j=1}^k
    \Gamma_{jR}\otimes\beta_{k-1,N}^{(\widehat j)}.
    \end{equation}
    Indeed, expand
\[
  \Psi_{LR}=\frac1d\sum_{a,b=1}^d E_{ab}\otimes E_{ab}
\]
and apply the exact leakage formula~\eqref{eq:exact-leakage} to each $E_{ab}$, to get
\[
  \frac1d\sum_{a,b}
  (E_{ab}-\delta_{ab}\tau)\otimes E_{ab}
  =\Psi_{LR}-\tau_L\otimes\tau_R.
\]
Defining 
\begin{equation}
\label{eq:Aj}
  A_j:=\Gamma_{jR}\otimes\tau_{\widehat j}^{\otimes(k-1)},
\end{equation}
we get
 \begin{equation}
\label{eq:Delta-leading}
  \Delta_{k,n}=\frac1n\sum_{j=1}^kA_j+O_{d,k}(n^{-2}), \quad \eta_{k,n}=\frac{I_{d^{k+1}}}{d^{k+1}}+O_{d,k}(n^{-1})
\end{equation}
from Eq.~\eqref{eq:exact-leakage}.
For $i\ne j$,
\begin{equation}
\label{eq:Aj-norm}
  \Tr(A_j^2)=\frac{d^2-1}{d^{k+1}}, \quad \Tr(A_iA_j)=0.
\end{equation}
Indeed, The partial trace of $\Gamma_{jR}$ over either of its factors vanishes. 
For $i\ne j$, the product $A_iA_j$ contains a factor $\tau_i$ multiplying $\Gamma_{iR}$ while the remaining operator is independent of subsystem $i$; 
tracing subsystem $i$ therefore gives zero. This proves orthogonality.
Next,
\[
  \Tr(\Gamma_{jR}^2)
  =\Tr(\Psi_{jR}^2)
   -2\Tr(\Psi_{jR}\tau_j\otimes\tau_R)
   +\Tr((\tau_j\otimes\tau_R)^2)
  =1-\frac1{d^2}.
\]
Since $\Tr(\tau^2)=1/d$, the second part of Eq.~\eqref{eq:Aj-norm} follows. Therefore
\begin{equation}
\label{eq:Delta-HS}
  \Tr(\Delta_{k,n}^2)
  =\frac{k(d^2-1)}{d^{k+1}n^2}+O_{d,k}(n^{-3}),
\end{equation}
which follows from Eq.~\eqref{eq:Delta-leading} and Eq.~\eqref{eq:Aj-norm}.

The spectral hypothesis of Lemma~\ref{lem:bures-expansion} holds from an explicit threshold on. Since $\tau=I_d/d$, 
Proposition~\ref{prop:exact-mean-channel} writes $\lambda_{k,n}-I/d^k$ as the convex combination
\[
  \lambda_{k,n}-\frac{I}{d^k}
  =\frac{n-k}{n}\left(\beta_{k,N}-\tau^{\otimes k}\right)
   +\frac1n\sum_{j=1}^k\left(\tau^{(j)}\otimes\beta_{k-1,N}^{(\widehat j)}-\tau^{\otimes k}\right),
\]
so Lemma~\ref{lem:collision} and $\binom{k-1}2\le\binom k2$ give
\[
  \left\|\eta_{k,n}-\frac{I}{d^{k+1}}\right\|_1
  =\left\|\lambda_{k,n}-\frac{I}{d^k}\right\|_1
  \le 2\binom k2\frac{d-1}{n-2},
\]
using $\eta_{k,n}=\lambda_{k,n}\otimes\tau_R$ and $\|\tau_R\|_1=1$. As $\|\cdot\|_\infty\le\|\cdot\|_1$,
\[
  \lambda_{\min}(\eta_{k,n})
  \ \ge\ \frac{1}{d^{k+1}}-2\binom k2\frac{d-1}{n-2}
  \ \ge\ \frac{1}{2d^{k+1}}
  \qquad\text{whenever}\qquad
  n\ \ge\ n_0(d,k):=2+2k(k-1)(d-1)d^{k+1}.
\]
In particular $\eta_{k,n}$ is then full rank, and Lemma~\ref{lem:bures-expansion} applies with $c=\tfrac12 d^{-(k+1)}$ for every $n\ge n_0(d,k)$. For $k=1$ no two sampled sites can collide, $\beta_{1,N}=\tau$ exactly, and $n_0(d,1)=2$.

Apply Lemma~\ref{lem:bures-expansion} with $r=d^{k+1}$, $\epsilon=1/n$, and then use Eq.~\eqref{eq:Delta-HS}:
\begin{align*}
  f(\sigma_{k,n},\eta_{k,n})
  &=1-\frac{d^{k+1}}8\Tr(\Delta_{k,n}^2)
       +O_{d,k}(n^{-3})\\
  &=1-\frac{k(d^2-1)}{8n^2}+O_{d,k}(n^{-3}).
\end{align*}
\end{proof}

\subsubsection{Haar-typical multiplicity vectors}\label{subsubsec:supple_haar_typical}

Let
\[
  r_k:=d^{k+1},
  \qquad
  M_{n,k}:=\binom nk d^{2k+2}=\binom nk r_k^2, \qquad D_q:=\operatorname{dim}\left[q+1, q^{d-1}\right].
\]
Fix orthonormal computational bases of $V_{\omega_1}$ and $V_{\omega_1}^{\otimes k}$. 
For $a,b\in[d]$ and output basis vectors $u,v\in[d]^k$, define
\begin{equation}
\label{eq:test-operator}
  A_{S;ab;uv}
  :=(\bra b\otimes I_{M_{\omega_1}})
  W_{\omega_1}^\dagger\bigl(\ket v\bra u_S\otimes I_{\hat S}\bigr)W_{\omega_1}
  (\ket a\otimes I_{M_{\omega_1}})
  \in\End(M_{\omega_1}).
\end{equation}

\begin{smlemma}
\label{lem:quadratic-form}
For every unit $m\in M_{\omega_1}$,
\begin{equation}
\label{eq:entry-quadratic}
  \bra u\Phi_{S,m}(E_{ab})\ket v
  =\bra m A_{S;ab;uv}\ket m,
\end{equation}
with $\|A_{S;ab;uv}\|_\infty\le1$. Moreover,
\begin{equation}
    \frac{\operatorname{Tr} A_{S ; a b ; u v}}{D_q}=\langle u| \overline{\Phi}_{k,n}\left(E_{a b}\right)|v\rangle .
\end{equation}
\end{smlemma}

\begin{proof}
We have
\begin{align*}
\langle u| \Phi_{S, m}\left(E_{a b}\right)|v\rangle&=\langle u|\operatorname{Tr}_{\hat S}\left[\mathcal{V}_m E_{a b} \mathcal{V}_m^{\dagger}\right]|v\rangle\\
&=\operatorname{Tr}\left[\left(|v\rangle\left\langle\left. u\right|_S \otimes I_{\hat S}\right) W_{\omega_1}\left(E_{a b} \otimes|m\rangle\langle m|\right) W_{\omega_1}^{\dagger}\right]\right.\\
&=\operatorname{Tr}\left[W_{\omega_1}^{\dagger}\left(|v\rangle\left\langle\left. u\right|_S \otimes I_{\hat S}\right) W_{\omega_1}\left(\ket{a}\bra{b} \otimes|m\rangle\langle m|\right)\right]\right.\\
&=(\langle b| \otimes I) W_{\omega_1}^{\dagger}\left(|v\rangle\left\langle\left. u\right|_S \otimes I_{\hat S}\right) W_{\omega_1}(|a\rangle \otimes I)\right.\\
&=A_{S ; a b ; u v},
\end{align*}
which proves Eq.~\eqref{eq:entry-quadratic}. Norm bound follows from $W_{\omega_1}$ being an isometry. The last statement follows from
\[
    \begin{aligned}
\mathbb{E}_m\langle m| A|m\rangle & =\operatorname{Tr}\left[A \mathbb{E}_m|m\rangle\langle m|\right] \\
& =\operatorname{Tr}\left[A \frac{I}{D_q}\right] \\
& =\frac{\operatorname{Tr} A}{D_q},
\end{aligned}
\]
Eq.~\eqref{eq:entry-quadratic} and 
\[
\mathbb{E}_m\langle u| \Phi_{S, m}\left(E_{a b}\right)|v\rangle=\langle u| \overline{\Phi}_{k,n}\left(E_{a b}\right)|v\rangle.
\]
\end{proof}

\begin{smlemma}
\label{lem:haar-second-moment}
Let $m$ be Haar-random on the unit sphere of a $D$-dimensional complex Hilbert space. For any $A\in\End(\C^D)$,
\begin{equation}
\label{eq:haar-variance}
  \mathbb{E}_m\left|
  \bra mA\ket m-\frac{\Tr A}{D}
  \right|^2
  =\frac{\Tr(AA^\dagger)-|\Tr A|^2/D}{D(D+1)}.
\end{equation}
In particular, if $\|A\|_\infty\le1$, the variance is at most $1/(D+1)$.
\end{smlemma}

\begin{proof}
Define 
\[
Z=\langle m| A|m\rangle-\frac{\operatorname{Tr} A}{D},
\]
and from the proof of the Lemma~\ref{lem:quadratic-form} we have $\mathbb{E}_m Z=0$. Then
\[
\left.\mathbb{E}|Z|^2=\mathbb{E}|\langle m| A| m\right\rangle\left.\right|^ 2-\frac{|\operatorname{Tr} A|^2}{D^2}.
\]
As for the first term, we have
\[
\mathbb{E}|\langle m| A| m\rangle\left.\right|^ 2=\operatorname{Tr}\left[\left(A \otimes A^{\dagger}\right)\mathbb{E}\left(|m\rangle\langle m| \otimes |m\rangle\langle m|\right)\right],
\]
where we use Haar second moment identity
\[
\mathbb{E}_m\left[(|m\rangle\langle m|)^{\otimes 2}\right]=\frac{I+\mathrm{SWAP}}{D(D+1)},
\]
from which it follows that
\[
   \mathbb{E}|\langle m| A| m\rangle\left.\right|^ 2= \frac{|\operatorname{Tr} A|^2+\operatorname{Tr}\left(A A^{\dagger}\right)}{D(D+1)},
\]
which concludes the proof of Eq.~\eqref{eq:haar-variance}. The last statement follows from $\operatorname{Tr}\left(A A^{\dagger}\right) \leq D\|A\|_{\infty}^{2} \leq D$.
\end{proof}

\begin{smlemma}
\label{lem:entry-to-diamond}
Let $T:\End(\C^d)\to\End((\C^d)^{\otimes k})$ be linear.  If every matrix element of every $T(E_{ab})$ has absolute value at most $t$, then
\begin{equation}
\label{eq:entry-diamond}
  \|T\|_\diamond\le r_k^{3/2}t,
  \qquad r_k=d^{k+1}.
\end{equation}
\end{smlemma}

\begin{proof}
Let $J(T)=\sum_{a,b}T(E_{ab})\otimes E_{ab}$ be the unnormalized Choi matrix.  It is an $r_k\times r_k$ matrix and all of its entries have magnitude at most $t$.  Therefore
\[
  \|J(T)\|_2\le r_k t,
  \qquad
  \|J(T)\|_1\le\sqrt{r_k}\,\|J(T)\|_2\le r_k^{3/2}t.
\]
The standard Choi bound $\|T\|_\diamond\le\|J(T)\|_1$ completes the proof.
\end{proof}

\begin{smtheorem}
\label{thm:haar-uniform}
Let $m$ be Haar-random in $M_{\omega_1}$.  Then
\begin{equation}
\label{eq:haar-probability}
  \mathbb P\!\left[
    \max_{|S|=k}
    \|\Phi_{S,m}-\overline\Phi_{k,n}\|_\diamond
    >r_k^{3/2}D_q^{-1/4}
  \right]
  \le
  \frac{M_{n,k}D_q^{1/2}}{D_q+1}.
\end{equation}
For fixed $d,k$, the right-hand side is $e^{-\Omega_d(n)}$.
\end{smtheorem}

\begin{proof}
    Define
    \[
        Z:=\langle u|\left(\Phi_{S, m}-\overline{\Phi}_{k,n}\right)\left(E_{a b}\right)|v\rangle=\langle m| A_{S ; a b ; u v}|m\rangle-\frac{\operatorname{Tr} A_{S ; a b ; u v}}{D_q},
    \]
    where the second equality follows from Lemma~\ref{lem:quadratic-form}. By Lemma~\ref{lem:haar-second-moment}, we have
    \[
        \mathbb{E}|Z|^2 \leq \frac{1}{D_q+1}.
    \]
    Applying Markov inequality to $|Z|^2$
    \[
        \mathbb{P}(|Z|>t)=\mathbb{P}\left(|Z|^2>t^2\right) \leq \frac{\mathbb{E}|Z|^2}{t^2},
    \]
    we get
    \[
        \mathbb{P}(|Z|>D_q^{-1/4}) \leq \frac{D_q^{1/2}}{D_q+1},
    \]
    where we took $t=D_q^{-1/4}$. Random variable $Z$ represents the difference between a matrix element of $\Phi_{S, m}\left(E_{a b}\right)$ and the corresponding matrix element of $\overline{\Phi}_{k,n}\left(E_{a b}\right)$. 
    The total number of such matrix elements is $M_{n,k}$.  By the union bound,
    \[
    \mathbb{P}(\text { some entry is bad }) \leq M_{n, k} \frac{D_q^{1 / 2}}{D_q+1}.
    \]
    For "good" events, we have
    \[
    \left.\left|\langle u|\left(\Phi_{S, m}-\overline{\Phi}_{k,n}\right)\left(E_{a b}\right)\right| v\right\rangle \mid \leq D_q^{-1 / 4},
    \]
    therefore by Lemma~\ref{lem:entry-to-diamond}, we obtain
    \[
    \max _{S \mid=k}\left\|\Phi_{S, m}-\overline{\Phi}_{k,n}\right\|_{\diamond} \leq r_k^{3 / 2} D_q^{-1 / 4},
    \]
    which concludes the proof.
\end{proof}

\subsubsection{Efficient sampling of multiplicity vectors}\label{sec:efficient_sampling}

We introduce Young--Yamanouchi basis for the Specht module $[\lambda]$ associated with a partition $\lambda\vdash n$. 
The Young diagram of $\lambda=(\lambda_1,\lambda_2,\ldots)$ consists of $\lambda_i$ boxes in its $i$th row. 
A standard Young tableau of shape $\lambda$ is a filling of these boxes with the integers $1,\ldots,n$, each appearing exactly once, such that the entries increase from left to right along every row and from top to bottom along every column. 
We denote the set of standard Young tableaux of shape $\lambda$ by $\operatorname{SYT}(\lambda)$. 
The Young--Yamanouchi basis of $[\lambda]$ is the orthonormal basis
\[
    \left\{\,\ket{T}:T\in\operatorname{SYT}(\lambda)\,\right\},
\]
whose basis vectors are indexed by the standard Young tableaux of shape $\lambda$. 
For a detailed discussion of Young tableaux, Specht modules, and the Young--Yamanouchi basis, see Ref.~\cite{Sagan2001}.

We will construct vectors
\begin{equation}\label{eq:family_of_multiplicity_vecs}
    \left|m_s\right\rangle=\frac{1}{\sqrt{D_q}} \sum_{T \in \operatorname{SYT}\left(\lambda_q\right)} z_s(T)|T\rangle
\end{equation}
where we defined
\[
    \lambda_q:=(q+1,q^{d-1}), \quad [\lambda_q]\cong M_{\omega_1},
\]
such that the seed $s$ has only $O(n \log d)$ bits and every phase
\[
z_s(T) \in\{1,-1, i,-i\}
\]
can be computed in polynomial time from $s$ and $T$.

\begin{smtheorem}\label{thrm:randomized_multiplicity_vector}
    Fix $d$ and $k$, and let $n=q d+1$. There exists an explicit randomized algorithm using $O(n \log d)$ random bits that outputs a seed $s$ specifying a flat multiplicity vector
    $$
    \left|m_s\right\rangle=\frac{1}{\sqrt{D_q}} \sum_{T \in \operatorname{SYT}\left(\lambda_q\right)} z_s(T)|T\rangle, \quad z_s(T) \in\{1,-1, i,-i\},
    $$
    such that each $z_s(T)$ is computable in time $\operatorname{poly}(n, \log d)$, and
    \begin{equation}\label{eq:randomized_multiplicity_vector_bound}
    \mathbb{P}_s\left[\max _{|S|=k}\left\|\Phi_{S, m_s}-\overline{\Phi}_{k,n}\right\|_{\diamond}>r_k^{3 / 2} D_q^{-1 / 4}\right] \leq \frac{M_{n, k}}{\sqrt{D_q}}=e^{-\Omega_d(n)}.
    \end{equation}
\end{smtheorem}
\begin{proof}
The idea is that we want to assign to every tableau $T$ a phase
$
z_s(T) \in\{1,-1, i,-i\},
$
but we do not want to store one independent random phase for every tableau, because there are
$
D_q=e^{\Theta(n)}
$
tableaux. So we generate all these phases from a very short random seed. The key idea is that we don't need independence of all phases, 
instead it is enough to have $4$-wise independence of the phases, which can be achieved with a seed of length $O(n \log d)$.

To efficiently label every tableau $T$ we define a \textit{tableau word} $w(T)$ of length $n$ as follows
\[
    w(T)=\left(r_1, \ldots, r_n\right),
\]
where $r_t\in\{1, \ldots, d\}$ is the row index of the box containing $t$ in $T$. This encoding requires $h:=n \log d$ bits.
We identify bitstring $x_T$ with corresponding element of $\mathbb{F}_{2^h}$.

Choose four independent uniformly random elements
$
c_0, c_1, c_2, c_3 \in \mathbb{F}_{2^h} .
$
These four field elements constitute the seed
$
s=\left(c_0, c_1, c_2, c_3\right) .
$
Thus the random seed contains exactly
$
4 h=O(n \log d)
$
bits.
Define
$$
p_s(x):=c_0+c_1 x+c_2 x^2+c_3 x^3 .
$$
Fix a surjective $\mathbb{F}_2$-linear map
$
L: \mathbb{F}_{2^h} \longrightarrow \mathbb{F}_2^2 .
$
For example, after choosing an $\mathbb{F}_2$-basis of the field, $L$ can simply return the first two coordinate bits.
Write
$$
L\left(p_s\left(x_T\right)\right)=\left(b_0(T), b_1(T)\right) .
$$
Now define
$$
z_s(T):=i^{b_0(T)+2 b_1(T)} .
$$
This produces one of the four phases. 

Take any four distinct tableaux
$
T_1, T_2, T_3, T_4 .
$
Their labels
$
x_{T_1}, \ldots, x_{T_4}
$
are distinct field elements.
The map
$$
\left(c_0, c_1, c_2, c_3\right) \longmapsto\left(p_s\left(x_{T_1}\right), p_s\left(x_{T_2}\right), p_s\left(x_{T_3}\right), p_s\left(x_{T_4}\right)\right)
$$
is a linear map over $\mathbb{F}_{2^h}$. Its matrix is the Vandermonde matrix
$$
\left(\begin{array}{cccc}
1 & x_{T_1} & x_{T_1}^2 & x_{T_1}^3 \\
1 & x_{T_2} & x_{T_2}^2 & x_{T_2}^3 \\
1 & x_{T_3} & x_{T_3}^2 & x_{T_3}^3 \\
1 & x_{T_4} & x_{T_4}^2 & x_{T_4}^3
\end{array}\right) .
$$
Its determinant is
$
\prod_{a<b}\left(x_{T_b}-x_{T_a}\right),
$
which is nonzero because the $x_{T_j}$ are distinct.
Consequently,
$
p_s\left(x_{T_1}\right), \ldots, p_s\left(x_{T_4}\right)
$
are independent uniform field elements.
Applying $L$ separately to each one preserves independence, and each output is uniform in $\mathbb{F}_2^2$. Therefore
$
z_s\left(T_1\right), \ldots, z_s\left(T_4\right)
$
are independent uniform fourth roots of unity.
The same is automatically true for any set of fewer than four tableaux.
Thus the family
$
\left\{z_s(T)\right\}_T
$
is exactly four-wise independent.

The second part of the proof follows the ideas introduced in Theorem~\ref{thm:haar-uniform}. We first define
\[
Z(s):=\langle u|\left(\Phi_{S, m_s}-\overline{\Phi}_{k,n}\right)\left(E_{a b}\right)|v\rangle=\langle m_s| A_{S ; a b ; u v}|m_s\rangle-\frac{\operatorname{Tr} A_{S ; a b ; u v}}{D_q},
\]
where the second equality follows from Lemma~\ref{lem:quadratic-form} and $m_s$ is defined in Eq.~\eqref{eq:family_of_multiplicity_vecs} with seed $s$ constructed just above. Write
$$
\left|Y_A(s)\right|^2=\frac{1}{D_q^2} \sum_{\substack{T \neq T^{\prime} \\ U \neq U^{\prime}}} \overline{z_s(T)} z_s\left(T^{\prime}\right) z_s(U) \overline{z_s\left(U^{\prime}\right)} A_{T T^{\prime}} \overline{A_{U U^{\prime}}}
$$
Every phase expectation here involves at most four different tableaux. Thus four-wise independence gives exactly the same calculation as full independence.
Because
$$
\mathbb{E} z_s(T)=0, \quad \mathbb{E} z_s(T)^2=0, \quad\left|z_s(T)\right|=1,
$$
the only surviving terms are
$
T=U, \quad T^{\prime}=U^{\prime} .
$
Hence
$$
\mathbb{E}_s\left|Y_A(s)\right|^2=\frac{1}{D_q^2} \sum_{T \neq T^{\prime}}\left|A_{T T^{\prime}}\right|^2\leq \frac{\operatorname{Tr}\left(A A^{\dagger}\right)}{D_q^2}\leq \frac{1}{D_q}.
$$
Finally, applying the Markov inequality, the union bound argument and Lemma~\ref{lem:entry-to-diamond} as in the proof of Theorem~\ref{thm:haar-uniform} gives Eq.~\ref{eq:randomized_multiplicity_vector_bound}.
The algorithm samples only
$
c_0, c_1, c_2, c_3 \in \mathbb{F}_{2^h} .
$
Therefore random-bit complexity $=4 h=O(n \log d)$.
Evaluating $z_s(T)$ requires:
encoding the row word of $T$;
three field multiplications;
several field additions;
extracting two linear bits.
This takes
$
\operatorname{poly}(n, \log d)
$
bit operations.
So the sampler is an unconditional randomized polynomial-time algorithm in $n$ for fixed $d$.
It is a Monte Carlo sampler:
it always runs in polynomial time;
the output is good with probability $1-e^{-\Omega_d(n)}$;
it does not need to compute or inspect the exponentially large test matrices;
it does not efficiently certify afterward that the particular seed was good.
\end{proof}

\subsection{Arbitrary flagged multi-qudit noise}\label{subsec:supple_arb_3_noise}

We first record the form of the complementary channel for flagged noise on a family of sets, Lemma~\ref{lem:dual_to_arb}, and the single-site estimate, Theorem~\ref{thrm:sudscale_general}, 
which needs only cyclic invariance of the code space and holds for any physical representation $V_{\omega_r}$. 
The remaining arguments never use that the noisy set has a particular size. 
They use only that the reduced state on the noisy set splits into a logical-independent part and a logical-dependent part that is $O_{d,k}(1/n)$ 
uniformly over the sets on which the noise can act. We therefore state that estimate once, for an arbitrary family $\mathcal S$ of subsets, 
and then apply it twice: to the exactly transitive codes of Section~\ref{subsec:supple_multi_erasure_noise}, which gives the statement for noise on at most three qudits, and to the codes of 
Section~\ref{subsec:supple_multiqudit_erasure}, which extends it to an arbitrary fixed number of noisy sites.

\begin{smlemma}\label{lem:dual_to_arb}
    Let $\mathcal{E}$ be the encoding map defined by $\mathcal{E}(\rho_L):=\mathcal{V}\rho_L \mathcal{V}^{\dagger}$ and
    $
        \mathcal N(\sigma)
        =
        \sum_{S\in\mathcal S}
        p_S
        |S\rangle\langle S|_F
        \otimes
        \left(\mathcal N_S\otimes \mathrm{id}_{\hat S}\right)(\sigma)
    $
    be a noise model, where $\mathcal{N}_S$ is an arbitrary noise channel acting on the subsystem $S$. 
    Then a complementary channel to $\mathcal{N}\circ \mathcal{E}$ is given as follows:
    \[
        \widehat{\mathcal{N} \circ \mathcal{E}}\left(\rho_L\right)
        =
        \sum_{S\in\mathcal S}
        p_S
        |S\rangle\langle S|_{F_E}
        \otimes
        \widehat{\mathcal{N}_S}\left(\rho^{(S)}\left(\rho_L\right)\right).
    \]
\end{smlemma}

\begin{proof}
    Consider a Stinespring dilation of the noise channel $\mathcal{N}$ of the form
    \[
        \mathcal{W}
        =
        \sum_{S\in\mathcal S}
        \sqrt{p_S}
        |S\rangle_F
        \otimes
        |S\rangle_{F_E}
        \otimes
        \left(W_S \otimes I_{\hat S}\right),
    \]
    where $W_S$ is a Stinespring dilation of $\mathcal{N}_S$. Notice that
\[
    \begin{aligned}
    \widehat{\mathcal{N}\circ\mathcal{E}}(\rho_L)
    &=
    \operatorname{Tr}_{\mathrm{rec}}
    \left[
    \mathcal W \mathcal{V}\rho_L\mathcal{V}^\dagger \mathcal W^\dagger
    \right]
    \\
    &=
    \sum_{S,T\in\mathcal S}
    \sqrt{p_Sp_T}\,
    \operatorname{Tr}_F(|S\rangle\langle T|_F)\,
    |S\rangle\langle T|_{F_E}
    \otimes
    \operatorname{Tr}_{B_S\hat S}
    \left[
    (W_S\otimes I_{\hat S})
    \mathcal{V}\rho_L\mathcal{V}^\dagger
    (W_T^\dagger\otimes I_{\hat T})
    \right]
    \\
    &=
    \sum_{S\in\mathcal S}
    p_S |S\rangle\langle S|_{F_E}
    \otimes
    \operatorname{Tr}_{B_S\hat S}
    \left[
    (W_S\otimes I_{\hat S})
    \mathcal{V}\rho_L\mathcal{V}^\dagger
    (W_S^\dagger\otimes I_{\hat S})
    \right]
    \\
    &=
    \sum_{S\in\mathcal S}
    p_S |S\rangle\langle S|_{F_E}
    \otimes
    \operatorname{Tr}_{B_S}
    \left[
    W_S
    \operatorname{Tr}_{\hat S}(\mathcal{V}\rho_L\mathcal{V}^\dagger)
    W_S^\dagger
    \right]
    \\
    &=
    \sum_{S\in\mathcal S}
    p_S |S\rangle\langle S|_{F_E}
    \otimes
    \widehat{\mathcal N_S}\bigl(\rho^{(S)}(\rho_L)\bigr),
    \end{aligned}
\]
where $B_S$ is the output Hilbert space of the local noise channel acting on the subsystem $S$.
\end{proof}
\begin{smtheorem}\label{thrm:sudscale_general}
            Let
$
    \mathcal H_L=V_{\omega_1},\;
    \mathcal H_P=(V_{\omega_r})^{\otimes n}
$
and $nr\equiv 1\;(\operatorname{mod} d)$.
Equip \(\mathcal H_L\) with the fundamental representation of
\(\mathfrak{su}(d)\), and equip \(\mathcal H_P\) with the transversal
representation given by
\begin{equation}\label{eq:global_generators}
\begin{gathered}
t_a \rightarrow T_a \equiv \sum_{i=1}^n T_a^{(i)}, \quad a=1,\ldots,d^2-1,
\end{gathered}
\end{equation} 
where $T_a^{(i)}$ is a representation of $t_a$ on the $i$-th physical space $V_{\omega_r}$, acting as in Eq.~\eqref{eq:action_on_extirior_power}. 
If encoding $\mathcal{E}$, which is $\mathfrak{su}(d)$-covariant with respect to defined physical and logical representations, 
defines a code space invariant under cyclic permutations of physical qudits, then $\mathcal{E}$ is $O_{d}\left(\frac{1}{\sqrt{n}}\right)$-approximate against arbitrary single-qudit noise 
channel $\mathcal{N}(\sigma)=\sum_{i=1}^n p_i|i\rangle\left\langle\left. i\right|_F \otimes\left(\mathcal{N}_i \otimes \mathrm{id}_{\hat{i}}\right)(\sigma),\right.$.
\end{smtheorem}
\begin{proof}
In Proposition~\ref{prop:SUdcov} we obtained the following form of the reduced state on the \(i\)-th physical qudit:
\[
\rho^{(i)}(\rho_L)
=
\frac{I_{V_{\omega_r}}}{\dim V_{\omega_r}}
+
\frac{1}{n}\mathcal A_i\left(\rho_L-\frac{I_d}{d}\right),
\]
where the linear map \(\mathcal A_i\) is defined on the traceless part by
\[
\mathcal A_i\left(\sum_a r_a t_a\right)
=
\frac{1}{2\kappa_{\omega_r}}\sum_a r_a T_a^{(i)} .
\]
Therefore, using Lemma~\ref{lem:dual_to_arb}, we can rewrite the complementary channel as
\[
\widehat{\mathcal{N} \circ \mathcal{E}}(\rho_L)
=
\sum_i p_i |i\rangle\langle i|_{F_E}\otimes
\left[
\widehat{\mathcal N}_i\left(\frac{I_{V_{\omega_r}}}{\dim V_{\omega_r}}\right)
+
\frac{1}{n}
\widehat{\mathcal N}_i\mathcal A_i\left(\rho_L-\frac{I_d}{d}\right)
\right].
\]
It is then natural to take
\[
\Lambda_0(\rho_L)
=
\operatorname{Tr}(\rho_L)
\sum_i p_i |i\rangle\langle i|_{F_E}\otimes
\widehat{\mathcal N}_i\left(\frac{I_{V_{\omega_r}}}{\dim V_{\omega_r}}\right),
\]
which is a constant channel, i. e. it is independent of \(\rho_L\). We now estimate \(F(\widehat{\mathcal{N} \circ \mathcal{E}},\Lambda_0)\).

We compute \(F_\rho(\widehat{\mathcal{N} \circ \mathcal{E}},\Lambda_0)=f(\sigma_\rho,\eta_\rho)\), where
\[
\sigma_\rho
:=
(\widehat{\mathcal{N}\circ\mathcal{E}}\otimes \mathrm{id})
(|\psi_\rho\rangle\langle\psi_\rho|),
\qquad
\eta_\rho
:=
(\Lambda_0\otimes \mathrm{id})
(|\psi_\rho\rangle\langle\psi_\rho|)
\]
and \(|\psi_\rho\rangle\) is a purification of \(\rho\). By the Fuchs--van de Graaf inequality, we have
\[
\begin{aligned}
1-f(\sigma_\rho,\eta_\rho)
&\leq
\frac{1}{2}\left\|\sigma_\rho-\eta_\rho\right\|_1 \\
&=
\frac{1}{2n}
\sum_i p_i
\left\|
\left((\widehat{\mathcal N}_i\circ \mathcal A_i)\otimes \mathrm{id}_R\right)
\left[
|\psi_\rho\rangle\langle\psi_\rho|
-
\frac{I_d}{d}\otimes \rho_R
\right]
\right\|_1 \\
&\leq
\frac{1}{2n}
\sum_i p_i
\left\|
(\mathcal A_i\otimes \mathrm{id}_R)
\left[
|\psi_\rho\rangle\langle\psi_\rho|
-
\frac{I_d}{d}\otimes \rho_R
\right]
\right\|_1 \\
&\leq
\frac{C_{\omega_r}}{2n}
\left\|
|\psi_\rho\rangle\langle\psi_\rho|
-
\frac{I_d}{d}\otimes \rho_R
\right\|_1 \\
&\leq
\frac{C_{\omega_r}}{n}.
\end{aligned}
\]
Here, in the third line, we used contractivity of the trace norm under the CPTP map \(\widehat{\mathcal N}_i\), 
and \(C_{\omega_r}:=\max_i|\mathcal A_i\otimes \mathrm{id}_R|_{1\to 1}\). 
This constant depends only on the fixed physical representation \(V_{\omega_r}\), and not on \(n\).

Therefore,
\[
F_\rho(\widehat{\mathcal{N} \circ \mathcal{E}},\Lambda_0)
\geq
1-\frac{C_{\omega_r}}{n}
\qquad
\forall \rho ,
\]
and hence
\[
F(\widehat{\mathcal{N} \circ \mathcal{E}},\Lambda_0)
=
\min_\rho F_\rho(\widehat{\mathcal{N} \circ \mathcal{E}},\Lambda_0)
\geq
1-\frac{C_{\omega_r}}{n}.
\]
By the same arguments as in the \(SU(2)\) case and from Proposition~\ref{prop:SUdcovfidelity}, there exists a recovery map \(\mathcal R\) such that
\[
d(\mathcal R\mathcal N\mathcal E,\mathrm{id})
\leq
\sqrt{\frac{2C_{\omega_r}}{n}}.
\]
Thus our encoding \(\mathcal E\) is \(\varepsilon\)-approximate with \(\varepsilon=O_{d}(n^{-1/2})\).

\end{proof}

\begin{smlemma}\label{lem:arb_noise_from_leakage}
    Let $\mathcal E(\rho_L)=\mathcal{V}\rho_L\mathcal{V}^\dagger$ be an isometric encoding of $\mathcal H_L$ into $\mathcal H_P=A_1\otimes\cdots\otimes A_n$, 
    and let $\mathcal S$ be a family of subsets of $\{1,\ldots,n\}$. 
    Suppose that for every $S\in\mathcal S$ the reduced state~\eqref{eq:reduced_state} splits as
    \begin{equation}
    \label{eq:leakage-splitting}
        \rho^{(S)}(\rho_L)=\tau_S+\Delta_S(\rho_L)
        \qquad\text{for every state }\rho_L\text{ on }\mathcal H_L,
    \end{equation}
    where $\tau_S$ is a state on $A_S$ that does not depend on $\rho_L$, so that the linear map $\Delta_S(X):=\rho^{(S)}(X)-\Tr(X)\,\tau_S$ carries all of the logical dependence. 
    Let $R\cong\mathcal H_L$ be a reference and put
    \begin{equation}
    \label{eq:uniform-leakage}
        \delta
        :=
        \sup_{S\in\mathcal S}\;
        \sup_{\rho}\;
        \left\|
            (\Delta_S\otimes\mathrm{id}_R)
            \left(\ket{\psi_\rho}\bra{\psi_\rho}\right)
        \right\|_1,
    \end{equation}
    the inner supremum running over states $\rho$ on $\mathcal H_L$ and their purifications $\ket{\psi_\rho}$. 
    Then, for every flagged noise channel of the form~\eqref{eq:arbitrary_multiqudit_noise} supported on $\mathcal S$,
    \[
        \mathcal N(\sigma)
        =
        \sum_{S\in\mathcal S}
        p_S
        \ket S\bra S_F
        \otimes
        \left(\mathcal N_S\otimes\mathrm{id}_{\hat S}\right)(\sigma),
    \]
    there exists a recovery channel $\mathcal R$ with
    \begin{equation}
    \label{eq:arb-noise-bound}
        d(\mathcal R\circ\mathcal N\circ\mathcal E,\mathrm{id})
        \le
        \sqrt{\frac\delta2}.
    \end{equation}
\end{smlemma}

\begin{proof}
    By Lemma~\ref{lem:dual_to_arb} and the splitting~\eqref{eq:leakage-splitting}, the complementary channel is
    \[
        \widehat{\mathcal N\circ\mathcal E}(\rho_L)
        =
        \sum_{S\in\mathcal S}
        p_S
        \ket S\bra S_{F_E}
        \otimes
        \left[
            \widehat{\mathcal N}_S(\tau_S)
            +
            \widehat{\mathcal N}_S(\Delta_S(\rho_L))
        \right].
    \]
    Only the second summand depends on the logical input, so we take
    \[
        \Lambda_0(\rho_L)
        =
        \operatorname{Tr}(\rho_L)
        \sum_{S\in\mathcal S}
        p_S
        \ket S\bra S_{F_E}
        \otimes
        \widehat{\mathcal N}_S(\tau_S),
    \]
    which is a constant channel in the sense of Eq.~\eqref{eq:constant_channel_def}, because each $\widehat{\mathcal N}_S(\tau_S)$ is a state independent of $\rho_L$. 
    Now we estimate $F(\widehat{\mathcal N\circ\mathcal E},\Lambda_0)$.  We compute
    \[
        F_\rho(\widehat{\mathcal N\circ\mathcal E},\Lambda_0)
        =
        f(\sigma_\rho,\eta_\rho),
        \quad
        \left\{
        \begin{aligned}
        \sigma_\rho
        &:=
        (\widehat{\mathcal N\circ\mathcal E}\otimes \mathrm{id}_R)
        \left(\ket{\psi_\rho}\bra{\psi_\rho}\right),\\
        \eta_\rho
        &:=
        (\Lambda_0\otimes \mathrm{id}_R)
        \left(\ket{\psi_\rho}\bra{\psi_\rho}\right).
        \end{aligned}
        \right.
    \]
    By the Fuchs--van de Graaf inequality,
    \[
    \begin{aligned}
    1-f(\sigma_\rho,\eta_\rho)
    &\leq
    \frac{1}{2}
    \left\|
        \sigma_\rho-\eta_\rho
    \right\|_1
    \\
    &=
    \frac{1}{2}
    \sum_{S\in\mathcal S}
    p_S
    \left\|
    (\widehat{\mathcal N}_S\otimes \mathrm{id}_R)
    \left[
        (\Delta_S\otimes \mathrm{id}_R)
        \left(\ket{\psi_\rho}\bra{\psi_\rho}\right)
    \right]
    \right\|_1
    \\
    &\leq
    \frac{1}{2}
    \sum_{S\in\mathcal S}
    p_S
    \left\|
        (\Delta_S\otimes \mathrm{id}_R)
        \left(\ket{\psi_\rho}\bra{\psi_\rho}\right)
    \right\|_1
    \\
    &\leq
    \frac{\delta}{2}.
    \end{aligned}
    \]
    The second line uses that the flag states $\ket S_{F_E}$ are orthogonal, so the trace norm splits as the $p$-weighted sum over $S$; 
    the third uses contractivity of the trace norm under CPTP maps; the fourth uses the definition~\eqref{eq:uniform-leakage} together with $\sum_Sp_S=1$.

    Therefore $F_\rho(\widehat{\mathcal N\circ\mathcal E},\Lambda_0)\ge1-\delta/2$ for every $\rho$, and taking the minimum over $\rho$,
    \[
        F(\widehat{\mathcal N\circ\mathcal E},\Lambda_0)
        =
        \min_\rho
        F_\rho(\widehat{\mathcal N\circ\mathcal E},\Lambda_0)
        \geq
        1-\frac{\delta}{2}.
    \]
    By the complementary-channel theorem, Eq.~\eqref{eq:beny}, and the definition of the channel distance in Eq.~\eqref{eq:distance_def}, 
    there exists a recovery map $\mathcal R$ with
    \[
        d(\mathcal R\circ\mathcal N\circ\mathcal E,\mathrm{id})
        \le
        d\left(\widehat{\mathcal N\circ\mathcal E},\Lambda_0\right)
        =
        \sqrt{1-F\left(\widehat{\mathcal N\circ\mathcal E},\Lambda_0\right)}
        \le
        \sqrt{\frac\delta2},
    \]
    which is~\eqref{eq:arb-noise-bound}.
\end{proof}

\begin{smtheorem}\label{thrm:3_transitive_scale}
    Let
    $
        \mathcal H_L=V_{\omega_1},\;
        \mathcal H_P=(V_{\omega_1})^{\otimes n}
    $
    with $d=p^r$ where $p$ is prime and $r$ is a positive integer, 
    and $n$ is such that $n-1=p^s$, where $s \geq r$ is a positive integer. 
    Equip \(\mathcal H_L\) with the fundamental representation of
    \(\mathfrak{su}(d)\), and equip \(\mathcal H_P\) with the transversal
    representation given by
    \begin{equation}
    \begin{gathered}
    t_a \rightarrow \sum_{i=1}^n t_a^{(i)}, \quad a=1,\ldots,d^2-1,
    \end{gathered}
    \end{equation} 
    where $t_a^{(i)}$ is a fundamental of $t_a$ on $i$-th physical space $V_{\omega_1}$. For sufficiently
    large $n$ there exists an $\mathfrak{su}(d)$-covariant encoding $\mathcal{E}$ with respect to these representations, 
    such that the code space is invariant under action of $\mathrm{PGL}(2, n-1)$ and such encoding $\mathcal{E}$ is $O_{d}\left(\frac{1}{\sqrt{n}}\right)$-approximate 
    against arbitrary flagged noise supported on sets of at most three sites, $\mathcal{N}(\sigma)=\sum_{1\leq|S|\leq3} p_S|S\rangle\langle S|_F \otimes\left(\mathcal{N}_S \otimes \mathrm{id}_{\hat S}\right)(\sigma)$, 
    with arbitrary channels $\mathcal N_S$ on $A_S$.
\end{smtheorem}

\begin{proof}
    Take $\mathcal E$ to be the $\mathrm{PGL}(2,n-1)$-invariant encoding of Theorem~\ref{thrm:3_transitive_covariant_code} and let $\mathcal S$ 
    be the family of all nonempty subsets of $\{1,\ldots,n\}$ of size at most three. We first treat the three-element sets. 
    In Proposition~\ref{prop:3site-fundamental} we obtained the splitting~\eqref{eq:leakage-splitting} on any three physical qudits \(S=\{i,j,k\}\):
    \[
        \rho^{(S)}(\rho_L)
        =
        \tau_S+\Delta_S(\rho_L),
    \]
    where the input-independent part is
    \[
    \begin{aligned}
        \tau_S
        &:=
        \frac{1}{d^3}\mathbf 1
        -\frac{2}{d^2 n}
        \sum_{m=1}^{d^2-1}
        \left(
            t_m^{(i)}t_m^{(j)}
            +
            t_m^{(i)}t_m^{(k)}
            +
            t_m^{(j)}t_m^{(k)}
        \right)
        \\
        &\quad
        +
        \frac{4}{d n(n-2)}
        \sum_{m,l,p=1}^{d^2-1}
        d_{mlp}
        t_m^{(i)}t_l^{(j)}t_p^{(k)},
    \end{aligned}
    \]
    and the logical-state-dependent part is
    \[
    \begin{aligned}
        \Delta_S(\rho_L)
        &:=
        \sum_b r_b
        \Bigg[
        \frac{1}{n d^2}
        \left(
            t_b^{(i)}+t_b^{(j)}+t_b^{(k)}
        \right)
        \\
        &\qquad\qquad
        -
        \frac{2}{d n(n-2)}
        \sum_{m=1}^{d^2-1}
        \left(
            t_b^{(i)}t_m^{(j)}t_m^{(k)}
            +
            t_m^{(i)}t_b^{(j)}t_m^{(k)}
            +
            t_m^{(i)}t_m^{(j)}t_b^{(k)}
        \right)
        \Bigg].
    \end{aligned}
    \]
    Since $\Delta_S$ depends on $\rho_L$ only through its Bloch components $r_b$, which vanish at $\rho_L=I_d/d$, we have $\tau_S=\rho^{(S)}(I_d/d)$, a state as required by Lemma~\ref{lem:arb_noise_from_leakage}. 
    Every term of $\Delta_S$ carries an explicit factor $1/n$, and the $r_b$ are bounded, so \(\Delta_S=O_{d}(n^{-1})\) uniformly in \(S\). 
    Hence there is a constant \(C_d>0\), independent of \(n\), \(S\), and \(\rho\), with $\delta\le C_d/n$ for the quantity~\eqref{eq:uniform-leakage} restricted to three-element sets.

    Sets of one or two sites obey the same bound.  If $S'\in\mathcal S$ has $|S'|<3$, choose a three-element set $S\supseteq S'$, which is possible for $n\ge3$. 
    Then $\rho^{(S')}=\Tr_{S\setminus S'}\rho^{(S)}$, so the splitting~\eqref{eq:leakage-splitting} is inherited with $\tau_{S'}=\Tr_{S\setminus S'}\tau_S$, again a state independent of $\rho_L$, and $\Delta_{S'}=\Tr_{S\setminus S'}\circ\,\Delta_S$.
    The partial trace does not increase the trace norm, so each of these sets contributes at most what $S$ contributes to~\eqref{eq:uniform-leakage}, and $\delta\le C_d/n$ over all of $\mathcal S$.

    Lemma~\ref{lem:arb_noise_from_leakage} then gives a recovery map \(\mathcal R\) with
    \[
        d(\mathcal R\circ\mathcal N\circ\mathcal E,\mathrm{id})
        \leq
        \sqrt{\frac{C_d}{2n}},
    \]
    so our encoding $\mathcal{E}$ is $\varepsilon$-approximate with $\varepsilon=O_{d}\left(\frac{1}{\sqrt{n}}\right)$.
\end{proof}

The same lemma applies to an arbitrary fixed number of noisy sites, once the uniform leakage bound~\eqref{eq:uniform-leakage} is supplied by the codes of Section~\ref{subsec:supple_multiqudit_erasure} 
instead of by exact transitivity.

\begin{smproposition}\label{prop:k-site-uniform-leakage}
    Fix $d\ge2$ and $k\ge1$, let $n=qd+1$, and let $\ket m\in M_{\omega_1}$ be a unit vector such that
    \begin{equation}
    \label{eq:diamond-hypothesis}
        \max_{|S|=k}
        \left\|\Phi_{S,m}-\overline\Phi_{k,n}\right\|_\diamond
        \le\epsilon .
    \end{equation}
    Let $\mathcal E_m(X)=\mathcal{V}_mX\mathcal{V}_m^\dagger$ and, for every $k$-element subset $S$, set
    \[
        \tau_S:=\lambda_{k,n},
        \qquad
        \Delta_S(\rho_L):=\Phi_{S,m}(\rho_L)-\lambda_{k,n},
    \]
    with $\lambda_{k,n}$ as in Eq.~\eqref{eq:lambda-kn}. Then the splitting~\eqref{eq:leakage-splitting} holds, and the uniform leakage~\eqref{eq:uniform-leakage} obeys
    \begin{equation}
    \label{eq:k-leakage-bound}
        \delta\le\epsilon+\frac{2k}n .
    \end{equation}
\end{smproposition}

\begin{proof}
    By Eq.~\eqref{eq:local-channel} and Eq.~\eqref{eq:reduced_state}, $\rho^{(S)}(\rho_L)=\Phi_{S,m}(\rho_L)$, so the splitting holds by construction, and $\lambda_{k,n}=\overline\Phi_{k,n}(\tau)$ is a state independent of $S$ and of $\rho_L$.  Since $\Tr\rho_L=1$ we have $\Lambda_{k,n}(\rho_L)=\lambda_{k,n}$, hence the decomposition
    \[
        \Delta_S
        =
        \left(\Phi_{S,m}-\overline\Phi_{k,n}\right)
        +
        \left(\overline\Phi_{k,n}-\Lambda_{k,n}\right)
        \qquad\text{on trace-one inputs.}
    \]
    For the first term, the diamond norm is by definition the completely bounded trace norm, so~\eqref{eq:diamond-hypothesis} gives
    \[
        \left\|
        \left(\left(\Phi_{S,m}-\overline\Phi_{k,n}\right)\otimes\mathrm{id}_R\right)
        \left(\ket{\psi_\rho}\bra{\psi_\rho}\right)
        \right\|_1
        \le\epsilon .
    \]
    For the second term, the exact leakage formula~\eqref{eq:exact-leakage} gives, after tensoring with the reference,
    \[
        \left(\left(\overline\Phi_{k,n}-\Lambda_{k,n}\right)\otimes\mathrm{id}_R\right)
        \left(\ket{\psi_\rho}\bra{\psi_\rho}\right)
        =
        \frac1n\sum_{j=1}^k
        \left(
            \psi_{jR}-\tau_j\otimes\rho_R
        \right)
        \otimes
        \beta_{k-1,N}^{(\widehat j)},
    \]
    where $\psi_{jR}$ denotes $\ket{\psi_\rho}\bra{\psi_\rho}$ with the logical factor carried by the $j$th retained site and $\rho_R:=\Tr_L\ket{\psi_\rho}\bra{\psi_\rho}$. 
    Each summand is a difference of two states tensored with the state $\beta_{k-1,N}^{(\widehat j)}$, so its trace norm is at most $2$, and the sum has trace norm at most $2k/n$. 
    The triangle inequality gives~\eqref{eq:k-leakage-bound}.
\end{proof}

\begin{smtheorem}\label{thrm:k_arb_scale}
    Fix $d\ge2$ and $k\ge1$, and let $n=qd+1$. 
    Let $\ket{m_s}$ be the multiplicity vector sampled by Theorem~\ref{thrm:randomized_multiplicity_vector} and let $\mathcal E_s(X)=\mathcal{V}_{m_s}X\mathcal{V}_{m_s}^\dagger$, 
    an exactly $SU(d)$-covariant encoding of $\mathcal H_L=V_{\omega_1}$ into $\mathcal H_P=(V_{\omega_1})^{\otimes n}$. 
    Then, with probability at least $1-e^{-\Omega_d(n)}$ over the seed $s$, the following holds: for every flagged noise channel acting on $k$ sites,
    \[
        \mathcal N(\sigma)
        =
        \sum_{|S|=k}
        p_S
        \ket S\bra S_F
        \otimes
        \left(\mathcal N_S\otimes\mathrm{id}_{\hat S}\right)(\sigma)
    \]
    with arbitrary channels $\mathcal N_S$ on $A_S$, there exists a recovery channel $\mathcal R$ such that
    \begin{equation}
    \label{eq:k-arb-bound}
        d(\mathcal R\circ\mathcal N\circ\mathcal E_s,\mathrm{id})
        \le
        \sqrt{\frac kn}
        +
        e^{-\Omega_d(n)}
        =
        O\!\left(\sqrt{\frac kn}\right).
    \end{equation}
    The same bound holds for a Haar-random $\ket m$, with probability at least $1-e^{-\Omega_d(n)}$ by Theorem~\ref{thm:haar-uniform}.
\end{smtheorem}

\begin{proof}
    Condition on the event of Theorem~\ref{thrm:randomized_multiplicity_vector}, which has probability at least $1-M_{n,k}D_q^{-1/2}=1-e^{-\Omega_d(n)}$ and on which the hypothesis~\eqref{eq:diamond-hypothesis} holds with $\epsilon=r_k^{3/2}D_q^{-1/4}$. 
    Proposition~\ref{prop:k-site-uniform-leakage} then bounds the uniform leakage by $\delta\le\epsilon+2k/n$, and Lemma~\ref{lem:arb_noise_from_leakage}, applied with $\mathcal S$ the family of all $k$-element subsets of $\{1,\ldots,n\}$, gives a recovery channel $\mathcal R$ with
    \[
        d(\mathcal R\circ\mathcal N\circ\mathcal E_s,\mathrm{id})
        \le
        \sqrt{\frac\delta2}
        \le
        \sqrt{\frac kn}
        +
        \sqrt{\frac\epsilon2},
    \]
    using $\sqrt{a+b}\le\sqrt a+\sqrt b$.  The last term is $e^{-\Omega_d(n)}$, because $r_k$ is fixed while $D_q$ grows exponentially in $n$ by Lemma~\ref{lem:dimension_bound}. 
    For the final claim, repeat the argument conditioning on the event of Theorem~\ref{thm:haar-uniform} instead, which supplies the same $\epsilon$.
\end{proof}

\section{Covariant Analog Simulations}\label{sec:supple_analog_sim}

\subsection{Dynamical Lie Algebra and Symmetries}\label{subsec:supple_finite_group_symmetry_DLA}

\begin{smlemma}\label{lem:reductive_decomposition}
    Let $G$ be a finite group, and let $\mathcal{H}$ be a unitary $G$-representation. If $\mathcal{H} \cong \bigoplus_{\alpha\in \mathcal{I}_{\mathcal{H}}}\left(V_\alpha \otimes M_{\alpha}\right)$, where $\mathcal{I}_{\mathcal{H}}$ 
    is the set of irreducible representations of $G$ present in $\mathcal{H}$, then
    $$
    \mathrm{u}(\mathcal{H})^G \cong \bigoplus_{\alpha\in \mathcal{I}_{\mathcal{H}}} \mathrm{u}\left(M_{\alpha}\right),\quad \mathfrak{s} \mathfrak{u}(\mathcal{H})^G \cong\left(\bigoplus_{\alpha\in \mathcal{I}_{\mathcal{H}}} \mathfrak{s} \mathfrak{u}\left(M_{\alpha}\right)\right) \oplus \mathfrak{u}(1)^{\oplus(|\mathcal{I}_{\mathcal{H}}|-1)} .
    $$
\end{smlemma}

\begin{proof}
    From double centralizer theorem (see, for example,~\cite{GoodmanWallach2009}) we know, that every element from commutant of $G$, meaning $\operatorname{End}_{\mathbb{C}}(\mathcal{H})^G$, acts non-trivially only on multiplicity spaces $M_{\alpha}$, so we have
    \[
        \operatorname{End}_{\mathbb{C}}(\mathcal{H})^G \cong \bigoplus_{\alpha \in \mathcal{I}_{\mathcal{H}}} \operatorname{End}\left(M_{\alpha}\right).
    \]
    For arbitrary $X \in \operatorname{End}_{\mathbb{C}}(\mathcal{H})^G$ we can write it as 
    $
    X=\bigoplus_{\alpha \in \mathcal{I}_{\mathcal{H}}} X_\alpha,
    $
    where $X_\alpha \in \operatorname{End}\left(M_{\alpha}\right)$ and
    $
    X^{\dagger}=\bigoplus_{\alpha \in \mathcal{I}_{\mathcal{H}}} X_\alpha^{\dagger},
    $
    so if $X\in \mathfrak{u}(\mathcal{H})^G=\mathfrak{u}(\mathcal{H}) \cap \operatorname{End}_{\mathbb{C}}(\mathcal{H})^G$, then $X=-X^{\dagger}$, meaning that $X_\alpha=-X_\alpha^{\dagger}$, so $X_\alpha \in \mathfrak{u}\left(M_{\alpha}\right)$, and we have
    \[
        \mathfrak{u}(\mathcal{H})^G \cong \bigoplus_{\alpha \in \mathcal{I}_{\mathcal{H}}} \mathfrak{u}\left(M_{\alpha}\right).
    \]
    Now, since $\mathfrak{s} \mathfrak{u}(\mathcal{H})^G=\mathfrak{s} \mathfrak{u}(\mathcal{H}) \cap \mathfrak{u}(\mathcal{H})^G$, for arbitrary $X \in \mathfrak{s} \mathfrak{u}(\mathcal{H})^G$ we have
    \[
        X=\bigoplus_{\alpha\in \mathcal{I}_{\mathcal{H}}}\left(I_{V_\alpha} \otimes A_\alpha\right), \quad A_\alpha \in \mathfrak{u}\left(M_{\alpha}\right)
    \]
    and 
    $
        \operatorname{tr}_{\mathcal{H}}(X)=\sum_{\alpha\in \mathcal{I}_{\mathcal{H}}} \operatorname{tr}\left(I_{V_\alpha}\right) \operatorname{tr}\left(A_\alpha\right)=\sum_{\alpha\in \mathcal{I}_{\mathcal{H}}} n_\alpha \operatorname{tr}\left(A_\alpha\right),
    $
    from which we get the desired isomorphism
    \[
        \mathfrak{s} \mathfrak{u}(\mathcal{H})^G \cong
        \left\{\left(A_1, \ldots, A_{|\mathcal{I}_{\mathcal{H}}|}\right) \in \bigoplus_{\alpha\in \mathcal{I}_{\mathcal{H}}} \mathfrak{u}\left(M_{\alpha}\right): \sum_{\alpha\in \mathcal{I}_{\mathcal{H}}} n_\alpha \operatorname{tr}\left(A_\alpha\right)=0\right\}\cong\left(\bigoplus_{\alpha\in \mathcal{I}_{\mathcal{H}}} \mathfrak{s} \mathfrak{u}\left(M_{\alpha}\right)\right) \oplus \mathfrak{u}(1)^{\oplus(|\mathcal{I}_{\mathcal{H}}|-1)},
    \]
    which completes our proof.
\end{proof}

\begin{smlemma}\label{lem:number_of_ireps}
Let \(S_N\) act on
\(
    \mathcal H=(\mathbb C^2)^{\otimes N}
\)
by permuting tensor factors. Then the distinct irreducible representations of
\(S_N\) appearing in \(\mathcal H\) are precisely those labeled by two-row
partitions
\[
    (N-k,k),
    \qquad
    k=0,1,\ldots,\left\lfloor\frac N2\right\rfloor .
\]
In particular, the number of distinct irreducible representations appearing in
\(\mathcal H\) is
\[
    |\mathcal I_{\mathcal H}|
    =
    \left\lfloor\frac N2\right\rfloor+1 .
\]
\end{smlemma}

\begin{proof}
This follows directly from Schur--Weyl duality for
\((\mathbb C^2)^{\otimes N}\). Since the local dimension is \(2\), only Young
diagrams with at most two rows appear. These are exactly the partitions
\((N-k,k)\) with \(0\leq k\leq \lfloor N/2\rfloor\).
\end{proof}

\begin{smlemma}\label{lem:multipl_dim}
The multiplicity of the \(S_N\)-irrep labelled by \((N-k,k)\) in
\((\mathbb C^2)^{\otimes N}\) is
\[
    m_k=N-2k+1,
    \qquad
    k=0,1,\ldots,\left\lfloor\frac N2\right\rfloor .
\]
Equivalently, the corresponding multiplicity space has dimension \(N-2k+1\).
\end{smlemma}

\begin{proof}
By Schur--Weyl duality,
\[
    (\mathbb C^2)^{\otimes N}
    \cong
    \bigoplus_{k=0}^{\lfloor N/2\rfloor}
    S_{(N-k,k)}(\C^2)
    \otimes
    [N-k, k].
\]
Therefore the multiplicity of \([N-k, k]\) is
\[
    \dim S_{(N-k,k)}(\C^2)=N-2k+1 .
\]
\end{proof}

The proof of the following lemma was obtained in~\cite{Albertini2018}, we reproduce it for the sake of consistency:
\begin{smlemma}\label{lem:Sn_full_dimension}
For the \(S_N\)-representation on \(\mathcal H=(\mathbb C^2)^{\otimes N}\), the real dimension of the Lie algebra is the following
\[
    \dim\left(\mathfrak{su}(m_1)\oplus\cdots\oplus
    \mathfrak{su}(m_{|\mathcal I_{\mathcal H}|})
    \oplus
    \mathfrak{u}(1)^{\oplus(|\mathcal I_{\mathcal H}|-1)}\right)=\binom{N+3}{3}-1.
\]
\end{smlemma}

\begin{proof}
Using the previous Lemma~\ref{lem:multipl_dim}, the multiplicities are
\[
    m_k=N-2k+1,
    \qquad
    k=0,1,\ldots,\left\lfloor\frac N2\right\rfloor .
\]
Therefore, using Lemma~\ref{lem:number_of_ireps} we get
\[
\begin{aligned}
    \dim_{\mathbb R}
    \left(
    \bigoplus_k \mathfrak{su}(m_k)
    \oplus
    \mathfrak{u}(1)^{\oplus(|\mathcal I_{\mathcal H}|-1)}
    \right)
    &=
    \sum_k (m_k^2-1)
    +
    \left(|\mathcal I_{\mathcal H}|-1\right)
    =
    \sum_k m_k^2-1 .
\end{aligned}
\]
Direct computations of the following sum
\[
    \sum_{k=0}^{\lfloor N/2\rfloor}(N-2k+1)^2
    =
    \frac{(N+1)(N+2)(N+3)}{6}
    =
    \binom{N+3}{3}.
\]
give desired dimension.
\end{proof}

\begin{smlemma}\label{lem:subgroup_su_inclusion}
Let \(G\subseteq S_N\) act on \(\mathcal H=(\mathbb C^2)^{\otimes N}\) by permuting tensor factors. Then
\[
    \mathfrak{su}(\mathcal H)^{S_N}
    \subseteq
    \mathfrak{su}(\mathcal H)^G .
\]
\end{smlemma}

\begin{proof}
By definition,
\[
    \mathfrak{su}(\mathcal H)^{S_N}
    =
    \{X\in\mathfrak{su}(\mathcal H): P_\sigma X P_\sigma^\dagger=X
    \ \text{for all } \sigma\in S_N\}.
\]
Since \(G\subseteq S_N\), any \(X\) invariant under every \(\sigma\in S_N\) is, in particular, invariant under every \(g\in G\). Therefore
\(
    X\in\mathfrak{su}(\mathcal H)^G,
\)
which proves the inclusion.
\end{proof}

\subsection{Covariant encodings for \texorpdfstring{%
$\mathfrak{su}(d_1)\oplus\cdots\oplus\mathfrak{su}(d_k)\oplus \mathfrak{u}(1)^{\oplus(r-1)}$%
}{su(d1) plus ... plus su(dk) plus u(1)(r-1)}}
\label{subsec:supple_cov_block_enc}

Importantly, we have to understand how to obtain this block-diagonal form in practice, because to perform our encoding, we need to know the block-diagonal form of each Hamiltonian we are interested in.
Below we provide an example of how to do it for $N=3$ qubits, where we use the method of Young symmetrizers proposed in~\cite{DAlessandroHartwig2021}.

\begin{smexample}\label{ex:Sn_block_diagonalization}
Basis, that respects the decomposition on isotypic components of the representation of $S_3$ 
\[
    \left(\mathbb{C}^2\right)^{\otimes 3} \cong V_{(3)} \otimes \mathbb{C}^4 \oplus V_{{(2,1)}} \otimes \mathbb{C}^2
\]
is defined as follows
\[
    \begin{gathered}
    \left|\frac{3}{2}, \frac{3}{2}\right\rangle=|000\rangle, \\
    \left|\frac{3}{2}, \frac{1}{2}\right\rangle=\frac{|100\rangle+|010\rangle+|001\rangle}{\sqrt{3}}, \\
    \left|\frac{3}{2}, -\frac{1}{2}\right\rangle=\frac{|011\rangle+|101\rangle+|110\rangle}{\sqrt{3}}, \\
    \left|\frac{3}{2}, -\frac{3}{2}\right\rangle=|111\rangle, \\
    \end{gathered}\quad
    \begin{gathered}
    \left|\frac{1}{2}, \frac{1}{2}\right\rangle_a=\frac{|100\rangle+|010\rangle-2|001\rangle}{\sqrt{6}}, \\
    \left|\frac{1}{2}, -\frac{1}{2}\right\rangle_a=\frac{2|110\rangle-|101\rangle-|011\rangle}{\sqrt{6}}, \\
    \left|\frac{1}{2}, \frac{1}{2}\right\rangle_b=\frac{|100\rangle-|010\rangle}{\sqrt{2}}, \\
    \left|\frac{1}{2}, -\frac{1}{2}\right\rangle_b=\frac{|101\rangle-|011\rangle}{\sqrt{2}} .
    \end{gathered}
\]
from which we can construct the unitary $U$ that transforms the standard basis to this basis
\begin{equation}\label{eq:basis_transform_example}
    U=\left(\begin{array}{cccccccc}
    1 & 0 & 0 & 0 & 0 & 0 & 0 & 0 \\
    0 & \frac{1}{\sqrt{3}} & 0 & 0 & -\frac{2}{\sqrt{6}} & 0 & 0 & 0 \\
    0 & \frac{1}{\sqrt{3}} & 0 & 0 & \frac{1}{\sqrt{6}} & 0 & -\frac{1}{\sqrt{2}} & 0 \\
    0 & 0 & \frac{1}{\sqrt{3}} & 0 & 0 & -\frac{1}{\sqrt{6}} & 0 & -\frac{1}{\sqrt{2}} \\
    0 & \frac{1}{\sqrt{3}} & 0 & 0 & \frac{1}{\sqrt{6}} & 0 & \frac{1}{\sqrt{2}} & 0 \\
    0 & 0 & \frac{1}{\sqrt{3}} & 0 & 0 & -\frac{1}{\sqrt{6}} & 0 & \frac{1}{\sqrt{2}} \\
    0 & 0 & \frac{1}{\sqrt{3}} & 0 & 0 & \frac{2}{\sqrt{6}} & 0 & 0 \\
    0 & 0 & 0 & 1 & 0 & 0 & 0 & 0
    \end{array}\right)
\end{equation}
which transforms the Hamiltonians of interest to the following block-diagonal form
\[
    U^{\dagger} H_\alpha U=H_\alpha^{(4)} \oplus H_\alpha^{(2)} \oplus H_\alpha^{(2)}, \quad \alpha \in\{x, y, z z\}
\]
precisely, we have
\[
\begin{gathered}
-i H_x=\left(\begin{array}{cccc}
0 & -i \sqrt{3} & 0 & 0 \\
-i \sqrt{3} & 0 & -2 i & 0 \\
0 & -2 i & 0 & -i \sqrt{3} \\
0 & 0 & -i \sqrt{3} & 0
\end{array}\right) \oplus\left(\begin{array}{cc}
0 & -i \\
-i & 0
\end{array}\right) \oplus\left(\begin{array}{cc}
0 & -i \\
-i & 0
\end{array}\right) . \\
-i H_y=\left(\begin{array}{cccc}
0 & \sqrt{3} & 0 & 0 \\
-\sqrt{3} & 0 & 2 & 0 \\
0 & -2 & 0 & \sqrt{3} \\
0 & 0 & -\sqrt{3} & 0
\end{array}\right) \oplus\left(\begin{array}{cc}
0 & 1 \\
-1 & 0
\end{array}\right) \oplus\left(\begin{array}{cc}
0 & 1 \\
-1 & 0
\end{array}\right) . \\
-i H_{z z}=\operatorname{diag}(-3 i, i, i,-3 i) \oplus i I_2 \oplus i I_2 .
\end{gathered}
\]

Notice, that so far it is a decomposition of the form
\[
    u(4) \oplus u(2),\text { where } \operatorname{diag}(A, B, B) \text { with } \operatorname{tr}(A)+2 \operatorname{tr}(B)=0.
\]
To get the desired decomposition of the form $\mathfrak{su}(4)\oplus \mathfrak{su}(2)\oplus \mathfrak{u}(1)$  we need to perform
\[
   \operatorname{diag}(A, B, B)=\operatorname{diag}(A+\mu i I_4, B-\mu i I_2, B-\mu i I_2) + \mu\operatorname{diag}\left(-i I_4, i I_2, i I_2\right),
\]
where $\mu=\frac{1}{2 i} \operatorname{tr}(B)$ and gives the following transformation $\mathfrak{L} \rightarrow \mathfrak{su}(4)\oplus \mathfrak{su}(2)\oplus \mathfrak{u}(1)$ is 
\[
\begin{gathered}
    -i H_x \longrightarrow \left(-2 i J_x^{(3 / 2)}, -i \sigma_x, 0 \right)\equiv \left(A_x, B_x, C_x \right) \\
    -i H_y \longrightarrow \left(-2 i J_y^{(3 / 2)}, -i \sigma_y, 0 \right)\equiv \left(A_y, B_y, C_y \right) \\
    -i H_{z z} \longrightarrow \left(-2 i J_z^2+\frac{5 i}{2} I_4, 0, 1\right)\equiv \left(A_{zz}, B_{zz}, C_{zz} \right)
\end{gathered}
\]
(we emphasize that although we use spin operators to represent first component operators, these operators have no relation to spin-representation -- they lie in $\mathfrak{su}(4)$ Lie algebra,
in fact $\Lie_{\R} \{-2 i J_x^{(3 / 2)}, -2 i J_y^{(3 / 2)}, -2 i J_z^2+\frac{5 i}{2} I_4\}=\mathfrak{su}(4)$, but we will elaborate on it later).
Our covariant encoding in this case is given by
\[
    -iH_k \longrightarrow \sum_{i=1}^{n_1} A_k^{(i)} \oplus \sum_{i=1}^{n_2} B_k^{(i+n_1)} \oplus C_k, \quad k\in\{x,y,zz\}
\]
which is explicitly
\[
\begin{gathered}
    -i H_x \longrightarrow \sum_{k=1}^{n_1} \left(-2 i J_x^{(3 / 2)}\right)^{(k)} \oplus \sum_{k=1}^{n_2} \left( -i \sigma_x\right)^{(k+n_1)}\oplus 0,\\
    -i H_y \longrightarrow \sum_{k=1}^{n_1} \left(-2 i J_y^{(3 / 2)}\right)^{(k)} \oplus \sum_{k=1}^{n_2} \left( -i \sigma_y\right)^{(k+n_1)}\oplus 0,\\
    -i H_{z z} \longrightarrow \sum_{k=1}^{n_1} \left(-2 i J_z^2+\frac{5 i}{2} I_4\right)^{(k)} \oplus 0 \oplus C_{zz},\\
\end{gathered}
\]
where we do not elaborate on the nature of $\mathfrak{u}(1)$ representation $C_{k}$, for the reasons explained in the main text of the paper.
\end{smexample}

\subsection{Correlated flagged two-qudit erasure}\label{subsec:supple_block_enc_noise}

\begin{smproposition}\label{prop:two_block_fidelity}
    Let
    $
        \mathcal H_L=V_{d_1, \omega_1}\otimes V_{d_2,\omega_1},\;
        \mathcal H_P=(V_{d_1,\omega_1})^{\otimes n_1}\otimes (V_{d_2,\omega_1})^{\otimes n_2}
    $
    with parameters $n_i=p_i^{k_i}$, where $p_i$ is a prime number and $k_i$ is a positive integer, s. t. $n_i\equiv 1 \pmod{d_i}$ for $i\in\{1,2\}$.
    Equip \(\mathcal H_L\) with the fundamental representation of
    \(\mathfrak{su}(d)\), and equip \(\mathcal H_P\) with the transversal
    representation given by
    \begin{equation}
    \begin{gathered}
    t_a \rightarrow \sum_{i=1}^{n_1} t_a^{(i)}, \quad s_a \rightarrow \sum_{i=1}^{n_2} s_a^{(i)}
    \end{gathered}
    \end{equation} 
    where $t_a^{(i)}$ and $s_a^{(i)}$ are fundamental representations on the $i$-th qudit of the generators $t_a$ and $s_a$ of $\mathfrak{su}(d_1)$ and $\mathfrak{su}(d_2)$, correspondingly. 
    If encoding $\mathcal{E}$, which is $\mathfrak{su}(d_1)\oplus\mathfrak{su}(d_2)$-covariant with respect to defined physical and logical representations, defines a code space where each block is invariant under action of $\mathrm{AGL}(1, n_i)$ (for the definition see Eq.~\eqref{eq:AGL}),
then under erasure noise on two sites $\mathcal{N}(\sigma)=\sum_{1 \leq i<j \leq n} q_{i j} \left|i, j\right\rangle\left\langle i, j\right|_F \otimes \left| e\right\rangle\left\langle e\right|_{i, j} \otimes \operatorname{Tr}_{i j}(\sigma)$ the entanglement fidelity has the following form
\[
    F\left(\widehat{\mathcal{N} \circ \mathcal{E}}, \Lambda_0\right)=1-\frac{p_1(d_1^2-1)}{4n_1^2}-\frac{p_2(d_2^2-1)}{4n_2^2}-p_{12}\left(\frac{d_1^2-1}{8n_1^2}+\frac{d_2^2-1}{8n_2^2}\right)+O_{d_1,d_2}\!\left(n_1^{-2}n_2^{-2}+n_1^{-3}+n_2^{-3}\right).
\]
where $p_1, p_2$ are combinations of $\{q_{ij}\}$ and $p_{12}:=1-p_1-p_2$ and we take $\Lambda_0$ as follows
\begin{equation}\label{eq:Lambda0_block}
    \Lambda_0(\rho)=\operatorname{tr}(\rho) \left(p_1\omega_{1, \text{in}} \otimes \tau_1+p_2\omega_{2, \text{in}} \otimes \tau_2+p_{12}\omega_{\text{inter}} \otimes \tau_{12}\right),
\end{equation}
where $\omega_{1, \text{in}}, \omega_{2, \text{in}}$ and $\omega_{\text{inter}}$ are some orthogonal flag states and
\begin{equation}\label{eq:omega_block}
    \tau_1:=\frac{I}{d_1^2}-\frac{2}{d_1 n_1} \sum_m t_m^{(i)} t_m^{(j)}, \quad \tau_2:=\frac{I}{d_2^2}-\frac{2}{d_2 n_2} \sum_m s_m^{(i)} s_m^{(j)}, \quad \tau_{12}:=\frac{I}{d_1 d_2}.
\end{equation}
\end{smproposition}

\begin{proof}
Instead of the generalized Bloch parameterization of logical states, we will use the following representation, that respects block structure of our problem
\[
    \rho_L=\frac{I}{d_1 d_2}+\frac{1}{d_2} \sum_a r_{1, a} \bar{t}_a \otimes I+\frac{1}{d_1} \sum_b r_{2, b} I \otimes \bar{s}_b+\sum_{a, b} C_{a b} \bar{t}_a \otimes \bar{s}_b,
\]
where $\bar{t}_a$ and $\bar{s}_b$ are generators of $\mathfrak{su}(d_1)$ and $\mathfrak{su}(d_2)$ correspondingly, $r_{1, a}$ and $r_{2, b}$ are components of the Bloch vectors of reduced states on each block, and $C_{a b}$ is a correlation matrix.

There are basically three types of reduced states that the environment "can see": in-block reduced two-qudit states (for each block different)
\[
\begin{gathered}
    \rho_{1,\text {in }}^{(ij)}(\rho_L)=\frac{1}{d_1^2} \mathbf{1}-\frac{2}{d_1 n_1} \sum_{m=1}^{d_1^2-1} t_m^{(i)} t_m^{(j)}+\frac{1}{n_1 d_1} \sum_{b=1}^{d_1^2-1} r_{1, b}\left(t_b^{(i)}+t_b^{(j)}\right),\\
    \rho_{2,\text {in }}^{(ij)}(\rho_L)=\frac{1}{d_2^2} \mathbf{1}-\frac{2}{d_2 n_2} \sum_{m=1}^{d_2^2-1} t_m^{(i)} t_m^{(j)}+\frac{1}{n_2 d_2} \sum_{b=1}^{d_2^2-1} r_{2, b}\left(t_b^{(i)}+t_b^{(j)}\right),
\end{gathered}
\]
which we obtained by tracing out one qudit from the three-qudit reduced state $\rho^{(ijk)}(\rho_L)$ obtained in Proposition~\ref{prop:3site-fundamental}. As for the inter-block reduced two-qudit state, we have
\[
    \rho_{\text {inter }}^{(ij)}\left(\rho_L\right)=\left(\rho_1^{(i)} \otimes \rho_2^{(j)}\right)\left(\rho_L\right),
\]
where $\rho_1^{(i)}$ and $\rho_2^{(j)}$ are single qudit reduced states channels of the $i$-th qudit of the first block and the $j$-th qudit of the second block correspondingly. Using formulas for single-qudit reduced states obtained in Proposition~\ref{prop:SUdcov}, namely 
\[
    \rho_k^{(i)}\left(\frac{I_{d_k}}{d_k}\right)=\frac{I_{d_k}}{d_k}, \quad \rho_k^{(i)}\left(\bar{t}_a\right)=\frac{1}{n_k} t_a^{(i)}, \quad k\in\{1,2\},
\]
we can obtain the following formula for the inter-block reduced state
\[
    \rho_{\text {inter }}^{(ij)}(\rho_L)=\frac{I}{d_1 d_2}+\frac{1}{n_1 d_2} \sum_a r_{1, a} t_a^{(i)} \otimes I+\frac{1}{n_2 d_1} \sum_b r_{2, b} I \otimes s_b^{(j)}+\frac{1}{n_1 n_2} \sum_{a, b} C_{a b} t_a^{(i)} \otimes s_b^{(j)}.
\]
(We further omit indices $i$ and $j$, since our one-block code spaces are symmetric). In this notation, the dual channel becomes
\[
    \widehat{\mathcal{N} \circ \mathcal{E}}(\rho)=p_1\;\omega_{1, \text{in}}\otimes \rho_{1,\text{in}} + p_2\;\omega_{2, \text {in}}\otimes \rho_{2,\text{in}} + (1-p_1-p_2)\;\omega_{\text {inter }} \otimes \rho_{\text{inter}},
\]
where $\omega_{1, \text{in}}, \omega_{2, \text{in}}$ and $\omega_{\text{inter}}$ are some orthogonal flag states. We define $\Lambda_0$ as in Eq.~\eqref{eq:Lambda0_block}.

From Lemma~\ref{fidelity_opt} we know that the optimum of fidelity is achieved on a symmetric maximally mixed logical input
$
    \frac{I_{d_1 d_2}}{d_1 d_2},
$
so our further computations will be conducted for this input. 
It is given by
$
    \left|\Phi\right\rangle=\frac{1}{\sqrt{D}} \sum_{x=1}^{d_1} \sum_{y=1}^{d_2}|x, y\rangle_L \otimes|x, y\rangle_R=\left|\Phi\right\rangle_{L_1 R_1} \otimes\left|\Phi\right\rangle_{L_2 R_2}.
$
If we define
\[
    \begin{gathered}
    \sigma_{1/D}:=\left(\widehat{\mathcal{N} \circ \mathcal{E}} \otimes \operatorname{id}_R\right)\left(\left|\Phi\right\rangle\left\langle\Phi\right|\right), \\
    \eta_{1/D}:=\left(\Lambda_0 \otimes \operatorname{id}_R\right)\left(\left|\Phi\right\rangle\left\langle\Phi\right|\right) .
    \end{gathered}
\]
then, because of the flagged erasure noise model, we have
\[
    \sigma_{1/D}=\bigoplus_{\alpha \in\{1,2,12\}} p_\alpha \sigma_\alpha, \quad \eta_{1/D}=\bigoplus_{\alpha \in\{1,2,12\}} p_\alpha \eta_\alpha,
\]
where 
\[
    \begin{aligned}
    & \sigma_1:=\left(\rho_{1, \text {in}} \otimes \mathrm{id}\right)\left(\left|\Phi\right\rangle\left\langle\Phi\right|\right), \\
    & \sigma_2:=\left(\rho_{2, \text {in}} \otimes \mathrm{id}\right)\left(\left|\Phi\right\rangle\left\langle\Phi\right|\right), \\
    & \sigma_{12}:=\left(\rho_{\text {inter }}\otimes \mathrm{id}\right)\left(\left|\Phi\right\rangle\left\langle\Phi\right|\right) .
    \end{aligned}
\]
and
\[
    \eta_1:=\tau_1 \otimes \frac{I_R}{D}, \quad \eta_2:=\tau_2 \otimes \frac{I_R}{D}, \quad \eta_{12}:=\tau_{12} \otimes \frac{I_R}{D}.
\]
(for the definition of $\tau_i$ see Eq.~\eqref{eq:omega_block}). Thus, our fidelity splits into three parts in a really convenient way
\[
    f\left(\sigma_{1/D}, \eta_{1/D}\right)=p_1 f\left(\sigma_1, \eta_1\right)+p_2 f\left(\sigma_2, \eta_2\right)+p_{12} f\left(\sigma_{12}, \eta_{12}\right).
\]
Now we are to compute each term separately. We use the following notation
$
    E_{r s}^{(i)}:=|r\rangle\langle s|, \quad F_{r s}^{(i)}:=E_{r s}^{(i)}-\delta_{r s} \frac{I_{d_i}}{d_i},
$ 
where $i$ is the block index, so the maximally mixed state on each block can be written as
$$
    |\Phi\rangle\langle\Phi|=\frac{1}{D} \sum_{r, s=1}^{d_1} \sum_{\mu, \nu=1}^{d_2} E_{r s}^{(1)} \otimes E_{\mu \nu}^{(2)} \otimes E_{r s}^{\left(R_1\right)} \otimes E_{\mu \nu}^{\left(R_2\right)}.
$$
Let's compute $f\left(\sigma_1, \eta_1\right)$ first. We have
$$
\rho_{1, \mathrm{in}}\left(E_{r s}^{(1)} \otimes E_{\mu \nu}^{(2)}\right)=\delta_{\mu \nu}\left[\delta_{r s} \tau_1+\frac{1}{n_1 d_1}\left(F_{r s}^{(i)}+F_{r s}^{(j)}\right)\right]
$$
from which we can obtain
$$
\sigma_1=\left(\tau_1 \otimes \frac{I_{R_1}}{d_1} \otimes \frac{I_{R_2}}{d_2}\right)+\left(\frac{1}{d_1 d_2} \sum_{r, s=1}^{d_1} \frac{1}{n_1 d_1}\left(F_{r s}^{(i)}+F_{r s}^{(j)}\right) \otimes E_{r s}^{\left(R_1\right)} \otimes I_{R_2}\right)\equiv \eta_1+\Delta_1.
$$
Notice, that $\operatorname{Tr} \Delta_1=0$ and $\eta_1=\frac{I}{d_1^3d_2}$. Therefore, we can use the same Taylor expansion arguments as in the proof of Proposition~\ref{prop:3site-fundamental} to obtain
$$
f\left(\sigma_1, \eta_1\right)=1-\frac{d_1^3 d_2}{8} \operatorname{Tr}\left(\Delta_1^2\right)+O_{d_1,d_2}\left(n_1^{-3}\right).
$$
One can easily check that $\operatorname{Tr}\left(\Delta_1^2\right)=\frac{2\left(d_1^2-1\right)}{d_1^3 d_2 n_1^2}$, so we have
$$
f\left(\sigma_1, \eta_1\right)=1-\frac{d_1^2-1}{4 n_1^2}+O_{d_1,d_2}\left(n_1^{-3}\right).
$$
By the same arguments, we can obtain
$$
f\left(\sigma_2, \eta_2\right)=1-\frac{d_2^2-1}{4 n_2^2}+O_{d_1,d_2}\left(n_2^{-3}\right).
$$
We compute the inter-block branch in the same fashion. We have
$$
\rho_{\text {inter }}\left(E_{r s}^{(1)} \otimes E_{\mu \nu}^{(2)}\right)=\left(\delta_{r s} \frac{I_{d_1}}{d_1}+\frac{1}{n_1} F_{r s}^{(i)}\right) \otimes\left(\delta_{\mu \nu} \frac{I_{d_2}}{d_2}+\frac{1}{n_2} F_{\mu \nu}^{(j)}\right)
$$
from which we can obtain
\[
\begin{aligned}
\sigma_{12}
=&
\frac{I}{d_1^2 d_2^2}
+
\frac{1}{n_1 d_1 d_2^2}
\left(
\sum_{r,s=1}^{d_1}
F_{rs}^{(i)}\otimes E_{rs}^{(R_1)}
\right)
\otimes I_{2R_2}
\\
&+
\frac{1}{n_2 d_1^2 d_2}
I_{1R_1}\otimes
\left(
\sum_{\mu,\nu=1}^{d_2}
F_{\mu\nu}^{(j)}\otimes E_{\mu\nu}^{(R_2)}
\right)
\\
&+
\frac{1}{n_1n_2d_1d_2}
\left(
\sum_{r,s=1}^{d_1}
F_{rs}^{(i)}\otimes E_{rs}^{(R_1)}
\right)
\otimes
\left(
\sum_{\mu,\nu=1}^{d_2}
F_{\mu\nu}^{(j)}\otimes E_{\mu\nu}^{(R_2)}
\right).
\end{aligned}
\]
Noticing, that $\operatorname{Tr}\left(\Delta_{12}\right)=0$ we get
\[
    f\left(\sigma_{12}, \eta_{12}\right)=1-\frac{d_1^2 d_2^2}{8} \operatorname{Tr}\left(\Delta_{12}^2\right)+O_{d_1,d_2}\left(\left\|\Delta_{12}\right\|^3\right).
\]
One can check that $\operatorname{Tr}\left(\Delta_{12}^2\right)=\frac{d_1^2-1}{d_1^2 d_2^2 n_1^2}+\frac{d_2^2-1}{d_1^2 d_2^2 n_2^2}+O_{d_1,d_2}\left(n_1^{-2} n_2^{-2}\right)$, so we have
\[
    f\left(\sigma_{12}, \eta_{12}\right)=1-\frac{d_1^2-1}{8 n_1^2}-\frac{d_2^2-1}{8 n_2^2}+O_{d_1,d_2}\left(n_1^{-2} n_2^{-2}+n_1^{-3}+n_2^{-3}\right).
\]
Combining all the branches together, we get the desired result.
\end{proof}

\section{Universal Analog Computations}\label{sec:supple_Universal_comp}

\subsection{Symmetry Breaking Hamiltonians as a resource in Fault-Tolerant Analog Computations}

\begin{smexample}\label{ex:s3_breaking_not_block_diagonal}
We now recall example~\ref{ex:Sn_block_diagonalization}, where we constructed $\mathfrak{su}(4)\oplus\mathfrak{su}(2)\oplus\mathfrak{u}(1)$-covariant encoding for $S_3$ symmetric system.
As it was pointed out in~\cite{Albertini2018}, in this case DLA spanned by $iH_x, iH_y, iH_{zz}$ actually coincides with invariant algebra
$\mathfrak{su}(4)\oplus\mathfrak{su}(2)\oplus \mathfrak{u}(1)$. Thus, if we find 
$\mathcal{G}^{\text{br}}$ such that the Lie algebra generated by $\mathfrak{su}(\HH)^G$ and $\mathcal{G}^{\text{br}}$ is universal, then 
we will know the universal generating set of Hamiltonians for this system. We will show in~\ref{sec:applications} that we can achieve universality by adding only one breaking Hamiltonian, namely $Z_1+X_2$.
So we conclude that for a system of three qubits with $S_3$ symmetry
\[
    \mathrm{Lie}_{\mathbb{R}}\{iH_x, iH_y, iH_{zz}, i(Z_1+X_2)\}\cong \mathfrak{su}(2^3).
\]
In the basis that respects $S_3$ isotopic decomposition, our breaking Hamiltonian has the following form
\[
    U^{\dagger}\left(Z_1+X_2\right) U=\left(\begin{array}{cccccccc}
    1 & \frac{1}{\sqrt{3}} & 0 & 0 & \frac{1}{\sqrt{6}} & 0 & -\frac{1}{\sqrt{2}} & 0 \\
    \frac{1}{\sqrt{3}} & \frac{1}{3} & \frac{2}{3} & 0 & -\frac{\sqrt{2}}{3} & \frac{1}{3 \sqrt{2}} & -\frac{\sqrt{6}}{3} & -\frac{1}{\sqrt{6}} \\
    0 & \frac{2}{3} & -\frac{1}{3} & \frac{1}{\sqrt{3}} & -\frac{1}{3 \sqrt{2}} & -\frac{\sqrt{2}}{3} & \frac{1}{\sqrt{6}} & -\frac{\sqrt{6}}{3} \\
    0 & 0 & \frac{1}{\sqrt{3}} & -1 & 0 & -\frac{1}{\sqrt{6}} & 0 & \frac{1}{\sqrt{2}} \\
    \frac{1}{\sqrt{6}} & -\frac{\sqrt{2}}{3} & -\frac{1}{3 \sqrt{2}} & 0 & \frac{2}{3} & \frac{2}{3} & -\frac{1}{\sqrt{3}} & \frac{1}{\sqrt{3}} \\
    0 & \frac{1}{3 \sqrt{2}} & -\frac{\sqrt{2}}{3} & -\frac{1}{\sqrt{6}} & \frac{2}{3} & -\frac{2}{3} & \frac{1}{\sqrt{3}} & \frac{1}{\sqrt{3}} \\
    -\frac{1}{\sqrt{2}} & -\frac{\sqrt{6}}{3} & \frac{1}{\sqrt{6}} & 0 & -\frac{1}{\sqrt{3}} & \frac{1}{\sqrt{3}} & 0 & 0 \\
    0 & -\frac{1}{\sqrt{6}} & -\frac{\sqrt{6}}{3} & \frac{1}{\sqrt{2}} & \frac{1}{\sqrt{3}} & \frac{1}{\sqrt{3}} & 0 & 0
    \end{array}\right),
\]
where $U$ is defined in Eq.~\eqref{eq:basis_transform_example}. Breaking Hamiltonian is obviously not block-diagonal, and thus it is not implemented transversally for the code space of the corresponding $\mathfrak{su}(4)\oplus\mathfrak{su}(2)\oplus \mathfrak{u}(1)$-covariant code.
The exact implementation of this Hamiltonian will depend on the physical system, because our encoding map is defined by Schur transform, which is a global unitary transformation and depends on
number of physical qubits, namely
\[
    H^{\text{br}}_P=\mathcal{V} U^{\dagger}\left(Z_1+X_2\right) U \mathcal{V}^{\dagger},
\] 
where $\mathcal{V}$ is the encoding map of the $\mathfrak{su}(4)\oplus\mathfrak{su}(2)\oplus \mathfrak{u}(1)$-code.
One should not confuse two different Schur transforms here: the one that defines the encoding map and the one that defines the basis in which we write symmetry-respected 
decomposition of logical Hilbert space.
\end{smexample}

\subsection{Sufficient condition for universality}

In this section we establish mathematical formalism for studying universality sufficient conditions,
which heavily relies on representation theory of symmetry group $G$. We remind that
definition of action and motivation for the symmetry group $G$ can be found in Section~V of the main text of the paper.

Let $\mathcal{I}$ be the set of irreducible representations (irreps) of $G$ appearing in $\HH$. We can decompose the Hilbert space as:
\begin{equation*}
    \HH=\bigoplus_{\alpha \in \mathcal{I}} \HH_\alpha, \quad \HH_\alpha \simeq V_\alpha \otimes M_\alpha,
\end{equation*}
where $V_\alpha$ is the vector space of the irrep $\alpha$ (dimension $d_\alpha$), and for computational purposes we will fix identification of the multiplicity space $M_\alpha \cong \mathbb{C}^{m_\alpha}$.

\begin{smlemma}
    Let $\HH=\bigoplus_{\alpha\in \mathcal{I}} \mathcal{H}_\alpha$ be the decomposition of the Hilbert space into isotypic components. As vector spaces, the endomorphism algebra decomposes as:
    \begin{equation}\label{eq:endomorphism_decomposition}
        \End_{\mathbb{C}}(\mathcal{H}) \cong \bigoplus_{\alpha, \beta \in \mathcal{I}} \mathcal{W}_{\beta\alpha},
    \end{equation}
    where the blocks are defined as $\mathcal{W}_{\beta\alpha} := \Hom_{\mathbb{C}}(\mathcal{H}_\alpha, \mathcal{H}_\beta)$.
\end{smlemma}

\begin{proof}
    We use the property that linear maps from a direct sum are determined by their components. Specifically, $\End(\mathcal{H}) \cong \mathcal{H} \otimes \mathcal{H}^*$. Substituting the decomposition of $\mathcal{H}$:
    \begin{align*}
        \End_{\mathbb{C}}(\mathcal{H}) &\cong \left( \bigoplus_{\alpha} \mathcal{H}_\alpha \right) \otimes \left( \bigoplus_{\beta} \mathcal{H}_\beta^* \right) \cong \bigoplus_{\alpha, \beta} \left( \mathcal{H}_\beta \otimes \mathcal{H}_\alpha^* \right) \cong \bigoplus_{\alpha, \beta} \Hom_{\mathbb{C}}(\mathcal{H}_\alpha, \mathcal{H}_\beta).
    \end{align*}
\end{proof}

\begin{smlemma}
    If $\HH\cong V_\alpha \otimes M_\alpha$, for every $\alpha \in \mathcal{I}$, then each block $\mathcal{W}_{\beta\alpha}:=\Hom_{\mathbb{C}}(\mathcal{H}_\alpha, \mathcal{H}_\beta)$ has the following tensor product structure: 
    \begin{equation*}
        \mathcal{W}_{\beta\alpha} \cong \Hom_{\mathbb{C}}(V_\alpha, V_\beta) \otimes \Hom_{\mathbb{C}}(M_\alpha, M_\beta).
    \end{equation*}
\end{smlemma}

\begin{proof}
    Consider the following series of basic isomorphisms:
    \begin{align*}
        \mathcal{H}_\beta \otimes \mathcal{H}_\alpha^* &\cong (V_\beta \otimes M_{\beta}) \otimes (V_\alpha \otimes M_{\alpha})^* \\
        &\cong (V_\beta \otimes M_{\beta}) \otimes (V_\alpha^* \otimes (M_{\alpha})^*) \\
        &\cong (V_\beta \otimes V_\alpha^*) \otimes (M_{\beta} \otimes (M_{\alpha})^*) \\
        &\cong \Hom_{\mathbb{C}}(V_\alpha, V_\beta) \otimes \Hom_{\mathbb{C}}(M_\alpha, M_\beta).
    \end{align*}
    This confirms the tensor product structure of the blocks.
\end{proof}

We define the two factors of this tensor product as:
\begin{equation}\label{eq:tensor_product_factors}
    \mathbf{T}_{\beta\alpha} := \Hom_\C(V_\alpha, V_\beta), \quad \mathbf{M}_{\beta\alpha} := \Hom_\C(M_\alpha, M_\beta).
\end{equation}
Thus, $\mathcal{W}_{\beta\alpha} \cong \mathbf{T}_{\beta\alpha} \otimes \mathbf{M}_{\beta\alpha}$. 
Notice that $\operatorname{End}_{\mathbb{C}}(\mathcal{H})^G$ is actually commutant of the representation of $G$ on $\HH$, so by the double-centralizer theorem (see, for example,~\cite{GoodmanWallach2009}) we have the following decomposition
\[
    \operatorname{End}_{\mathbb{C}}(\mathcal{H})^G \cong \bigoplus_{\alpha \in \mathcal{I}} \operatorname{id}_{V_\alpha} \otimes \operatorname{End}_{\mathbb{C}}\left(M_\alpha\right),
\]
which reflects the fact that the commutant acts only on multiplicity spaces and does not touch the irreps of $G$.

Even though the underlying objects that we usually work with in physics are real (e. g. Lie algebras $\mathfrak{su}(\HH), \mathfrak{u}(\HH)$) in the upcoming sections it will be more convenient to work with complexified objects -- we will 
treat them as representations of $G$ and representations over complex field are much easier to handle. In the end, we will show how to descend back to the real form and get the final result about the universality of dynamics.

\begin{smdefinition}\label{def:complexification}
    Let $\mathfrak{g}$ be a real Lie algebra. Its complexification $\mathfrak{g}_\mathbb{C} := \mathfrak{g} \otimes_{\mathbb{R}} \mathbb{C}$ is a complex Lie algebra obtained by allowing complex coefficients.
\end{smdefinition}

\begin{smlemma}\label{lem:complexification}
    The following isomorphism holds for the complexification of the unitary Lie algebra and the special unitary Lie algebra:
    \begin{equation*}
        \uu(\mathcal{H}) \otimes_{\mathbb{R}} \mathbb{C} \cong \gl(\mathcal{H}), \quad \su(\mathcal{H}) \otimes_{\mathbb{R}} \mathbb{C} \cong \slc(\mathcal{H}).
    \end{equation*}
\end{smlemma}

\begin{proof}
    Any operator $A \in \gl(\mathcal{H})$ can be uniquely decomposed into its skew-Hermitian parts (Cartesian decomposition):
    \[ A = X + iY, \quad \text{where } X = \frac{A - A^\dagger}{2}, Y = \frac{A + A^\dagger}{2i}. \]
    Since $X, Y \in \uu(\mathcal{H})$, the complex span of $\uu(\mathcal{H})$ covers all of $\gl(\mathcal{H})$. The dimensions match: $\dim_{\mathbb{R}}\uu(\mathcal{H}) = d^2 = \dim_{\mathbb{C}}\gl(\mathcal{H})$.
    Proof for $\su(\mathcal{H})$ instantly follows after restricting to traceless operators.
\end{proof}

We denote the complexification of the invariant Lie subalgebra $\mathfrak{su}(\HH)^G$ as
\[
    \mathfrak{l}_\C := \mathfrak{su}(\HH)^G \otimes_{\R} \C,
\]
and one can easily check from the definition of $\operatorname{End}_{\C}(\HH)^G$ and from the fact that $\mathfrak{su}(\HH) \otimes_{\R} \C \cong \mathfrak{sl}(\HH)$ from Lemma~\ref{lem:complexification} that
\begin{equation*}
    \mathfrak{l}_\C = \left(\operatorname{End}_{\C}(\HH)^G \cap \mathfrak{su}(\HH) \right) = \operatorname{End}_{\C}(\HH)^G\cap \mathfrak{sl}(\HH)\cong\left(\bigoplus_{\alpha \in \mathcal{I}} \mathrm{id}_{V_\alpha} \otimes \operatorname{End}_{ \mathbb{C}}\left(M_\alpha\right)\right) \cap \mathfrak{sl}(\HH),
\end{equation*}
so convenient, operator description would be
\begin{equation}\label{eq:compl_inv_alg}
    \mathfrak{l}_\C \cong \left\{\bigoplus_{\alpha \in \mathcal{I}} \mathrm{id}_{V_\alpha} \otimes A_\alpha: A_\alpha \in \operatorname{End}_\C(M_{\alpha}), \quad \sum_\alpha d_\alpha \operatorname{Tr}\left(A_\alpha\right)=0\right\}.
\end{equation}

The block decomposition
$
\text{End}_{\mathrm{C}}(\HH) \cong \bigoplus_{\alpha, \beta \in \mathcal{I}} \mathcal{W}_{\beta \alpha}
$
is useful because the $G$-invariant complexified Hamiltonians (denoted $\mathfrak{l}_\C$) are precisely those operators with no off-diagonal isotypic blocks and whose diagonal blocks are of the form
$$
\mathrm{id}_{V_\alpha} \otimes A_\alpha, \quad A_\alpha \in \operatorname{End}_{\mathrm{C}}\left(M_\alpha\right), \quad \sum_\alpha d_\alpha \operatorname{Tr}\left(A_\alpha\right)=0 .
$$
By contrast, the symmetry-breaking Hamiltonians may have nonzero components in off-diagonal blocks $\mathcal{W}_{\beta \alpha}$ with $\beta \neq \alpha$ -- in fact, since we are able to choose these Hamiltonians we would want them to have such off-diagonal components. 

The relevant action here is the adjoint action by commutators. 
Indeed, the dynamical Lie algebra is generated from the invariant controls $\mathfrak{su}(\HH)^G$ and the breakers
$\mathcal{G}^{\text{br}}$ by repeated commutators, so to understand what additional generators can be produced 
from a breaker component lying in a given block $\mathcal{W}_{\beta \alpha^{\prime}}$, one must understand how the invariant algebra 
acts on that block under $\operatorname{ad}_L(X)=[L, X]$. Since every element of $\mathfrak{l}_{\C}$ (complexification of $\mathfrak{su}(\HH)^G$) is block-diagonal and has the form $\mathrm{1}_{\gamma} \otimes A_\gamma$, 
its commutator action preserves each block $\mathcal{W}_{\beta \alpha} \cong \mathbf{T}_{\beta \alpha} \otimes \mathbf{M}_{\beta \alpha}$, 
acts trivially on the irrep-intertwining factor $\mathbf{T}_{\beta \alpha}$, and acts non-trivially on the multiplicity factor $\mathbf{M}_{\beta \alpha}$. Thus, although the set of breakers is finite and may contribute only a few components to a given off-diagonal block, 
the invariant algebra still supplies a large family of commutators inside that block. This makes it plausible that, once the breakers provide sufficiently many first-factor directions in $\mathbf{T}_{\beta \alpha}$, the adjoint action of $\mathfrak{l}_{\C}$ 
can generate the entire off-diagonal block $\mathcal{W}_{\beta \alpha}$ (see below for precise definitions and details).

The strategy of the proof is first to use the breakers to access suitable off-diagonal blocks, then to use the invariant algebra $\mathfrak{l}_\C$ to generate the full multiplicity factor $\mathbf{M}_{\beta \alpha}$ 
inside those blocks while preserving the available directions in the irrep-intertwining factor $\mathbf{T}_{\beta \alpha}$, and finally to propagate between blocks $\mathcal{W}_{\beta \alpha}$ by taking commutators, 
thereby recovering the full complex operator algebra $\mathfrak{sl}(\HH)$ (and then descending to the real form).

Now let's further formalize the above intuition and give precise definitions of the relevant objects.
\begin{smdefinition}\label{def:projection_onto_block}
    For $\alpha, \beta \in \mathcal{I}$, let $P_\alpha$ be the projector onto the isotypic component $\HH_\alpha$. The projection of an operator $X$ onto the block $\mathcal{W}_{\beta\alpha}$ is defined as:
    \begin{equation*}
        \Pi_{\beta\alpha}(X) := P_\beta X P_\alpha.
    \end{equation*}
\end{smdefinition}
\begin{smremark}
    On practice projection $P_{\alpha}$ can be computed by using the character table of the group $G$ and the formula
    \begin{equation*}\label{eq:projector_formula}
        P_\alpha=\frac{d_\alpha}{|G|} \sum_{g \in G} \chi_\alpha\left(g^{-1}\right) U(g),
    \end{equation*}
    where $\chi_\alpha$ is the character of the irrep $\alpha$. 
\end{smremark}

\begin{smdefinition}\label{def:coupling_graph}
    Let $\mathcal{G}^{\text{br}} \subset \mathfrak{su}(\HH)$ be the set of breakers. \textit{The Coupling graph} $\Gamma_{\mathrm{co}}$ has vertex set $\mathcal{I}$, and an (undirected) edge $\{\alpha, \beta\}$ is present if there exists $H \in \mathcal{G}^{\text{br}}$ such that
$
\Pi_{\beta \alpha}(H) \neq 0 .
$    
\end{smdefinition}

Basically, the coupling graph records which isotypic sectors are coupled by the symmetry-breaking Hamiltonians.

\begin{smremark}
    If $H$ is skew-Hermitian, then
    $
    \Pi_{\alpha \beta}(H)=P_\alpha H P_\beta=-\left(P_\beta H P_\alpha\right)^{\dagger}.
    $
    So
    $
    \Pi_{\beta \alpha}(H) \neq 0 \Longleftrightarrow \Pi_{\alpha \beta}(H) \neq 0 .
    $
    Hence the edge relation is symmetric, and the graph is naturally undirected.
\end{smremark}

\begin{smdefinition}\label{def:breaker_generated_subspace}
    Fix $\alpha, \beta \in \mathcal{I}$, and identification
    $
    \mathcal{W}_{\beta \alpha} \cong \mathbf{T}_{\beta \alpha} \otimes \mathbf{M}_{\beta \alpha}.
    $
    We call the following subspace of $\mathcal{W}_{\beta \alpha}$ 
    $$
    \mathcal{B}_{\beta \alpha}:=\operatorname{span}_{\mathrm{C}}\left\{\Pi_{\beta \alpha}(H): H \in \mathcal{G}_{\text {break }}\right\} \subseteq \mathcal{W}_{\beta \alpha} .
    $$
    a \textit{breaker-generated subspace}. Define \textit{the first-factor support} of the breakers on the block $(\alpha, \beta)$, which we denote by $S_{\beta \alpha}$, as the smallest subspace $S \subseteq \mathbf{T}_{\beta \alpha}$ such that
    $$
    \mathcal{B}_{\beta \alpha} \subseteq S \otimes \mathbf{M}_{\beta \alpha}
    $$
    Such a smallest subspace exists, and is given explicitly by the contractions of the breakers against the second factor,
    $$
    S_{\beta\alpha}=\operatorname{span}_\C\left\{\left(\operatorname{id}\otimes\psi\right)b\;:\;b\in\mathcal B_{\beta\alpha},\ \psi\in\mathbf M_{\beta\alpha}^*\right\},
    $$
    where $\operatorname{id}\otimes\psi:\mathbf T_{\beta\alpha}\otimes\mathbf M_{\beta\alpha}\to\mathbf T_{\beta\alpha}$ sends $t\otimes m\mapsto\psi(m)\,t$. Indeed, if $\mathcal B_{\beta\alpha}\subseteq S\otimes\mathbf M_{\beta\alpha}$ then every such contraction lies in $S$; conversely, writing $b=\sum_i t_i\otimes m_i$ with the $m_i$ linearly independent and choosing $\psi_i$ dual to them recovers each $t_i$, so $b$ lies in the span above tensored with $\mathbf M_{\beta\alpha}$.
\end{smdefinition}

\begin{smdefinition}\label{def:first_factor}
    We say that \textit{the full first-factor span condition} holds on the edge ${\alpha, \beta}$ if for first-factor support it holds
    \[
        S_{\beta \alpha}=\mathbf{T}_{\beta \alpha}=\operatorname{Hom}_{\mathrm{C}}\left(V_\alpha, V_\beta\right)
    \]
\end{smdefinition}

The point of the first-factor support is to measure which directions in the irrep-intertwining space $\mathbf{T}_{\beta\alpha}=\operatorname{Hom}_{\mathrm{C}}\left(V_\alpha, V_\beta\right)$ are actually supplied by the breakers, 
independently of how a given block element is written as a sum of simple tensors. Indeed, an element of $\mathcal{W}_{\beta \alpha} \cong \mathbf{T}_{\beta \alpha} \otimes \mathbf{M}_{\beta \alpha}$ 
may admit many different tensor decompositions, so the "first factors appearing in a decomposition" are not intrinsically defined. The subspace $S_{\beta \alpha}$ avoids this ambiguity by collecting 
all first-factor directions that can be extracted from breaker projections by contracting the multiplicity factor. 
This is exactly the quantity needed in the proof, because the invariant algebra acts only on the multiplicity factor $\mathbf{M}_{\beta \alpha}$, so once the breakers provide a subspace $S_{\beta \alpha} \subseteq\mathbf{T}_{\beta \alpha}$, 
the adjoint action of the invariant algebra propagates it to $S_{\beta \alpha} \otimes \mathbf{M}_{\beta \alpha}$ (one non-trivial vector from $\mathbf{M}_{\beta \alpha}$ turns out to be enough as we will see). Thus the condition $S_{\beta \alpha}=\mathbf{T}_{\beta \alpha}$ 
is precisely what ensures that the entire block $\mathcal{W}_{\beta \alpha}$ can be generated.

Now we are ready to formulate the sufficient conditions. For our proof, we denote DLA as follows
\begin{equation}\label{eq:DLA_with_breakers}
    \mathcal{L}:=\operatorname{Lie}_{\R}\{\mathfrak{su}(\HH)^G \cup \mathcal{G}^{\text{br}}\} \subseteq \mathfrak{su}(\HH).
\end{equation}
\begin{smtheorem} \label{thm:main_real}
    If the graph $\Gamma_{\mathrm{co}}$ is connected and the Full First-Factor Span condition holds for every edge in $\Gamma_{\mathrm{co}}$, then:
    \begin{equation*}
        \mathcal{L} = \mathfrak{su}(\HH).
    \end{equation*}
\end{smtheorem}
\begin{smremark}
    For abelian groups all irreducible representations are 1-dimensional, so $\mathbf{T}_{\beta\alpha} \cong \C$, so this condition simply means the projection $\Pi_{\beta \alpha}$ of $\mathcal{G}^{\text{br}}$ is non-zero
    for every pair $\{\alpha, \beta\}$.
\end{smremark}

\begin{smremark}
    The practical way to check the full first-factor span condition in the theorem is the following:
    \begin{enumerate}
    \item Use the projector formula~\eqref{eq:projector_formula} to get the isotypic blocks:
        $
        \Pi_{\beta \alpha}(H)=P_\beta H P_\alpha ;
        $
    \item Choose basis associated with the Schur-transform
        $
        \HH_\alpha \cong V_\alpha \otimes M_\alpha, \quad \HH_\beta \cong V_\beta \otimes M_\beta ;
        $
    \item Rewrite
        $
        \Pi_{\beta \alpha}(H) \in \operatorname{Hom}\left(\HH_\alpha, \HH_\beta\right) \cong \operatorname{Hom}\left(V_\alpha, V_\beta\right) \otimes \operatorname{Hom}\left(M_\alpha, M_\beta\right) ;
        $
    \end{enumerate}
    This gives us an explicit decomposition of $\Pi_{\beta \alpha}(H)$ into simple tensors, so we can read off the first factor and check if they span the entire $\mathbf{T}_{\beta \alpha}=\operatorname{Hom}_\C(V_\alpha, V_\beta)$.
    We will demonstrate this procedure on the example of $S_3$-symmetric system in section \ref{sec:H_br_applications}.
\end{smremark}

Now we proceed with the detailed proof. Schematic illustration of objects, used in the proof, can be seen in pic~\ref{fig:proof_schematic}.

\begin{figure}[htbp]
    \centering
    \IfFileExists{figs/H_br.pdf}{%
        \includegraphics[width=0.7\textwidth]{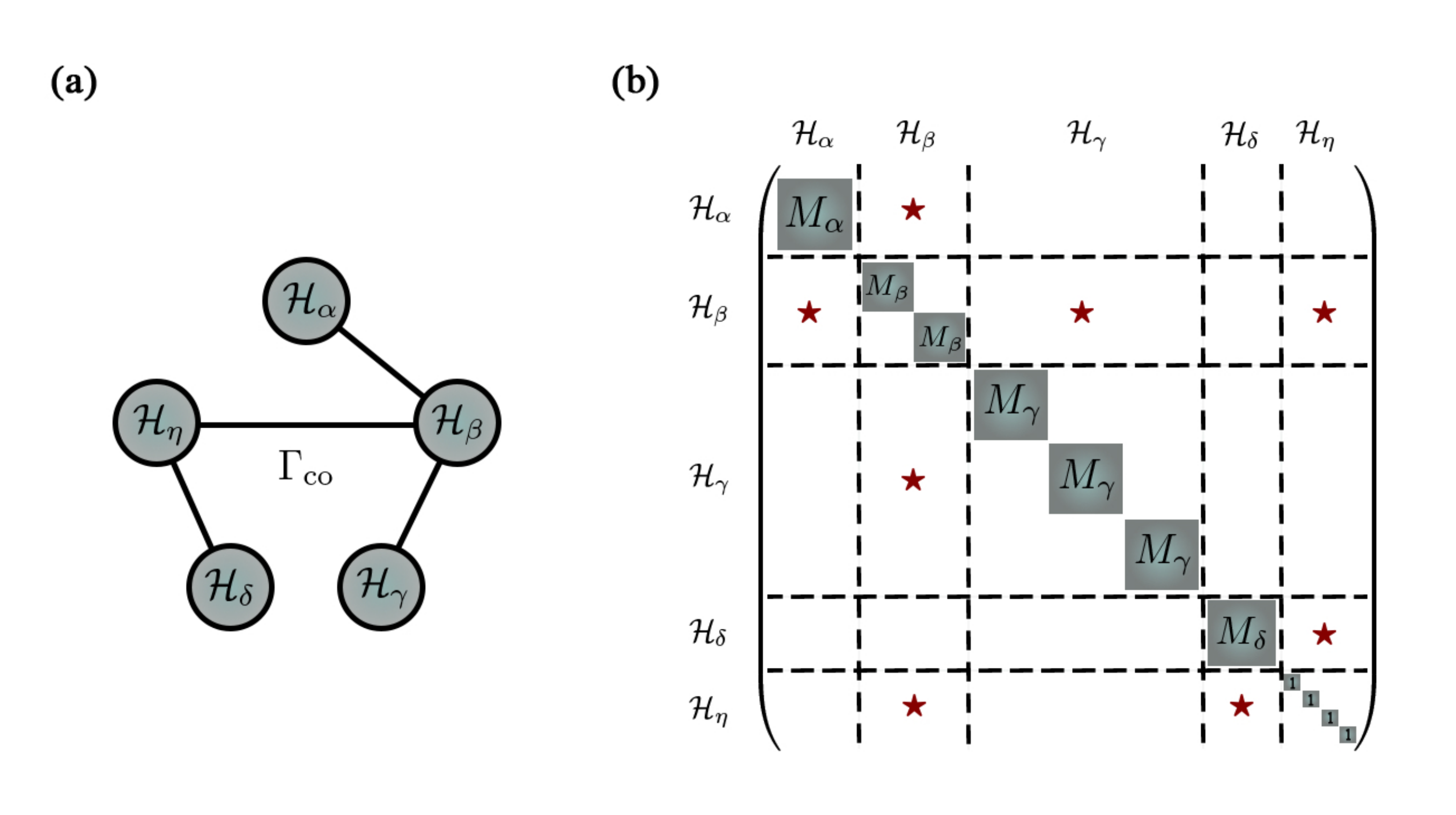}%
    }{%
        \fbox{\rule{0pt}{2in}\rule{0.7\textwidth}{0pt}}
    }
    \caption{Schematic illustration of objects, used in the proof of Theorem~\ref{thm:main_real}. (a) The coupling graph $\Gamma_{\mathrm{co}}$ from Definition~\ref{def:coupling_graph}. 
    (b) Complexified Invariant algebra $\mathfrak{su}(\mathcal{H})^G$, meaning $\mathfrak{l}_{\C}$, is shown in blue blocks on diagonal (see Eq.~\eqref{eq:compl_inv_alg}); 
    red stars represent blocks $\operatorname{Hom}(\mathcal{H}_{\alpha}, \mathcal{H}_{\beta})$, where $\mathcal{G}^{\text{br}}$ has non-zero support, 
    corresponding to the coupling graph $\Gamma_{\mathrm{co}}$ on the left.}
    \label{fig:proof_schematic}
\end{figure}

\begin{proof}

\subsubsection*{Step 1: Complexification}

We define the complexification of the DLA in the following way:
\begin{equation}\label{eq:complexified_DLA_with_breakers}
    \mathcal{L}_{\mathbb{C}}=\mathcal{L} \otimes_{\mathbb{R}} \mathbb{C}.
\end{equation}
From now on, we will work over a complex field -- all representation theory and Lie algebraic arguments will be done in the complexified setting.

\subsubsection*{Step 2: Invariant Lie Algebra Action}

Since commutators with scalar operators are trivial, only the non-central part of the invariant algebra $\mathfrak{l}_\C$ contributes to the adjoint action on the blocks. 
This effective part is precisely the so-called derived algebra 
$\mathfrak{l}'_\C$, i. e.
\[
    \mathfrak{l}'_\C = [\mathfrak{l}_\C, \mathfrak{l}_\C] \cong \bigoplus_{\alpha \in \mathcal{I}} \mathrm{id}_{V_\alpha} \otimes \mathfrak{sl}\left(M_\alpha\right).
\]
Consider the action of $\mathfrak{l}'_\C$ on a specific block $\mathcal{W}_{\beta\alpha} \cong \mathbf{T}_{\beta\alpha} \otimes \mathbf{M}_{\beta\alpha}$.
Let $L \in \mathfrak{l}'_\C$. We can write $L = \sum_\gamma \id_{V_\gamma} \otimes X_\gamma$, where $X_\gamma \in \mathfrak{sl}(M_\gamma)$.
For an element $T \otimes M \in \mathcal{W}_{\beta\alpha}$ (where $T: V_\alpha \to V_\beta$ and $M: M_\alpha \rightarrow M_\beta$), the commutator is:
\begin{align*}
    [L, T \otimes M] &= (\id_{V_\beta} \otimes X_\beta)(T \otimes M) - (T \otimes M)(\id_{V_\alpha} \otimes X_\alpha) \nonumber \\
    &= T \otimes (X_\beta M - M X_\alpha).
\end{align*}
Thus, the invariant algebra acts trivially on the first factor $\mathbf{T}_{\beta\alpha}$ and acts via the representation $(X_\beta, X_\alpha) \cdot M := X_\beta M - M X_\alpha$
on the second factor $\mathbf{M}_{\beta\alpha}$, where
$(X_\beta, X_\alpha) \in \mathfrak{sl}\left(M_\beta\right) \oplus \mathfrak{sl}\left(M_\alpha\right)$.
Notice, that $\left[L, \mathcal{W}_{\beta \alpha}\right] \subseteq \mathcal{W}_{\beta \alpha}$.

\subsubsection*{Step 3: Irreducibility of Blocks.}

\begin{smlemma}
    Let $\mathfrak{g}_1, \mathfrak{g}_2$ be Lie algebras acting irreducibly on finite-dimensional vector spaces $U, \mathcal{V}$ respectively. 
    Then the external tensor product module $U \otimes \mathcal{V}$ is irreducible under the action of $\mathfrak{g}_1 \oplus \mathfrak{g}_2$.
\end{smlemma}

\begin{proof}
    Let $S \subseteq U \otimes \mathcal{V}$ be a non-zero submodule. We must show $S = U \otimes \mathcal{V}$.
    Choose a non-zero element $x \in S$. We can write $x = \sum_{i=1}^k u_i \otimes v_i$, where $\{v_i\}$ are linearly independent in $\mathcal{V}$. We choose $x$ such that the rank $k$ is minimal.
    
    By the Density Theorem, the action of the enveloping algebra $\mathcal{U}(\mathfrak{g}_1)$ generates all linear operators $\End(U)$. 
    Thus, there exists an operator $L_1 \in \mathcal{U}(\mathfrak{g}_1)$ such that $L_1 u_1 = u'$ (any arbitrary vector) and $L_1 u_i = 0$ for $i > 1$.
    Applying $(L_1, 0)$ to $x$:
    \[ (L_1 \otimes \I) x = \sum (L_1 u_i) \otimes v_i = u' \otimes v_1 \in S. \]
    Since $u'$ is arbitrary, $U \otimes \{v_1\} \subseteq S$.
    
    Similarly, using the density of $\mathcal{U}(\mathfrak{g}_2)$ on $\mathcal{V}$, we can apply $(\I \otimes L_2)$ to elements of $U \otimes \{v_1\}$ to generate $U \otimes \mathcal{V}$. Thus $S = U \otimes \mathcal{V}$.
\end{proof}

In our case, the derived invariant algebra $\mathfrak{l}'_\C$ contains the summand $\mathfrak{sl}(M_{\beta}) \oplus \mathfrak{sl}(M_{\alpha})$. 
The space $\mathbf{M}_{\beta\alpha} = \Hom(M_{\alpha}, M_{\beta}) \cong M_{\beta} \otimes (M_{\alpha})^*$.
Since the fundamental representation $M_{\alpha}$ and its dual $(M_{\alpha})^*$ are irreducible representations of $\mathfrak{sl}(M_{\alpha})$, 
the lemma implies $\mathbf{M}_{\beta\alpha}$ is irreducible.
Thus, we have the following

\begin{smlemma}\label{lem:irreducibility}
Let $\alpha,\beta\in\mathcal I$ with $\alpha\neq \beta$. Then
$
\mathbf M_{\beta\alpha}=\Hom_\C(M_\alpha,M_\beta)
$
is an irreducible module under the action of
$
\mathfrak{sl}\left(M_\beta\right) \oplus \mathfrak{sl}\left(M_\alpha\right)
$
defined by
$
(X_\beta,X_\alpha)\cdot M := X_\beta M - M X_\alpha .
$
\end{smlemma}

\subsubsection*{Step 4: Generation of off-diagonal blocks.}

\begin{smlemma}\label{prop:block_generation}
Let $\alpha,\beta\in\mathcal I$ with $\alpha\neq\beta$. Consider the
breaker-generated subspace $\mathcal B_{\beta\alpha}$ and first-factor support $S_{\beta\alpha}$, defined in Def.~\ref{def:first_factor}. Then the Lie algebra generated by
$
\mathfrak l_\C' \cup \mathcal B_{\beta\alpha}
$
contains
$
S_{\beta\alpha}\otimes \mathbf M_{\beta\alpha}.
$
In particular, if the full first-factor span condition holds on the edge $\{\alpha,\beta\}$, i.e.
$
S_{\beta\alpha}=\mathbf T_{\beta\alpha},
$
then
\[
\mathcal W_{\beta\alpha}
=
\mathbf T_{\beta\alpha}\otimes \mathbf M_{\beta\alpha}
\subseteq
\Lie_{\R}\{\mathfrak l_\C' \cup \mathcal B_{\beta\alpha}\}
\subseteq \mathcal L_\C.
\]
\end{smlemma}

\begin{proof}
Since $\alpha\neq\beta$, Lemma~\ref{lem:irreducibility} implies that $\mathbf M_{\beta\alpha}$ is an irreducible
$\mathfrak{sl}(M_\beta)\oplus \mathfrak{sl}(M_\alpha)$-module.
By the Density Theorem, applied as in the proof above, the enveloping algebra $\mathcal U(\mathfrak l_\C')$ therefore acts on $\mathbf M_{\beta\alpha}$ 
as the full algebra $\End_\C(\mathbf M_{\beta\alpha})$. In particular every rank-one operator
$
w\psi:m\mapsto\psi(m)\,w,
$
with $w\in\mathbf M_{\beta\alpha}$ and $\psi\in\mathbf M_{\beta\alpha}^*$, is realised by some element of $\mathcal U(\mathfrak l_\C')$.
The Lie algebra generated by $\mathfrak l_\C'\cup\mathcal B_{\beta\alpha}$ is closed under iterated brackets with elements of $\mathfrak l_\C'$, 
and by Step 2 the bracket with $L=\sum_\gamma\id_{V_\gamma}\otimes X_\gamma$ acts on $\mathcal W_{\beta\alpha}\cong\mathbf T_{\beta\alpha}\otimes\mathbf M_{\beta\alpha}$ as $T\otimes M\mapsto T\otimes(X_\beta M-MX_\alpha)$, 
that is, trivially on the first factor. Iterating, it therefore contains $(\operatorname{id}\otimes a)b$ for every $a\in\End_\C(\mathbf M_{\beta\alpha})$ 
and every $b\in\mathcal B_{\beta\alpha}$. Taking $a=w\psi$ and writing $b=\sum_i t_i\otimes m_i$,
\[
\left(\operatorname{id}\otimes w\psi\right)b
=\sum_i t_i\otimes\psi(m_i)\,w
=\left(\left(\operatorname{id}\otimes\psi\right)b\right)\otimes w .
\]
Since the elements $(\operatorname{id}\otimes\psi)b$ span $S_{\beta\alpha}$ by Definition~\ref{def:breaker_generated_subspace} and $w\in\mathbf M_{\beta\alpha}$ is arbitrary, the Lie algebra generated by
$
\mathfrak l_\C' \cup \mathcal B_{\beta\alpha}
$
contains
$
S_{\beta\alpha}\otimes \mathbf M_{\beta\alpha}.
$
If the full first-factor span condition holds, then
$
S_{\beta\alpha}=\mathbf T_{\beta\alpha},
$
and therefore
\[
\mathcal W_{\beta\alpha}
=
\mathbf T_{\beta\alpha}\otimes \mathbf M_{\beta\alpha}
\subseteq
\Lie_{\R}\{\mathfrak l_\C' \cup \mathcal B_{\beta\alpha}\}.
\]
This proves the claim.
\end{proof}

\subsubsection*{Step 5: Propagation along the coupling graph $\Gamma_{\mathrm{co}}$.}

\begin{smlemma}
    $[\mathcal{W}_{\beta\gamma}, \mathcal{W}_{\gamma\alpha}] = \mathcal{W}_{\beta\alpha}$ for pairwise distinct indices.
\end{smlemma}
\begin{proof}
Let
\[
A\in \mathcal W_{\beta\gamma}=\Hom_\C(\HH_\gamma,\HH_\beta),
\qquad
B\in \mathcal W_{\gamma\alpha}=\Hom_\C(\HH_\alpha,\HH_\gamma).
\]
Then
$
AB\in \Hom_\C(\HH_\alpha,\HH_\beta)=\mathcal W_{\beta\alpha}.
$
Since \(\alpha,\beta,\gamma\) are pairwise distinct, the reverse product vanishes:
$
BA=0,
$
because \(A\) is zero outside \(\HH_\gamma\) and takes values in \(\HH_\beta\), while \(B\) is zero outside \(\HH_\alpha\), 
and \(\HH_\alpha\cap \HH_\beta=\{0\}\).
Hence
$
[A,B]=AB\in \mathcal W_{\beta\alpha},
$
so
$
[\mathcal W_{\beta\gamma},\mathcal W_{\gamma\alpha}]
\subseteq
\mathcal W_{\beta\alpha}.
$

For the reverse inclusion, it suffices to show that every rank-one map in
\(\mathcal W_{\beta\alpha}\) is such a commutator. Let
$
R\in \mathcal W_{\beta\alpha}=\Hom_\C(\HH_\alpha,\HH_\beta)
$
be rank one. Then there exist \(v\in \HH_\beta\) and \(u\in \HH_\alpha^*\) such that
$
R=v\otimes u.
$
Choose a nonzero vector \(w\in \HH_\gamma\), and let \(\eta\in \HH_\gamma^*\) satisfy
\(\eta(w)=1\). Define
\[
A:=v\otimes \eta \in \Hom_\C(\HH_\gamma,\HH_\beta)=\mathcal W_{\beta\gamma},
\qquad
B:=w\otimes u \in \Hom_\C(\HH_\alpha,\HH_\gamma)=\mathcal W_{\gamma\alpha}.
\]
Then
$
AB=(v\otimes \eta)(w\otimes u)=\eta(w)\,v\otimes u=R.
$
As above, \(BA=0\), and therefore
$
[A,B]=AB=R.
$
Since rank-one maps span \(\mathcal W_{\beta\alpha}\), we conclude that
$
\mathcal W_{\beta\alpha}\subseteq
[\mathcal W_{\beta\gamma},\mathcal W_{\gamma\alpha}].
$
Thus
$
[\mathcal W_{\beta\gamma},\mathcal W_{\gamma\alpha}]=\mathcal W_{\beta\alpha}.
$
\end{proof}

Since the graph $\Gamma_{\mathrm{co}}$ is connected, and we can generate the blocks corresponding to edges (by Step 4),
we can generate the block $\mathcal{W}_{\beta\alpha}$ for any pair $(\alpha, \beta)$ by taking nested commutators along a path connecting them. 
Thus, all off-diagonal blocks are in $\mathcal{L}_\C$.

\subsubsection*{Step 6: Generation of diagonal traceless blocks.}

\begin{smlemma}\label{lem:diag_blocks}
For any distinct \(\alpha,\beta\in\mathcal I\),
\[
[\mathcal W_{\alpha\beta},\mathcal W_{\beta\alpha}]
\subseteq
\mathcal W_{\alpha\alpha}\oplus \mathcal W_{\beta\beta}.
\]
Moreover, if \(\mathcal W_{\alpha\beta},\mathcal W_{\beta\alpha}\subseteq \mathcal L_\C\), then
their commutators generate the traceless part of
\(\mathcal W_{\alpha\alpha}\oplus \mathcal W_{\beta\beta}\).
\end{smlemma}

\begin{proof}
The inclusion is immediate from block multiplication:
\[
\mathcal W_{\alpha\beta}\mathcal W_{\beta\alpha}\subseteq \mathcal W_{\alpha\alpha},
\qquad
\mathcal W_{\beta\alpha}\mathcal W_{\alpha\beta}\subseteq \mathcal W_{\beta\beta}.
\]
Choose bases of \(\HH_\alpha\) and \(\HH_\beta\), and let
$
E^{\alpha\beta}_{ip}\in \mathcal W_{\alpha\beta},
\qquad
E^{\beta\alpha}_{qj}\in \mathcal W_{\beta\alpha}
$
be the corresponding matrix units. Then
$
[E^{\alpha\beta}_{ip},E^{\beta\alpha}_{pj}]
=
E^{\alpha\alpha}_{ij}
\qquad (i\neq j),
$
and
$
[E^{\alpha\beta}_{ip},E^{\beta\alpha}_{pi}]
=
E^{\alpha\alpha}_{ii}-E^{\beta\beta}_{pp}.
$
Thus the commutators generate all off-diagonal matrix units inside
\(\mathcal W_{\alpha\alpha}\) and \(\mathcal W_{\beta\beta}\), together with all diagonal differences. Hence they generate the traceless part of
\(\mathcal W_{\alpha\alpha}\oplus \mathcal W_{\beta\beta}\).
\end{proof}

\subsubsection*{Step 7: Descend to the Real field}

Combining the previous steps, \(\mathcal L_\C\) contains all off-diagonal blocks and the full diagonal traceless part. Hence
$
\mathcal L_\C=\mathfrak{sl}(\HH).
$ 
Since $\mathcal L_\C = \mathbb{C}\otimes_\R \mathcal L$, we have
$
\dim_\R \mathcal L=\dim_\C \mathfrak{sl}(\HH)=d^2-1.
$
But \(\mathcal L\subseteq \su(\HH)\), because it is generated by traceless skew-Hermitian operators, and
$
\dim_\R \su(\HH)=d^2-1.
$
Hence
$
\mathcal L=\su(\HH).
$
This concludes the proof of the theorem.
\end{proof}

\subsection{Applications}\label{sec:H_br_applications}
\label{sec:applications}

\subsubsection*{Example 1: $S_3$ symmetry.}

Consider the symmetric group $S_3$ acting on three qubit system $(\mathbb{C}^2)^{\otimes 3}$ by permuting them, i. e.
\[
U(\sigma)|\psi_1\rangle|\psi_2\rangle|\psi_3\rangle=|\psi_{\sigma^{-1}(1)}\rangle|\psi_{\sigma^{-1}(2)}\rangle|\psi_{\sigma^{-1}(3)}\rangle, \quad \sigma \in S_3.
\]
The Hilbert space decomposes into irreps of $S_3$ as:
\[
    \HH \cong V_{(3)} \otimes M_{(3)} \oplus V_{(2,1)} \otimes M_{(2,1)}
\]
Corresponding dimensions are
\[
    \begin{aligned}
    & \operatorname{dim} V_{(3)}=1, \quad \operatorname{dim} V_{(2,1)}=2, \\
    & \operatorname{dim} M_{(3)}=4, \quad \operatorname{dim} M_{(2,1)}=2,
    \end{aligned}
\]
The coupling graph $\Gamma_{\mathrm{co}}$ has two vertices corresponding to $(2,1)$ and $(3)$ irreps.
The relevant first-factor space is
\[
    \mathbf{T}_{(2,1),(3)}=\operatorname{Hom}\left(V_{(3)}, V_{(2,1)}\right) \cong V_{(2,1)}
\]
so $\operatorname{dim} \mathbf{T}_{(2,1),(3)}=2$. Which means that we only have to check the theorem conditions 
for one edge $(2,1)\to(3)$. Thus, a single breaker $iH_{\mathrm{br}}$ can work if projection
\[
    \Pi_{(2,1),(3)}\left(iH_{\mathrm{br}}\right) \in \mathbf{T}_{(2,1),(3)} \otimes \mathbf{M}_{(2,1),(3)}
\]
has first-factor support equal to all of $\mathbf{T}_{(2,1),(3)}$.
Let's prove that $iH^{\mathrm{br}}=iZ_1+iX_2$ is enough to achieve universality. From now on, we will work with Hamiltonians instead of skew-Hermitian generators, 
so we will omit the factor of $i$ for brevity. 

We first compute the block-projection explicitly
\[
    \Pi_{(2,1),(3)}(Z_1+X_2)=P_{(2,1)} (Z_1+X_2) P_{(3)}.
\]
In this example, it is much easier, because we have to compute only $P_{(3)}$, for example, and then
$P_{(2,1)}=\mathrm{1}-P_{(3)}$. We can compute $P_{(3)}$ by using Eq.~\eqref{eq:projector_formula}
\[
    P_{(3)}=\frac{1}{6}(I+U(12)+U(13)+U(23)+U(123)+U(132))
\]
Basis in $\HH_{(3)}$ isotypic component is given by the following vectors:
\[
    \begin{aligned}
    & s_0=|000\rangle, \quad s_1=\frac{|100\rangle+|010\rangle+|001\rangle}{\sqrt{3}}, \\
    & s_2=\frac{|110\rangle+|101\rangle+|011\rangle}{\sqrt{3}}, \quad s_3=|111\rangle .
    \end{aligned}
\]
As for $\HH_{(2,1)}$, we choose the following basis:
\[
    \begin{aligned}
    &u_1=\frac{2|100\rangle-|010\rangle-|001\rangle}{\sqrt{6}}, \quad u_2=\frac{|010\rangle-|001\rangle}{\sqrt{2}},\\
    &v_1=\frac{2|011\rangle-|101\rangle-|110\rangle}{\sqrt{6}}, \quad v_2=\frac{|101\rangle-|110\rangle}{\sqrt{2}},
    \end{aligned}
\]
Knowing the action of $Z_1$, $X_2$ and $P_{(3)}$ on the computational basis, we can compute the matrix representation of $\Pi_{(2,1),(3)}(Z_1+X_2)$ in the bases defined above. We get
\[
    \left[\Pi_{(2,1),(3)}(Z_1+X_2)\right]_{\left(u_1, v_1, u_2, v_2\right) \leftarrow\left(s_0, s_1, s_2, s_3\right)}=\left(\begin{array}{cccc}
    -\frac{\sqrt{6}}{6} & -\frac{2 \sqrt{2}}{3} & \frac{\sqrt{2}}{6} & 0 \\
    0 & \frac{\sqrt{2}}{6} & \frac{2 \sqrt{2}}{3} & -\frac{\sqrt{6}}{6} \\
    \frac{\sqrt{2}}{2} & 0 & -\frac{\sqrt{6}}{6} & 0 \\
    0 & -\frac{\sqrt{6}}{6} & 0 & \frac{\sqrt{2}}{2}
    \end{array}\right).
\]
We choose a basis $\left\{e_1, e_2\right\}$ of $V_{(2,1)}$ and a basis $\left\{m_1, m_2\right\}$ of the multiplicity space $M_{(2,1)}$, and identify 
\[
    u_1=e_1 \otimes m_1, \quad u_2=e_2 \otimes m_1, \quad v_1=e_1 \otimes m_2, \quad v_2=e_2 \otimes m_2,
\]
from which it becomes apparent that
\[
    \Pi_{(2,1),(3)}(Z_1+X_2)=e_1 \otimes F_1+e_2 \otimes F_2,
\]
where $F_1, F_2 \in \mathbf{M}_{(2,1),(3)}=\operatorname{Hom}\left(M_{(3)}, M_{(2,1)}\right)$ such that
\[
    \begin{gathered}
    F_1=\left(\begin{array}{cccc}
    -\frac{\sqrt{6}}{6} & -\frac{2 \sqrt{2}}{3} & \frac{\sqrt{2}}{6} & 0 \\
    0 & \frac{\sqrt{2}}{6} & \frac{2 \sqrt{2}}{3} & -\frac{\sqrt{6}}{6}
    \end{array}\right), \\
    F_2=\left(\begin{array}{cccc}
    \frac{\sqrt{2}}{2} & 0 & -\frac{\sqrt{6}}{6} & 0 \\
    0 & -\frac{\sqrt{6}}{6} & 0 & \frac{\sqrt{2}}{2}
    \end{array}\right) .
    \end{gathered}
\]  
Notice that $e_1$ and $e_2$ are linearly independent by construction. $F_1$ and $F_2$ are linearly independent, so breaker satisfies the full first-factor span condition, i. e.
\[
    S_{(2,1),(3)}=\operatorname{span}\left\{e_1, e_2\right\}=\mathbf{T}_{(2,1),(3)}.
\]
We conclude that
\[
    \Lie_{\R}\left\{\mathfrak{su}\left((\mathbb{C}^2)^{\otimes 3}\right)^{S_3} \cup\{iZ_1+iX_2\}\right\}=\mathfrak{su}\left((\mathbb{C}^2)^{\otimes 3}\right).
\]

\begin{smremark}
    The proof would be simpler if we split the $Z_1+X_2$ Hamiltonian into two breakers, i. e. consider the breaking set $\mathcal{G}^{\text {br }}=\left\{i Z_1, i X_2\right\}$. 
    Indeed, with two breakers, one only needs independence of the first-factor directions; however, with one combined breaker, one also needs independence on the multiplicity side to check the full first-factor span condition.
    Without linear independence of $F_1$ and $F_2$, the breaker would be $\Pi_{(2,1),(3)}(Z_1+X_2)=(e_1+ \lambda e_2) \otimes F_1$,
    so the first-factor support would be only one-dimensional, despite $e_1$ and $e_2$ being linearly independent.
\end{smremark}

\end{appendix}

\end{document}